\pdfoutput=1

\PassOptionsToPackage{pagebackref,draft=false}{hyperref}
\PassOptionsToPackage{dvipsnames}{xcolor}

\documentclass[letterpaper,12pt,titlepage,oneside,final]{book}
 
\newcommand{\href}[1]{#1} 

\usepackage{ifthen}
\newboolean{PrintVersion}
\setboolean{PrintVersion}{false} 

\usepackage[pdftex]{graphicx} 
\usepackage[utf8]{inputenc}

\usepackage{adjustbox}
\usepackage{amsmath}			
\usepackage{amssymb}			
\usepackage{amsthm}				
\usepackage{array}
\usepackage{bbm}
\usepackage{bm}				
\usepackage{dutchcal}
\usepackage{calc}
\usepackage{caption}
\usepackage{centernot}			
\usepackage[shortlabels]{enumitem}

	\setenumerate{itemsep=2pt,topsep=2pt}
	\setenumerate[1]{label=(\roman*)}

	\setitemize{itemsep=2pt,topsep=2pt}

\usepackage{etoolbox}			
\usepackage{ifthen}					
\usepackage{mathtools}
\usepackage{scalerel}
\usepackage{subcaption}			
\usepackage{tikz-cd} 				
\usepackage{tikzit}				
	\usetikzlibrary{decorations.pathreplacing,decorations.pathmorphing,arrows,shapes.geometric,positioning,calc}

\tikzstyle{morphism}=[fill=white, draw=black, shape=rectangle]
\tikzstyle{medium box}=[fill=white, draw=black, shape=rectangle, minimum width=0.8cm, minimum height=0.9cm]
\tikzstyle{large morphism}=[fill=white, draw=black, shape=rectangle, minimum width=1.7cm, minimum height=1cm]
\tikzstyle{bn}=[fill=black, draw=black, shape=circle, inner sep=1.5pt]
\tikzstyle{state}=[fill=white, draw=black, regular polygon, regular polygon sides=3, minimum width=0.8cm, shape border rotate=180, inner sep=0pt]
\tikzstyle{medium state}=[fill=white, draw=black, regular polygon, regular polygon sides=3, minimum width=1.3cm, inner sep=0pt, shape border rotate=180]
\tikzstyle{large state}=[fill=white, draw=black, regular polygon, regular polygon sides=3, minimum width=2.2cm, shape border rotate=180, inner sep=0pt]
\tikzstyle{wn}=[fill=white, draw=black, shape=circle, inner sep=1.5pt]
\tikzstyle{blue morphism}=[fill=white, draw={rgb,255: red,15; green,0; blue,150}, shape=rectangle, text={rgb,255: red,15; green,0; blue,150}, tikzit category=blue]
\tikzstyle{blue state}=[fill=white, draw={rgb,255: red,15; green,0; blue,150}, shape=circle, regular polygon, regular polygon sides=3, minimum width=0.8cm, shape border rotate=180, inner sep=0pt, text={rgb,255: red,15; green,0; blue,150}, tikzit category=blue]
\tikzstyle{blue node}=[fill={rgb,255: red,15; green,0; blue,150}, draw={rgb,255: red,15; green,0; blue,150}, shape=circle, tikzit category=blue, inner sep=1.5pt]
\tikzstyle{blue}=[text={rgb,255: red,15; green,0; blue,150}, tikzit draw={rgb,255: red,191; green,191; blue,191}, tikzit category=blue, tikzit fill=white, inner sep=0mm]
\tikzstyle{red node}=[fill={rgb,255: red,150; green,0; blue,2}, draw={rgb,255: red,150; green,0; blue,2}, shape=circle, inner sep=1.5pt]
\tikzstyle{Purple node}=[fill={rgb,255: red,150; green,0; blue,150}, draw={rgb,255: red,150; green,0; blue,150}, shape=circle, inner sep=1.5pt]
\tikzstyle{red}=[text={rgb,255: red,150; green,0; blue,2}, inner sep=0mm, tikzit fill=white, tikzit draw={rgb,255: red,191; green,191; blue,191}]
\tikzstyle{purple}=[text={rgb,255: red,150; green,0; blue,150}, inner sep=0mm, tikzit fill=white, tikzit draw={rgb,255: red,191; green,191; blue,191}]
\tikzstyle{wide state}=[fill=white, draw=black, shape=isosceles triangle, minimum width=0.8cm, shape border rotate=270, inner sep=1.4pt, minimum height=0.5cm, isosceles triangle apex angle=80]
\tikzstyle{white morphism}=[fill=white, draw=white, shape=rectangle, tikzit draw={rgb,255: red,139; green,139; blue,139}]
\tikzstyle{node}=[fill=none, draw=black, shape=circle, inner sep=1.5pt, tikzit fill=white]
\tikzstyle{node_fix}=[fill=none, draw=black, shape=circle, tikzit fill=white, minimum size=15pt, inner sep=0.5pt]
\tikzstyle{node_med}=[fill=none, draw=black, shape=circle, tikzit fill=white, minimum size=17pt, inner sep=1pt]
\tikzstyle{triangle_fix}=[regular polygon, regular polygon sides=3, fill=none, draw=black, minimum size=25pt, inner sep=0.5pt, tikzit fill=white]
\tikzstyle{green point}=[fill={rgb,255: red,33; green,153; blue,51}, draw=black, shape=circle, inner sep=1.5pt]
\tikzstyle{small node}=[fill=black, draw=black, shape=circle, inner sep=1pt]
\tikzstyle{grey}=[text={rgb,255: red,64; green,64; blue,64}, inner sep=0mm, shape=circle, tikzit fill=white, tikzit draw={rgb,255: red,64; green,64; blue,64}]
\tikzstyle{name_wire}=[fill=white, draw=black, shape=circle, inner sep=1pt, minimum size=9pt, font={\scriptsize}]

\tikzstyle{arrow}=[->]
\tikzstyle{dashed box}=[-, dashed]
\tikzstyle{blue arrow}=[-, draw={rgb,255: red,15; green,0; blue,150}, tikzit category=blue]
\tikzstyle{mapsto}=[{|->}]
\tikzstyle{double wire}=[-, double]
\tikzstyle{curly brace}=[decorate,decoration={brace,amplitude=10pt}]
\tikzstyle{curly brace 2}=[decorate,decoration={brace,amplitude=5pt}]
\tikzstyle{fill_ar}=[->, >=stealth']
\tikzstyle{fill_ar2}=[<->, >=stealth']
\tikzstyle{white}=[-, draw=white, tikzit draw={rgb,255: red,141; green,141; blue,141}]
\tikzstyle{white_ar}=[tikzit draw={rgb,255: red,141; green,141; blue,141}, draw=white, ->, >=stealth']
\tikzstyle{ellipse}=[-Ellipse]
\tikzstyle{segment}=[{|-|}]
\tikzstyle{thick_ar}=[->, thick]
\tikzstyle{orange edge}=[-, draw={rgb,255: red,215; green,143; blue,0}, thick]
\tikzstyle{purple edge}=[-, draw={rgb,255: red,125; green,0; blue,250}, thick]
\tikzstyle{thick edge}=[-, draw={rgb,255: red,145; green,145; blue,145}, fill={rgb,255: red,220; green,220; blue,220}, tikzit draw=black, tikzit fill={rgb,255: red,198; green,198; blue,198}]
\tikzstyle{thick grey}=[-, draw={rgb,255: red,145; green,145; blue,145}, fill={rgb,255: red,249; green,245; blue,183}, tikzit draw=black, tikzit fill={rgb,255: red,249; green,245; blue,183}]
\tikzstyle{protected}=[-, preaction={{ultra thick,white,draw}}]
\tikzstyle{blue fill}=[-, draw=none, fill={rgb,255: red,233; green,244; blue,255}, tikzit draw={rgb,255: red,0; green,125; blue,178}]
\tikzstyle{yellow fill}=[-, draw=none, fill={rgb,255: red,249; green,249; blue,214}, tikzit draw={rgb,255: red,225; green,218; blue,0}]
\tikzstyle{protected blue}=[-, preaction={{ultra thick,draw={rgb,255: red,233; green,244; blue,255}}}]
\tikzstyle{snake}=[decorate,decoration={snake}]
\tikzstyle{dense}=[densely dashed]
\tikzstyle{blue trans}=[-, draw=none, fill={rgb,255: red,211; green,233; blue,255}, tikzit draw={rgb,255: red,0; green,125; blue,178}, opacity=0.5]
\tikzstyle{white fill}=[-, fill=white]
\tikzstyle{behind}=[-, densely dashed, opacity=0.5, draw=black, very thin]	

\def\rotateclockwise#1{
  \newdimen\xrw
  \pgfextractx{\xrw}{#1}
  \newdimen\yrw
  \pgfextracty{\yrw}{#1}
  \pgfpoint{\yrw}{-\xrw}
}

\def\rotatecounterclockwise#1{
  \newdimen\xrcw
  \pgfextractx{\xrcw}{#1}
  \newdimen\yrcw
  \pgfextracty{\yrcw}{#1}
  \pgfpoint{-\yrcw}{\xrcw}
}

\def\outsidespacerpgfclockwise#1#2#3{
  \pgfpointscale{#3}{
    \rotateclockwise{
      \pgfpointnormalised{
        \pgfpointdiff{#1}{#2}}}}
}

\def\outsidespacerpgfcounterclockwise#1#2#3{
  \pgfpointscale{#3}{
    \rotatecounterclockwise{
      \pgfpointnormalised{
        \pgfpointdiff{#1}{#2}}}}
}

\def\outsidepgfclockwise#1#2#3{
  \pgfpointadd{#2}{\outsidespacerpgfclockwise{#1}{#2}{#3}}
}

\def\outsidepgfcounterclockwise#1#2#3{
  \pgfpointadd{#2}{\outsidespacerpgfcounterclockwise{#1}{#2}{#3}}
}

\def\outside#1#2#3{
  ($ (#2) ! #3 ! -90 : (#1) $)
}

\def\cornerpgf#1#2#3#4{
  \pgfextra{
    \pgfmathanglebetweenpoints{#2}{\outsidepgfcounterclockwise{#1}{#2}{#4}}
    \let\anglea\pgfmathresult
    \let\startangle\pgfmathresult

    \pgfmathanglebetweenpoints{#2}{\outsidepgfclockwise{#3}{#2}{#4}}
    \pgfmathparse{\pgfmathresult - \anglea}
    \pgfmathroundto{\pgfmathresult}
    \let\arcangle\pgfmathresult
    \ifthenelse{180=\arcangle \or 180<\arcangle}{
      \pgfmathparse{-360 + \arcangle}}{
      \pgfmathparse{\arcangle}}
    \let\deltaangle\pgfmathresult

    \newdimen\x
    \pgfextractx{\x}{\outsidepgfcounterclockwise{#1}{#2}{#4}}
    \newdimen\y
    \pgfextracty{\y}{\outsidepgfcounterclockwise{#1}{#2}{#4}}
  }
  -- (\x,\y) arc [start angle=\startangle, delta angle=\deltaangle, radius=#4]
}

\def\corner#1#2#3#4{
  \cornerpgf{\pgfpointanchor{#1}{center}}{\pgfpointanchor{#2}{center}}{\pgfpointanchor{#3}{center}}{#4}
}

\def\hedgem#1#2#3#4{
  
  \outside{#1}{#2}{#4}
  \pgfextra{
    \def\hgnodea{#1}
    \def\hgnodeb{#2}
  }
  foreach \c in {#3} {
    \corner{\hgnodea}{\hgnodeb}{\c}{#4}
    \pgfextra{
      \global\let\hgnodea\hgnodeb
      \global\let\hgnodeb\c
    }
  }
  \corner{\hgnodea}{\hgnodeb}{#1}{#4}
  \corner{\hgnodeb}{#1}{#2}{#4}
  -- cycle
}

\def\hedgeii#1#2#3{
  \hedgem{#1}{#2}{}{#3}
}

	\newcommand{\togo}[1]{}
\usepackage{xcolor}
\usepackage{xparse}				

\usepackage[pdftex]{hyperref} 
	\hypersetup{
	    plainpages=false,       
	    unicode=false,          
	    pdftoolbar=true,        
	    pdfmenubar=true,        
	    pdffitwindow=false,     
	    pdfstartview={FitH},    
	    pdftitle={Resource Theories as Quantale Modules},    
	    pdfauthor={Tomáš Gonda},    
		pdfsubject={PhD Thesis},  
		pdfkeywords={resource theory} {quantale} {lattice} {ordered algebraic structures} {quantum information theory} {quantale module} {monotone} {hypothesis testing} {process theory} {transition systems} {linear logic}, 
	    pdfnewwindow=true,      
	    colorlinks=true,        
	    linkcolor=blue,         
	    citecolor=magenta,        
	    filecolor=magenta,      
	    urlcolor=cyan           
	}

	\ifthenelse{\boolean{PrintVersion}}{   
		\hypersetup{	
			pagebackref=false
		    citecolor=black,%
		    filecolor=black,%
		    linkcolor=black,%
		    urlcolor=black}
	}{} 
\usepackage[noabbrev]{cleveref}

\let\origdoublepage\cleardoublepage
\newcommand{\clearemptydoublepage}{%
  \clearpage{\pagestyle{empty}\origdoublepage}}
\let\cleardoublepage\clearemptydoublepage

\numberwithin{equation}{chapter}

\newcommand{\hyphenationsetting}{%
	\emergencystretch=0pt	
	\tolerance=2000			
	\pretolerance=1000		
	\righthyphenmin=4		
	\lefthyphenmin=4			
}
\hyphenationsetting

\newlength{\myparskip}
\newlength{\scale}
	\newcommand*{\coloniff}{\mathrel{\vcentcolon\Longleftrightarrow}}

	\DeclarePairedDelimiter{\abs}{\lvert}{\rvert}
	\DeclarePairedDelimiter{\norm}{\lVert}{\rVert}
	\DeclarePairedDelimiterXPP{\pnorm}[2]{}{\lVert}{\rVert}{_{#1}}{#2}
	\makeatletter
		\let\oldabs\abs
		\def\abs{\@ifstar{\oldabs}{\oldabs*}}
		\let\oldnorm\norm
		\def\norm{\@ifstar{\oldnorm}{\oldnorm*}}
		\let\oldpnorm\pnorm
		\def\pnorm{\@ifstar{\oldpnorm}{\oldpnorm*}}
	\makeatother

	\providecommand{\given}{}			
	\newcommand{\SetSymbol}[1][]{%
		\nonscript\;\,#1\vert
		\allowbreak
		\nonscript\;\,
		\mathopen{}
	}
	\DeclarePairedDelimiterX{\Set}[1]{\{}{\}}{%
		\renewcommand{\given}{\SetSymbol[\delimsize]}
		\, #1 \,
	}
	\makeatletter
		\let\oldSet\Set
		\def\Set{\@ifstar{\oldSet}{\oldSet*}}
	\makeatother
	
	\DeclarePairedDelimiterX{\Family}[1]{(}{)}{%
		\renewcommand{\given}{\SetSymbol[\delimsize]}
		\, #1 \,
	}
	\makeatletter
		\let\oldFamily\Family
		\def\Family{\@ifstar{\oldFamily}{\oldFamily*}}
	\makeatother

	\newsavebox{\numbox}%
	\newsavebox{\slashbox}%
	\newsavebox{\denbox}%
	\newlength{\slashlength}%
	\newlength{\faktorscale}%
	\DeclareDocumentCommand{\newfaktor}{m O{0.35} m O{-0.35}}{
		\savebox{\numbox}{\ensuremath{#1}}
		\savebox{\slashbox}{\ensuremath{\diagup}}
		\savebox{\denbox}{\ensuremath{#3}}
		\setlength{\faktorscale}{0.5\ht\numbox+0.5\ht\denbox}%
		\setlength{\slashlength}{2pt+0.8\faktorscale+#2\faktorscale-#4\faktorscale}%
		\raisebox{#2\ht\slashbox}{\usebox{\numbox}}
		\mkern-2mu%
		\rotatebox{-30}{\rule[#4\ht\denbox]{0.4pt}{\slashlength}}
		\mkern9mu%
		\hspace{-0.44\slashlength}%
		\raisebox{#4\ht\denbox}{\usebox{\denbox}}
	}
	\DeclareDocumentCommand{\linefaktor}{m O{0.08} m O{-0.08}}{
		\savebox{\numbox}{\ensuremath{#1}}
		\savebox{\slashbox}{\ensuremath{\diagup}}
		\savebox{\denbox}{\ensuremath{#3}}
		\setlength{\faktorscale}{0.5\ht\numbox+0.5\ht\denbox}%
		\setlength{\slashlength}{0.2\faktorscale+0.8\baselineskip}%
		\raisebox{#2\ht\slashbox}{\usebox{\numbox}}
		\raisebox{-0.8pt}{%
			\rotatebox{-30}{\rule[#4\ht\denbox]{0.4pt}{\slashlength}} 
		}%
		\mkern-2mu%
		\hspace{-0.25\slashlength}%
		\raisebox{#4\ht\denbox}{\usebox{\denbox}}
	}

	\makeatletter
		\renewcommand*\env@matrix[1][\arraystretch]{%
			\edef\arraystretch{#1}%
			\hskip -\arraycolsep
			\let\@ifnextchar\new@ifnextchar
			\array{*\c@MaxMatrixCols c}
		}
	\makeatother
	
	\makeatletter
		\newcommand{\raisedtimes}[1]{\mathpalette\do@raisedtimes{#1}\relax}
		\newcommand{\do@raisedtimes}[2]{\raisebox{#2}{$\m@th#1\times$}}
		\newcommand{\scaledtimes}[2]{\scaleobj{#1}{\raisedtimes{#2}}}
		
		\newcommand{\sqcuptimes}{\mathrel{\mathpalette\do@sqcuptimes\relax}}
		\newcommand{\do@sqcuptimes}[2]{%
			\ooalign{%
				\hidewidth$#1\m@th\scaledtimes{0.7}{2pt}$\hidewidth\cr
				$#1\m@th\sqcup$\cr
			}%
		}
		\newcommand{\bigsqcuptimes}{\mathop{\mathpalette\do@bigsqcuptimes\relax}}
		\newcommand{\do@bigsqcuptimes}[2]{%
			\ooalign{%
				\hidewidth$#1\m@th\scaledtimes{1.1}{0pt}$\hidewidth\cr
				$#1\m@th\bigsqcup$\cr
			}%
		}
	\makeatother
	
	\newcommand{\id}{\mathsf{id}}
	
	\newcommand*{\ditto}{---\texttt{"}---}
	
	\newcommand{\ph}{\mathord{\rule[-0.05em]{0.6em}{0.05em}}}		
	\newcommand{\phsm}{\mathord{\rule[-0.035em]{0.4em}{0.035em}}}		

	\newcommand{\down}{\mathord{\downarrow}}
	\newcommand{\up}{\mathord{\uparrow}}

	\newcommand{\downwrt}[1]{\mathord{\downarrow_{#1}}}
	\newcommand{\upwrt}[1]{\mathord{\uparrow_{#1}}}

	\newcommand{\cost}[2]{#1\textup{-}\mathsf{cost}_{#2}}
	\newcommand{\yield}[2]{#1\textup{-}\mathsf{yield}_{#2}}

	\newcommand{\Rfree}[1]{\left( \mathcal{R}_{\rm free} \right)^{\sqcuptimes k}}

	\newcommand{\Interesting}[2]{#1\textup{-}\mathsf{Interesting}\left( #2 \right) }
	
	\newcommand{\cmps}{\star}				
	\newcommand{\apply}{\triangleright}		
	
	\newcommand{\enhgeq}{\geq_{\rm enh}}
	\newcommand{\deggeq}{\geq_{\rm deg}}
	\newcommand{\enhconv}{\freeconv_{\rm enh}}
	\newcommand{\degconv}{\freeconv_{\rm deg}}

	\newcommand{\freeconv}{\succeq}
	\newcommand{\reals}{\overline{\mathbb{R}}}
	\newcommand{\ordreals}{(\reals,\geq)}
	
	\newcommand{\UCRT}[1]{(\mathcal{#1}_{\rm free}, \mathcal{#1}, \boxtimes)}

	\newcommand{\res}{\mathcal{R}}
	\newcommand{\resx}{X}
	
	\newcommand{\ordres}{(\res,\freeconv)}
	\newcommand{\ordresx}{(\resx,\freeconv)}

	\newcommand{\atoms}{\mathsf{Atoms}}
	
	\newcommand{\rI}{\mathrm{I}}
	\newcommand{\rO}{\mathrm{O}}
	
	\DeclareRobustCommand{\bitcoin}{{%
		\normalfont\sffamily
		\raisebox{-.05ex}{\makebox[.1\width][l]{-\kern-.2em-}}B%
	}}

	\newtheorem{theorem}{Theorem}[chapter]
	\newtheorem{corollary}[theorem]{Corollary}
	\newtheorem{lemma}[theorem]{Lemma}
		\newenvironment{lemma'}[1]{\renewcommand{\thetheorem}{\ref*{#1}$'$}\addtocounter{theorem}{-1}\begin{lemma}}{\end{lemma}}		
	\newtheorem{proposition}[theorem]{Proposition}
	
\theoremstyle{definition}

	\newtheorem{definition}[theorem]{Definition}
	\newtheorem{example}[theorem]{Example}
		\newenvironment{example'}[1]{\renewcommand{\thetheorem}{\ref*{#1}$'$}\addtocounter{theorem}{-1}\begin{example}}{\end{example}}		

	\newtheorem*{namedthm}{\namedthmname}	
		\newcounter{namedthm}
		\makeatletter
		\newenvironment{named}[1]
			{\def\namedthmname{#1}%
				\refstepcounter{namedthm}%
				\namedthm\def\@currentlabel{#1}}
			{\endnamedthm}
		\makeatother

	\newtheorem{remark}[theorem]{Remark}
\newtheoremstyle{question}{\myparskip}{\myparskip}{\color{BrickRed}\normalfont}{0pt}{\bfseries}{.}{5pt plus 1pt minus 1pt}{}
\theoremstyle{question}
	
\newtheoremstyle{answer}{\myparskip}{\myparskip}{\color{PineGreen}\normalfont}{0pt}{\bfseries}{.}{5pt plus 1pt minus 1pt}{\thmname{#1}\thmnote{ \bfseries #3}}
\theoremstyle{answer}
	
\theoremstyle{remark}

\begin{document}


\pagestyle{empty}
\pagenumbering{roman}

\begin{titlepage}
        \begin{center}
        \vspace*{1.0cm}

        \Huge
        {\bf Resource Theories \\as Quantale Modules}

        \vspace*{1.0cm}

        \normalsize
        by \\

        \vspace*{1.0cm}

        \Large
        Tomáš Gonda \\

        \vspace*{3.0cm}

        \normalsize
        A thesis \\
        presented to the University of Waterloo \\ 
        in fulfillment of the \\
        thesis requirement for the degree of \\
        Doctor of Philosophy \\
        in \\
        Physics \\

        \vspace*{2.0cm}

        Waterloo, Ontario, Canada, 2021 \\

        \vspace*{1.0cm}

        \copyright\ Tomáš Gonda 2021 \\
        \end{center}
\end{titlepage}

\pagestyle{plain}
\setcounter{page}{2}

\cleardoublepage 

\begin{center}\textbf{Examining Committee Membership}\end{center}
  \noindent
	The following served on the Examining Committee for this thesis. The decision of the Examining Committee is by majority vote.
  \bigskip
  
  \noindent
\begin{tabbing}
Internal-External Member: \=  \kill 
External Examiner: \>  Aleks Kissinger \\ 
\> Associate Professor, Department of Computer Science, \\
\> University of Oxford \\
\end{tabbing} 
  \bigskip
  
  \noindent
\begin{tabbing}
Internal-External Member: \=  \kill 
Supervisors: \> Robert W. Spekkens \\
\> Senior Faculty, Perimeter Institute for Theoretical Physics \\
\> \\
\> Florian Girelli \\
\> Associate Professor, Department of Applied Mathematics, \\
\> University of Waterloo \\
\end{tabbing}
  \bigskip
  
  \noindent
\begin{tabbing}
Internal-External Member: \=  \kill 
Internal Member: \> Raymond LaFlamme \\
\> Professor, Department of Physics and Astronomy, \\
\> University of Waterloo \\
\end{tabbing}
  \bigskip
  
  \noindent
\begin{tabbing}
Internal-External Member: \=  \kill 
Internal-External Member: \> Jon Yard \\
\> Associate Professor, Department of Combinatorics and \\
\> Optimization, University of Waterloo \\
\end{tabbing}
  \bigskip
  
  \noindent
\begin{tabbing}
Internal-External Member: \=  \kill 
Other Member: \> William Slofstra \\
\> Assistant Professor, Department of Pure Mathematics, \\
\> University of Waterloo \\
\end{tabbing}

\cleardoublepage

\begin{center}\textbf{Author's Declaration}\end{center}

  \noindent
This thesis consists of material all of which I authored or co-authored: see Statement of Contributions included in the thesis. 
This is a true copy of the thesis, including any required final revisions, as accepted by my examiners.
  \bigskip
  
  \noindent
I understand that my thesis may be made electronically available to the public.

\cleardoublepage


\begin{center}\textbf{Statement of Contributions}\end{center}

	While none of the work presented in this thesis has been published prior to its writing, \cref{sec:monotones} is based to a large extent on the content of the article titled ``Monotones in General Resource Theories'' \cite{Monotones}.
	\Cref{sec:distinguishability}, on the other hand, is based on yet unpublished work.
	I have worked on both of these projects in collaboration with Robert W.\ Spekkens. 
	\Cref{sec:math_prelim,sec:resource_theories} are based on work I have undertaken as a natural continuation of \cite{Monotones}.
	Besides these, I have also worked on several other projects that did not make their way into this thesis in order to keep reasonably concise and on point.
	The ones that were finished during the course of my PhD programme are \cite{gonda2018almost,fritz2020representable,fritz2021finetti}.
	The first one is unrelated to this thesis, while the other two concern topics of abstract probabilistic reasoning. 
	Insofar as resources are often of probabilistic nature, the projects of building an abstract framework for resouning about resources (such as the one presented here) and of building an abstract framework for reaoning probabilistically (such as those presented in \cite{fritz2020representable,fritz2021finetti}) support each other.
	Indeed, process diagrams used in \cref{sec:distinguishability} are precisely those used to depict generalized probabilistic morphisms in a Markov category \cite{fritz2019synthetic}.

\cleardoublepage


\begin{center}
	\textbf{Abstract} 
	\\
	(short version for arXiv)
\end{center}


	We aim to counter the tendency for specialization in science by advancing a language that can facilitate the translation of ideas and methods between disparate contexts. 
	The focus is on questions of ``resource-theoretic nature''. 
	In a resource theory, one identifies resources and allowed manipulations that can be used to transform them. 
	Some of the main questions are: 
	How to optimize resources? 
	What are the trade-offs between them? 
	Can a given resource be converted to another one via the allowed manipulations? 
	
	Because of their ubiquity, methods for answering such questions in one context can be used to tackle corresponding questions in new contexts. 
	The translation occurs in two stages. 
	Firstly, concrete methods are generalized to the abstract language to find under what conditions they are applicable.
	Then, one can determine whether potentially novel contexts satisfy these conditions. 
		
	We focus on the first stage, by introducing two variants of an abstract framework in which existing and yet unidentified resource theories can be represented. 
	Using these, the task of generalizing concrete methods is tackled in \cref{sec:monotones} by studying the ways in which meaningful measures of resources may be constructed. 
	
	One construction expresses a notion of \emph{cost} (or \emph{yield}) of a resource. 
	Among other applications, it may be used to extend measures from a subset of resources to a larger domain.
	
	Another construction allows the translation of resource measures between resource theories. 
	Special cases include resource robustness and weight measures, as well as relative entropy based measures quantifying minimal distinguishability from freely available resources.
	
	We instantiate some of these ideas in a resource theory of distinguishability in \cref{sec:distinguishability}. 
	It describes the utility of systems with probabilistic behavior for the task of distinguishing between hypotheses, which said behavior may depend on.

\cleardoublepage


\begin{center}\textbf{Acknowledgements}\end{center}

	Above all, I am indebted to my supervisor, Robert Spekkens, for his continued support and guidance throughout the years. 
	My ideas have been shaped by his insight and vision, and I will carry the lessons I have learned for the rest of my life.
	In addition, I highly value the extensive freedom I received, which allowed me to grow as an independent researcher and otherwise.
	It has been a pleasure to work in Rob's research group thanks to the inclusive and respectful environment he fosters.
	
	As a result, I have benefitted from nearly endless stream of discussions with other members of the Quantum Foundations and other groups at Perimeter Institute, at group meeting, lunches, seminars, or just in one of the idiosyncratic corridors of the building.
	Thank you, David Schmid, TC Fraser, Nitica Sakharwalde, Ravi Kunjwal, Belén Sainz, Elie Wolfe, Tom Galley, Markus Müller, and others!
	I have managed to collaborate more intensely with some, but less than I would have liked---the scientific talent concentrated at PI is enormous.
	Other collaborators I have and continue to learn from are Paolo Perrone, Eigil Rischel, Lorenzo Catani, Luke Burns, and Justin Dressel.
	A special thanks goes to Tobias Fritz, both for sharing some of his extensive knowledge and for his mentorship.
	
	Perimeter Institute is a unique place and I still find it hard to believe that I have been granted the opportunity to spend 5 years there.
	I am grateful to everyone who contributed towards making it such an incredible hub of ideas and scientific progress.
	This includes the very important people in adminstrative departments and the Black hole bistro.
	Specifically, I want to thank all the people behind the PSI programme; the PSI fellows, James Forrest, Debbie Guenther, and all the students; without which my research and life direction would be quite different, to say the least.
	Big thanks goes to Aggie Branczyk and Maïté Dupuis in particular, with whom I had the pleasure to work on teaching side, and from whom I learned a lot.
	I would also like to thank people who have and those who continue to work hard on making PI ever more inclusive and respectful environment.
	
	I spent most of my PhD years at PI, but I am also grateful to Lídia del Rio and Nuriya Nurgalieva for their hospitality and the opportunity to visit ETH Zürich for one month in the summer of 2019.
	This was a time during which we engaged in a lot of new ideas.
	While none made it into this thesis explicitly, they have certainly had an indirect influence on it and I hope that the work in the thesis can allow for some of these to come to fruition in future.
	
	Additionally, I gained a lot from the many conferences, workshops, and summer schools I was given the chance to participate in. 
	Great thanks goes to all the organizers.
	Among the people I have not mentioned yet, I have been lucky to benefit from discussions with many influential people, among them Bob Coecke, Alexander Kurz, Drew Moshier, Giulio Chiribella, Paolo Perinotti, and many others.
	Additionally, I have learned a lot of quantum information theory from the great course given by John Watrous at University of Waterloo.
	The thesis and its presentation have been improved thanks to the work of Aleks Kissinger on \href{https://tikzit.github.io/}{TikZiT}, a software with which most of the diagrams were made.
	
	Certainly, some of the most important people that enabled me to push through the hard times and enjoy the beautiful ones are my friends and family, of which there are too many to mention.
	However, there are a couple of people that have been particularly instrumental for my ability to complete a PhD thesis at the time of a global pandemic.
	I cannot thank you enough, Soham Mukherjee, Ashley Hynd, Fiona McCarthy, Michael Vasmer, and Denis Rosset!
	
	A huge thanks is reserved to my parents for their endless support in spite of the `abstract nonsense' I decided to engage in.
	Without their support, I could not even dream of being where I am.
	
\cleardoublepage

%
%

\renewcommand\contentsname{Table of Contents}
\tableofcontents
\cleardoublepage
\phantomsection    

\addcontentsline{toc}{chapter}{List of Tables}
\listoftables
\cleardoublepage
\phantomsection		
%
\addcontentsline{toc}{chapter}{List of Figures}
\listoffigures
\cleardoublepage
\phantomsection		


\pagenumbering{arabic}


\chapter{Introduction}\label{sec:introduction}
		
	\subsubsection{Motivations and brief history.}
	
	Physics, like all natural sciences, interpolates between truth-seeking and tool-seeking.
	In the first mode, we try to improve our sense of what is true about the natural world.
	In the latter, which we might call the pragmatic approach, we are focused on directly improving our ability to achieve desired goals.
	While learning about what the world is like tends to improve one's engineering ability, it is not uncommon for key concepts from pragmatic approaches to permeate the foundations of physics as well.
	That is to say, these two modes support each other.
	Among other features, a pragmatic perspective can sometimes help us look past some of our deeply ingrained assumptions and interpretations.
	The pragmatic approach is therefore especially useful\footnotemark{} 
	\footnotetext{Useful here refers to its applicability to the truth-seeking mode.}%
		in times when bottom-up (or reductionist)\footnotemark{} approaches are limited by the current natural interpretations.
	\footnotetext{Pragmatic perspective is not necessarily in opposition, and certainly not the only alternative, to reductionism, but it is the one that's relevant to our discussion.}%
	
	One of the examples of practical questions leading to foundational insights is the study of heat engines and the advent of thermodynamics.
	The goal there is to optimize the useful work that can be extracted from systems, possibly when one is allowed to access a heat bath or two.
	Nowadays, it is hard to imagine fundamental physics without concepts such as energy, temperature and entropy; all of which rose to prominence in this context.
	
	In thermodynamics, valuable resources for the task at hand are systems in states out of thermal equilibrium.
	Since the extraction of work can be conceptualized as a state of some system out of equilibrium (such as a battery), the relevant questions boil down to those of manipulation of resources.
	This point of view serves as a basis for the resource-theoretic treatment of thermodynamics \cite{lostaglio2019introductory}. 
	
	\begin{example}\label{ex:rt_athermality_intro}
		A common approach for treating thermodynamics in this way is via a resource theory of athermality \cite{Janzing2000,Brandao2013}.
		Its name stems from the fact that the state of a system, in which the system is in thermal equilibrium with a given heat bath, is sometimes called the thermal state.
		Other states, ones that are valuable by virtue of being out of equilibrium, are thus athermal.
		The allowed manipulations of resources (called \emph{thermal operations}) include the use of systems initialized in their thermal states, application of reversible energy-preserving operations, and discarding of systems.
		Since these operations alone do not allow one to generate work, any useful work extracted in a protocol involving thermal operations and athermal states can be attributed to the latter---i.e.\ resources used in the protocol.
		We return to this resource theory in greater detail in \cref{ex:rt_athermality} as well as in other examples throughout the thesis.
	\end{example}
	
	It has eventually been understood that resources of thermal nonequilibrium could be of informational character. 
	One of the most famous examples is the Szilard engine \cite{Szilard1929}, which uses information about the state of a system in order to perform work.	
	Here, the pragmatic point of view has a lot to offer and complements the recognition that `information is physical' \cite{landauer1991information}.
	In particular, we can recognize information as a resource that can be interconverted with the resource of (ordered) energy.
	Rather than focusing on what this means at the `fundamental level', a resource theory studies the features of such convertibility and the properties of systems that allow for such connections to be made.
	Consequently, it can help us solidify our understanding of concepts (such as that of energy and information) in a relational rather than analytic manner.
	
	Resource-theoretic thinking has expanded with the development of full-fledged information theory.
	The pioneering work of Claude Shannon \cite{Shannon1948} is centered around questions regarding the convertibility of resources that are communication channels.
	It is not surprising, therefore, that with the rising prominence of information theory in physics in recent years, the application of resource-theoretic ideas in physics has also been on the rise.  
	
	One of the key moments on the road towards establishing resource theories as a research topic in its own right is the development of quantum information theory. 
	The resource-theoretic approach has become particularly influential in the study of quantum entanglement \cite{Horodecki2009,Devetak2008}.
	In a resource theory of entanglement, one studies the manipulation of multipartite quantum states by means of local operations aided by classical communication (LOCC) among the parties \cite{Chitambar2014}.
	In this context, viewing entangled states as resources for communication tasks (via teleportation), dense coding, or secret key generation, solidifies the modern understanding of entanglement as an attribute of physical systems.
	Moreover, by comparison to analogous resource theories with only classical resources available, it concretizes how the concept of entanglement diverges from classical correlations. 

	The success of entanglement theory inspired the study of other properties of quantum states and channels as resources relative to manipulation by a chosen set of transformations.
	For example, symmetry-breaking states can be characterized as resources relative to covariant operations \cite{Marvian2013} and quantum states of thermal nonequilibrium can be characterized as resources relative to thermal operations \cite{Brandao2013}. 
	These resource theories have led to surprising and important conceptual insights, such as a generalization of Noether's theorem \cite{Marvian2014} and a refinement of our understanding of the second law of \mbox{thermodynamics \cite{Brandao2015}}.
	Many other examples of the use of resource theories within quantum information theory have followed \cite{Chitambar2018}. 
	
	\subsubsection{High level introduction to resource theories.}
	
	A resource theory provides a description of phenomena in terms of answers to questions regarding the manipulation of resources. 
	Its basic elements include resources (such as gas in a chamber) subject to transformations (such as heating the gas).
	We identify a special class of transformations, with which the resources can be freely manipulated.
	They are ``free'' in the sense that their cost need not be accounted for, whether as a consequence of technological, fundamental or circumstantial considerations.
	For instance, if we are interested in questions of work extraction, free transformations \emph{should not} include those that require work to be exerted---all of the ordered energy, incoming or outgoing, has to be accounted for.
	In a resource theory of athermality (\cref{ex:rt_athermality_intro}), free transformations are given by the thermal operations that satisfy this requirement indeed.
	
	One objective of a resource theory is to turn simple qualitative distinctions (free vs.\ non-free transformations) into more refined qualitative statements (possibility and quality of resource conversions) as well as quantitative ones (amount of resources needed to achieve a certain task).
	In this way, it concentrates the assumptions being made, primarily in the form of 
	\begin{enumerate}
		\item assumptions inherent to a framework of resource theories (and its interpretation),
		\item assumptions that inform the explicit description of resources and transformation used, and
		\item assumptions motivating the choice of free transformations.
	\end{enumerate}
	There is not necessarily a clear distinction between the first two types of assumptions.
	The most apparent example of these in the context of quantum resource theories (of states) is the assumption that we can describe resources as quantum states and that their transformations are well-described by quantum channels.
	Of course, these include any assumptions implicit in quantum theory and a chosen interpretation thereof.
	The third type of assumptions is the one which resource theories are good at exposing.
	
	For instance, consider the concept of quantum entanglement.
	When we say that a quantum state is entangled with respect to a chosen bipartition, it is clear what we mean---namely that the state is not separable.
	However, there are multiple resource theories in which the set of free states (i.e.\ ones that can be generated by free transformations from nothing) are the separable states.
	One is the resource theory with free transformations given by LOCC operations mentioned above. 
	The other is a resource theory of LOSR (local operations and shared randomness) in which no communication is free---the two parties are merely allowed to access pre-shared classically correlated states for free on top of arbitrary local operations \cite{buscemi2012all}.
	The two resource theories make different statements about how to compare the entanglement of quantum states---they give rise to two \emph{different} concepts of quantum entanglement.
	
	Whether it is appropriate to say that a given quantum state $\rho$ is more entangled (i.e.\ a better resource) than another state $\sigma$ 
	\begin{itemize}
		\item if there is an LOCC operation mapping $\rho$ to $\sigma$, or  
		\item merely if there is an LOSR operation that achieves the transition
	\end{itemize}
	depends on the situation we are considering \cite{schmid2020understanding}.
	In particular, it depends on whether assumptions motivating one of the choices of free transformations are better suited to the situation at hand.
	For example, when describing the utility of quantum states for tasks in quantum internet infrastructure, it is natural to assume the existence of classical internet that allows for classical communication.
	In such a case, LOCC operations are the sensible choice for the set of free transformations.
	On the other hand, in a Bell experiment\footnotemark{} we assume that the two parties involved do not communicate with each other. 
	\footnotetext{A Bell experiment is a scenario that allows one to experimentally distinguish predictions of quantum theory from those of local hidden variable models via a Bell inequality \cite{bell1964einstein,bell1966problem}.
		It is an experiment with two parties that make measurements on their respective part of a bipartite quantum system (such as a pair of potentially entangled particles).
		The measurements they make depend on the choice of settings that are given to them locally (see diagram \eqref{eq:Bell_DAG} for a graphical depiction of this scenario).
		Crucially, once the settings for each party are given, the two are not allowed to communicate by any means.}%
	Entanglement with respect to LOSR operations is the relevant concept in such a context.

	Other objectives of developing a particular resource theory relate to practical questions of how to optimize the amount of resources needed to achieve a task, which could be described either as a transformation (powering an electric motor) or as another resource (a charged battery).
	This includes the aim to understand the trade-offs between different kinds of resources one could try to use.
	
	A common attitude towards resource theories, and in particular the choice of free transformations, is that they specify a practical restriction arising from the situation one wishes to study.
	In thermodynamics, work is a valuable (and scarce) asset.
	In entanglement theory, it is the noiseless quantum communication that is deemed costly.
	
	However, an alternative interpretation is also possible, in which the choice of free transformations encodes a physical principle, which could be emergent or fundamental.
	For example, invertible transformations (e.g.\ unitary channels in quantum theory) that are free in a resource theory of athermality (\cref{ex:rt_athermality_intro,ex:rt_athermality}) are required to satisfy energy conservation \cite[section I.B]{lostaglio2019introductory}.
	Other restrictions on agents' capabilities can be also often seen in such light whenever they encode what \emph{any} agent in a given situation can or cannot achieve \emph{in principle}. 
	A particularly clear example of such a principle is the inability to communicate faster than via the speed of light.
	This second attitude is explored deeper in constructor theory \cite{deutsch2013constructor}---a framework for building fundamental theories based on (im)possibility.
	
	\subsubsection{Merits of the resource-theoretic perspective.}
	
	A significant value of a resource-theoretic perspective is the unification and organization of knowledge about related phenomena.
	Entangled states may be used for varied tasks, some of which we mentioned explicitly already.
	However, in a resource theory of entanglement, they can all be jointly studied and related as manifestations of the resource of entanglement.
	As such, a resource theory can be thought of not just as a means for modelling and optimization, but also as a methodological tool.
	Because it aims to understand different 
		phenomena in a unified explanatory scheme, it can also impose a certain level of consistency among definitions and conventions used to describe these.
	\begin{example}
		Consider a resource theory of nonlocality,\footnotemark{} which can be used to understand violations of Bell inequalities within a resource-theoretic framework.
		We provide more details in \cref{ex:quantum_correlations}, but here we do not need those.
		\footnotetext{Following \cite{Wolfe2019}, it is perhaps more appropriate to call it the resource theory of nonclassicality of common-cause boxes.}%
		Since quantum entanglement is necessary to witness a Bell inequality violation, it has long been thought that there ought to be a very close relationship between the concepts of quantum entanglement\footnotemark{} and quantum nonlocality.
		\footnotetext{The implicit assumption was that the relevant concept is that of quantum entanglement with respect to LOCC operations, which is usually thought to be the default one.}%
		In contradiction, it has later been shown that quantum states with \emph{less} LOCC-entanglement can be responsible for \emph{more} nonlocality \cite{methot2006anomaly,vidick2011more}.
		The reason for this inconsistency becomes apparent if one tries to build a unified resource theory of which quantum entanglement and nonclassicality of common-cause boxes are special cases.
		This is not possible if one interprets quantum entanglement as LOCC-entanglement because communication is not allowed in a Bell scenario.
		On the other hand, there is no inconsistency between LOSR-entanglement and Bell nonlocality.
		Indeed, there is a unified resource theory of which the two are special cases \cite{schmid2020understanding,schmid2020type}.
		The resolution of (apparent) paradoxes can be thus achieved by embracing the resource-theoretic approach and exposing implicit assumptions behind concepts such as that of quantum entanglement in this situation.
	\end{example}
	Additionally, by concentrating the assumptions as we alluded to before, the resource-theoretic point of view can often help expose hidden assumptions to the benefit of greater transparency and future revisions.
	
	In view of the above discussion, resource theories fit within the program of knowledge integration aiming to counter the natural process of knowledge disintegration that accompanies specialization. 
	One of the tools of knowledge integration is the development of languages that can express and relate problems from distinct domains of knowledge.
	This is especially pertinent if disparate areas of science and engineering aim to answer analogous questions, but use different terminology and therefore cannot simply engage in the exchange of ideas. 
	An objective of such a language is to be able to translate results between, and import methods used to address, these analogous questions from one area to another. 
	Goals of numerous research efforts are in line with the above ideas, one of the most notable examples of which is applied category theory \cite{fong2018seven,Coecke2009}.
	Resource theories in our sense may be thought of as a subfield thereof.
		
	\subsubsection{Merits of an abstract framework.}
	
	Topics currently studied under the hood of resource theories are largely limited to physics and quantum information theory in particular.
	However, questions of resource-theoretic nature are by no means exclusive to physics.
	This sentiment is echoed in other works on the topic of resource theories from a broad perspective, such as the frameworks of \cite{DelRio2015,Coecke2016,Fritz2017,marsden2018quantitative}.
	There are undoubtedly many parallel formulations of similar ideas in other domains of inquiry, some of which we mention towards the end of \cref{sec:future} (see also \cite[section 10]{Fritz2017}).
	Going beyond natural sciences and engineering, the ubiquity of majorization-like conditions \cite{Marshall1979} throughout applied mathematics---in the context of economics, ecology, genetics, statistics and others (see sections 11 and 12 of \cite{arnold2018majorization} for an overview)---underscores this point.
	However, because of potentially independent development and terminology, they may not be easy to identify and connect in stronger than heuristic ways.
	
	This is where a framework that generalizes the current instances of resource theories may help.
	By abstracting some of the key features along with their interpretations, it can accommodate a wider range of situations and facilitate connections that would otherwise be tricky to establish formally.
	These connections can occur within physics, thus recontextualizing concepts in a new light and potentially proposing new ways to move past dated ontological commitments.
	At an abstract level, they can also allow translation of methods even if the concepts and their interpretations do not carry over.
	While this thesis develops an abstract resource-theoretic framework, it contains very little material of this kind that would be directly applicable.
	We leave such work for future studies.
	
	On another note, an abstract framework can advance the understanding of resource theories as construed at present.
	It can help us recognize patterns and formulate methods in a more unified manner.
	Moreover, it allows the study of properties of resource theories that are necessary or sufficient for these general methods to be applicable.
	\Cref{sec:monotones} is concerned with such applications that aim to generalize and understand methods used in resource theories currently.
	More details are provided in the part of the Introduction on resource measures below.
	
	General frameworks can also set the stage for building new and alternative resource theories via general constructions framed independently of the particular instance.
	\begin{example}\label{ex:destinguishability}
		In \cref{sec:distinguishability}, we introduce a resource theory of distinguishability of probabilistic behaviours.
		The resource objects are tuples of probability distributions on a system $A$, and they are valuable if the different elements of the tuple can be distinguished with high confidence by observing $A$.
		Besides the relevance of this theory for questions in statistics, we can also use it to generate resource measures for any resource theory based on classical probability theory by translating measures of distinguishability.
		We give more intuition on the translation of resource measures towards the end of the Introduction.
		
		Given an abstract framework, we can then define a resource theory of distinguishability in general (see \cref{sec:Encodings}), whose resource objects are tuples of resources from some prespecified resource theory of interest.
		This construction includes, as instances, resource theories of distinguishability of quantum states \cite{Wang2019} or channels \cite{wang2019resource} as well as many others.
		Once again, these resource theories of distinguishability can be studied in their own right, but they can also be used to generate resource measures in the original resource theory by translating measures of distinguishability (\cref{sec:Monotones From Contractions in General}).
		The procedure is analogous to the classical case.
	\end{example}
	\begin{remark}\label{rem:iid}
		There are other ways to obtain new resource theories in the abstract that we do not elaborate on in this thesis.
		For instance, instead of a tuple of resources representing distinct alternatives as in the case of \cref{ex:destinguishability}, we can model uncertainty about the identity of the resource by other means.
		A probability distribution over the space of resources is one of the common representations of uncertainty, but there are many others as well.
		
		One can also take a resource theory and study its many-copy (or IID) version in which the resource objects are countable powers $x^{\mathbb{N}}$ of every resource $x$.
		A successful transformation $x^{\mathbb{N}} \mapsto y^{\mathbb{N}}$ is specified by a sequence of transformations indexed by $n \in \mathbb{N}$ that achieve the conversion $x^n \mapsto y^{\lfloor \alpha n \rfloor}$.
		Here, $\alpha \in [0,\infty]$ is a label that indicates the rate of conversion of said many-copy transformation.
		Usually, one would also require the conversions to be \emph{approximate} with an asymptotically vanishing error.
		One of the reasons why this construction cannot be described in our abstract framework is that we do not include any metric structure that would allow us to talk about approximations.
	\end{remark}
		
		\vspace{\intextsep}
		\begin{samepage}
		\bgroup
		\begin{adjustbox}{center}
			\renewcommand{\arraystretch}{1.5}%
			\begin{tabular}
			{>{\raggedright}p{0.25\textwidth}>{\centering}p{0.12\textwidth}>{\centering}p{0.12\textwidth}>{\centering}p{0.15\textwidth}>{\centering}p{0.14\textwidth}>{\centering\arraybackslash}p{0.15\textwidth}}
				\hline
				\textbf{Framework}		&  	\textbf{Transfor-mations}	&		\textbf{Resource types}		&  	\textbf{Resource conjunction} 	&  	\textbf{Ambiguity}		&  	\textbf{Disjunction of access} 	\\ 
			\end{tabular}
		\end{adjustbox}\\%
		\begin{adjustbox}{center}
			\renewcommand{\arraystretch}{1.5}%
			\begin{tabular}
			{>{\raggedright}m{0.25\textwidth}>{\centering}m{0.12\textwidth}>{\centering}m{0.12\textwidth}>{\centering}m{0.15\textwidth}>{\centering}m{0.14\textwidth}>{\centering\arraybackslash}m{0.15\textwidth}}
				\hline
				Partitioned process theories \cite{Coecke2016}		&	\checkmark	&	\checkmark	&	\checkmark	&	 				& 	 				\\	
				Ordered commutative monoids \cite{Fritz2017}		&					&					&	\checkmark	&	 				& 	 				\\	
				Resource theories of knowledge \cite{DelRio2015}	&	\checkmark	&					&	 				&	\checkmark 	& 	 				\\	
				Universally combinable resource theories (\cref{sec:ucrt})									&					&					&	\checkmark	&					& 	\checkmark	\\	
				Quantale modules (\cref{sec:rt})						&	\checkmark	&					&					&		 			& 	\checkmark 	\\	
				\hline
			\end{tabular}
		\end{adjustbox}
		\captionof{table}[Frameworks for Resource Theories]{Rough overview of some of the resource-theoretic frameworks and the features they model explicitly (rather than implicitly or not at all).
			\emph{Transformations} refer to the inclusion of protocols for converting resources, whose identity is distinct from the identity of resources so that they can be studied and optimized in their own right.
			With \emph{resource types}, one can perform high-level inference and verification because they restrict the manipulations to those that respect the types.
			An explicit structure of \emph{resource conjunctions} allows one to consider composite resources whose parts are themselves resources of potentially distinct types.
			A structure representing \emph{ambiguity} deals with the lack of knowledge about the identity of resources (such as when we only care about approximate conversions).
			Finally, \emph{disjunction of access} to resources is relevant when we wish to contrast the power of agents who find themselves in different circumstances (hypothetical or not).
			It should be noted that a partitioned process theory as such does not include a specification of resource transformations and conjunctions.
			Instead, this is achieved by identifying a resource theory associated to it, such as a resource theory of states \cite[section 3.2]{Coecke2016}, a resource theory of parallel combinable processes \cite[section 3.3]{Coecke2016}, or a resource theory of universally combinable processes \cite[section 3.4]{Coecke2016}.}
		\label{tab:frameworks}
		\egroup
		\end{samepage}
		\vspace{\intextsep}
	
	\subsubsection{Resource theories as quantale modules.}
	
	In this thesis, we present an abstract framework for resource theories, complementing previous works of this kind \cite{DelRio2015,Coecke2016,Fritz2017}.\footnotemark{}
	\footnotetext{See also the recent approach of \cite{fraser2020functoriality}, which shares some of the motivations, but makes use of modal logic and model theory in order to arrive at an abstract description.}%
	The idea guiding the specific mathematical structures chosen here is to be able to model varying degrees of access that an agent may have to resources.
	To illustrate this idea, imagine a collection of resources (denoted by $X$) as well as a collection of transformations (denoted by $T$).
	An element $x$ of the set $X$ encodes that the agent in question has access to the resource $x$, but not to any other ones.
	Access to free transformations---something that we commonly assume implicitly in a resource theory---is specified by a subset $T_{\rm free}$ of all free transformations.
	There are other sets we may be interested in too.
	Given a resource $x$, there is a set that describes those resources which can be obtained provided an agent has access to \emph{both} $x$ and $T_{\rm free}$.
	We denote such a set of resources by $T_{\rm free} \apply x$.

	\begin{example}
		Resource theory of magic states \cite{howard2017application} is an instance of a resource theory describing quantum computation in the quantum circuit paradigm.
		Free transformations constitute those that are efficiently simulable on a classical computer---they include the preparation of stabilizer states, implementation of Clifford unitary gates and of Pauli measurements.
		In the magic state paradigm, we assume that these are operations that can be executed in a fault-tolerant manner.
		Among the resources that are not free are so-called magic states.
		Sets of resources such as $T_{\rm free} \apply x$ for a given magic state $x$ are of particular interest as they describe quantum computations accessible if we allow the injection of a magic state on top of the classically-simulable operations.
	\end{example}

	For a free transformation $t$, the subset inclusion $\{t\} \subseteq T_{\rm free}$ is interpreted as follows.
	An agent who can choose to use any of the free transformations can also choose to use the transformation $t$.
	That is, access to $T_{\rm free}$ entails access to $\{t\}$.
	The same interpretation is given to other subset inclusions. 
	As we will see later, from the definition of a resource theory it follows, for example, that we have $x \subseteq T_{\rm free} \apply x$.\footnotemark{}
	\footnotetext{In order to improve readability, we generally use an abusive notation in which singletons such as $\{x\}$ are merely denoted by $x$.}%
	We thus interpret this fact (in a resource theory of magic states) as saying that an agent who has access to the injection of a given magic state $x$ on top of classically-simulable operations may also choose to just inject the magic state without applying any other transformations.

	Notions of degrees of access to resources and their entailment mentioned above are often used implicitly in resource-theoretic reasoning. 
	The most obvious example is the question of whether it is possible to convert a resource $x$ to a resource $y$ via free transformations, which can be formulated as: 
	``Does joint access to $x$ and $T_{\rm free}$ entail access to $y$?''
	In other words, we may write it as:
	``Is $y \subseteq T_{\rm free} \apply x$ true?''
	
	We are passing from resources to \emph{sets of} resources and from set membership $\in$ to subset inclusion $\subseteq$.
	This is viewed as a move to use the power set $\mathcal{P}(X)$ instead of the collection $X$ of all resources to describe their value and interactions.
	However, the fact that we \emph{can} describe resource theories in this way does not mean that we \emph{should}.
	The justification for developing the abstract language presented in \cref{sec:resource_theories} is twofold.
	
	First of all, its development goes hand in hand with the understanding of the methods from \cref{sec:monotones} on resource measures.
	By expressing the constructions of resource measures abstractly, we not only understand \emph{that} they are valid measures, but we can also see more clearly \emph{why}.
	Consequently, this reformulation allows us to generalize them, as we explain in greater detail in \cref{sec:monotones}.
	
	\begin{example}
		As one example for many, consider the case of robustness measures.
		Resource robustness \cite{vidal1999robustness} and global robustness \cite{harrow2003robustness} are measures that play an important role in many resource theories.
		For instance, global robustness can be used to characterize many-copy conversions \cite[theorem 1]{brandao2015reversible}.
		The former quantifies the endurance of resources against noise in the form of mixing with other free resources while the latter allows for noise in the form of mixing with arbitrary resources.
		By proving that robustness is a valid measure in the abstract language, it is easy to recognize that both robustness measures are special cases of a more general construction (\cref{sec:weight and robustness}).
		Indeed, it is not important that the ``noise'' comes from either the set of free resources or the set of all resources.
		The crucial property is that the set is \emph{downward closed} (\cref{def:downset}).
	\end{example}
	
	The second case for abstraction comes from the aim to expand the scope of resource theories.
	Compared to the aforementioned generalization of concrete methods, to which we give a lot of space in this thesis, this objective is largely undeveloped here.
	One of the ways in which we strive to extend resource-theoretic thinking beyond its present context is by developing the resource-theoretic language in more general terms.
	That is, once we pass from a set of resources $X$ to the power set $\mathcal{P}(X)$, it is a small step to consider other \emph{lattices} (\cref{sec:lattice}) than the power set lattice.
	Lattices are mathematical structures that allow us to talk about entailment and disjunction of access to resources, just like the operations $\subseteq$ and $\cup$ from the power set do, but they have much wider applicability.
	Throughout the manuscript, we mention a few simple examples of resource theories that cannot be described by a power set lattice (\cref{ex:ucrt_from_ocm,ex:AIT}), but we expect many more examples to appear in the future.
	
	Some of these examples are undoubtedly already developed in the literature that is not commonly associated with resource-theoretic studies.
	This is the second way in which resource theories can extend their scope.
	Namely, by developing an abstract framework, we may recognize that investigations in other research fields are instances of resource-theoretic methods.
	The most immediate ones to look for, given the nature of our framework, are notions from theoretical computer science and logic (see \cref{sec:future}).
	Thus, the abstract framework we develop here should be particularly suited for establishing connections between physics and computer science. 
	
	Of course, a general-purpose framework for resource theories ought to capture other features of resource theories and their interpretations than merely the varying degrees of access to resources.
	In order to keep it simple and restrict its scope, we do not necessarily aim for all of those.
	Some of the main omissions are mentioned in \cref{sec:future}, and outlined also in \cref{tab:frameworks}.
	What we do include is a description of resources, transformations, and two notions of composition.
	One type of composition, denoted by $\star$, describes how transformations can be combined to form new ones.
	In resource theories of quantum states, where the transformations are channels, $\star$ prescribes how to compose channels.
	Channels may generally be composed sequentially and in parallel, but their composition may also be restricted further, e.g.\ by causal requirements.
	The other type of composition, denoted by $\apply$, models how transformations transform resources.
	For example, it may specify, given a quantum channel and a state, what is the output state after the channel is applied.
	On top of these, we require a specification of free transformations that prescribe how resources can be manipulated for free.
	More details on resource theories viewed as \emph{quantale modules}, in line with the above, can be found in \cref{sec:rt}.
	In particular, the central definition of a \emph{resource theory}, as used in the context of this thesis, is \cref{def:rt}.
	
	Besides the above structure that loosely corresponds to the partitioned process theories from \cite{Coecke2016}, but with less focus on compositionality, we also include an account of universally combinable resource theories as a special case in \cref{sec:ucrt}.
	The idea of universally combinable resource theories dates back to \cite{Coecke2016} and it is closely connected to the ordered commutative monoids of \cite{Fritz2017} (see \cref{sec:ucrt_order} for more details).
	In this thesis, a universally combinable resource theory is one in which there is no conceptual distinction between resources and transformations.
	Manipulation of resources is achieved by a composition of resources, denoted by a \emph{commutative} binary operation $\boxtimes$. 
	The reason to require commutativity is that we interpret $r \boxtimes s$ for two resources $r$ and $s$ as their universal combination.
	That is, $r \boxtimes s$ specifies \emph{all} conceivable ways of combining $r$ and $s$, which coincides with $s \boxtimes r$ that specifies all conceivable ways of combining $s$ and $r$.
	Universally combinable resource theories have been used in the work concerning monotones in resource theories \cite{Monotones}, which forms a basis for a large part of this thesis.
	In a sense, one could say that the two operations $\star$ and $\apply$---telling us how transformations can be combined and how they act on resources---reduce to a single operation $\boxtimes$ that describes both how resources can be combined and how they act on other resources. 
	Additionally, $\boxtimes$ provides a notion of conjunction of resources, which the framework of quantale module does not have explicitly (see \cref{tab:frameworks}).
	Namely, $r \boxtimes s$ is interpreted as allowing access to both $r$ and $s$ in conjunction (as well as any other resource that results from combining the two).
	
	\subsubsection{Resource measures.}
	
	Regardless of the framework, one of the key questions to answer in a resource theory is to ascertain whether a free conversion between two given resources is possible.
	This problem can be viewed as that of characterizing the convertibility relation which provides an ordering, or hierarchy, among resources.
	Two resources, $r$ and $s$, are ordered as long as $r$ can be converted to $s$ by a free transformation.
	A useful tool for learning about the resource ordering is an algorithm that can answer the question: 
	``Does a free transformation mapping a resource $r$ to a resource $s$ exist?''
	
	An alternative approach to understand the resource ordering is to identify ways in which we can assign real values to resources.
	In particular, we are generally interested in value-assignments that preserve all the relations between resources.
	That is, a meaningful resource measure $f$ is one that satisfies $f(r) \geq f(s)$ whenever the resource $r$ is freely convertible to the resource $s$.
	We thus avoid assigning higher value to a resource that is less useful than another one.
	Such resource evaluations are called \emph{resource monotones}. 

	While a resource monotone preserves all the relations between resources, it generally also adds new ones.
	That is, knowing that $f(r) \geq f(s)$ holds provides no guarantee that the resource $r$ is more useful than the resource $s$. 
	As such, a single resource monotone contains only partial information about the ordering of resources.
	However, a sufficient number of resource monotones \emph{can} be used to completely characterize the resource ordering.
	This happens if knowing that $f_i(r) \geq f_i(s)$ holds \emph{for every} resource monotone $f_i$ in the collection is in fact a guarantee that $r$ and $s$ are ordered.
	In this case $r$ can be freely converted to $s$ \emph{if and only if} $f_i(r) \geq f_i(s)$ holds for all $f_i$.
	We then speak of a \emph{complete set} of monotones $\{f_i\}_i$.
	
	Monotones play a key role in concisely capturing a substantial amount of information about a resource theory.
	They can be used to quickly answer questions about the resource order and can be used to gain intuition about the resources in question, particularly when they can be represented graphically.
	\begin{example}
		A classical resource theory of nonuniformity \cite{Gour2015,Horodecki2003} describes the degree to which probability distributions differ from the uniform distribution.
		It can be viewed as a resource theory of athermality with respect to a heat bath at infinite temperature (so that thermal states are given by uniform distributions).
		The resource ordering of nonuniformity can be shown to be identical to the well-known majorization ordering \cite{Marshall1979}.
		One of the equivalent ways to describe majorization is in terms of the so-called Lorenz curve \cite{lorenz1905methods,arnold1987majorization}.
		That is, one can determine the convertibility between probability distributions in a classical resource theory of nonuniformity simply by drawing their Lorenz curves  (see \cref{fig:Lorenz_curves} for a concrete example).
	\end{example}
		
		\begin{figure}[h]
			\begin{center}
					\begin{tikzpicture}
						\begin{pgfonlayer}{nodelayer}
							\node [style=none] (0) at (-2, -2) {};
							\node [style=none] (1) at (-2, 2) {};
							\node [style=none] (2) at (2, 2) {};
							\node [style=none] (3) at (2, -2) {};
							\node [style=none] (4) at (3, -2) {};
							\node [style=none] (5) at (-2, 3) {};
							\node [style=none] (7) at (-2, -2.5) {$0$};
							\node [style=none] (8) at (2, -2) {};
							\node [style=none] (11) at (0, 1.5) {};
							\node [style=none] (12) at (-1, 1) {};
							\node [style=none] (14) at (-2, 1.5) {};
							\node [style=none] (15) at (0, -2) {};
							\node [style=none] (16) at (0, -2.5) {$1/2$};
							\node [style=none] (17) at (2, -2.5) {$1$};
							\node [style=none] (18) at (-2.5, 1.5) {$7/8$};
							\node [style=none] (23) at (-2, 1) {};
							\node [style=none] (24) at (-2.5, 1) {$3/4$};
							\node [style=none] (25) at (-1, -2) {};
							\node [style=none] (26) at (-1, -2.5) {$1/4$};
						\end{pgfonlayer}
						\begin{pgfonlayer}{edgelayer}
							\draw [style=arrow] (8.center) to (4.center);
							\draw [style=arrow] (1.center) to (5.center);
							\draw [style=blue arrow, thick] (0.center) -- (11.center) -- (2.center); 
							\draw [-, draw={rgb,255: red,150; green,0; blue,15}, thick] (0.center) -- (12.center) -- (2.center); 
							\draw [style=segment] (0.center) to (25.center);
							\draw [style=segment] (25.center) to (15.center);
							\draw [style=segment] (15.center) to (8.center);
							\draw [style=segment] (0.center) to (23.center);
							\draw [style=segment] (23.center) to (14.center);
							\draw [style=segment] (14.center) to (1.center);
							\draw [style=dashed box] (1.center) to (2.center);
							\draw [style=dashed box] (3.center) to (2.center);
						\end{pgfonlayer}
					\end{tikzpicture}
			\end{center}
			\caption[Example of Lorenz Curves]{An example of two Lorenz curves.
			Red curve corresponds to distribution $x \coloneqq (3/4,1/12,1/12,1/12)$ while the blue one to distribution $y \coloneqq (7/8,1/8)$.
			The height of the respective Lorenz curve at any point along the horizontal axis is a resource monotone.
			The fact that the red curve is higher than the blue one at $1/4$ implies that $y$ is \emph{cannot} be converted to $x$ by free transformations in the resource theory of nonuniformity.
			On the other hand, comparing the fact that the blue curve is higher than the red one at $1/2$ implies the lack of convertibility in the opposite direction.
			Moreover, Lorenz curves have the appealing property that they procide a complete set of nonuniformity monotones, so that no information about the majorization order is lost in this representation.}
			\label{fig:Lorenz_curves}
		\end{figure}
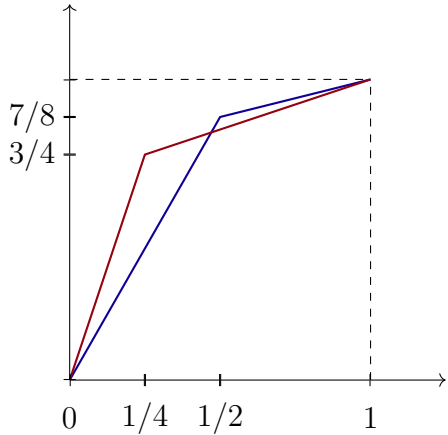
	Resource monotones can also sometimes acquire a concrete interpretation as quantifying the utility of any given resource for a particular task \cite{Takagi2019,Ducuara2019}. 
	Additionally, one can use them in algorithmic characterizations of the resource ordering.
	
	In practice, one can identify patterns among monotones used in various resource theories, which begs the question:
	Are they instances of constructions that are independent of the details of a particular resource theory?
	This is what we investigate in \cref{sec:monotones}, in the context of the aforementioned abstract framework.
	Following the approach of \cite{Monotones}, we introduce a general scheme for modularizing monotone constructions (see \ref{broad scheme}).	
	The idea is that the problem of finding monotones can be broken down into identifying
	\begin{itemize}
		\item a mediating map from the set of resources to a distinct ordered set, ideally one that is well-understood, and
		\item a monotone for the latter which is then pulled back to the resource ordering.
	\end{itemize}
	\begin{example}
	\label{ex:Nielsen's Theorem}
		As shown by Nielsen in \cite{Nielsen1999}, entanglement properties are connected to majorization \cite{Marshall1979}. 
		Specifically, the ordering of pure bipartite quantum states with respect to conversions by LOCC operations is equivalent to the (reverse) ordering of probability distributions by majorization. 
		It follows that all such measures of nonuniformity of distributions can be pulled back to measures of entanglement via the mediating map, as pointed out in \cite{Coecke2010} where the mediating map was termed a qualitative measure of entanglement.
	\end{example}
	
	Consequently, we identify several ways to generate resource monotones.
	This allows us to
	\begin{enumerate}
		\item classify concrete resource monotones according to how they fit with the \ref{broad scheme},
		\item have a structured procedure for generating new, alternative monotones starting from existing constructions,
		\item generalize monotones beyond the resource theories in which they were introduced and identify sufficient conditions under which they exist,
		\item and thus to translate resource measures between resource theories.
	\end{enumerate}
	Specifically, \cref{sec:cost_yield} concerns monotone constructions dubbed \emph{resource cost} and \emph{resource yield}.
	They measure the minimal cost and the maximum yield of the resource in question respectively.
	These include a wide class of examples, among them entanglement rank and entanglement of formation. 
	They have been used explicitly in the study of quantum (and other nonclassical) correlations from the resource-theoretic perspective (see \cref{ex:quantum_correlations}).
	A special case of the cost and yield constructions describes extremal extensions of monotones from a subset to all resources, as introduced in \cite[section 3.2]{Monotones} and subsequently also considered in \cite{gour2020optimal}.
	For instance, the extension of monotones for states, which are often easier to study, to other types of resources such as channels (see \cref{ex:Changing type,sec:dist_channels}) can be particularly useful.
	A general construction of costs and yields is summarized in the form of \cref{thm:cost} and \cref{thm:yield} respectively.
	
	Then, in \cref{sec:translating}, we first study monotone constructions that pull back measures of distinguishability.	
	These can be thought of as examples of the translation of monotones, as the codomain of the mediating map is a resource theory of distinguishability (see \cref{ex:destinguishability} and \cref{sec:Encodings}).
	We also show how standard monotone constructions, such as weight, robustness and measures based on relative entropy, can all be recast in this way, which leads to natural generalizations of them.
	
	A concrete instance of a resource theory of distinguishability, one in which the resources are tuples of probability distributions, is studied in greater detail in \cref{sec:distinguishability}.
	We interpret the tuples as specifying counter-factual possibilities (or hypotheses) of what the actual behaviour is.
	A particular tuple 
	\begin{equation}
		(x_1, x_2, \ldots, x_k)
	\end{equation}
	of distributions of some classical system $A$ is a resource of distinguishability insofar as it allows one to distinguish the alternative hypotheses labelled by $\{1, 2, \ldots, k\}$, e.g.\ by performing measurements of the system $A$.
	The transformations allowed for free are those that are completely uninformative about (i.e.\ independent of) the hypotheses.
	We can view this resource theory as a reframing of the theory of comparison of statistical experiments \cite{blackwell1951comparison,le1996comparison,torgersen1991comparison}, and in particular its formulation in terms of matrix majorization \cite{Marshall1979,Dahl1999}.
	In \cref{sec:distinguishability}, we aim to illuminate the idea of resource theories of distinguishability through this concrete example.
	Additionally, we use it to show how some of the abstract monotone constructions can be used in this context.
	
	Finally, in \cref{sec:monotone_comparison}, we introduce an ordering among monotones that captures how much information they have about the resource ordering. 
	This concept is another tool one may use for the classification of monotone constructions.
	We make partial progress along these lines by proving \cref{thm:yield and cost from more informative functions}. 
	It says: If we think of cost and yield constructions as extending the domain of a partial monotone $f$, then the more informative $f$ is, the more informative the resulting extensions are.
		
	\subsubsection{Reading suggestions and notation.}
	
		\Cref{sec:math_prelim} introduces mathematical structures used throughout the rest of the thesis.
		It is fairly dense with definitions and lemmas whose motivation often becomes apparent only in later chapters.
		As such, the reader may consider continuing with \cref{sec:resource_theories} instead and referring to \cref{sec:math_prelim} when in need of clarification only.
		
		As a final point before we delve into the details, let us make a remark on notation that we use.
		While it may be broken on occassion, the general rule is that lowercase latin letters 
			denote individual resources or transformations, uppercase latin letters 
			denote sets of resources or objects with similar interpretation, and uppercase caligraphic letters 
			represent collections of sets of resources. 
		Keeping this in mind may be helpful when we refer to an object like $S$ as a resource, even though in most straightforward interpretations of the formalism, one would think of such an $S$ as a \emph{set of resources}.
		This convention carries over to generalizations, such as when we replace power set lattices with more general lattices.
		
		Similarly, $\subseteq$ and $\supseteq$ denote set inclusion as a relation on the power set lattices, but also the ordering relation for general lattices.
		In the same vein, $\cup$ and $\cap$ represent set union and intersection respectively, just as well as suprema and infima in a generic lattice.
		
		Order relations that are interpreted as resource orderings are usually denoted by $\succeq$.
		When referring to a generic order relation---one which does not carry an interpretation as an ordering of resources---we tend to use the symbol $\geq$, which is not meant to suggest that they have anything to do with numbers a priori.
		
		\Cref{tab:notation} displays further conventions that we follow concerning the use of different Latin letters in particular.
		
		\bgroup
		\renewcommand{\arraystretch}{1.2}%
		\begin{table}[]
		\centering
			\begin{tabular}
			{m{0.1\textwidth}|>{\raggedright}m{0.25\textwidth}>{\arraybackslash}m{0.65\textwidth}}
				\hline
				\textbf{Symbol} 		&  	\textbf{Meaning(s)}  &  	\textbf{Interpretation(s)}  										\\
				\hline
				$r$, $x$, \dots        	&  	Element of a set  	& 	Individual resource or a transformation						\\
				\ditto				   	&  	Atom of a lattice  	& 	\ditto				 												\\ 
				\hline
				$R$        				& 	A set					& 	Set of all resources in a UCRT 									\\
				$S$, $T$  			& 	\ditto					& 	Set of resources in a UCRT 										\\
				$R_{\rm free}$		& 	\ditto					& 	Set of free resources in a UCRT 								\\
				$X$, $W$        		& 	\ditto					& 	Set of all resources in a concrete resource theory 			\\
				$Y$, $Z$		 		& 	\ditto					& 	Set of resources in a resource theory 							\\
				$Y$, $Z$, $X$  					& 	Suplattice element	& 	\ditto 																\\
				$A$ 					& 	A set					& 	Domain of definition of a partial function 						\\
				$T$, $U$        		& 	\ditto					& 	Set of all transformations in a concrete resource theory 	\\
				$S$, $U$		    		& 	\ditto					& 	Set of transformations in a resource theory 					\\
				$S$, $T$, $U$  		& 	Suplattice element	& 	\ditto 																\\
				$T_{\rm free}$		& 	A set					& 	Set of free transformations in a resource theory 				\\
				$D$ 					& 	\ditto					& 	A downset of resources 											\\
				$D$					& 	\ditto					& 	A right invariant set of transformations 						\\
				$S$ 					& 	\ditto					& 	A left invariant set of transformations							\\ 
				\hline
				$\mathcal{R}$ 		& 	A quantale			& 	Suplattice of resources in a UCRT								\\
				$\mathcal{P}(R)$ 	& 	A power set			& 	Suplattice of resources in a concrete UCRT					\\
				$\boxtimes$ 			& 	commutative quantale operation	& 	Combination of resources in a UCRT			\\
				$\mathcal{X}$ 		& 	A suplattice			& 	Suplattice of resources and their disjunctions					\\
				$\mathcal{P}(X)$ 	& 	A power set			& 	Suplattice of resources in a concrete resource theory		\\
				$\mathcal{T}$ 		& 	A quantale			& 	Quantale of transformations and their disjunctions			\\
				$\mathcal{P}(T)$ 	& 	A power set			& 	Quantale of transformations in a concrete resource theory	\\
				$\star$				&	quantale operation	&	Combination of transformations in a resource theory			\\
				$\apply$				&	quantale action		&	Action of transformations on resources						\\
				\hline
			\end{tabular}
			\caption[Notation Overview]{Notational conventions for commonly used symbols, with interpretations in the context of resource theories as presented in sections \cref{sec:resource_theories,sec:monotones}.
			The abbreviation `UCRT' refers to a `universally combinable resource theory'.}
			\label{tab:notation}
		\end{table}
		\egroup

\chapter{Mathematical Preliminaries}\label{sec:math_prelim}

	In this chapter, we introduce various more or less well-known mathematical structures and concepts.
	Some of them, particularly those in \cref{sec:ordered_lattice}, may be new given the not-quite-standard interpretation we propose here.
	However, it could very well be the case that they have appeared elsewhere unbeknownst to us, possibly disguised under a different language.
	The purpose of this chapter is to make the thesis reasonably self-contained and to act as a resource for the reader to look up relevant definitions whose notational conventions here tend to match those in the rest of the manuscript.
	
	\section{Process Theories}\label{sec:proc_th}
	
		We start with a concept that is not going to be relevant for any of the technical details we present in this thesis.
		However, process theories are tied with resource theories in multiple ways.
		Here, they serve to provide examples that elucidate the intended meaning of several of our definitions. 
		They form an important class of examples that connect our framework to the practice of resource theories because process theories are often used as a set-up therefor. 
		In large part, this is the case thanks to the influential framework for resource theories in terms of \emph{partitioned} process theories \cite{Coecke2016}.
		Another point of contact is a result of many resource theories originating in the community of quantum information theory, which has assimilated the process-theoretic perspective as exhibited in \cite{coecke2018picturing}.
		More fundamentally, one could say that process theories and resource theories also share some of their conceptual underpinnings, namely the adherence to pragmatism, structural realism,\footnotemark{} and process philosophy \cite{sep-process-philosophy}.
		\footnotetext{In the case of resource theories, structural perspective enters because resources attain their attributes only in relation to other resources and the choice of free manipulations.
		Properties of an object, which its value as a resource, are not intrinsic.}
		
		We avoid the category-theoretic technicalities behind the scenes and introduce process theories in a diagrammatic language \cite{selinger2010survey}, following the approach of \cite{coecke2015categorical}.
		A process $t$ is thus denoted by a box:
		\begin{equation}\label{eq:process}
			\input{process.tikz}
		\end{equation}
		Wires on the bottom represent incoming systems, while those on top are outgoing systems.
		Moreover, wires carry labels that identify the type of the system in question.
		We generally do not need to consider wire types in this manuscript and thus leave the labels implicit.

		Such processes can be \emph{contracted} by joining a pair of an incoming and outgoing wire (of maching type) to form diagrams such as
		\begin{equation}\label{eq:diagram}
			\input{diagram.tikz}	
		\end{equation}
		These embody the principle that ``only connectivity matters'' in the sense that moving boxes and wires around on the page does not change the identity of a diagram, as long as the connectivity is preserved.
		Furthermore, we also require that there are no cycles in the admissible diagrams \cite[definition 2.3]{coecke2015categorical}, i.e.\ we restrict our attention to \emph{circuit} diagrams.
		
		There are two special kinds of wirings, $\otimes$ and $\circ$, which carry the interpretation of parallel and sequential composition respectively.
		For example, we have
		\begin{equation}\label{eq:parallel_sequential}
			\input{parallel_sequential.tikz}
		\end{equation}
		and, assuming that the output wire type of $s$ matches the input wire type of $t$, we also have
		\begin{equation}\label{eq:parallel_sequential_2}
			\input{parallel_sequential_2.tikz}
		\end{equation}
		These two operations generate all admissible diagrams and thus can be seen as building blocks of process contractions.
		With these in mind, a \textbf{process theory} is a collection of processes such that every admissible wiring of processes corresponds to another process in the theory.
		
		In this thesis, we often refer to the concept of a partitioned process theory from \cite{Coecke2016}.
		A \textbf{partitioned process theory} is nothing but a process theory with a distinguished subset of so-called free processes.
		The key requirement is that the subset of free processes forms a process theory itself, i.e.\ it should be closed under both parallel and sequential compositions.
		
	\section{Ordered Sets}\label{sec:order}
	
		One of the central concepts used in resource theories is that of an ordered set.
		We use it in two ways with quite distinct interpretations.
		Some order relations refer to the fact of a resource being more valuable than another one (such as the resource ordering $\freeconv$ on $X$).
		Others describe the fact that having access to a given set of resources may subsume the access to another set of resources.
		
		\begin{definition}\label{def:preorder}
			A relation $\geq$ on a set $X$ is said to be a \textbf{preorder} if it satisfies
			\begin{enumerate}[series=order]
				\item Transitivity: $x \geq y$ and $y \geq z$ implies $x \geq z$.
				\item Reflexivity: $x \geq x$.
			\end{enumerate}
			The pair $(X, \geq)$ is then called a \textbf{preordered set}.
		\end{definition}
		Every preorder generates an equivalence relation indentifying elements that are interchangeable as far as their place in the ordering is concerned.
		That is, we say that $x,y \in X$ are equivalent with respect to $\geq$ if we have
		\begin{equation}
			x \geq y  \qquad \text{ and } \qquad  y \geq x.
		\end{equation}
		In situations in which one is not interested in the identity of elements of $X$ beyond their placement according to $\geq$, we can quotient out the variability within the equivalence classes to obtain the corresponding partial order.
		\begin{definition}\label{def:partial_order}
			A relation $\geq$ on a set $X$ is said to be a \textbf{partial order} if it is a preorder and furthermore satisfies
			\begin{enumerate}[resume=order]
				\item Antisymmetry: $x \geq y$ and $y \geq x$ implies $x = y$.
			\end{enumerate}
			The pair $(X, \geq)$ is then called a \textbf{partially ordered set} (often shortened to \textbf{poset}).
		\end{definition}
		
		Since, in the study of resource theories, we are often interested in the interplay between different orderings, we usually prefer to work with the more basic notion of preorders rather than partial orders.
		\begin{definition}\label{def:total_order}
			A preorder $\geq$ on $X$ is termed a \textbf{total order} if any two elements of $X$ are comparable, i.e.\ if for any $x,y \in X$ we have either
			\begin{equation}
				x \geq y  \qquad \text{ or } \qquad  y \geq x
			\end{equation}
			(or both).
		\end{definition}
		\begin{example}\label{ex:reals}
			A canonical example of a totally ordered set for our purposes is that of (extended) real numbers $\ordreals$.
			The underlying set consists of all real numbers as well as $-\infty$ and $\infty$, which are ordered as expected.
		\end{example}
		Generally speaking, however, ordered sets will have incomparable pairs $(x,y)$, for which neither $x$ is above $y$ nor is $y$ above $x$, according to $\geq$.
		
		\begin{definition}\label{def:isotone}
			Given preordered sets $(X,\geq)$ and $(Y,\geq)$, a function $f \colon X \to Y$ is said to be an \textbf{isotone} if it satisfies
			\begin{equation}\label{eq:isotone}
				x \geq x' \quad \implies  \quad  f(x) \geq f(x')
			\end{equation}
			for all $x,x' \in X$.
			Similarly, a partial function $f \colon X \to Y$ is termed a \textbf{partial isotone} if it satisfies \eqref{eq:isotone} for all $x$ and $x'$ within its domain.
		\end{definition}
		That is, isotones are order-preserving functions.
		Traditionally, such functions are also referred to as monotones.
		However, we prefer the term isotone.
		It allows us to distinguish a specific class of order-preserving functions as follows.
		\begin{definition}\label{def:monotone}
			An order-preserving function $f \colon X \to Y$ is called a \textbf{monotone} if $(Y,\geq)$ is a totally ordered set.
		\end{definition}
		Isotones are the most basic structure-preserving maps in the theory of ordered sets.
		While they may coarse-grain the information in their domain, we can still recover some of the original structure from the codomain.
		Specifically, for an isotone $f \colon X \to Y$, the contrapositive statement of \eqref{eq:isotone} is
		\begin{equation}
			 f(x) \not\geq f(x')  \quad \implies \quad  x \not\geq x'
		\end{equation}
		That is, lack of ordering in the image of $f$ lets one learn about the lack of ordering in the domain.
		
		Another useful concept is that of subsets of $X$ which include all the points below their elements---downsets---and their dual versions---upsets.
		\begin{definition}\label{def:downset}
			Let $(X, \geq)$ be a preordered set.
			A subset $D$ thereof is called a \textbf{downward closed set} (or downset for short) if for all $x \in X$ and $d \in D$ the implication
			\begin{equation}
				 d \geq x  \quad \implies \quad x \in D
			\end{equation}
			holds.
			The set of all downward closed subsets of $X$ is denoted by $\mathcal{DC}(X)$.

			On the other hand, a set $U \subseteq X$ is an \textbf{upward closed set} (or upset for short) if for all $x \in X$ and $u \in U$ the implication
			\begin{equation}
				 x \geq u  \quad \implies \quad  x \in U
			\end{equation}
			holds.
			The set of all upward closed subsets of $X$ is denoted by $\mathcal{UC}(X)$.
	 	\end{definition}
	 	
 		\begin{figure}[!tb]
		\begin{center}
			\begin{subfigure}[b]{.4\textwidth}\centering
				\begin{tikzpicture}[align=center,thick,>=stealth,node distance=\scale and 0.5*\scale,
						resource/.style={circle,fill=black,draw=none},
						freeconv/.style={->,thick}
					]
					\node[resource]	(a)		at (0,0)					{};
					\node[resource]	(1b)		[below left=of a]			{};
					\node[resource]	(b1)		[below right=of a]		{};
					\node[resource]	(2c)		[below left=of 1b]		{};
					\node[resource]	(c)		[below left=of b1]		{};
					\node[resource]	(c2)		[below right=of b1]		{};
					\node[resource]	(3d)		[below left=of 2c]		{};
					\node[resource]	(1d)		[below left=of c]			{};
					\node[resource]	(d1)		[below left=of c2]		{};
					\node[resource]	(d3)		[below right=of c2]		{};
					
					\draw[freeconv]	(a) -- (1b);
					\draw[freeconv]	(a) -- (b1);
					\draw[freeconv]	(1b) -- (2c);
					\draw[freeconv]	(1b) -- (c);
					\draw[freeconv]	(b1) -- (c);
					\draw[freeconv]	(b1) -- (c2);
					\draw[freeconv]	(2c) -- (3d);
					\draw[freeconv]	(2c) -- (1d);
					\draw[freeconv]	(c) -- (1d);
					\draw[freeconv]	(c) -- (d1);
					\draw[freeconv]	(c2) -- (d1);
					\draw[freeconv]	(c2) -- (d3);
					
					\node[text=white]			at (b1)								{$r$};
					
					\begin{pgfonlayer}{background}
						\foreach \nodename in {a,1b,b1,2c,c,c2,3d,1d,d1,d3} {\coordinate (\nodename') at (\nodename);} 
					
						\path[fill=SeaGreen,draw=Emerald,line width=1.17*\scale,line cap=round,line join=round]  (b1') to (1d') to (d3') to (b1') -- cycle;
						\path[fill=SeaGreen,draw=SeaGreen,line width=1.06*\scale,line cap=round,line join=round]  (b1') to (1d') to (d3') to (b1') -- cycle;
						
						\path[fill=Bittersweet!80,draw=Brown,line width=0.94*\scale,line cap=round,line join=round]  (c2') to (d1') to (1d') to (2c') to (3d') to (d3') to (c2') -- cycle;
						\path[fill=Bittersweet!80,draw=Bittersweet!80,line width=0.82*\scale,line cap=round,line join=round]  (c2') to (d1') to (1d') to (2c') to (3d') to (d3') to (c2') -- cycle;
					\end{pgfonlayer}
				\end{tikzpicture}
				\caption{Example of a downset (brown region) and the downward closure $\down(r)$ of a particular resource $r$ (turquoise) in a simple preordered set.}\label{fig:DC}
			\end{subfigure}\hspace{0.05\textwidth}
			\begin{subfigure}[b]{.4\textwidth}\centering
				\centering
				\begin{tikzpicture}[align=center,thick,>=stealth,node distance=\scale and 0.5*\scale,
						resource/.style={circle,fill=black,draw=none},
						freeconv/.style={->,thick}
					]
					\node[resource]	(a)		at (0,0)					{};
					\node[resource]	(1b)		[below left=of a]			{};
					\node[resource]	(b1)		[below right=of a]		{};
					\node[resource]	(2c)		[below left=of 1b]		{};
					\node[resource]	(c)		[below left=of b1]		{};
					\node[resource]	(c2)		[below right=of b1]		{};
					\node[resource]	(3d)		[below left=of 2c]		{};
					\node[resource]	(1d)		[below left=of c]			{};
					\node[resource]	(d1)		[below left=of c2]		{};
					\node[resource]	(d3)		[below right=of c2]		{};
					
					\draw[freeconv]	(a) -- (1b);
					\draw[freeconv]	(a) -- (b1);
					\draw[freeconv]	(1b) -- (2c);
					\draw[freeconv]	(1b) -- (c);
					\draw[freeconv]	(b1) -- (c);
					\draw[freeconv]	(b1) -- (c2);
					\draw[freeconv]	(2c) -- (3d);
					\draw[freeconv]	(2c) -- (1d);
					\draw[freeconv]	(c) -- (1d);
					\draw[freeconv]	(c) -- (d1);
					\draw[freeconv]	(c2) -- (d1);
					\draw[freeconv]	(c2) -- (d3);
					
					\node[text=white]			at (b1)								{$r$};
					
					\begin{pgfonlayer}{background}
						\foreach \nodename in {a,1b,b1,2c,c,c2,3d,1d,d1,d3} {\coordinate (\nodename') at (\nodename);} 
						
						\path[fill=Bittersweet!80,draw=Brown,line width=1.17*\scale,line cap=round,line join=round]  (c2') to (a') to (2c') to (1b') to (b1') to (c2') -- cycle;
						\path[fill=Bittersweet!80,draw=Bittersweet!80,line width=1.06*\scale,line cap=round,line join=round]  (c2') to (a') to (2c') to (1b') to (b1') to (c2') -- cycle;
						
						\path[fill=SeaGreen,draw=Emerald,line width=0.94*\scale,line cap=round,line join=round]  (b1') to (a') to (b1') -- cycle;
						\path[fill=SeaGreen,draw=SeaGreen,line width=0.82*\scale,line cap=round,line join=round]  (b1') to (a') to (b1') -- cycle;
					\end{pgfonlayer}
				\end{tikzpicture}
				\caption[Example of Upsets and Downsets]{Example of an upset (brown) and the upward closure $\up(r)$ of an $r \in R$ (turquoise) in a simple preordered set.\\}\label{fig:UC}
			\end{subfigure}
		\end{center}
		\end{figure}
		
	 	\begin{definition}\label{def:order_closure}
	 		Given a preordered set $(X, \geq)$, the \textbf{downward closure} $\down \colon X \to \mathcal{DC}(X)$ is defined by\footnotemark{}
	 		\footnotetext{The set $\down(x)$ is also called the principal ideal of $x$.}%
			\begin{equation}
				\label{eq:downward closure}
				\down (x) = \Set{ y \in X  \given x \geq y}.
			\end{equation}
			Similarly, the \textbf{upward closure} $\up \colon X \to \mathcal{UC}(X)$ is defined by
			\begin{equation}
				\up (x) = \Set{ y \in X  \given  y \geq x}.
			\end{equation}
		\end{definition}
		
		When we intersect the downward closure of some $x$ with the upward closure of some $y$, we get all the elements that lie in between the two.
		These subsets are called \textbf{intervals} and we denote them as
		\begin{equation}
			\langle y, x \rangle \coloneqq \down(x) \cap \up(y).
		\end{equation}

		Notice that both $\mathcal{DC}(X)$ and $\mathcal{UC}(X)$ have a natural ordering in terms of subset inclusion which makes $\up$ and $\down$ into order-preserving maps.  
		In particular, it is easy to see that the two maps of preordered sets
		\begin{align}
			\down &\colon (X,\geq) \to (\mathcal{DC}(X), \supseteq)  &  \up &\colon (X,\geq) \to (\mathcal{UC}(X), \subseteq)
		\end{align}
		are both isotones.
		
		Downward and upward closures can also be extended to act on sets of resources by requiring compatibility with unions.
		That is, for any $Z \in \mathcal{P}(X)$ we define 
		\begin{align}\label{eq:down_on_sets}
			\down (Z) &\coloneqq \bigcup_{z \in Z} \down (z) 	& 	\up (Z) &\coloneqq \bigcup_{z \in Z} \up(z).
		\end{align}
		
		Thus, a set $D$ is downward closed if and only if $\down (D) = D$ holds, while the fact that a set $U$ is upward closed can be stated as $\up (U) = U$.

	\section{Lattices}\label{sec:lattice}
	
		Preorders offer a suitable level of generality to describe the relationship between different resources and their value.
		Even though resource orderings that arise in practice often come with extra structure or properties, a great deal of the general theory can be laid out in the abstract, as we do in this thesis.
		
		However, when it comes to order relations representing degrees of possibility, it is useful to make some further assumptions.
		For example, notice that the power set $\mathcal{P}(X)$ ordered by set inclusion has a bottom element (the empty set) and a top element (the full set $X$).
		Moreover, any two elements have a least upper bound (their union) as well as a greatest lower bound (their intersection).
		Least upper bounds and greatest lower bounds exist for \emph{any} collection of elements of $\mathcal{P}(X)$.
		In other words, $(\mathcal{P}(X),\supseteq)$ is a complete lattice.
		
		\begin{definition}\label{def:sup_inf}
			Let $(\mathcal{X}, \supseteq)$ be a poset.
			Given a subset $\mathcal{Z}$ of $\mathcal{X}$, consider the collections of all upper bounds and all lower bounds thereof respectively:
			\begin{equation}
				\begin{split}
					\mathcal{U}(\mathcal{Z}) &\coloneqq \Set*[\big]{  X \in \mathcal{X}  \given  \forall \, Z \in \mathcal{Z} \; : \; X \supseteq Z  } \\
					\mathcal{L}(\mathcal{Z}) &\coloneqq \Set*[\big]{  X \in \mathcal{X}  \given  \forall \, Z \in \mathcal{Z} \; : \; X \subseteq Z  }
				\end{split}
			\end{equation}
			If $\mathcal{U}(\mathcal{Z})$ has a bottom element (the least upper bound), we term it the \textbf{supremum} of $\mathcal{Z}$, denoted by $\bigcup \mathcal{Z}$.
			Similarly, if $\mathcal{L}(\mathcal{Z})$ has a top element (the greatest lower bound), we term it the \textbf{infimum} of $\mathcal{Z}$, denoted by $\bigcap \mathcal{Z}$.
		\end{definition}
		A complete lattice is a poset in which every subset\footnotemark{} has both a supremum and an infimum.
		\footnotetext{Note that a subset $\mathcal{Z}$ of $\mathcal{X} = \mathcal{P}(X)$ is a \emph{collection of subsets} of $X$.}%
		Clearly, a complete lattice has both a top element and a bottom element, as they arise as suprema and infima given the choice of $\mathcal{Z} = \mathcal{X}$ in \cref{def:sup_inf}:
		\begin{equation}
			\top = \bigcup \mathcal{X}  \qquad \qquad  \bot = \bigcap \mathcal{X}
		\end{equation}
		By convention, the top element is also the infimum of the empty subset of $\mathcal{X}$, i.e.\ $\mathcal{Z} = \emptyset$, and similarly for the dual notion:
		\begin{equation}\label{eq:empty_bound}
			\top = \bigcap \emptyset  \qquad \qquad  \bot = \bigcup \emptyset
		\end{equation}
		
		Since any infimum can be expressed as a supremum of the set of lower bounds, it suffices to assume the existence of all suprema to obtain a complete lattice.
		A poset that has all suprema is called a suplattice (see \cref{def:suplattice}) and this observation means that any given poset is a complete lattice if and only if it is a (bounded) suplattice.
		Nevertheless, the two concepts are not quite interchangeable.
		This is because maps that preserve suprema \emph{do not} coincide with those that preserve both suprema and infima.
		In other words, the homomorphisms of suplattices differ from the homomorphisms of complete lattices.
		In our context, we think of suprema as being the fundamental part of the structure we are interested in, and infima as being a derived notion.
		That is why we use suplattices instead of complete lattices.
		\begin{definition}\label{def:suplattice}
			A poset $\mathcal{X}$ is a (bounded)\footnotemark{} \textbf{suplattice} if it has a minimal element $\bot$ and every subset $\mathcal{Z}$ of $\mathcal{X}$ has a supremum $\bigcup \mathcal{Z} \in \mathcal{X}$.
			We thus denote a suplattice by the triple $(\mathcal{X}, \bigcup, \bot)$ or just $\mathcal{X}$ whenever the other structures are clear from the context.
			\footnotetext{We will drop the explicit mention, but throughout we assume every suplattice to be bounded.
			In particular, this means that we require morphisms of suplattices to preserve the bottom element.}%
		\end{definition}
		
		The traditional perspective in mathematical logic is to interpret the relation $Y \supseteq Z$ in a lattice as entailment ($Z$ entails $Y$) and the operations $Y \cup Z$, $Y \cap Z$ as predicate disjunction and conjunction respectively.
		In this thesis we take a different point of view.
		Namely, we also interpret $Y \supseteq Z$ as a kind of entailment, but in the \emph{opposite} direction as follows.
		\begin{quote}
			$Y \supseteq Z$ means: If an agent has access to anything in $Y$, then they also have access to anything $Z$.
		\end{quote}
		The supremum $Y \cup Z$ then represents the joint access to anything in $Y$ or in $Z$, while the infimum $Y \cap Z$ represents access to whatever $Y$ and $Z$ have in common.
		For example, in the case of a power set lattice, $Y$ and $Z$ are sets of resources whose common elements are given by their intersection.
		That is, 
		\begin{quote}
			having access to $Y \in \mathcal{P}(X)$ means: An agent may use one (and only one) of the elements of $Y$, the choice of which is theirs to make.
		\end{quote}
		Of course, there are other interpretations that the mathematical structures allow for, but the reader should keep in mind that many definitions are introduced and justified with this one in mind.
		
		It is worth noting that in works that have used lattices in a resource-theoretic context in the past, such as \cite{DelRio2015,Sadrzadeh2006}, the lattice structure is introduced to represent \emph{knowledge} about resources. 
		That is, the relation $Y \supseteq Z$ indicates that knowing $Z$ entails knowing $Y$.
		Given such interpretation, $Z$ ought to be considered a ``better resource'' than $Y$ in the sense that it represents an agent with a higher amount of information about the resource in question.
		On the contrary, in our case $Y \supseteq Z$ implies that $Y$ is a ``better resource'' than $Z$ as it gives an agent greater (or equal) access.
		In other words, an agent having access to resources in $Y$ can choose from a wider range of actions to achieve their goal than an agent having access to resources in $Z$.
		
		\begin{definition}\label{def:lat_morphism}
			Given suplattices $(\mathcal{X}, \bigcup, \bot)$ and $(\mathcal{Y}, \bigvee, \bot)$, a function $f \colon \mathcal{X} \to \mathcal{Y}$ is a \textbf{suplattice homomorphism} if it satisfies 
			\begin{equation}\label{eq:lat_morphism}
				f \left( \bigcup \mathcal{Z} \right) = \bigvee_{Z \in \mathcal{Z}} f(Z)
			\end{equation}
			for all subsets $\mathcal{Z}$ of $\mathcal{X}$.
		\end{definition}
		Choosing $\mathcal{Z}$ to be the empty set shows that a suplattice homomorphism preserves the bottom element, i.e. $f(\bot) = \bot$.
		Necessarily, they are also order-preserving, which can be derived as follows:
		\begin{equation}\label{eq:hom_pres_ord}
			\begin{split}
				X_1 \supseteq X_2 \quad 	&\implies \quad  f(X_1) = f \left( X_1 \cup X_2 \right) = f(X_1) \vee f(X_2) \\
															&\implies \quad  f(X_1) \supseteq f(X_2)
			\end{split}
		\end{equation}
		However, suplattice homomorphisms preserve neither infima nor the top element in general.
		
		The power set lattice has many special properties, among which is the special role of singletons.
		They are the minimal elements if we exclude the bottom element $\emptyset$ and every element of the power set can be expressed as the union of a collection of singletons.
		Moreover, this collection is unique.
		In general terms, we can phrase these properties as follows.
		\begin{definition}\label{def:atom} 
			Let $(\mathcal{X}, \bigcup, \bot)$ be a suplattice with the partial order denoted by $\supseteq$.
			An element $X \in \mathcal{X}$ distinct from $\bot$ is an \textbf{atom} if for all $Y \in \mathcal{X}$, $X \supseteq Y$ implies $Y \in \{\bot,X\}$.
			The set of all atoms of $\mathcal{X}$ is denoted by $\atoms(\mathcal{X})$.
			$\mathcal{X}$ is said to be an \textbf{atomistic lattice} if every element thereof is the supremum of some subset of its atoms. 
			That is, for each $X \in \mathcal{X}$ there ought to be a set $A_X \subseteq \atoms(\mathcal{X})$ such that we have
			\begin{equation}
				X = \bigcup A_X.
			\end{equation}
			Finally, $\mathcal{X}$ is said to be \textbf{uniquely atomistic} if the set $A_X \subseteq \atoms(\mathcal{X})$ is unique for \mbox{each $X$}.
		\end{definition}
		As an example of an atomistic lattice, consider the following one that can be viewed as the power set of three elements but for one subset removed:
		\begin{equation}\label{eq:lattice}
			\input{lattice.tikz}
		\end{equation}
		We can see that $\top$ is the supremum of both $\{a_1,a_3\}$ and $\{a_1,a_2,a_3\}$, which is why the lattice is not uniquely atomistic.
		Many lattices are not even atomistic.
		For instance, removing any single atom from the one depicted in diagram \eqref{eq:lattice} leads to one.
		As another example, the ordered set of extended real numbers from \cref{ex:reals} has no atoms at all.
		
		Uniquely atomistic lattices are precisely the ones that can be interpreted as power sets.
		\begin{lemma}\label{lem:UAL=CABA}
			Every uniquely atomistic suplattice $\mathcal{X}$ is isomorphic to the power set lattice $(\mathcal{P}(X), \bigcup, \emptyset)$ where $X \coloneqq \atoms(\mathcal{X})$.
		\end{lemma}
		\begin{proof}
			This argument can be found in~\cite{Wofsey2016} for example.
			Let's consider a function ${\bigcup \colon \mathcal{P}(X) \to \mathcal{X}}$ given by $A \mapsto \bigcup A$.
			It preserves the order relation, suprema, as well as the bottom element.
			Moreover, by the assumption of unique atomisticity, it has a two-sided inverse $Y \mapsto A_Y$ where $A_Y$ denotes the unique set of atoms that ``make up'' $Y$.
		\end{proof}
		By equation \eqref{eq:lat_morphism}, homomorphisms between uniquely atomistic suplattices correspond to functions of type $X \to \mathcal{P}(Y)$ that assign a set of atoms in the codomain to each atom in the domain.
		
		Sometimes we want to use lattices whose suprema and infima ``behave like'' unions and intersections, but assuming the existence of atoms is unnecessarily strong or undesirable.
		A commonly used structure that formalizes this idea is that of a locale (also known as a frame).
		\begin{definition}\label{def:locale}
			A suplattice $(\mathcal{X}, \bigcup, \bot)$ is called a \textbf{locale} if it satisfies the (infinite) distributive law
			\begin{equation}
				X \cap \left( \bigcup \mathcal{Y} \right)  =  \bigcup_{Y \in \mathcal{Y}} (X \cap Y)
			\end{equation}
			for all $X \in \mathcal{X}$ and all subsets $\mathcal{Y}$ of $\mathcal{X}$.
		\end{definition}
		
	\section{Ordered Lattices}\label{sec:ordered_lattice}
	
		In this thesis, we are interested in the interplay between the resource ordering induced by free operations and the lattice of resource disjunctions.
		The structure that arises is that of an ordered lattice.
		\begin{definition}\label{def:ordered_lattice}
			A suplattice $(\mathcal{X}, \bigcup, \bot)$ equipped with a preorder $\geq$ on $\mathcal{X}$ is a \textbf{preordered suplattice} if for all $X,Y \in \mathcal{X}$ and all subsets $\mathcal{Z}$ of $\mathcal{X}$, we have
			\begin{equation}
				X \geq Y \quad \implies \quad X \cup \mathcal{Z} \geq Y \cup \mathcal{Z}.\footnotemark{}
			\end{equation}
			\footnotetext{Note that $X \cup \mathcal{Z}$ denotes the supremum (in $\mathcal{X}$) of the union of $\{ X \}$ and $\mathcal{Z}$.}%
		\end{definition}
		\begin{definition}\label{def:extended_lattice}
			Given a suplattice $(\mathcal{X}, \bigcup, \bot)$, a preorder $\geq$ on $\mathcal{X}$ is said to be an \textbf{extension} of $\supseteq$ if for all $X,Y \in \mathcal{X}$ and all subsets $\mathcal{Z}$ of $\mathcal{X}$
			\begin{enumerate}
				\item \label{it:elim_ord} $X \supseteq Y$ implies $X \geq Y$, as well as
				\item \label{it:sup_ord} $X \geq Z$ for every $Z \in \mathcal{Z}$ implies $X \geq \bigcup \mathcal{Z}$.
			\end{enumerate}
			The pair is then called a \textbf{preorder-extended suplattice}.
		\end{definition}
		In particular, condition \ref{it:sup_ord} says that suprema with respect to $\geq$ (if they exist) are $\geq$-above suprema with respect to $\supseteq$.
		The properties required by \cref{def:extended_lattice} also readily entail
		\begin{align}\label{eq:bot_is_bot}
			X &\geq \bot 	&	&\text{and} 	&	 X \geq Y \quad &\implies \quad X \cup Z \geq Y \cup Z
		\end{align}
		for all $X, Y, Z \in \mathcal{X}$.
		Since the latter can be shown to hold also for arbitrary suprema, it is indeed the case that every preorder-extended suplattice is a preordered suplattice.
		However, the converse is not true and we will see a relevant example of a preordered lattice that is not preorder-extended in \cref{def:deg_ord}.
		
		In this work, we are especially interested in preordered lattices that arise from an underlying preordered set $(X,\freeconv)$ of resources by passing to the power set lattice $\mathcal{P}(X)$.
		These are indeed preorder-extended lattices in the sense of \cref{def:extended_lattice} as we show in \cref{lem:enh_ord_lat}.
		Let us comment on the meaning of their defining properties in this context.
		\begin{remark}
			One might be tempted to interpret condition \ref{it:elim_ord} and its consequence $x \geq \bot$ as saying that discarding resources is a free operation.
			If this were the case, such a structure would be inappropriate for describing costs associated with the destruction of resources, which need not be negligible in many circumstances of interest.
			However, our perspective is different.
			The lattice ordering is interpreted as a relation that describes what an agent has access to, with $\bot$ representing the contrived situation of not having access to anything at all. 
			As such, even though an agent may have no access to $x$, and thus describe their state by $\bot$, this does not mean that $x$ has been destroyed.
			Thus, passing from $x$ to $\bot$ is interpreted as the fact that losing access to $x$ is free.
			If we were interested in modelling externalities and other costs associated with destruction, we would instead have to include an explicit description of a ``clean slate'', distinct from $\bot$, in which all relevant garbage has been dealt with.
		\end{remark}
		\begin{remark}
			Condition \ref{it:sup_ord} has a transparent interpretation if $\mathcal{Z}$ is a collection of atoms 	and $X$ is itself an atom.
			For example, consider $\mathcal{Z}$ consisting of two individual resources $z_1$ and $z_2$. 
			In this case, the condition says that if $x$ can be converted to both $z_1$ and $z_2$ by possibly distinct free conversions $\phi_1$ and $\phi_2$, then $x$ can be converted to $\{z_1\} \cup \{z_2\}$.
			The non-atomic resource $\{z_1, z_2\}$ represents the possibility of using either $z_1$ or $z_2$, the choice of which corresponds to applying either $\phi_1$ or $\phi_2$ to $x$.
		\end{remark}
		
		Given a preorder of individual resources, $(X,\freeconv)$, we may be interested in comparing not just elements of $X$, but also its subsets.
		There is not a canonical way to do so.
		Whether two subsets $Y,Z$ of $(X, \geq)$ should be ordered or not depends on the interpretation of the ordering and set union in the particular situation one is interested in. 
		
		\subsubsection{Enhancement preorder.}
		
		We could say that $Y$ should be above $Z$, denoted by $Y \freeconv_{\rm enh} Z$, if every element of $Z$ lies below an element of $Y$.
		In the following toy example (in which nodes are elements of $X$ and arrows depict order relations)
		\begin{equation}\label{eq:enh_deg}
			\input{enh_deg.tikz}
		\end{equation}
		we would have that $\{y_1,y_2\}$ is above $\{z_3,z_4\}$, but not above $\{z_1,z_2,z_3\}$.
		With respect to the interpretation of the elements of $\mathcal{P}(X)$: 
		\begin{quote}
			$Y \freeconv_{\rm enh} Z$ means: For every element of $Z$ that an agent could make use of, there is an element of $Y$ which is at least as valuable.
		\end{quote}
		This notion of ordering of subsets of $(X,\geq)$ can be expressed via the existence of an enhancement map, defined as follows.
		\begin{definition}\label{def:enh}
			Let $(X,\geq)$ be a preordered set with two subsets $Y,Z$.
			A function $\mathsf{enh} \colon Z \to Y$ is termed an \textbf{enhancement} if we have
			\begin{equation}
				\mathsf{enh}(z) \geq z \quad \forall \, z \in Z.
			\end{equation}
		\end{definition}	
		
		\begin{definition}\label{def:enh_ord}
			Given a preordered set $(X,\freeconv)$, define the \textbf{enhancement preorder} $\enhconv$ on $\mathcal{P}(X)$ by
			\begin{equation}
				Y \enhconv Z  \quad \coloniff \quad  \text{there exists an enhancement } Z \to Y.
			\end{equation}
			whenever $Z$ is non-empty and
			\begin{equation}
				Y \enhconv \emptyset
			\end{equation}
			for all $Y \in \mathcal{P}(X)$.
		\end{definition}
		
		Given $Y \enhconv Z$, access to resources in $Y$ can be reduced to the access to resources in $Z$ by ignoring elements of $Y$ outside the image of $\mathsf{enh}$.
		In this way, we can obtain a preordered lattice $(\mathcal{P}(X), \enhconv)$ from a preordered set $(X, \freeconv)$, such that the latter is isomorphic to the ordering of the atoms in $\mathcal{P}(X)$.
		The same construction can be used in any atomistic lattice given a preorder among its atoms, as well as any strongly coherent lattice \cite[definition 5.3]{Tsinakis2004} given a preorder among its completely join prime elements.
		As we show next, the enhancement ordering constitutes a consistent extension of the subset inclusion.
		
		\begin{lemma}\label{lem:enh_ord_lat}
			The enhancement preorder makes $(\mathcal{P}(X), \supseteq)$ into a preorder-extended suplattice.
		\end{lemma}
		\begin{proof}
			Indeed, if $Y$ is a superset of $Z$, then the inclusion of $Z$ within $Y$ is an enhancement.
			Moreover, if there is an enhancement $Z_i \to Y$ for every $Z_i$ in a family $\mathcal{Z}$ of subsets of $X$, then there also exists an enhancement 
			\begin{equation}
				\bigcup_i Z_i  \to  Y,
			\end{equation}
			which gives condition \ref{it:sup_ord}.
		\end{proof}
		
		The enhancement ordering is not the only consistent extension of $(X,\freeconv)$.
		However, it is the one that postulates ``fewest'' relations among the elements of $\mathcal{P}(X)$.
		
		\begin{lemma}\label{prop:enh_coarse}
			Let $\mathcal{X}$ be an atomistic, preorder-extended suplattice and denote by $(\mathcal{P}(\atoms(\mathcal{X})),\enhgeq)$ the enhancement preorder induced by the restriction of $\geq$ to the atoms of $\mathcal{X}$.
			For any $Y,Z \in \mathcal{X}$ we have
			\begin{equation}\label{eq:enh_coarse}
				A_Y \enhgeq A_Z  \quad \implies \quad  Y \geq Z.
			\end{equation}
			where $A_Y$ and $A_Z$ denote sets of atoms that make up $Y$ and $Z$ respectively.
		\end{lemma}
		\begin{proof}
			The antecedent stipulates that there is an enhancement $\mathsf{enh} \colon A_Z \to A_Y$.
			Thus, for every atom $z \in A_Z$, there is $\mathsf{enh}(z) \in A_Y$ which is above it according to the preorder $\geq$.
			Since $Y$ is an upper bound of $A_Y$, we have
			\begin{equation}
				Y \supseteq \mathsf{enh}(z) \geq z.
			\end{equation}
			Thus, by condition \ref{it:elim_ord}, $Y \geq z$ holds and by condition \ref{it:sup_ord} we obtain $Y \geq \bigcup A_Z  = Z$.
		\end{proof}
		
		In order to make the enhancement ordering the unique extension of a preorder among the atoms, one would have to further assume that the implication
		\begin{equation}\label{eq:fragmentable}
			\bigcup \mathcal{Y} \geq z  \quad \implies \quad  \exists \, Y \in \mathcal{Y} \; : \; Y \geq z
		\end{equation}
		holds for all $z \in \atoms(\mathcal{X})$ and all subsets $\mathcal{Y}$ of $\atoms(\mathcal{X})$.
		This is automatically satisfied if we replace atoms by completely join prime elements and atomistic lattices by strongly coherent ones, in which case \eqref{eq:enh_coarse} is therefore an equivalence.
		
		The enhancement preorder can be equivalently expressed in terms of the downward closure operator. 
		
		\begin{lemma}\label{lem:enh_down}
			Let $Y,Z$ be two subsets of a preordered set $(X,\freeconv)$.
			Then we have
			\begin{equation}
				Y \enhconv Z  \quad \iff \quad  \down (Y) \supseteq \down (Z)
			\end{equation}
			where the action of $\down$ is given as in \eqref{eq:down_on_sets}.
		\end{lemma}
		\begin{proof}
			If $Y \enhconv Z$ holds, then there is an enhancement $\mathsf{enh} \colon Z \to Y$.
			Denoting its image within $Y$ by $\mathsf{enh}(Z)$, we have
			\begin{equation}
				\down (Y) \supseteq \down \bigl( \mathsf{enh}(Z) \bigr) = \down \down \bigl( \mathsf{enh}(Z) \bigr) \supseteq \down (Z),
			\end{equation}
			since $\down \bigl( \mathsf{enh}(Z) \bigr) \supseteq Z$ follows from the definition of an enhancement.
			
			Conversely, if $\down (Y) \supseteq \down (Z)$ holds, then we have $\down (Y) \supseteq Z$.
			That is, for every $z \in Z$, there exists some $y_z \in Y$ such that $y_z \freeconv z$.
			Thus, the function given by $z \mapsto y_z$ is an enhancement of type $Z \to Y$.
		\end{proof}
		
		\begin{corollary}\label{lem:down_isotone}
			The map $\down \colon (\mathcal{P}(X), \enhconv) \to (\mathcal{P}(X), \supseteq)$ is an isotone.
		\end{corollary}
		
		\subsubsection{Degradation preorder.}
		
		On the other hand, we could also say that $Y$ should be above $Z$, denoted by $Y \freeconv_{\rm deg} Z$, if every element of $Y$ lies above an element of $Z$.
		Then we would have that $\{y_1,y_2\}$ is above $\{z_1,z_2,z_3\}$, but not above $\{z_3,z_4\}$ in diagram \eqref{eq:enh_deg}.
		\begin{definition}\label{def:deg}
			Let $(X,\geq)$ be a preordered set with two subsets $Y,Z$.
			A function $\mathsf{deg} \colon Y \to Z$ is termed a \textbf{degradation} if we have
			\begin{equation}
				y \geq \mathsf{deg}(y) \quad \forall \, y \in Y.
			\end{equation}
		\end{definition}

		\begin{definition}\label{def:deg_ord}
			Given a preordered set $(X, \freeconv)$, define the \textbf{degradation preorder} $\degconv$ on $\mathcal{P}(X)$ by
			\begin{equation}
				Y \degconv Z  \quad \coloniff \quad  \text{there exists a degradation } Y \to Z.
			\end{equation}
			whenever $Y$ is non-empty and
			\begin{equation}
				\emptyset \degconv Z
			\end{equation}
			for all $Z \in \mathcal{P}(X)$.
		\end{definition}
		
		In other words, the enhancement and degradation can be also expressed explicitly as follows.
		\begin{equation}
			\begin{split}
				Y \enhconv Z  \quad &\iff \quad  \forall \, z \in Z \text{, } \exists \, y \in Y \text{ such that } y \succeq z \\
				Y \degconv Z  \quad &\iff \quad  \forall \, y \in Y \text{, } \exists \, z \in Z \text{ such that } y \succeq z
			\end{split}
		\end{equation}
		
		It is easy to see that $(\mathcal{P}(X), \supseteq)$ equipped with $\degconv$ is a preordered suplattice.
		However, degradation ordering does not necessarily have a meaningful resource-theoretic interpretation and does not yield a preorder-extended suplattice.
		Therefore, we do not think of it as expressing that $Y$ is more valuable than $Z$, even though such an interpretation may be viable in specific contexts, one of which is as follows. 
		
		Imagine that Alice and Bob play the following game. 
		Alice has to choose a set of resources $Y \subseteq X$, while Bob receives a resource $x \in X$ from the referee.
		If Bob can recover an element of $Y$ from $x$ (for free), he wins.
		Otherwise, Alice wins.
		The relation $Y \degconv Z$ means that $Y$ is not worse than $Z$ from Alice's point of view, for any distribution of referee's choices.
				
		Whatever the interpretation, the preordered set $(\mathcal{P}(X), \degconv)$ will be useful when we study resource monotones in \cref{sec:monotones}.
		One way to understand it is via a dual version of \cref{lem:enh_down}.
		
		\begin{lemma}\label{lem:deg_up}
			Let $Y,Z$ be two subsets of a preordered set $\ordres$.
			Then we have
			\begin{equation}
				Y \degconv Z  \quad \iff \quad  \up (Y) \subseteq \up (Z)
			\end{equation}
			where the action of $\up$ is given as in \eqref{eq:down_on_sets}.
		\end{lemma}
		\begin{proof}
			If $Y \degconv Z$ holds, then there is a degradation $\mathsf{deg} \colon Y \to Z$.
			Denoting its image within $Z$ by $\mathsf{deg}(Y)$, we have
			\begin{equation}
				\up (Z) \supseteq \up \bigl( \mathsf{deg}(Y) \bigr) = \up \up \bigl( \mathsf{deg}(Y) \bigr) \supseteq \up (Y),
			\end{equation}
			since $\up \bigl( \mathsf{deg}(Y) \bigr) \supseteq Y$ follows from the definition of a degradation.
			
			Conversely, if $\up (Y) \subseteq \up (Z)$ holds, then we also have $Y \subseteq \up (Z)$.
			That is, for every $y \in Y$, there exists some $z_y \in Z$ such that $y \freeconv z_y$.
			Thus, the function given by $y \mapsto z_y$ is a degradation of type $Y \to Z$.
		\end{proof}
		
		\begin{corollary}\label{lem:up_isotone}
			The map $\up \colon (\mathcal{P}(X), \degconv) \to (\mathcal{P}(X), \subseteq)$\footnotemark{} is an isotone.
			\footnotetext{Note the opposite direction of the ordering in the codomain relative to prior occurences.
			If we were to keep the convention from earlier, we would instead say that $\up$ is an antitone---an order-reversing function.}
		\end{corollary}
		
		\subsubsection{Isotones for enhancement and degradation preorders.}
		
		Another way to get some intuition about the enhancement and degradation preorders is to look at their form when constructed from total orders.
		\begin{example}\label{ex:enh_reals}
			Specifically, the enhancement preorder for sets of extended real numbers is given by the comparison of their suprema.
			That is, for $Y, Z \in \mathcal{P}(\reals)$ we have 
			\begin{equation}
				Y \enhgeq Z \quad \iff \quad  \sup Y \geq \sup Z.
			\end{equation}
			On the other hand, the degradation preorder is given by infima.
			In this case, both $\enhgeq$ and $\deggeq$ are thus total preorders.
		\end{example}
		As an immediate consequence of this example, we conclude that 
		\begin{equation}
			\begin{split}
				\mathrm{sup} &\colon \bigl( \mathcal{P}(\reals), \enhgeq \bigr) \to \bigl(\reals , \geq \bigr) \\
				\mathrm{inf} &\colon \bigl( \mathcal{P}(\reals), \deggeq \bigr) \to \bigl(\reals , \geq \bigr)
			\end{split}
		\end{equation}
		are both monotones.
		
		\begin{lemma}\label{lem:sup_isotone}
			Given a partial isotone $f \colon (X, \freeconv) \to (W, \geq)$ with upward closed domain of definition $A \in \mathcal{UC}(X)$, the function $f_* \colon \mathcal{P}(X) \to \mathcal{P}(W)$ that maps each $Y$ to its image under $f$ is an isotone with respect to the corresponding enhancement preorders.
		\end{lemma}
		\begin{proof}
			Consider $Y,Z \in \mathcal{P}(X)$ and an enhancement $\mathsf{enh} \colon Z \to Y$.
			If $f_* (Z)$ is non-empty, we can construct an enhancement $f_* (Z) \to f_ * (Y)$ as follows.
			Given $w \in f_*(Z)$, pick an arbitrary element $z \in Z$ such that $f(z) = w$ holds.
			Since $z$ is in the upset $A$, $f$ is also defined for its enhancement so that we can let the image of $w$ be $f(\mathsf{enh} (z)) \in f_ * (Y)$.
			Since $f$ is an isotone, we have
			\begin{equation}
				f \bigl( \mathsf{enh} (z) \bigr) \geq f(z) = r.
			\end{equation}
			Thus, any such function is an enhancement and we conclude that $Y \succeq_{\rm enh} Z$ implies $f_ * (Y) \geq_{\rm enh} f_ * (Z)$.
		\end{proof}
		\begin{lemma}\label{lem:sup_isotone_deg}
			Given a partial isotone $h \colon (X, \freeconv) \to (W, \geq)$ with downward closed domain of definition $A \in \mathcal{DC}(X)$, the function $h_* \colon \mathcal{P}(X) \to \mathcal{P}(W)$ is an isotone with respect to the corresponding degradation preorders.
		\end{lemma}
		\begin{proof}
			The proof is analogous to that of \cref{lem:sup_isotone}; replacing $\mathsf{enh}$ with a degradation $\mathsf{deg} \colon Y \to Z$ yields a degradation $h_* (Y) \to h_ * (Z)$.
		\end{proof}
		\begin{corollary}\label{lem:monotone_extensions}
			Given partial monotones $f,h \colon (X, \freeconv) \to \ordreals$ defined on upward (downward) closed subsets of $X$ respectively, the two maps
			\begin{equation}
				\begin{split}
				\mathrm{sup} \circ f_* &\colon \bigl( \mathcal{P}(X), \enhgeq \bigr) \to \bigl(\reals , \geq \bigr) \\
				\mathrm{inf} \circ h_* &\colon \bigl( \mathcal{P}(X), \deggeq \bigr) \to \bigl(\reals , \geq \bigr)
				\end{split}
			\end{equation}
			are both monotones.
		\end{corollary}
		In the case of a trivial preorder on $X$, under which any two distinct elements are incomparable, the corresponding enhancement and degradation preorders are just $\supseteq$ and $\subseteq$ respectively.
		Moreover, every subset of $X$ is both a downset and an upset.
		Thus, we get the following special case.
		\begin{corollary}\label{lem:function_extensions}
			Given a partial function $f \colon X \to \reals$, the two maps 
			\begin{equation}
				\begin{split}
					\mathrm{sup} \circ f_* &\colon \bigl( \mathcal{P}(X), \supseteq \bigr) \to \bigl(\reals , \geq \bigr) \\
					\mathrm{inf} \circ f_* &\colon \bigl( \mathcal{P}(X), \subseteq \bigr) \to \bigl(\reals , \geq \bigr)
				\end{split}
			\end{equation}
			are both monotones.
		\end{corollary}

	\section{Quantales}\label{sec:quantale}
	
		Preordered lattices or analogous structures may suffice for studying the ordering of resources and their disjunctions.
		However, when we are interested in studying how the resource order arises as a convertibility relation via free transformations, we need an additional piece of structure.
		It is especially relevant if we care about how transformations can be constructed from elementary building blocks.
		At the most rudimentary level, it can come in the form of a binary operation that represents the composition of transformations.
		This leads to the notion of a quantale---a suplattice with a monoid operation.	
		\begin{definition}
		\label{def:quantale}
			A \textbf{quantale} $(\mathcal{Q}, \bigvee, \bot, \star, 1)$ is a poset $\mathcal{Q}$ equipped with additional structure so that
			\begin{enumerate}
				\item $(\mathcal{Q}, \bigvee, \bot)$ is a suplattice,
				
				\item $(\mathcal{Q}, \star, 1)$ is a monoid, and
				
				\item $\bigvee$ distributes over $\star$.
				That is, for all $T \in \mathcal{Q}$ and all subsets $\mathcal{S}$ of $\mathcal{Q}$, we have
					\begin{align}\label{eq:quantale_distributivity}
						\left( \bigvee \mathcal{S} \right ) \star T &= \bigvee_{S \in \mathcal{S}} (S \star T)  	&  	T \star \left( \bigvee \mathcal{S} \right ) &= \bigvee_{S \in \mathcal{S}} (T \star S) 
					\end{align}
			\end{enumerate}
			If $(\mathcal{Q}, \star, 1)$ is a commutative monoid, we say that $\mathcal{Q}$ is a \textbf{commutative quantale}.
		\end{definition}
		Commonly, quantales as defined above are referred to as \emph{unital} quantales. 
		Since we exclusively work with quantales that have a unit, we drop the adjective throughout.
		Choosing $\mathcal{S}$ in \eqref{eq:quantale_distributivity} to be $\{U, V\}$ and $\emptyset$ respectively lets one derive
		\begin{align}
			U \supseteq V  \quad &\implies \quad  U \star T \supseteq V \star T  	& &\text{and} &	  T \star \bot &= \bot.
		\end{align}
		That is, $\star$ preserves the lattice ordering (in both arguments) and $\bot$ is an ideal element (see \cref{def:ideal}).
		Rather than listing all the basic consequences of \cref{def:quantale}, we refer the reader to \cite{rosenthal1990quantales,eklund2018semigroups}, and focus on the concepts and properties that are relevant for our purposes.
		
		There are many examples of quantales that have nothing to do with resource theories, an extensive overview can be found in \cite{rosenthal1990quantales} for example.
		The most basic class of examples relevant to us are quantales associated to a set of transformations that can be composed, so as to form a monoid.
		\begin{example}\label{ex:free_quantale}
			For every monoid $(Q, \cdot, 1)$, there is the corresponding \textbf{free quantale} $(\mathcal{P}(Q), \bigvee, \emptyset, \cdot, 1)$ where the monoid operation $\cdot$ is extended to the power set by requiring that unions distribute over it.
			That is, for $S,T \in \mathcal{P}(Q)$, we define
			\begin{equation}
				S \cdot T \coloneqq \bigvee_{s,t} s \cdot t
			\end{equation}
			where $s$ and $t$ range over $S$ and $T$ respectively.
			We also have a version of \cref{lem:UAL=CABA}, which says that every uniquely atomistic quantale is isomorphic to the free quantale of some monoid.
		\end{example}
		A similar construction can be used if there are multiple ways to compose transformations.
		\begin{example}\label{ex:free_quantale_multi}
			Given a family of semigroups\footnotemark{} $(Q, \cdot_i)$, we can construct a corresponding (non-unital) quantale $(\mathcal{P}(Q), \bigvee, \emptyset, \star)$ with composition defined by
			\footnotetext{A semigroup is a monoid that is not required to have a unit element.}%
			\begin{equation}
				S \star T \coloneqq \bigvee_{i} S \cdot_i T = \bigvee_{i,s,z} s \cdot_i t.
			\end{equation}
			If there is a subset $1 \subseteq Q$ such that 
			\begin{equation}
				1 \star S = S = S \star 1
			\end{equation}
			holds for all $S \subseteq Q$, then we get a (unital) quantale as well with unit given by $1$.
			
		\end{example}
		Indeed the semigroups could even be partial, meaning that the operation $\cdot_i$ need not be always defined.
		Given $s,t \in Q$ for which $s \cdot_i t$ is undefined, the corresponding composition $s \cdot_i t$ in the power set quantale yields the empty set.
		\begin{example}\label{ex:free_quantale_cat}
			It is now easy to recognize that every small category can be likewise associated with a quantale whose underlying set is the power set of all morphisms.
			In particular, every such category is a partial semigroup whose operation is the morphism composition.
			The unit of this quatale is just the collection of all identity morphisms.
			Of course, this is often not the most convenient way to describe the data of a category, since the information about types of morphisms is now implicit rather than explicit.
			See \cref{sec:future} for a discussion of ways to introduce types for quantales.
		\end{example}
		
		Since process theories are themselves small categories, we can view them also as quantales via \cref{ex:free_quantale_cat}.
		However, this perspective completely misses most of the ways in which processes can be composed.
		The question of which notions of composition should be allowed when specifying $\star$ depends on the particular situation being modelled and constitutes part of the data identifying a resource theory.
		\begin{example}\label{ex:quantale_for_UCRT1}
			For instance, we may choose $\star$ to incorporate both parallel and sequential compositions in a process theory.
			Given two processes $s$ and $t$, we can think of their composition $s \star t$ as the collection of all diagrams obtained by admissible\footnotemark{} ways of contracting the outputs of $t$ with inputs of $s$.
			\footnotetext{The valid ways of plugging wires are restricted by wire types (see \cref{sec:proc_th}), we presume that all the types match in our examples.}%
			For example, given a process $t$ with two outputs and a process $s$ with one input, we may depict the composition $s \star t$ (suppressing the fact that each diagram is in fact a singleton \emph{set of diagrams}) as follows:
			\begin{equation}\label{eq:universal_composition}
				\input{universal_composition.tikz}
			\end{equation}
			Here, we presume that the input of $s$ and the two outputs of $t$ have matching types, while the other wires are of arbitrary type.
			Note that the collection on the right-hand side of \eqref{eq:universal_composition} includes, as one of its elements, the diagram consisting of $s$ and $t$ with no contraction between them---the parallel composition.
			The other elements then depend on the specifics of the processes involved.
			They can be decomposed into the ``elementary'' parallel and sequential compositions with the aid of identity process, e.g.,
			\begin{equation}
				\input{decomp.tikz}
			\end{equation}
			but this decomposition is not unique.
			
			The quantale unit, in this case, includes only one element---the empty diagram.
			One can think of $(s,t)$ as a pair of ``unplaced'' processes and the operation given by $\star \colon (s,t) \mapsto s \star t$ as first placing them in a common diagram followed by listing all the contractions:
			\begin{equation}\label{eq:universal_composition_2}
				\input{universal_composition_2.tikz}
			\end{equation}
			A similar idea has been presented in \cite{kristjansson2020resource} in the context of composing communication channels between parties.
			One can readily extend the construction illustrated in \eqref{eq:universal_composition_2} which allows one to construct $s \star t$ to an arbitrary number of processes to be composed.
			In particular, this perspective makes it clear that such an operation is associative.
		\end{example}
		At first sight, allowing to plug any outputs into any inputs as long as the types match appears to sidestep any causal considerations that one may want to impose.
		However, causal constraints can be incorporated too.
		One way to do so is to have richer types that incorporate causal constraints. 
		In particular, given any causal structure represented by a directed acyclic graph, one can use the framework of \cite{kissinger2017categorical} to build a corresponding type of processes that captures those processes that are compatible with the given causal structure (see section 6.2 therein).

		A related, albeit distinct, approach would be to supplement processes with information specifying causal relations in the form of \emph{routed circuits} in the sense of \cite{vanrietvelde2021routed} or \emph{causal channels} in the sense of \cite[section 4.1]{houghton2021mathematical}.
		We expect that the latter especially can lead to a fruitful investigation of resource theories with causal constraints in very general terms.
		
		\begin{definition}\label{def:idmptnt}
			An element $Q$ of a quantale is called 
			\begin{itemize}
				\item \textbf{transitive} if it satisfies $Q \supseteq Q \star Q$,
				\item \textbf{reflexive} if $Q \supseteq 1$ holds.
			\end{itemize}
		\end{definition}
		Note that if $Q$ is both reflexive and transitive, then it is also \textbf{idempotent}, i.e.\ $Q$ equals $Q \star Q$.
		
		\begin{definition}\label{def:quantale_morphism}
			Given quantales $\mathcal{T}$ and $\mathcal{Q}$, a function $f \colon \mathcal{T} \to \mathcal{Q}$ is a \textbf{quantale homomorphism} if it is a homomorphism of their suplattices, maps the unit of $\mathcal{T}$ to the unit of $\mathcal{Q}$, and satisfies
			\begin{equation}\label{eq:quantale_morphism}
				f(S \star T) = f(S) \star f(T)
			\end{equation}
			for all $S,T \in \mathcal{T}$.
		\end{definition}
		
		We now describe the standard notions of subobjects and quotients.
		In the literature on quantales, and more broadly lattice theory, certain types of (co)closure operators called (co)nuclei are commonly employed to describe these \cite[section 3]{rosenthal1990quantales}.
		However, since we will not need nuclei otherwise, we opt for an alternative presentation.
		\begin{definition}\label{def:subquantale}
			A \textbf{subquantale} is a quantale homomorphism $i \colon \mathcal{T} \to \mathcal{Q}$ which is injective as a function of the underlying sets.
			In this case, we also often speak of $\mathcal{T}$ as a subquantale of $\mathcal{Q}$.
		\end{definition}
		\begin{definition}\label{def:congruence}
			An equivalence relation $\sim$ on a quantale $\mathcal{Q}$ is a \textbf{quantale congruence} if it satisfies the following implications:
			\begin{equation}
				\begin{alignedat}{999}
					S_1 &\sim S_2  	&\quad 	&\implies \quad 	& 	S_1 \star T &\sim S_2 \star T \\
					T_1 &\sim T_2  	&\quad 	&\implies \quad 	& 	S \star T_1 &\sim S \star T_2 \\
					S_i &\sim T_i  	&\quad 	&\implies \quad 	& 	\bigvee_i S_i &\sim \bigvee_i T_i
				\end{alignedat}
			\end{equation}
		\end{definition}
		Equivalence classes $[S]$ with respect to a congruence form a quantale themselves, which is the \textbf{quotient quantale} $(\linefaktor{\mathcal{Q}}{{\sim}}, \mathbin{\&}, [\bot], \bullet, [1])$ with suprema and composition given by
		\begin{equation}\label{eq:quantale_quotient}
			\begin{split}
				[S] \mathbin{\&} [T] &\coloneqq [S \vee T] \\
				[S] \bullet [T] &\coloneqq [S \star T],
			\end{split}
		\end{equation}
		from which one can infer the lattice ordering to be
		\begin{equation}\label{eq:quantale_quotient2}
			[S] \geq [T]  \quad \iff \quad  S \vee T \sim S.
		\end{equation}
		Every quotient quantale comes with the quotient projection $\mathcal{Q} \to \linefaktor{\mathcal{Q}}{{\sim}}$ which is a surjective homomorphism.
		Conversely, the image of any homomorphism can be identified with a quotient of its domain with respect to the congruence identifying its fibers.
		
		In the context of commutative quantales, every element induces a congruence and thus a corresponding quotient.
		\begin{definition}\label{def:augment}
			Given an element $F$ of a commutative quantale $\mathcal{Q}$, $\bm{F}$\textbf{-augmentation} is the map $\mathsf{Aug}_F \colon \mathcal{Q} \to \mathcal{Q}$ given by 
			\begin{equation}
				S \mapsto F \star S
			\end{equation}
		\end{definition}
		One can easily check that the relation
		\begin{equation}\label{eq:augment_cong}
			S \sim_F T  \quad \iff \quad  F \star S = F \star T.
		\end{equation}
		is a congruence.
		Thus, even though $\mathsf{Aug}_F \colon \mathcal{Q} \to \mathcal{Q}$ is not guaranteed to be a homomorphism (unless $F$ is idempotent for example), the projection $\pi_F \colon \mathcal{Q} \to \linefaktor{\mathcal{Q}}{{\sim_F}}$ is.
		For non-commutative quantales, left and right augmentations have to be distinguished and their fibres do not correspond to congruences in general.
	
	\section{Quantale Modules}\label{sec:quantale_module}
	
		Note that for any suplattice $\mathcal{X}$, the set $\mathrm{End}(\mathcal{X})$ of all suplattice homomorphisms $\mathcal{X} \to \mathcal{X}$ is a quantale.
		In particular, the suprema are given pointwise: 
		For any $f,g \in \mathrm{End}(\mathcal{X})$ and any $Y \in \mathcal{X}$ we let
		\begin{equation}
			( f \cup g ) (Y) \coloneqq f(Y) \cup g(Y).
		\end{equation}
		The quantale composition is just the sequential composition of functions.
		
		A subquantale $\mathcal{Q} \hookrightarrow \mathrm{End}(\mathcal{X})$ can be then thought of specifying the transitions between elements of $\mathcal{X}$.
		In this case, elements of $\mathcal{Q}$ (which are eventually interpreted as resource transformations) are given in terms of the transitions they induce.
		More generally, we can consider them as independent entities and related to transitions via a quantale homomorphism $\mathcal{Q} \to \mathrm{End}(\mathcal{X})$.
		Such a structure corresponds to a general notion of a transition system \cite{Abramsky1993} and can be equivalently described as a quantale module \cite[Proposition 3.1.3]{eklund2018semigroups}.
		
		\begin{definition}\label{def:quantale_module}
			A \textbf{quantale module} $(\mathcal{Q}, \mathcal{X}, \apply)$ consists of a quantale $(\mathcal{Q}, \bigvee, \bot, \star, 1)$, a suplattice $(\mathcal{X}, \bigcup, \bot)$, and an action
				$\apply \colon \mathcal{Q} \times \mathcal{X} \to \mathcal{X}$ that satisfies
			\begin{subequations}\label{eq:qmodule}
			\begin{align}
				\label{eq:qmodule_comp_action}  S \apply (T \apply Y) &= (S \star T) \apply Y, \\
				\label{eq:qmodule_res_disj}  T \apply \left( \bigcup \mathcal{Z} \right) &= \bigcup_{Z \in \mathcal{Z}} (T \apply Z), \\
				\label{eq:qmodule_trans_disj}  \left( \bigvee \mathcal{S} \right) \apply Y &= \bigcup_{S \in \mathcal{S}} (S \apply Y), \\
				\label{eq:qmodule_identity}  1 \apply Y &= Y 
			\end{align}	
			\end{subequations}
			for all $S,T \in \mathcal{Q}$ and all $Y \in \mathcal{X}$.
		\end{definition}
		Apart from the books on quantales mentioned in \cref{sec:quantale}, one can find a concise introduction to quantale modules in \cite{Abramsky1993,russo2010quantale}.
		
		Whenever it is sufficient to keep the relevant structure implicit, we refer to $\mathcal{X}$ itself as a $\mathcal{Q}$-module.
		Just like for quantale compositions, every action $\apply$ necessarily has to preserve the lattice orderings of both $\mathcal{Q}$ and $\mathcal{X}$.
		Furthermore, we also have 
		\begin{align}
			\bot \apply Y &= \bot  	&	  S \apply \bot = \bot
		\end{align}
		for all $S \in \mathcal{Q}$ and all $Y \in \mathcal{X}$.
		
		\begin{definition}
			A suplattice homomorphism $f \colon \mathcal{Y} \to \mathcal{X}$ is a homomorphism of $\mathcal{Q}$-modules if it satisfies
			\begin{equation}
				f(Q \apply Y) = Q \apply f(Y)
			\end{equation}
			for all $Q \in \mathcal{Q}$ and all $Y \in \mathcal{Y}$.
		\end{definition}
		\begin{example}\label{ex:module_from_morphism}
			Note that any quantale defines a $\mathcal{Q}$-module $(\mathcal{Q},\mathcal{Q},\star)$ by action on itself via composition.
			Moreover, for any quantale homomorphism $\ell \colon \mathcal{T} \to \mathcal{Q}$, there is a quantale module $(\mathcal{T},\mathcal{Q},\star_\ell)$ with the action given for any $T \in \mathcal{T}$ and $Q \in \mathcal{Q}$ by
			\begin{equation}
				T \star_\ell Q \coloneqq \ell(T) \star Q.
			\end{equation}
			In particular $\mathcal{T}$ may be a subquantale of $\mathcal{Q}$.
		\end{example}
		The above procedure is known as restriction of scalars in module theory.
		As usual, such an $\ell$ can be used to associate a $\mathcal{T}$-module to \emph{any} $\mathcal{Q}$-module in a similar fashion.
		\begin{definition}\label{def:qmod_morphism}
			Consider two quantale modules $(\mathcal{T}, \mathcal{Y}, \blacktriangleright)$ and $(\mathcal{Q}, \mathcal{X}, \apply)$.
			A \textbf{quantale module homomorphism} is a pair of a quantale homomorphism $\ell \colon \mathcal{T} \to \mathcal{Q}$ and a suplattice homomorphism $f \colon \mathcal{Y} \to \mathcal{X}$ such that we have
			\begin{equation}\label{eq:qmod_morphism}
				f \bigl( T \blacktriangleright Y \bigr) = \ell (T) \apply f(Y)
			\end{equation}
			for all $T \in \mathcal{T}$ and all $Y \in \mathcal{Y}$.
			In other words, the diagram
			\begin{equation}
				\begin{tikzcd}
					\mathcal{Q} \times \mathcal{Y} \arrow[rr, " f"]  &                               & \mathcal{Q} \times \mathcal{X} \arrow[d, "\apply"] \\
					\mathcal{T} \times \mathcal{Y} \arrow[r, "\blacktriangleright"'] \arrow[u, "\ell"']                                                   & \mathcal{Y} \arrow[r, "f"'] & \mathcal{X}                                           
				\end{tikzcd}
			\end{equation}
			commutes.
		\end{definition}
		More generally, a \textbf{lax} (or \textbf{oplax}) \textbf{quantale module homomorphism} is instead only required to satisfy
		\begin{equation}\label{eq:lax_qmod_morphism}
			f \bigl( T \blacktriangleright Y \bigr) \supseteq \ell (T) \apply f(Y)  \qquad \text{or} \qquad   f \bigl( T \blacktriangleright Y \bigr) \subseteq \ell (T) \apply f(Y)
		\end{equation}
		respectively.
		In both of these weakenings, it makes sense to also consider the quantale homomorphism $\ell$ to be lax (or oplax), such that only
		\begin{equation}
			\ell (S \star T) \supseteq  \ell(S) \star \ell(T)  	\qquad \text{or} \qquad   \ell (S \star T) \subseteq  \ell(S) \star \ell(T)
		\end{equation}
		is required to hold, as opposed to \cref{eq:quantale_morphism}.
		
		\begin{example}
			Just like every monoid has an associated power set quantale, given a monoid $(Q, \cdot, 1)$ acting on a set $X$ via an action $\apply \colon Q \times X \to X$, we get a quantale module $(\mathcal{P}(Q), \mathcal{P}(X), \apply)$.
			In particular, the action is extended to sets by compatibility with unions:
			\begin{equation}
				S \apply Y \coloneqq \bigcup_{s,y} s \apply y
			\end{equation}
			for any $S \in \mathcal{P}(Q)$ and any $Y \in \mathcal{P}(X)$.
		\end{example}
		
		\begin{definition}\label{def:ideal}
			Let $F$ be an element of a quantale $\mathcal{Q}$.
			An element $X$ of a $\mathcal{Q}$-module is said to be $\bm{F}$\textbf{-ideal} if it satisfies $F \apply X \subseteq X$.
			If $S \apply X \subseteq X$ holds for all $S \in \mathcal{Q}$, then we say that $X$ is an \textbf{ideal element}.
		\end{definition}
		If $\top$ denotes the top element of $\mathcal{Q}$, the condition that $X$ be ideal is equivalent to 
		\begin{equation}
			\top \apply X = X,
		\end{equation}
		since $\top \supseteq 1$ implies $\top \apply X \supseteq X$ for an arbitrary $X \in \mathcal{X}$.
		Thus, ideal elements are the same as $\top$-ideal ones.
		The reason for calling these ideal elements is that they are in bijective correspondence with ideals of the $\mathcal{Q}$-module \cite{Russo2017}.
		
		Given a $\mathcal{Q}$-module $\mathcal{X}$, we can define preorders on the suplattice that describe transitions between elements of $\mathcal{X}$ with respect to elements of $\mathcal{Q}$ closed under composition in the following sense.
		\begin{definition}\label{def:reachability}
			Consider a $\mathcal{Q}$-module $\mathcal{X}$ and a reflexive, transitive $F \in \mathcal{Q}$.
			The $\bm{F}$\textbf{-reachability preorder} $\freeconv_F$ on $\mathcal{X}$ is defined by 
			\begin{equation}\label{eq:reachability}
				Y \freeconv_F Z \quad \coloniff \quad  F \apply Y \supseteq Z.
			\end{equation}
		\end{definition}
		Every such $F$ corresponds to a subquantale $\langle \bot,F \rangle$ of $\mathcal{Q}$ which consists of all $S \in \mathcal{Q}$ satisfying
		\begin{equation}
			F \supseteq S \supseteq \bot.
		\end{equation}
		Therefore, we can view the $F$-reachability preorder in the $\langle \bot,F \rangle$-module $\mathcal{X}$ as defined via
		\begin{equation}
			\begin{split}
				Y \freeconv_F Z \quad &\iff \quad  \top \apply Y \supseteq Z \\
					&\iff \quad  \exists \, S \in \langle \bot,F \rangle  \; : \; \,  S \apply Y \supseteq Z .
			\end{split}
		\end{equation}
		where the top element of $\langle \bot, F \rangle$ is just $F$ of course.
		However, since we are interested in studying distinct reachability preorders of a (potentially) fixed $\mathcal{Q}$-module, \cref{def:reachability} is the more relevant way to think of these order relations.
		
		\begin{lemma}\label{lem:access_is_sensible}
			The $F$-reachability preorder of a $\mathcal{Q}$-module $\mathcal{X}$ is an extension of its lattice ordering $\supseteq$.
		\end{lemma}
		\begin{proof}
			First of all, let us check that $\succeq$ is indeed a preorder as the name suggests.
			Reflexivity of the order follows from reflexivity of $F$ via the fact that $\apply$ preserves the lattice ordering of $\mathcal{Q}$.
			For transitivity, assume both $X \freeconv_F Y$ and $Y \freeconv_F Z$ hold.
			Then we have
			\begin{equation}
				F \apply X \supseteq (F \star F) \apply X = F \apply (F \apply X) \supseteq F \apply Y \supseteq Z
			\end{equation}
			as required.
			
			The fact that $\mathcal{X}$ forms a preorder-extended suplattice follows directly from \eqref{eq:reachability}.
			Indeed, if $Y \supseteq Z$ holds, then
			\begin{equation}
				F \apply Y \supseteq Y \supseteq Z
			\end{equation}
			does as well since since $\freeconv_F$ is reflexive as wer already established.
			Furthermore, if we have $F \apply X \supseteq Z$ for all $Z$ in some subset $\mathcal{Z}$ of $\mathcal{X}$, then $F \apply X$ is an upper bound of $\mathcal{Z}$ and thus satisfies 
			\begin{equation}
				F \apply X \supseteq \bigcup \mathcal{Z}.
			\end{equation}
		\end{proof}
		
		\begin{lemma}\label{lem:access=enh}
			Let $\mathcal{X}$ be an atomistic $\mathcal{Q}$-module.
			Let $F$ be a reflexive, transitive element of $\mathcal{Q}$ and denote by $\enhconv$ the enhancement preorder on $\mathcal{X}$ constructed from the restriction of the $F$-reachability preorder to the atoms of $\mathcal{X}$.
			Then we have 
			\begin{equation}\label{eq:access=enh}
				Y \enhconv Z \quad \iff \quad  Y \freeconv_F Z
			\end{equation}
			for all $Y, Z \in \mathcal{X}$.
		\end{lemma}
		\begin{proof}
			By \cref{lem:enh_down}, it suffices to show that 
			\begin{equation}\label{eq:access=down}
				\downwrt{F} (Y) \supseteq \downwrt{F}(Z) \quad \iff \quad  Y \freeconv_F Z
			\end{equation}
			holds, where $\downwrt{F}$ is the downward closure operator defined as
			\begin{equation}\label{eq:down_def}
				\downwrt{F} (Y) = \bigcup_{y \in A_Y} \downwrt{F} (y) = \bigcup_{y \in A_Y} F \apply y
			\end{equation}
			for any choice of a set $A_Y$ of atoms whose supremum is $Y$.
			Since the rightmost expression in \eqref{eq:down_def} equals $F \apply Y$, we can rewrite the desired equivalence \eqref{eq:access=down} as 
			\begin{equation}\label{eq:reachability2}
				F \apply Y \supseteq F \apply Z \quad \iff \quad  F \apply Y \supseteq Z.
			\end{equation}
			The forward implication is ensured by $F \supseteq 1$, while the converse can be derived by applying another $F$ to $F \apply Y \supseteq Z$ and observing that
			\begin{equation}
				F \apply (F \apply Y) = (F \star F) \apply Y = F \apply Y
			\end{equation}
			holds thanks to the idempotence of $F$.
		\end{proof}
		
		We now describe quotient quantale modules, once again in terms of congruences rather than closure operators.
		\begin{definition}\label{def:module_congruence}
			Consider a $\mathcal{Q}$-module $\mathcal{X}$.
			An equivalence relation $\sim$ on $\mathcal{X}$ is a $\mathcal{Q}$-module \textbf{congruence} if it satisfies the following implications:
			\begin{equation}
				\begin{alignedat}{999}
					X_1 &\sim X_2  	&\quad &\implies \quad & 	Q \apply X_1 &\sim Q \apply X_2 \\
					X_i &\sim Y_i  	&\quad &\implies \quad & 	\bigcup_i X_i &\sim \bigcup_i Y_i
				\end{alignedat}
			\end{equation}
		\end{definition}
		Equivalence classes $[X]$ with respect to a congruence form a $\mathcal{Q}$-module themselves, which is the \textbf{quotient $\bm{\mathcal{Q}}$-module} $\linefaktor{\mathcal{X}}{{\sim}}$ with suprema and composition given by
		\begin{equation}\label{eq:qmodule_quotient}
			\begin{split}
				[X] \cup [Y] &\coloneqq [X \cup Y] \\
				Q \apply [X] &\coloneqq [Q \apply X].
			\end{split}
		\end{equation}
		Every quotient quantale comes with the quotient projection $\mathcal{X} \to \linefaktor{\mathcal{X}}{{\sim}}$ which is a surjective $\mathcal{Q}$-module homomorphism.
		
		More generally, we can have a quantale congruence on $\mathcal{Q}$ in conjunction with a $\mathcal{Q}$-module congruence on $\mathcal{X}$.
		If the former satisfies
		\begin{equation}
			S \sim T 	 	\quad \implies \quad  	S \apply X \sim T \apply X
		\end{equation}
		for all $S,T \in \mathcal{Q}$ and all $X \in \mathcal{X}$, then the pair forms a \textbf{quantale module congruence}.
		As one would expect, we can then form the quotient quantale module $(\linefaktor{\mathcal{Q}}{{\sim}},\linefaktor{\mathcal{X}}{{\sim}},\apply)$ with 
		\begin{equation}
			[Q] \apply [X]  \coloneqq [Q \apply X].
		\end{equation}
		
		Just like for quantales, we can define augmentation for quantale modules too.
		\begin{definition}\label{def:centralizer}
			Given a quantale $\mathcal{Q}$ and a subset $\mathcal{Z}$ thereof, we define the \textbf{centralizer} of $\mathcal{Z}$ in $\mathcal{Q}$, denoted $\mathrm{C}_{\mathcal{Q}}(\mathcal{Z})$, to be
			\begin{equation}\label{eq:centralizer}
				\mathrm{C}_{\mathcal{Q}}(\mathcal{Z}) \coloneqq  \Set*[\big]{ Q  \given  Q \star Z = Z \star Q \,\text{ for all }\, Z \in \mathcal{Z}  }.
			\end{equation}
		\end{definition}
		One can observe that every centralizer is a subquantale of the original quantale, since it contains $\bot$ and $1$ and it is closed under $\star$ as well as suprema.
		\begin{definition}\label{def:module_augment}
			Let $\mathcal{X}$ be a $\mathcal{Q}$-module and $F$ an element of $\mathcal{Q}$.
			$\bm{F}$\textbf{-augmentation} of $\mathcal{X}$ is the $\mathrm{C}_{\mathcal{Q}}(F)$-module $\mathsf{Aug}_F(\mathcal{X})$ given by its underlying set
			\begin{equation}
				\mathsf{Aug}_F(\mathcal{X}) \coloneqq \Set{ F \apply X \given X \in \mathcal{X} }
			\end{equation}
			and all structure inherited from the original module.
		\end{definition}
		Once again, the relation
		\begin{equation}\label{eq:augment_cong_2}
			Y \sim_F Z  \quad \iff \quad  F \apply Y = F \apply Z.
		\end{equation}
		is a congruence, but only when we think of $\mathcal{X}$ as a $\mathrm{C}_{\mathcal{Q}}(F)$-module.

\chapter{Resource Theories, Abstractly}\label{sec:resource_theories}
	
	\section{The Big Picture}
		
		In this chapter, we return to the topic of resource theories.
		Gradually, we solidify connections between resource-theoretic concepts and the mathematical structures introduced in \cref{sec:math_prelim} and provide the relevant interpretation and motivation.
		
		\subsubsection{Overview of the chapter.}
		
		In \cref{sec:ucrt}, we use commutative quantales to describe universally combinable resource theories.
		They correspond to situations in which there is no need for a conceptual distinction between resources and transformations.
		The term ``universally combinable resource theories'' originates in \cite{Coecke2016}, where the universal combination refers to the idea that processes can be composed arbitrarily.
		Such unconstrained composition defines a particular commutative quantale (see \cref{ex:quantale_for_UCRT2} for more details).
		This set-up mimics the approach of \cite{Fritz2017} in terms of ordered commutative monoids and generalizes the framework used in \cite{Monotones}.

		We also discuss the appropriate notion of the ordering of resources in \cref{sec:ucrt_order}.
		The following section then introduces ways in which a new set of free resources can be constructed from others. 
		For instance, LOCC processes in a resource theory of entanglement can be viewed as a conjunction of local operations (LO) and classical communication (CC).
		\Cref{sec:ucrt_catalysis} then describes a way to model conversions between resources with the help of catalysts.
		We can think of catalysts as tools which do not degrade throughout their use.
		A toy example to keep in mind is that of an axe which can be used to convert stored energy and a source of unprocessed wood into processed wood and which is (to a good approximation) returned in its original state after use.
		
		Later, in \cref{sec:rt}, we discuss the general framework of resource theories in terms of quantale modules.
		In this context, there is a conceptual distinction between resources and transformations. 
		The fact that the composition of transformations here is not required to be commutative allows for many more examples as compared to universally combinable resource theories.		
		In \cref{sec:rt_examples}, we provide several ways to turn a partitioned process theory \cite{Coecke2016} into a resource theory in our sense.
		We formalize the notion of a resource theory of asymmetry in the present context in \cref{sec:rt_asymmetry}.
		Quantum resources of asymmetry have been studied extensively in many works previously \cite{Marvian2012,Marvian2013,Marvian2014,Gour2009}.
		However, as we show in \cref{sec:monotones from twirling generalized} for instance, there are methods for studying resources of asymmetry that generalize beyond the context of quantum theory.
		Finally, \cref{sec:rt_morphisms} introduces a notion of morphisms that can be used to relate different resource theories.
		We provide examples of these for resource theories of asymmetry therein, but further examples of other morphisms appear in \cref{sec:monotones from twirling generalized,sec:dist_poss}.
		
		Even though both of the frameworks use quantales to describe ways in which resources (or their transformations) may be combined, the commutativity of the quantale operation $\boxtimes$ underlying a universally combinable resource theory means that its interpretation is typically different from the quantale operation $\star$ of resource transformations in the quantale module framework.
		In particular, if the combination of a resource $r$ with a resource $s$, $r \boxtimes s$ coincides with the swapped version $s \boxtimes r$ (which we assume in a universally combinable resource theory), then the object $r \boxtimes s$ carries no information about the sequential (or temporal) order in which $r$ and $s$ are used.
		Therefore, universally combinable resource theories are suitable if we do not need to model such sequential information explicitly via the quantale operation.
		Indeed, we may wish to allow that the set of possible transformations afforded by using a resource transformation $t$ before transformation $u$ is different from their combined use in reverse order.
		In this case, the framework of general quantale modules is the more suitable one.
		
		Other general constructions in the context of these abstract frameworks, such as ones of resource theories with a convex structure (\cref{sec:rt_convex}) or resource theories of distinguishability (\cref{sec:Encodings}) are only introduced once they are needed for our discussion of resource measures.
		Let us begin with a more detailed description of two resource theories we mentioned already, one of which describes resources of states out of equilibrium \cite{Janzing2000} and the other is about resources of quantum entanglement.
		
		\subsubsection{Examples of resource theories.}
		
		\begin{example}[resource theory of athermality]\label{ex:rt_athermality}
			There are multiple resource theories within which thermodynamic phenomena can be studied.
			This is a common feature of resource theories and it highlights the fact that the same concept can attain distinct meaning depending on the context.
			The choice of a resource theory, and in particular of the free transformations, is a specification of the context relative to which any conclusions about the values of thermodynamic resources are made (see also our discussion in the Introduction).
			Moreover, by identifying and justifying a particular choice of free transformations, one has to elucidate some of the assumptions being made.
			
			Here, we describe a resource theory of athermality with respect to a fixed background temperature, see \cite{lostaglio2019introductory} for an extensive discussion of the motivations, limitations and alternatives. 
			The objects of study are states of quantum systems, and we generally do not assume that they are close to equilibrium or be well-approximated by the `thermodynamic limit'. 
			Thus, we would only expect the standard notions of statistical mechanics to arise in the domain of manipulations of composite systems with a large number of constituents.
			The crucial feature of systems is that they come equipped with a Hamiltonian $H$ which specifies their time evolution in isolation.
			When in equilibrium with a heat bath identified with inverse temperature $\beta$, each system occupies the thermal state $e^{-\beta H} / Z$, where $Z$ is the partition function.
			
			The free manipulations (called thermal operations) include arbitrary unitary maps that commute with the Hamiltonian.
			Additionally, we assume that the agent can discard and bring in any system, as long as the latter is initialized in its thermal state. 
			In this context, the ability of the agent to extract work can be modelled in terms of the conversions allowed by thermal operations.
			For example, the accounting of energetic balance can be described by the energy eigenstates of a simple `battery' system.
			The conversion
			\begin{equation}
				\rho \otimes | E_0 \rangle \langle E_0 |  \freeconv  \sigma \otimes | E_1 \rangle \langle E_1 |
			\end{equation}
			models the fact that we can extract $E_1 - E_0$ of useful work by converting $\rho$ to $\sigma$.
		\end{example}
		
		\begin{example}[resource theory of entanglement]\label{ex:rt_entanglement}
			When thinking of entangled states of quantum systems as resources, we may consider tasks such as quantum teleportation \cite{bennett1993teleporting}, superdense coding \cite{bennett1992communication}, and others.
			For instance, the teleportation protocol asserts that it is possible to convert a quantum state maximally entangled between two parties, Alice and Bob, to a single use of a quantum channel between them.
			The statement is made non-trivial by requiring that the manipulations allowed exclude any quantum communication between the parties---they can only send classical information.
			That is, the free transformations (called LOCC operations) can be generated by arbitrary local quantum processing and arbitrary classical communications.
			
			We can think of other choices of free transformations that give rise to resource theories of distributed aspects of bipartite quantum states, such as
			\begin{enumerate}
				\item local operations (LO) with no causal relationship between Alice and Bob in the past or present,
				\item local operations with shared randomness (LOSR), which allow for a classically correlated pair of systems to be supplied to Alice and Bob \cite{schmid2020understanding},
				\item separable operations \cite{vedral1997quantifying}, which form a strict superset of LOCC, but are in many ways easier to characterize. 
			\end{enumerate}
			Nevertheless, the resource-theoretic framework places a non-trivial restriction on what the free transformations could be, even prior to considerations of what aspect we aim to capture with the specific choice.
			For example, the following sets of operations \emph{do not} constitute valid candidates for a resource theory of quantum entanglement:
			\begin{enumerate}
				\item local operations with a single round of two-way classical communication between Alice and Bob,
				\item operations that preserve the set of separable states of their input system.
			\end{enumerate}
			This is because we expect each free transformation to be available in unlimited amount, in accordance with the interpretation that they have ``zero cost''. 
			In the first case, multiple rounds of communication lead to a wider class of transformations with greater convertibility power---thus these are not closed under composition.
			In the second case, separable-preserving operations are not closed under conjunction with identity channels.
			That is, given a separable-preserving operation, such as the swapping of Alice's and Bob's systems, its tensor product with an identity channel on another bipartite system is not separable-preserving.
			This fact can be observed diagrammatically via
			\begin{equation}
				\input{swap_x_id_2.tikz}
			\end{equation}
			where the `cups' represent maximally entangled states $\Psi$.
			That is, starting with a product state $\Psi \otimes \Psi$, we end up with a maximally entangled state on the right.
		\end{example}
		
		\subsubsection{Two types of abstract resource theories.}
		
		Broadly speaking, we can identify two approaches in the pursuit to formalize the idea of a resource theory.
		Of course, to answer concrete questions, more structure is usually afforded than what we outline below, but this minimal description suffices to draw the distinction.
		
		One of the approaches postulates the existence of objects (the resources) which are to be manipulated.
		The manipulation of resources is described by a class of free transformations, usually as a strict subset of all conceivable transformations.
		In this sense, it is closer to the framework of quantale modules (see \cref{sec:rt_intuition}) among the two we present here.
		
		Choosing such free transformations is the crucial step in identifying a specific resource theory.
		Their interpretation can be manifold.
		Usually, free transformations are thought to be either
		\begin{enumerate}
			\item \emph{abundant in their occurence}, such as in the case of a resource theory modelling ``natural'' or ``spontaneous'' processes with no agents interested in engineering a specific outcome (for instance, in a resource theory modelling the potential emergence of life in Earth's early oceans, one may deem transformations that describe the use of abundant molecules to be free);
			
			\item \emph{cheap or easy to implement}, such as in the case of a resource theory modelling the capabilities of concrete agents (for instance, in a resource theory of computational complexity, one may deem computations that can be done in polynomially increasing time in terms of the problem size to be free);
			
			\item \emph{implementable with limited information}, such as in the case of agents with restriction on their knowledge (for instance, in a resource theory of distinguishability described in \cref{sec:Encodings}, a transformation is free only if it can be implemented with no information about the hypotheses that are to be distinguished);
			
			\item \emph{constructible from a basic set of generators}, such as in the case of agents granted access to elementary operations (for instance, in a resource theory of quantum athermality, one may allow the use of thermal states, unitaries that commute with the Hamiltonian and discarding systems to be the set of generators, from which more general free operations can be built),\footnotemark{}
			\footnotetext{Even though the problem of convertibility by free operations generated from an elementary set is undecidable even in the context of resource theories of quantum states \cite{scandi2021undecidability}, this presentation often offers a useful perspective on free transformations.}%
			
			\item \emph{consistent with a physical principle}, such as in the case of a resource theory describing fundamental limitations (for instance, in a resource theory of nonclassical correlations described in \cref{ex:quantum_correlations}, any transformation that is compatible with the presumed causal structure is deemed to be free);
		\end{enumerate}
		or a combination thereof.
		In some circumstances, we can find multiple justifications for the same set of free transformations.
		However, the different explanatory avenues often do lead to different resource theories with inequivalent content.
		The resource-theoretic point of view thus helps us crystallize our understanding of the phenomena by exposing the explanatory mechanism.
		
		Having identified the free operations, one may be interested in learning about how resources can be manipulated by the free operations.
		The basic questions ask whether it is possible to convert one resource to another and, if so, how.
		These correspond to the problem of characterizing the convertibility preorder among the resources and building protocols for specific types of conversion.
		Modern investigations of resource theories in quantum information science \cite{Chitambar2018} tend to be of this kind.
		At the level of frameworks, one could include resource theories of states \cite[section 3.2]{Coecke2016} or of parallel combinable processes \cite[section 3.3]{Coecke2016} in the context of process theories here.
		
		The other approach abstracts away the free operations and postulates the existence of resources equipped with an order relation that contains the information about their convertibility.
		Some structure describing how resources can be combined is usually given as well.
		In this sense, it is closer to the framework of universally combinable resource theories (see \cref{sec:ucrt}) among the two we present here.
		However, universally combinable resource theories do not quite fit into either of the two approaches described here.
		
		With this often more abstract setup, it is easier to investigate more intricate and generalizable features of the resource preorder.
		For example, one can study the convertibility between many copies of resources or convertibility in the presence of catalysts \cite{Fritz2017,fritz2020local} in fairly general terms.
		Sometimes, it is helpful to view such ordered sets as $\{true, false\}$-enriched categories \cite[definition 4.1]{Coecke2016}.\footnote{See also \cref{ex:ocm_as_rt2} for a related point of view.}
		If the mere ``possible/impossible'' distinction is insufficient, one can change the base of the enrichment to include more intricate information about the conversions.
		In \cite{marsden2018quantitative}, this idea is explored to frame resource theories in the context of categories enriched over a commutative quantale, which specifies possibly more intricate information about the relationship between resources.
		Alternatively, this approach allows one to define general concepts, postulate additional axioms and derive their consequences purely in terms of (im)possibility, as constructor theory does \cite{deutsch2013constructor}.
		
		\subsubsection{Aims of our work.}
		
		Our goal is to give a framework that does not necessarily take one approach or the other.
		Instead, it should (eventually) have enough expressive power and generality to be able to strengthen the connection between them.
		Indeed, even though we mention above that the quantale module framework is closely tied to the first approach, it can also express many instances of the latter one (cf. \cref{ex:ocm_as_rt2}).
		This guiding idea connects to one of the primary aims of the work, which is to build a framework within which methods and results from specific resource theories can be generalized and tied together.
		By doing so, they may find new applications and form bridges between seemingly disparate domains of inquiry with potentially distinct languages and terminology.
		
		One could, and many people certainly do, pursue similar avenues outside of the domain of resource theories.
		However, we think that resource theories are particularly suited due to their proven record as a methodological tool for consolidating knowledge.
		The most notable example is the resource theory of quantum entanglement \cite{Horodecki2009}, but there are many others too.
		Additionally, questions that arise at the heart of resource theories are, when abstracted, abundant and appear in many different contexts.
		One manifestation of this fact is the wide range of interpretations of free operations as listed above.
		
		A secondary goal of the framework is to allow a more systematic study of relationships between different resource theories.
		Apart from aiding in the translation of methods as alluded to, this could also facilitate the study of 
		\begin{itemize}
			\item trade-offs between different kinds of resources and
			\item consequences of and interactions between physical principles
		\end{itemize}
		among other potential applications.
		
	\section{Universally Combinable Resource Theories}\label{sec:ucrt}
		
		\subsection{The Intuition}\label{sec:ucrt_intuition}
			
			Let us turn to a concrete description of the desired features of ``universally combinable resource theories'', described briefly in the Introduction.
			These provide one example of a class of resource theories that the framework introduced in \cref{sec:rt} can accommodate.
			
			First of all, a resource theory should describe a collection $R$ of resources---the objects of study.
			Secondly, it should describe a way of combining these, which we model by an associative, binary\footnotemark{} operation $\boxtimes$. 
			\footnotetext{The use of a \emph{binary} operation in this context is merely of formal convenience. 
				Conceptually, the notion of composition of $n$ resources that we have in mind is having access to all of them and having the ability to combine them in various ways.}%
			The idea of a resource theory which is \emph{universally combinable} is that there is no restriction on how resources may be used, besides ones arising directly from the nature of the resources.
			For example, given two resources $r$ and $s$, combining them into $r \boxtimes s$ should allow one to use $s$ prior to $r$ just as well as in the opposite arrangement.
			Therefore, $\boxtimes$ is assumed to be commutative.
			Finally, a universally combinable resource theory should incorporate a structure that specifies which conversions between the resources are possible and under what conditions.
			This is achieved by specifying a subset $R_{\rm free}$ of resources deemed to be \emph{free}.
			
			Generically, processes can be wired together in multiple ways.
			To capture this feature, the object $r \boxtimes s$ describing the combination of resources $r$ and $s$ is not another resource, but rather a set of resources, each providing a particular way of combining $r$ \mbox{and $s$.}
			We can thus view $r \boxtimes s$ as an element of $\mathcal{P}(R)$, the power set of $R$, which is a suplattice (see \cref{def:suplattice}) with the supremum operation being the union of sets.
			
			The interpretation of $r \boxtimes s$ relative to the individual resources $r$ and $s$ is that an agent has access to both $r$ and $s$ (and they can be combined in various ways).
			On the other hand, the union of $\{r\}$ and $\{s\}$, i.e.\ the set $\{r,s\} \in \mathcal{P}(\mathcal{R})$, represents an agent having access to either resource $r$ or resource $s$, but not both simultaneously.
			Because of this interpretation, we require that the $\boxtimes$ operation distributes over unions, just like conjunction distributes over disjunction.
			That is, for any two sets of resources $S,T \in \mathcal{P}(R)$, we have
			\begin{equation}
				\label{eq:Combination of sets of resources}
				S \boxtimes T = \bigcup_{s\in S, t \in T} s \boxtimes t,
			\end{equation}
			where $s \boxtimes t$ is a shorthand notation for $\{s\} \boxtimes \{t\}$.
			Equation \eqref{eq:Combination of sets of resources} expresses that agent's choices of a resource from $S$ and of a resource from $T$ are \emph{independent} of each other.
			For example, if both $S$ and $T$ had three elements each, then the union on the right-hand side ranges over nine possible choices of pairs of resources to be combined.
			Conceptually, it is important to distinguish this interpretation from collections of resources that correspond to counterfactual possibilities of what the `actual' resource is.
			The latter is explored in resource theories of knowledge \cite{DelRio2015} where the possibilities are thought to represent an agent's lack of knowledge about the resource; and in resource theories of distinguishability (\cref{sec:Encodings}) where they represent distinct hypotheses about an unknown variable.
			
			Furthermore, we assume that there is a neutral set of resources denoted by $1 \in \mathcal{P}(R)$ that satisfies
			\begin{equation}
				1 \boxtimes S = S = S \boxtimes 1
			\end{equation}
			for all $S \in \mathcal{P}(R)$.
			We may interpret the neutral set as consisting of resources that can only be used for trivial conversions.
			
			Notice that the neutral set differs significantly from the empty set $\emptyset \in \mathcal{P}(R)$, which satisfies
			\begin{equation}
				\emptyset \boxtimes S = \emptyset = S \boxtimes \emptyset
			\end{equation}
			and which is the bottom element of the suplattice $\mathcal{P}(R)$ ordered by set inclusion.
			If we have $r \boxtimes s = \emptyset$, the interpretation is that the resources $r$ and $s$ are mutually incompatible---there is no way to combine them.
			For example, in the context of deterministic computation, if $r$ and $s$ denote mutually exclusive states of a single register, then they cannot 
				coexist and therefore they cannot be combined.
			On the other hand, if $r \boxtimes s \subseteq 1$ holds, we would say that any way to combine $r$ and $s$ produces a resource in the neutral set, thus effectively discarding them.
		
			Altogether, we get a commutative monoid $(\mathcal{P}(R), \boxtimes, 1)$ with a monoid operation $\boxtimes$ that distributes over unions.
			In other words, it is a uniquely atomistic and commutative quantale (see \cref{sec:quantale}).
			
			When defining a resource theory, we also identify a distinguished subset of resources, denoted by $R_{\rm free}$.
			These resources are free in the sense that one can access them in unlimited supply, whence we impose that combining free resources cannot yield a non-free resource.
			That is, we require $R_{\rm free} \boxtimes R_{\rm free} \subseteq R_{\rm free}$.
			Since we also impose that the neutral set of resources is free, i.e.\ $1 \subseteq R_{\rm free}$, we can express this condition equivalently as
			\begin{equation}
				\label{eq:Compatibility of free set and star 2}
				R_{\rm free} \boxtimes R_{\rm free} = R_{\rm free}.
			\end{equation}
			In summary, we identify a concrete universally combinable resource theory as a uniquely atomistic, commutative quantale $\mathcal{P}(R)$ equipped with a reflexive, idempotent element $R_{\rm free}$.
			Note that $\mathcal{P}(R_{\rm free})$ is a subquantale of $\mathcal{P}(R)$.
			The quantale $\mathcal{P}(R)$ of all resources is commonly referred to as the ``enveloping theory'', which may be shared by multiple resource theories identified by distinct choices of the subquantale $\mathcal{P}(R_{\rm free})$ of free resources---the ``free subtheory''.
			
			More generally, a \textbf{universally combinable resource theory} is a commutative quantale $\mathcal{R}$ equipped with a reflexive, idempotent element $R_{\rm free}$.
			Including lattices that are not uniquely atomistic in the definitions allows for a broader class of examples, many of which arise as quotients (see \cref{sec:ucrt_order,sec:ucrt_catalysis,sec:rt_order}).
			By the construction as in \cref{ex:module_from_morphism}, every universally combinable resource theory defines a quantale module $(\mathcal{R}_{\rm free}, \mathcal{R}, \boxtimes)$ where $\mathcal{R}_{\rm free}$ is the interval $\langle \bot, R_{\rm free} \rangle$ within $\mathcal{R}$.
			Of course, if $\mathcal{R}$ is uniquely atomistic and thus a power set lattice, then $\mathcal{R}_{\rm free}$ is isomorphic to $\mathcal{P}(R_{\rm free})$.
			
			One might think that it would be worth considering free resources that correspond to more general subquantales than intervals of the form $\langle \bot, R_{\rm free} \rangle$ for some $R_{\rm free}$.
			However, because of their interpretation, for any $S \in \mathcal{R}_{\rm free}$ we require that all elements of $\mathcal{R}$ below it are also free.
			In the concrete case, this is very clear: 
			If a set of individual resources is free, then every subset thereof should likewise be free.
			
			Our go-to examples of universally combinable resource theories arise from process theories.
			There, $\boxtimes$ can be interpreted as mapping a pair of processes to an object that captures all the ways in which the two could be composed.
			Making this idea more precise is not trivial; we merely sketch the intuition here.
			
			In \cref{ex:quantale_for_UCRT1}, we saw a way to take a process theory and define a monoidal operation that includes both sequential and parallel composition.
			However, this construction does not necessarily capture the idea of an operation that encodes \emph{all} the ways in which $s$ and $t$ could be composed.
			In particular, at every stage of an iterated composition like $r \star (s \star t)$ defined there, the resulting processes are assumed to be of \emph{first-order}.
			That is, they take a certain collection of inputs and subsequently produce corresponding outputs, but leave no possibility to adjust their internal behaviour subsequently.
			In order to allow for future compositions to ``slot in'' an additional process in between two existing ones, we are pushed to think of \emph{higher-order} processes in one way or another.			
			\begin{example}\label{ex:quantale_for_UCRT2}
				The commutative version of the construction in \cref{ex:quantale_for_UCRT1} is probably the easiest to understand in terms of the depiction in \eqref{eq:universal_composition_2}.
				Specifically, rather than placing the processes in order, we can consider all the possible permutations of their placement
				\begin{equation}\label{eq:universal_composition_3}
					\input{universal_composition_3.tikz}
				\end{equation}
				followed by all possible contractions once again.
				This gives a valid notion of universal composition $s \boxtimes t$ of first-order processes $s$ and $t$ in the sense of \cite{Coecke2016} and generalizes to any number of processes.\footnotemark{}
				\footnotetext{Strictly speaking, $s \boxtimes t$ in \cite{Coecke2016} does not denote a collection of processes.
					It merely stands for a multiset of unplaced processes.
					However, their placement ultimately leads to an embedding into a background causal structure.
					Just like in our example, the choice of the embedding is not fixed among the different placements.}%
				However, $\boxtimes$ can only reduce to a non-trivial associative binary operation if we include higher-order processes \cite{wilson2021mathematical}.
				For instance, consider three process $s$, $t$ and $u$ (each with a single input and a single output, say).
				Associativity of the operation $\boxtimes$ would mean that we have
				\begin{equation}
					(s \boxtimes t) \boxtimes u = s \boxtimes (t \boxtimes u)
				\end{equation}
				and commutativity of $\boxtimes$ implies\footnotemark{} that the sequential composite $t \circ s \circ u$ is an element of the left-hand side $(s \boxtimes t) \boxtimes u$.
				\footnotetext{Here we assume, for concreteness, that $\boxtimes$ subsumes sequential composition of processes.}%
				The contradiction emerges by recognizing that there is no first-order process in $t \boxtimes u$ (as given by \eqref{eq:universal_composition_2} for example), which could generate $t \circ s \circ u$ by a composition with the process $s$.
				In other words, here the contractions should be allowed to leave ``holes'', unlike those in \eqref{eq:universal_composition_2}.
				Schematically, we can depict the contractions for one of the permutations in \eqref{eq:universal_composition_3} as 
				\begin{equation}\label{eq:universal_composition_4}
					\input{universal_composition_4.tikz}
				\end{equation}
				where, even though the first and last element of the right-hand side display the same connectivity as diagrams, we distinguish them, and thus have to go beyond the notion of a process theory as introduced in \cref{sec:proc_th}.
				Specifically, we think of the first one as a first-order process and the other as a higher-order process---in this case a so-called \mbox{comb \cite{chiribella2009theoretical}}.
				The light blue boxes in the background serve merely as visual cues here, but they may be interpreted as devices whose inner workings are fixed as far as future manipulation is concerned.
				Composition of higher-order processes of this kind can be then described in terms of $\boxtimes$ for first-order processes.
				Heuristically, we can describe it as
				\begin{equation}\label{eq:universal_composition_5}
					\input{universal_composition_5.tikz}
				\end{equation}
				where the box labelled as $s \boxtimes u$, for instance, does not refer to a process.
				It is a collection of processes (both first-order and higher-order) as specified by $s \boxtimes u$.
			\end{example}
			\Cref{ex:quantale_for_UCRT2} is just one of the possible ways of generating a commutative quantale from a process theory.
			Alternatively, one could include the possibility to form of combs with side-channels and compositions that generate more intricate causal orderings than total orders.
			A useful concept for building a descriptive language of such theories is that of causal channels \cite[section 4.1]{houghton2021mathematical}, appropriately generalized to higher-order processes \cite{kissinger2017categorical,wilson2021causality}.
			
		\subsection{Resource Ordering}\label{sec:ucrt_order}
	
			The allowed (or free) conversions between resources are those that arise via a composition with elements of $R_{\rm free}$.
			Given a concrete universally combinable resource theory, we define the order relation among individual resources $r,s \in R$ by
			\begin{equation}
				\label{eq:Order of resources}
				r \succeq s  \quad \coloniff \quad  R_{\rm free} \boxtimes r \ni s.
			\end{equation}
			The order relation captures whether $r$ can be converted to $s$ by composition with free resources.
			It can be used to determine the value of resources with respect to the choice of the partition of $R$ into free and non-free resources.
			If $r$ can be converted to $s$ for free, denoted $r \succeq s$ as above, then $r$ is no worse than $s$ as a resource in the resource theory.
			
			Similarly, we can define the ordering of sets of resources by
			\begin{equation}
				\label{eq:Order of sets of resources}
				S \succeq T  \quad \coloniff \quad   R_{\rm free} \boxtimes S \supseteq T
			\end{equation}
			for any $S,T \in \mathcal{R}$.
			Indeed, this is nothing but the $R_{\rm free}$-reachability preorder (see \cref{def:reachability}).
			\begin{definition}\label{def:ucrt_order}
				The \textbf{resource ordering} of a universally combinable resource theory $(\mathcal{R}, \boxtimes)$ with free resources $R_{\rm free}$ is the $R_{\rm free}$-reachability preorder. 
			\end{definition}
			Therefore, by \cref{lem:access_is_sensible}, $(\mathcal{R}, \succeq)$ is a preordered-extended suplattice. 
			
			In the framework of ordered commutative monoids \cite{Fritz2017}, there is a compatibility between the order relation and the monoid operation as one of the axioms.
			Specifically, one demands that if $S$ is a better resource than $T$ is, then the conjunction $S \boxtimes U$ of $S$ with a fixed resource $U$ is a better resource than $T \boxtimes U$ is.
			Here, we can derive a corresponding property from the definition of a universally combinable resource theory.
			\begin{lemma}
				\label{thm:Compatibility of star and order}
				Let $\UCRT{R}$ be a universally combinable resource theory with the corresponding order relation $\succeq$ defined as in~\eqref{eq:Order of sets of resources}.
				For any three elements $S,T,U$ of $\mathcal{R}$, we have
				\begin{equation}
					S \succeq T \implies S \boxtimes U \succeq T \boxtimes U.
				\end{equation}
			\end{lemma}
			\begin{proof}
				By the definition of $\succeq$, we have $S \succeq T \iff R_{\rm free} \boxtimes S \supseteq T $, which implies 
				\begin{equation}
					(R_{\rm free} \boxtimes S) \boxtimes U \supseteq T \boxtimes U.
				\end{equation}
				Via the associativity of $\boxtimes$, we can then conclude that $S \boxtimes U \succeq T \boxtimes U$ must hold whenever $S$ is above $T$ according to the order relation $\succeq$.
			\end{proof}
		
			Given the preorder $\ordres$, two resources $r$ and $s$ are said to be equivalent if they can be converted one to another for free.
			That is, we write $r \sim s$ if both $r \succeq s$ and $s \succeq r$ hold. 
			From the point of view of questions regarding resource convertibility, there is no need to distinguish equivalent resources unless distinguishing them provides a more convenient representation.
			\begin{lemma}\label{lem:ucrt_cong}
				The equivalence relation $\sim$ associated to the resource ordering of a universally combinable resource theory $\UCRT{R}$ is a quantale congruence.
			\end{lemma}
			\begin{proof}
				By idempotence of $R_{\rm free}$, we have that $S \freeconv T$ implies
				\begin{equation}
					R_{\rm free} \boxtimes S = R_{\rm free} \boxtimes R_{\rm free} \boxtimes S \supseteq R_{\rm free} \boxtimes T
				\end{equation}
				so that $S \sim T$ implies 
				\begin{equation}
					R_{\rm free} \boxtimes S = R_{\rm free} \boxtimes T.
				\end{equation}
				Since $R_{\rm free} \supseteq 1$ holds by definition, we obtain the converse.
				In other words, $\sim$ coincides with the relation obtained from $R_{\rm free}$-augmentation (\cref{def:augment}) and it is thus a congruence.
			\end{proof}
			Therefore, we get the quotient quantale $\linefaktor{\mathcal{R}}{{\sim}}$ whose elements are equivalence classes of resources in $\mathcal{R}$.
			As described in equations \eqref{eq:quantale_quotient} and \eqref{eq:quantale_quotient2}, the quantale composition (denoted by $+$) and lattice ordering (denoted by $\geq$) are given by
			\begin{equation}
				\begin{gathered}
					[S] + [T] = [S \boxtimes T], \\
					[S] \geq [T] \quad \iff \quad R_{\rm free} \boxtimes S \supseteq R_{\rm free} \boxtimes T.
				\end{gathered}
			\end{equation}
			In particular, the lattice ordering in $\linefaktor{\mathcal{R}}{{\sim}}$ is equivalent to the resource ordering in $\mathcal{R}$! 
			That is, we have $[S] \geq [T] \iff S \freeconv T$.
			Since $[R_{\rm free}]$ is the unit of $\linefaktor{\mathcal{R}}{{\sim}}$, we can view $\geq$ equally well as a resource ordering in the universally combinable resource theory specified by $\linefaktor{\mathcal{R}}{{\sim}}$.
			
			Altogether, we obtain an ordered commutative monoid $(\linefaktor{\mathcal{R}}{{\sim}}, +, \geq)$ with appropriate resource-theoretic interpretation that can be studied with methods introduced in \cite{Fritz2017} for instance.
			Since it is also a complete lattice, instead of thinking of it as a generic ordered commutative monoid, it is perhaps best construed as a completion of one.
			
			\begin{example}\label{ex:ucrt_from_ocm}
				Conversely, to any ordered commutative monoid $(A, +, \geq)$, one can associate a commutative quantale $(\mathcal{DC}(A), \bigcup, \boxtimes)$ of downsets of $A$, which arises as a quotient of the free quantale associated to $A$.
				It is ordered by set inclusion with composition given by
				\begin{equation}
					S \boxtimes T \coloneqq  \down (S + T)
				\end{equation}
				for any $S,T \in \mathcal{DC}(A)$.
				More details about this construction and its context can be found in \cite{Tsinakis2004}, whose theorem 5.1 is particularly relevant.
				One basic property that may be worth noting is that the associated lattice of downsets is not atomistic in general.
			\end{example}
			
		\subsection{Free Subtheories}\label{sec:ucrt_free}
			
			Given a commutative quantale $\mathcal{R}$, many resource theories can arise via different choices of the free resources.
			For instance, in quantum theory, the distinction between a resource theory of athermality and that of bipartite entanglement is entirely due to the choice of the free subtheory.
			Therefore, understanding the relationships between different choices of free resources (and their impact on the resource ordering) is a crucial step towards understanding the interactions of different kinds of resources.\footnotemark{}
			\footnotetext{Here, ``kind of resources'' refers not to elements of $\mathcal{R}$ but rather to the attributes of these that make them useful in a particular resource theory; e.g.\ asymmetry, athermality, nonuniformity, etc.}%
			
			Let $\mathcal{F}(\mathcal{R})$ be the set of all reflexive and idempotent elements of $\mathcal{R}$.
			That is,
			\begin{equation}
				\label{eq:Free sets definition}
				\mathcal{F}(\mathcal{R}) \coloneqq \Set*[\big]{ B \in \mathcal{R}  \given  1 \subseteq B \text{ and } B \boxtimes B = B }.
			\end{equation}
			These are in bijective correspondence with the universally combinable resource theories whose enveloping theory is $\mathcal{R}$.
			We now show that $\mathcal{F}(\mathcal{R})$ is a monoid and a complete lattice and phrase this result for arbitrary (not necessarily commutative) quantales.
			\begin{proposition}
			\label{thm:Free sets properties}
				Given a quantale $(\mathcal{Q}, \bigvee, \star)$ with top element $\top$, $\mathcal{F}(\mathcal{Q})$ has the following properties:
				\begin{enumerate}
					\item \label{item:Free sets 1}
						Both $1$ and $\top$ are elements of $\mathcal{F}(\mathcal{Q})$.
					
					\item \label{it:infimum} $\mathcal{F}(\mathcal{Q})$ is closed under arbitrary infima (in $\mathcal{Q}$).
					
					\item \label{it:monoid} If $\mathcal{Q}$ is commutative, then $\mathcal{F}(\mathcal{Q})$ is closed under $\star$.
				\end{enumerate}
			\end{proposition}
			\begin{proof}
				Properties \ref{item:Free sets 1} and \ref{it:monoid} are clear from the definition of $\mathcal{F}(\mathcal{Q})$.
				
				We can show that $\mathcal{F}(\mathcal{Q})$ is closed under infima as follows.
				First of all note that in a suplattice, such as $(\mathcal{Q}, \bigvee)$, we can express infimum of any set $\mathcal{B} \subseteq \mathcal{Q}$ as the supremum of $\mathcal{L}(\mathcal{B})$---the set of all its lower bounds defined as
				\begin{equation}
					\mathcal{L}(\mathcal{B}) = \Set*[\big]{ L \given \forall \, B \in \mathcal{B} : L \subseteq B }.
				\end{equation}
				For any $B \in \mathcal{F}(\mathcal{Q})$ and any $L$ such that $L \subseteq B$ holds, we also have
				\begin{equation}
					L \star L \subseteq B \star B = B,
				\end{equation}
				Thus, if $\mathcal{B}$ is a subset of $\mathcal{F}(\mathcal{Q})$, then for any lower bound $L$ of $\mathcal{B}$, we have $L \star L \in \mathcal{L}(\mathcal{B})$ and also
				\begin{equation}
					\begin{split}
						\left( \bigwedge \mathcal{B} \right) \star \left( \bigwedge \mathcal{B} \right) &=  \left( \bigvee \mathcal{L}(\mathcal{B}) \right)  \star \left( \bigvee \mathcal{L}(\mathcal{B}) \right) \\
							&=  \bigvee_{L \in \mathcal{L}(\mathcal{B})} L \star L  \;\subseteq\;  \bigvee_{L' \in \mathcal{L}(\mathcal{B}) } L' \\
							&= \bigwedge \mathcal{B} .
					\end{split}
				\end{equation}
				Furthermore, if we have $B \supseteq 1$ for all $B \in \mathcal{B}$, then $1$ is a lower bound of $\mathcal{B}$ so that 
				\begin{equation}
					\bigwedge \mathcal{B} \supseteq 1
				\end{equation}
				holds and $\bigwedge \mathcal{B}$ is an element of $\mathcal{F}(\mathcal{Q})$.
				Thus $\mathcal{F}(\mathcal{Q})$ is closed under $\bigwedge$.
			\end{proof}
			Let us try to interpret the concrete and commutative case. 
			The above result shows that $\mathcal{F}(\mathcal{R})$ embeds into $\mathcal{R}$ both as a monoid under $\boxtimes$ and as an inflattice (the dual concept to a suplattice) under intersections.
			Consequently, it also has arbitrary suprema.
			However, these need not, and usually do not, coincide with unions in $\mathcal{R}$!
			Instead, given $B, B' \in \mathcal{F}(\mathcal{R})$, the supremum of $\{B,B'\}$ in $\mathcal{F}(\mathcal{R})$ represents the smallest set of resources closed under composition that subsumes both $B$ and $B'$.
			In other words, it is the infimum of all upper bounds of $\{B,B'\}$ in $\mathcal{F}(\mathcal{R})$.
			Thus, we would not expect $\mathcal{F}(\mathcal{R})$ to be a subquantale of $\mathcal{R}$ and in fact it need not be a quantale at all.

		\subsection{Catalysis}\label{sec:ucrt_catalysis}
		
			Besides questions about convertibility by free operations, another class of basic questions that can be posed in the abstract regards the use of catalysts.
			In this case, a non-free resource $c$ may be used to aid in the conversion, but we require that it is returned unchanged after the conversion is done.
			Thus, it can also be reused in future conversions.
			In particular, in a concrete universally combinable resource theory $\UCRT{R}$, we can say that it is possible to convert $S$ to $T$ with the help of a catalyst $c \in R$ if
			\begin{equation}\label{eq:catalytic_order}
				c \boxtimes S \freeconv c \boxtimes T
			\end{equation}
			or equivalently 
			\begin{equation}
				c \boxtimes ( R_{\rm free} \boxtimes S )  \supseteq  c \boxtimes T
			\end{equation}
			holds.
			We denote the possibility of a $c$-catalytic conversion by $S \freeconv_c T$.
			This notion is of interest because in many resource theories (such as those of entanglement \cite{turgut2007catalytic} and athermality \cite{Brandao2015}), there are conversions that are possible with a catalyst, but impossible otherwise.
			Moreover, the resource ordering tends to simplify greatly if we allow for an \emph{arbitrary} catalyst $c$ to aid in the conversion, as explored in the abstract setting in \cite{Fritz2017}.
			
			However, note that in our set-up the use of an arbitrary catalyst \emph{cannot} be expressed by replacing $c$ in \eqref{eq:catalytic_order} with the set of all resources $R$.
			More concretely, if we want to express the possibility to convert $S$ to $T$ with the help of \emph{either} $c$ or $c'$, both individual resources, we cannot write this as 
			\begin{equation}\label{eq:catalyst_set}
				(c \cup c') \boxtimes S \freeconv (c \cup c') \boxtimes T.
			\end{equation}
			This is because the condition in \eqref{eq:catalyst_set} is true whenever both
			\begin{align}
				c \boxtimes S &\freeconv c' \boxtimes T  & &\textit{and} &  c' \boxtimes S &\freeconv c \boxtimes T
			\end{align}
			hold, for example.
			However, order relation \eqref{eq:catalyst_set} does not hold in case $S \freeconv_c T$ holds but neither $S \freeconv_{c'} T$ nor $c \boxtimes S \freeconv c' \boxtimes T$ do.
			Therefore, catalytic convertibility as expressed in \eqref{eq:catalytic_order} only attains the correct interpretation if $c$ is an individual resource, or else if $c$ can be written as a universal combination of individual resources.
			In the latter case, the catalyst allows the joint use of $c$ and $c'$ together as opposed to a choice of one of them.
			Instead of \eqref{eq:catalyst_set}, the preorder that represents the free conversion of $S$ to $T$ with the help of \emph{either} $c$ or $c'$ is the union of the two relations $\freeconv_c$ and $\freeconv_{c'}$.
			An extensive study of catalytic convertibility and its consequences can be found in \cite[section 4]{Fritz2017}.
			Specifically, the case considered there is the one in which arbitrary catalyst is allowed.
			
			\begin{remark}
				Recall that we can define the quotient quantale $\linefaktor{\mathcal{R}}{{\sim_C}}$ where $\sim_C$ denotes the congruence induced by $C$-augmentation as introduced via \eqref{eq:augment_cong}.
				Specifically, notice that the lattice ordering $\supseteq_C$ of $\linefaktor{\mathcal{R}}{{\sim_C}}$ is given by
				\begin{equation}
					[S]_C \supseteq_C [T]_C  \quad \iff \quad  C \boxtimes S \supseteq C \boxtimes T
				\end{equation}
				for arbitrary elements $[S]_C, [T]_C$ of the quotient.
				Therefore, the resource ordering $\freeconv_C$ in the universally combinable resource theory specified by the quotient $\linefaktor{\mathcal{R}}{{\sim_C}}$ with free resources $[R_{\rm free}]_C$ satisfies
				\begin{equation}\label{eq:quot_catal}
					[S]_C \freeconv_C [T]_C  \quad \iff \quad  S \freeconv_C T
				\end{equation}
				because it coincides with $\supseteq_C$.
				The relation $\freeconv_C$ on the right hand side of \eqref{eq:quot_catal} denotes the $C$-catalytic ordering in the original resource theory as specified by \eqref{eq:catalytic_order}.
			\end{remark}
	
	\section{General Resource Theories}\label{sec:rt}
	
		\subsection{The Intuition}\label{sec:rt_intuition}
			
			Many concrete quantum resource theories have been studied in recent years~\cite{Chitambar2018}. 
			For each of these, one may identify a choice of free quantum processes and a corresponding universally combinable resource theory that describes it.
			However, studies of specific resource theories rarely take universal combination as the relevant notion of composition.
			Instead, the general approach is to restrict resources to be of a specific type.
			In process theories, these would most commonly be either states or channels.
			Processes that model the conversions between resources are then typically of a different type.
			If the resources of interest are restricted to states, it suffices to consider conversions by channels; while if the resources are channels, we may choose the conversions between them to be combs \cite{chiribella2009theoretical}.
			
			Therefore, to establish a connection with the framework of universally combinable resource theories, one would generally have to study a much larger theory that is universally combinable or restrict attention to a subtheory, which is not a universally combinable resource theory itself.
			Combined with the issues that arise when constructing a well-characterized and well-defined universally combinable resource theory, this suggests that universally combinable resource theories do not present a general-purpose framework for describing resource theories.
			In this section, we introduce a more flexible framework. 
			
			The first departure from the presentation of universally combinable resource theories is that, besides the set of resources $X$, we consider a separate set of transformations $T$.
			This mimics the conceptual distinction that occurs in practice---between the entities that constitute resources (which we are interested in comparing and quantifying) and the entities that mediate their conversions.
			Despite such differentiation, it is often the case that resources can themselves be used for conversions.
			For instance, if the resources are quantum states, then any state $y \in X$ may be used to facilitate the conversion $x \mapsto x \otimes y$ and thus it also constitutes a \emph{transformation} of resources.
			To accommodate this, one could identify the set of resources $X$ with a subset of the set of transformations $T$, but we do not make such an assumption explicitly.
			
			Just as in a universally combinable resource theory, we presume that transformations can be composed to form new ones, possibly in multiple ways specified by a quantale $\mathcal{T}$ (see \cref{def:quantale}).
			The simplest non-trivial example would be that of a power set lattice $\mathcal{P}(T)$ equipped with a monoid operation among its atoms (see \cref{ex:free_quantale}).
			However, unlike for universal compositions described in \cref{sec:ucrt}, we do not wish to restrict our attention to commutative quantales.
			This is to allow an interpretation of the composition $S \star U$ of two sets of transformations $S,U \in \mathcal{P}(T)$ in a sequential fashion, so that application of a transformation from the set $U$ can be thought to precede the application of a transformation from $S$ (e.g., see \cref{ex:quantale_for_UCRT1}).
			
			Regarding the set of resources, one might want to stipulate a composition $\otimes$ among them that would describe conjunction in the sense that $x \otimes y$ represents an agent's access to both $x$ and $y$ for the purposes of future manipulation.
			In this work, we choose to deny resources this structure in order to focus the investigation on conclusions that can be made in its absence.
			Since this is undoubtedly an important piece of the puzzle in many resource theories, it should be incorporated in some of the future studies.
			Here, we merely assume one can form ``disjunctions'' $x \cup y$, whose meaning is that an agent can choose to access one of either the resource $x$ or the resource $y$ for manipulation by transformations.
			As such, we take the resources to be represented by a suplattice (see \cref{def:suplattice}).
			
			The final piece of structure that we need is one that tells us how to convert resources.
			Whenever both lattices are uniquely atomistic and thus power sets, we expect that for each transformation $t \in T$ there is a corresponding assignment $t \apply \ph \colon X \to \mathcal{P}(X)$ that maps $x$ to the set of resources one may obtain by applying $t$ to $x$.
			We can think of $\apply$ as an action $\mathcal{P}(T) \times \mathcal{P}(X) \to \mathcal{P}(X)$ which should satisfy that 
			\begin{enumerate}
				\item \label{it:rt_seq} applying two transformations in succession leads to a composition\footnotemark{} of the corresponding maps of resources,
				\footnotetext{One can think of it as a composition of relations $X \to X$ or as a composition in the Kleisli category of the powerset monad.
				In the context of our set-up, we view it as a composition of suplattice homomorphisms $\mathcal{P}(X) \to \mathcal{P}(X)$.}%
				\item being able to choose whether to apply $t$ to $x$ or $y$ leads to a union of potential outcomes $t \apply x$ and $t \apply y$, 
				\item being able to choose whether to apply $t$ or $s$ to $x$ leads to a union of potential outcomes $t \apply x$ and $s \apply x$, and
				\item applying the identity transformation does not change the resource.
			\end{enumerate}
			These conditions correspond to equations \eqref{eq:qmodule}, i.e.\ they ensure we have a quantale module $(\mathcal{P}(T), \mathcal{P}(X), \apply)$ where $\mathcal{P}(T)$ acts on $\mathcal{P}(X)$ via $\apply$.
			
			Generalizing to the case of lattices that need not be uniquely atomistic, we get a quantale module $(\mathcal{T}, \mathcal{X}, \apply)$.
			It specifies the resources, their transformations and the way in which differential access to them influences their compositionality.
			For each quantale module, however, there can be many resource theories.
			To specify a resource theory, we also need to identify the free transformations within $\mathcal{T}$.
			For example, both a (quantum) resource theory of athermality (\cref{ex:rt_athermality}) and a resource theory of entanglement (\cref{ex:rt_entanglement}) arise from the same quantale module in which resources are quantum states and transformations are quantum channels.
			They differ by a choice of free transformations---i.e.\ thermal operations or LOCC channels.
			We thus get the following definition of a resource theory in terms of quantale modules.
			This is the most abstract notion of a resource theory that we consider in this thesis.
			\begin{definition}\label{def:rt}
				A \textbf{resource theory} is a quantale module $(\mathcal{T}, \mathcal{X}, \apply)$ equipped with a reflexive and idempotent element $T_{\rm free}$ of $\mathcal{T}$. 
				Whenever both $\mathcal{T}$ and $\mathcal{X}$ are uniquely atomistic lattices, we speak of a \textbf{concrete resource theory} and $T_{\rm free}$ is the \textbf{set of free transformations}.
			\end{definition}
			Note that since the free element $T_{\rm free}$ is assumed to be idempotent (see \cref{def:idmptnt}) and reflexive (i.e.\ it subsumes the unit $1 \in \mathcal{T}$), the downward closure thereof with respect to the lattice ordering $\supseteq$ is a subquantale of $\mathcal{T}$.
			The notation we use for this subquantale is $\mathcal{T}_{\rm free}$ and we refer to it as the \textbf{subquantale of free transformations}.
			However, it is not a generic subquantale.
			An alternative, albeit equivalent, way to express \cref{def:rt} is as follows, taking into account that $\mathcal{T}_{\rm free}$ should be a downset with respect to $\supseteq$:
			\begin{quote}
				A resource theory (according to \cref{def:rt}) is a quantale module $(\mathcal{T}, \mathcal{X}, \apply)$ equipped with a \emph{downward closed subquantale} $\mathcal{T}_{\rm free}$ of $\mathcal{T}$.
			\end{quote}
			Since a subquantale has a top element (the supremum of all its elements), a downward closed subquantale necessarily coincides with the downward closure of its top element.
			In this context, we can denote its top element by $T_{\rm free}$, which clarifies the sense in which the two definitions of a resource theory are equivalent.
			In a concrete resource theory (as given by \cref{def:rt}), the downward closed subquantale $\mathcal{T}_{\rm free}$ consists of all subsets of $T_{\rm free}$---the set of all free transformations.
			That is, the subquantale of free transformations $\mathcal{T}_{\rm free}$ is given by the power set $\mathcal{P}(T_{\rm free})$.
			
			For many questions in resource theories, it suffices to specify the quantale module $(\mathcal{T}_{\rm free}, \mathcal{X}, \apply)$. 
			The reason is that the answers to these questions are independent of the quantale $\mathcal{T}$ in which the subquantale of free transformations is embedded.
			Even more abstractly, one may forget about the origin of the convertibility altogether and study the emerging order relation among resources on its own merits \cite{Fritz2017}.
			Our decision to include the enveloping quantale $\mathcal{T}$ in the description of a resource theory, as the authors of \cite{Coecke2016} do, is motivated on the following grounds:
			\begin{enumerate}
				\item It is more closely connected to the practice of the study of resource theories.
				\item Understanding the structure of non-free transformations and their interactions with free transformations can be used to obtain resource measures (\cref{ex:Advantage of generalized yield}).
				\item Specification of $\mathcal{T}$ is necessary to define certain classes of interesting resource theories, such as resource theories of asymmetry (\cref{sec:rt_asymmetry}).
				\item Knowing the enveloping theory is crucial also when we want to compare alternative choices of free transformations and when we want to study interactions among different resource theories.
				\item It allows us to study protocols for generating non-free transformations by combining free transformations with given resources algebraically at the level of the quantale of transformations. 
				Additionally, the resulting non-free transformation may be used for other kinds of resource conversions than the one we were interested in achieving in the first place, thus allowing us to repurpose protocols in a new context. 
			\end{enumerate}

		\subsection{Examples}\label{sec:rt_examples}
			
			Our presentation in \cref{sec:ucrt} makes it clear that every universally combinable resource theory is also a resource theory in the sense of \cref{def:rt}, i.e.\ a quantale module with a distinguished subquantale.
			Since every ordered commutative monoid defines a corresponding universally combinable resource theory as we describe at the end of \cref{sec:ucrt_order}, we can thus think of it as a resource theory too.
			However, there are two other, quite different, ways to do so.
			\begin{example}\label{ex:ocm_as_rt}
				Given an ordered commutative monoid $(A, +, \geq)$, let $\mathcal{X}$ be the power set of $A$ and let $\mathcal{T}$ be the quantale of relations of type $A \to A$.
				Their suprema are given by unions when thought of as subsets of $A \times A$ and their composition is just the standard composition of relations.
				The unit relation is of course given by the identity function.
				Every relation $f \subseteq A \times A$ gives rise to a suplattice homomorphism $f_* \colon \mathcal{P}(A) \to \mathcal{P}(A)$ via
				\begin{equation}
					f_* (Y) = \Set*[\big]{ z \in A \given  (y,z) \in f \text{ for some } y \in Y }.
				\end{equation}
				The assignment $f \mapsto f_*$ preserves composition and thus we get a quantale module $(\mathcal{T}, \mathcal{X}, \apply)$ with
				\begin{equation}
					f \apply Y \coloneqq f_* (Y).
				\end{equation}
				The (partial order) relation $\geq$ subsumes the unit relation by reflexivity and it is idempotent by transitivity.
				In the associated resource theory with free transformations given by $\geq$, the resource ordering that we introduce in the next section coincides with $\geq$ on singletons and can be described as the enhancement preorder (\cref{def:enh_ord}) on arbitrary sets.
				Even though it sheds light on some of the definitions, this is not a useful representation of the ordered commutative monoid per se since it eliminates the compositional structure of resources given by the monoid \mbox{operation $+$}. 
				That is, in a quantale module constructed from an ordered commutative monoid in this way, there is no explicit way to model the `conjunction' $a+b$ of individual resources $a$ and $b$.
				In order to do so, we would need to supplement the framework introduced in \cref{sec:rt_intuition} with an additional operation on the resources capturing their conjunctions.
				While this is certainly an important piece of structure in many circumstances, it goes beyond the considerations of this thesis and its analysis is left for future investigations.
				Besides catalysis, it would also allow one to describe many other concepts of interest: catalysis, asymptotic (many-copy) convertibility, universality, and others.
				As such, it is one of the most important extensions of the framework as presented here.
			\end{example}
			\begin{example}\label{ex:ocm_as_rt2}
				Alternatively, we may include the information about possible transformations in the action $\apply$ rather than in the identity of transformations.
				That is, let the quantale of transformations be $\mathcal{B} = \{0,1\}$, interpreted as Boolean truth values with supremum $\cup$ given by the disjunction $\lor$ (logical OR operation) and quantale multiplication $\star$ given by the conjunction $\land$ (logical AND operation).
				Then an ordered commutative monoid $(A, +, \geq)$ can be associated with a quantale module $(\mathcal{B}, \mathcal{P}(A), \apply)$ with 
				\begin{equation}
					\begin{split}
						0 \apply Y &= \emptyset \\
						1 \apply Y &= \down (Y)
					\end{split}
				\end{equation}
				for any subset $Y$ of $A$.
				Here, $\down$ is the downward closure operator (\cref{def:order_closure}) with respect to $\geq$.
			\end{example}
			
			Next, we provide several ways in which resource theories constructed from process theories \cite{Coecke2016} fit into our framework.
			The scope of these is general enough to cover most, if not all, resource theories studied in quantum information theory at present.
			\begin{example}\label{ex:rt_states}
				In a resource theory of states associated to a partitioned process theory, the lattice of resources $\mathcal{X}$ is the power set of all states, i.e.\ processes with a trivial input.
				The underlying lattice of the quantale of transformations is likewise a power set lattice.
				Its atoms are processes with arbitrary input and output.
				In the most elementary version of this construction, composition of transformations in the quantale is merely given by their sequential composition as in \cref{ex:free_quantale_cat}.
				Similarly, they act on states by sequential composition, as depicted below (assuming the relevant wire types match).
				\begin{equation}\label{eq:rt_states_1}
					\input{rt_states_1.tikz}
				\end{equation}
			\end{example}
			\begin{example}\label{ex:rt_states_2}
				Alternatively, we can take the quantale $\mathcal{T}$ of resource transformations to be the one introduced in \cref{ex:quantale_for_UCRT1}.
				In this case, both parallel and sequential compositions are incorporated in the description of $\star$, e.g., the left equation of \eqref{eq:rt_states_1} becomes
				\begin{equation}\label{eq:rt_states_5}
					\input{rt_states_5.tikz}
				\end{equation}
				Note that due to our choice of composition operation $\star$, i.e.\ one that includes all wirings irrespective of the positioning of wires on the page, we need not distinguish the diagram
				\begin{equation}\label{eq:rt_states_6}
					\input{rt_states_6.tikz}
				\end{equation}
				A more complicated example of composition of transformations is
				\begin{equation}\label{eq:rt_states_3}
					\input{rt_states_3.tikz}
				\end{equation}
				where greek letters denote types of wires which serve to restrict the valid contractions.
				In this construction, transformations act on resources by arbitrary wirings that yield a state (as opposed to a process with non-trivial inputs), e.g.,
				\begin{equation}\label{eq:rt_states_2}
					\input{rt_states_2.tikz}
				\end{equation}
				That is, all inputs of a channel must be contracted with some output of the state it is acting on, but there may be some outputs of the state which are left uncontracted.
				This action differs substantially from the one described in \cref{ex:rt_states}.
				Since the input type $\alpha \otimes \beta$ of $s$ in \cref{eq:rt_states_2} does not match the output type $\alpha \otimes \beta \otimes \alpha$ of $x$, the right hand side of \cref{eq:rt_states_2} would be the empty set in the case of the construction from \cref{ex:rt_states}.
				
				Note that the quantale unit here is a singleton---it only contains the empty diagram $\id \colon I \to I$.
				Nevertheless, its action via \eqref{eq:rt_states_2} has every resource as a fixed point as the quantale module definition demands.
			\end{example}
			
			Given a quantale module, identifying a specific resource theory then involves a choice of a subquantale of free transformations.
			In the case of the quantale module constructed from a process theory as in \cref{ex:rt_states}, we can define a resource theory of states by choosing a subcategory.\footnotemark{}
			\footnotetext{Strictly speaking, we want it to be a \emph{wide} subcategory, one that has the same objects (i.e.\ wire labels) as its enveloping category (given by the process theory).}%
			Letting the set of free transformations $T_{\rm free}$ consist of all processes in this subcategory ensures all the defining properties of a resource theory are satisfied.
			In particular, a subcategory contains all identity processes, so that $T_{\rm free}$ is reflexive.
			Closure under sequential composition then ensures that $T_{\rm free}$ is also idempotent.
			
			In the case of the quantale module constructed from a process theory as in \cref{ex:rt_states_2}, defining a resource theory of states amounts to choosing a sub-process theory thereof.
			That is, $T_{\rm free}$ is a set of processes within the original process theory, which includes the empty diagram and is closed under wiring of diagrams.\footnotemark{}
			\footnotetext{Note that the diagrams in question are really circuit diagrams and thus the wirings also have to be such that no cycles are generated, see \cref{sec:proc_th} for more details.}%
			By virtue of these properties of a sub-process theory, $T_{\rm free}$ is both reflexive and idempotent.
				 
			\begin{example}\label{ex:rt_channels}
				Resource theories of channels follow very much the same suit as those of states.
				We just need to replace states in the lattice of resources $\mathcal{X}$ by arbitrary channels.
				If all second-order processes can be reduced to a composition of first-order processes, then channels in the quantale of transformations $\mathcal{T}$ are replaced by 1-combs of the form
				\begin{equation}\label{eq:comb}
					\input{comb.tikz}
				\end{equation}
				They represent joint pre- and post-processing of a resource (i.e.\ a channel) with a side-channel.
				Composition $t \star s$ of combs is perhaps easiest to understand in the same way as the composition of channels in \eqref{eq:universal_composition_2}.
				That is, we first place them (in sequential order) and then perform all valid contractions:
				\begin{equation*}\label{eq:comb_comp}
					\input{comb_comp_2.tikz}
				\end{equation*}
				Namely, the four contractions are (in order as above): no contraction; plugging the right output of $t_1$ to the input of $s_1$; plugging the output of $s_2$ to the right input of $t_2$; and finally using both available contractions.
				The action of combs on channels is obtained in an analogous way.
				Here, we have to require that the resulting process is once again a channel.
				For instance, we have
				\begin{equation}\label{eq:rt_channels}
					\input{rt_channels.tikz}
				\end{equation}
				if all the wire types in the above diagram match.
				Choosing a free subquantale of transformations (i.e.\ a set of free 1-combs) can be done in multiple ways. 
				For instance, we can restrict both channels appearing in the decomposition of a free 1-comb to be elements of a given sub-process theory.
				However, we can also impose different restrictions separately to the pre- and post-processings.
				Additionally, one can limit the free 1-combs to include no side channel, as is the case in Shannon's theory of communication \cite{Shannon1948,shannon1958note} and its quantum version \cite{Devetak2008}.
				As another example, if the free 1-combs are chosen such that
				\begin{itemize}
					\item they have no side-channels,
					\item pre-processings are restricted to identity channels, and
					\item post-processings are unconstrained,
				\end{itemize}
				then, within the process theory of stochastic matrices (i.e.\ classical probability theory of finite sample spaces), we get a resource theory of distinguishability of probabilistic behaviours (see \cref{sec:distinguishability}). 
			\end{example}
	
		\subsection{Resource Ordering}\label{sec:rt_order}
			
			Just like for a universally combinable resource theory, we can now formalize the concept of a resource being ``better'' than another one within a given resource theory by introducing the following ordering of resources.
			\begin{definition}\label{def:free_conv}
				Consider a resource theory associated to a quantale module $(\mathcal{T},\mathcal{X}, \apply)$ with free transformations $T_{\rm free} \in \mathcal{T}$.
				The \textbf{resource ordering} thereof is the $T_{\rm free}$-reachability preorder (see \cref{def:reachability}).
				That is, $Y$ can be converted to $Z$ in the resource theory, denoted $Y \freeconv Z$, if we have
				\begin{equation}\label{eq:free_conv}
					T_{\rm free} \apply Y \supseteq Z.
				\end{equation}
			\end{definition}
			There are multiple viable interpretations of $\freeconv$. 
			Let us focus on the concrete case.
			The most immediate one is that $y \freeconv z$ holds among atoms $y$ and $z$ if there exists a free process that turns $y$ into $z$.
			In this sense, having access to $y$ is better than having access to $z$, since we can always choose to replace $y$ with $z$ at no cost.
			When dealing with resources of information in particular, $y \freeconv z$ may mean that $y$ can simulate $z$ with the free transformations specifying the meaning of ``allowed'' simulations.
			If the resources represent a strategy to play a game, the ordering may specify the ability to transform one strategy to another without increasing the pay-off or utility functions.
			
			In a resource theory of states as introduced in \cref{ex:rt_states}, ordering of idividual resources can be phrased as:
			\begin{equation}\label{eq:rt_states_order}
				\input{rt_states_order.tikz}
			\end{equation}
			
			We can extend any of the above interpretations  to arbitrary elements of $\mathcal{P}(X)$ via the enhancement preorder which coincides with the resource ordering by \cref{lem:access=enh}.
			Specifically, to say that $Y \freeconv Z$ holds is to postulate that for every choice of an element $z$ of $Z$, one can choose a resource $y \in Y$ and a free transformation that converts $y$ to $z$.
			By \eqref{eq:reachability2}, we can also equivalently express $Y \freeconv Z$ via 
			\begin{equation}\label{eq:free_image}
				T_{\rm free} \apply Y \supseteq T_{\rm free} \apply Z.
			\end{equation}
			We call $T_{\rm free} \apply Y$ the \textbf{free image} of $Y$ as it represents the set of resources one has access to if provided with $Y$ and arbitrary free operations.
			Yet another way to express that $Y$ is better than $Z$ is to say that the free image of $Y$ subsumes that of $Z$.
			Note that a set of resources is a free image of another one if and only if it is downward closed (see \cref{def:downset}) according to the resource ordering, or equivalently that it is a $T_{\rm free}$-ideal element (see \cref{def:ideal}).
			
			In the context of universally combinable resource theories, we saw in \cref{thm:Compatibility of star and order} that combining resources with a fixed resource $U$ preserves the resource ordering $\succeq$.
			At times, we can interpret the map $Y \mapsto U \apply Y$ as combining resources with a fixed transformation $U$.
			However, the application of transformations in a quantale module is not commutative in general.
			Applying $U$ to distinct resources may degrade them in ``incompatible'' ways.
			That is, even if $Y$ can be freely converted to some $Z$, there is no guarantee that $U \apply Y$ is still freely convertible to $U \apply Z$.
			Of course, if applying $U$ after free transformations coincides with the transformations afforded by $U$ followed by free transformations, i.e.\ if
			\begin{equation}
				T_{\rm free} \star U = U \star T_{\rm free}
			\end{equation}
			holds, then $U \apply Y$ is guaranteed to be freely convertible to $U \apply Z$ provided that $Y \succeq Z$ holds.
			More generally, we have the following result, from which the earlier \cref{thm:Compatibility of star and order} can be derived.
			\begin{lemma}
				\label{lem:Compatibility of order and combination}
				Consider a resource theory associated to a quantale module $(\mathcal{T},\mathcal{X}, \apply)$ with free transformations $T_{\rm free} \in \mathcal{T}$.
				For any $Y,Z \in \mathcal{X}$ and any $U \in \mathcal{T}$ satisfying $T_{\rm free} \star U \supseteq U \star T_{\rm free}$, we have
				\begin{equation}
					\label{eq:Compatibility of order and combination}
					Y \succeq Z  \quad \implies \quad  U \apply Y \succeq U \apply Z.
				\end{equation}
			\end{lemma}
			\begin{proof}
				Given that $Y \succeq Z$ holds, we can conclude 
				\begin{equation}
					\begin{split}
						T_{\rm free} \apply (U \apply Y) &= (T_{\rm free} \cmps U) \apply Y \\
							&\supseteq (U \cmps T_{\rm free}) \apply Y \\
							&= U \apply (T_{\rm free} \apply Y) \\
							&\supseteq U \apply Z,
					\end{split}
				\end{equation}
				which means that $U \apply Y \succeq U \apply Z$ holds as well.
			\end{proof}
			\begin{example}\label{ex:image_incompatible}
				For a concrete example of a situation in which the lemma above does not apply, consider a toy theory with three resources $X \coloneqq \{ 0, 1 , 2\}$, in which $T$ is the monoid of all functions $X \to X$. 
				Let the free transformations $T_{\rm free}$ consist of all non-increasing functions.
				Further define $u$ to be the function that fixes $1$ and swaps the other two resources.
				Pictorially, we can view the transitions allowed by these as
				\begin{equation}\label{eq:toy_3_levels}
					\input{toy_3_levels.tikz}
				\end{equation}
				Now it is clear that $u \star T_{\rm free}$ allows the conversion of $1$ to $2$ unlike $T_{\rm free} \star u$, so that $T_{\rm free} \star u \supseteq u \star T_{\rm free}$ does \emph{not} hold in this case.
				More generally, the conversions allowed by the two sets of transformations are
				\begin{equation}\label{eq:toy_3_levels_2}
					\input{toy_3_levels_2.tikz}
				\end{equation}
			\end{example}
			
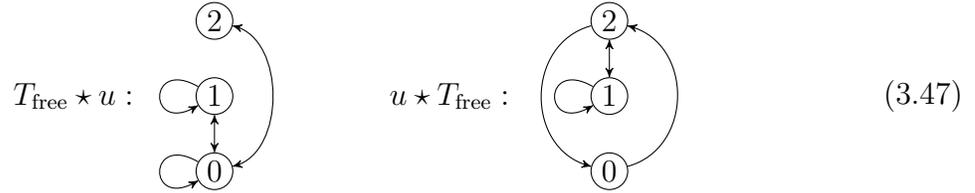
			\begin{definition}\label{def:invariant_trans}
				A transformation $U \in \mathcal{T}$ is termed
				\begin{itemize}
					\item \textbf{left invariant} if $T_{\rm free} \star U = U$ holds, and 
					\item \textbf{right invariant} if $U \star T_{\rm free} = U$ holds.
				\end{itemize}
			\end{definition}
			Given arbitrary transformation $S$, we can generate a left invariant one $T_{\rm free} \star S$ by left augmentation.
			Similarly, by right augmentation, one obtains a right invariant $S \star T_{\rm free}$.
			\begin{example}
				We have already seen a few examples of invariant sets in the context of the resource theory of entanglement with respect to LOCC operations (\cref{ex:rt_entanglement}).
				The separable operations are both left and right invariant.
				Since they are also reflexive and idempotent, they define a valid resource theory according to \cref{def:rt}.
				Likewise, operations that preserve the set of separable states are left and right invariant.
				However, these are not idempotent, as we mention in \cref{ex:rt_entanglement}.
				Another example of a right invariant set of transformations is the collection of all discarding maps---one for each system type.
				These are neither left invariant nor reflexive.
			\end{example}
			
			In a universally combinable resource theory, left and right invariant sets of resources coincide.
			If the resource theory is concrete, then they correspond precisely to sets of individual resources that are downward closed relative to the resource ordering.
			This is why, in the earlier work utilizing universally combinable resource theories only \cite{Monotones}, these three concepts are not distinguished.
			\begin{remark}
				\label{rem:Compatibility of order and combination for right invariant}
				Since $T_{\rm free} \cmps U \supseteq U$ holds for any $U \in \mathcal{P}(T)$, implication \eqref{eq:Compatibility of order and combination} holds for any right invariant set of individual transformations $U$ by \cref{lem:Compatibility of order and combination}.
			\end{remark}
			
			As in \cref{sec:ucrt_order}, we can define an equivalence relation among resources that identifies those which give us identical access to other resources, when augmented by free transformations.
			In other words, $Y$ and $Z$ are (order-)equivalent resources, denoted $Y \sim Z$, if both
			\begin{align}
				Y &\freeconv Z  &  &\text{and}  &   Z &\freeconv Y
			\end{align}
			hold, or in other words, if their free images coincide:
			\begin{equation}\label{eq:free_augment}
				T_{\rm free} \apply Y = T_{\rm free} \apply Z.
			\end{equation}
			
			Writing the free convertibility relation as in \eqref{eq:free_image} suggests that we should be able to quotient by the free transformations and obtain a resource theory in which the resource ordering coincides with the lattice ordering.
			Specifically, we can replace resources in $\mathcal{X}$ with their free images via the $T_{\rm free}$-augmentation construction (\cref{def:module_augment}).
			We define the suplattice 
			\begin{equation}
				\mathrm{Aug}_{\rm free}(\mathcal{X}) = \Set*[\big]{ T_{\rm free} \apply Y \given Y \in \mathcal{X} } ,
			\end{equation}
			which is clearly a sublattice of $\mathcal{X}$.
			However, $\mathcal{X}$ is generally not a $\mathcal{T}$-module since there may be a transformation $S$ (e.g., $u$ in \cref{ex:image_incompatible}) such that
			\begin{equation}
				S \apply (T_{\rm free} \apply Y)
			\end{equation}
			is not a free image of any element of $\mathcal{X}$.
			Such an $S$ would not constitute a well-defined transformation of free images.
			We can remedy this issue in a fairly transparent way by augmenting the transformations with arbitrary free pre- and post-processing.
			That is, we define a new quantale
			\begin{equation}
				\mathrm{Aug}_{\rm free}(\mathcal{T}) \coloneqq  \Set*[\big]{ T_{\rm free} \star S \star T_{\rm free}  \given S \in \mathcal{T} } 
			\end{equation}
			with suprema and composition inherited from $\mathcal{T}$. 
			It is straightforward to check that $(\mathrm{Aug}_{\rm free}(\mathcal{T}), \mathrm{Aug}_{\rm free}(\mathcal{X}), \apply)$ is another quantale module.
			We may think of it as the quotient of the original resource theory by the free transformations since $T_{\rm free}$ is now the unit of $\mathrm{Aug}_{\rm free}(\mathcal{T})$.
			
			\begin{remark}
				Strictly speaking, $\mathrm{Aug}_{\rm free}(\mathcal{T})$ is a quotient of the centralizer of $T_{\rm free}$ defined as
				\begin{equation}
					\mathrm{C}_{\mathcal{T}}(T_{\rm free}) \coloneqq  \Set*[\big]{ S  \given  S \star T_{\rm free} = T_{\rm free} \star S },
				\end{equation}
				rather than of $\mathcal{T}$.
				The equivalence relation on $\mathrm{C}_{\mathcal{T}}(T_{\rm free})$ is given by 
				\begin{equation}
					\begin{alignedat}{3}
						S_1 \sim_{\rm free} S_2  \quad &\coloniff \quad  & T_{\rm free} \star S_1 \star T_{\rm free} &= T_{\rm free} \star S_2 \star T_{\rm free} \\
							& \iff \quad & T_{\rm free} \star S_1 &=T_{\rm free} \star S_2
					\end{alignedat}
				\end{equation}
				It follows from idempotence of $T_{\rm free}$ that $\sim_{\rm free}$ is a quantale congruence.
				In fact, one could study the quantale module $(\mathrm{C}_{\mathcal{T}}(T_{\rm free}), \mathrm{Aug}_{\rm free}(\mathcal{X}), \apply)$.
				However, transformations equivalent under $\sim_{\rm free}$ induce \emph{identical} maps of resources, so it suffices to model representatives of their equivalence classes in the form of $T_{\rm free} \star S \star T_{\rm free}$, which is what $\mathrm{Aug}_{\rm free}(\mathcal{T})$ does.
			\end{remark}
			
			The procedure of passing from 
			\begin{itemize}
				\item a resource theory $(\mathcal{T}, \mathcal{X}, \apply)$ with free transformations given by $T_{\rm free}$ to
				\item a quantale module $(\mathrm{Aug}_{\rm free}(\mathcal{T}), \mathrm{Aug}_{\rm free}(\mathcal{X}), \apply)$ with an equivalent resource ordering, whose subquantale of free transformations is given by the unit $\mathrm{Aug}_{\rm free}(1)$ of the quantale $\mathrm{Aug}_{\rm free}(\mathcal{T})$
			\end{itemize}
			can be modularized.
			It can be carried out for any reflexive, idempotent element $B$ since that is all we assumed about $T_{\rm free}$.
			If we have such a $B$ which also satisfies $B \subseteq T_{\rm free}$, then we can construct the corresponding resource theory $(\mathrm{Aug}_{B}(\mathcal{T}), \mathrm{Aug}_{B}(\mathcal{X}), \apply)$ and iterate the process, if desired.
			
			The potential benefit of doing so is that finding an explicit characterization of $\mathrm{Aug}_{\rm free}(\mathcal{X})$ may be tricky in practice.
			Indeed, it amounts to solving one of the major questions in a resource theory---the question of characterizing the resource ordering.
			However, quotienting by a subclass $B$ of free transformations may be more tractable and lead to an iterative process of ``solving'' a resource theory.
			Alternatively, the $B$-augmented representation may sometimes offer a conceptually more interesting point of view.
			\begin{example}\label{ex:AIT}
				As an illustrative example, consider a toy version of algorithmic information theory relative to a given universal monotone Turing machine \cite[section 4.5.2]{Li2008}.
				In this context, we can let $\mathcal{X}$ be the power set of $X$, where $X$ is given by the disjoint union of two sets $X_\rI$ and $X_\rO$.
				They each contain finite and countably infinite strings of characters ($0$ and $1$, say) and correspond to the input and output tape respectively.
				In other words, we can interpret $X_\rI$ as the set of programs for the Turing machine and $X_\rO$ as the set of possible results of a computation.
				The behaviour of the Turing machine can be described by a partial function of type $X_\rI \to X_\rO$.
				It is the only non-trivial free transformation in this toy example.
				Whenever undefined, it acts on $X$ by returning $\emptyset \in \mathcal{X}$.
				
				Apart from the obvious identity function $X \to X$, there are a few other transformations one may want to include as free.
				In particular, consider ``cutting maps'' $\mathsf{Cut}_m \colon X_\rI \to X_\rI$ (one for each natural number $m$) that act on finite strings as follows:
				\begin{equation}
					\mathsf{Cut}_m(x_1 x_2 \cdots x_k) \coloneqq 
						\begin{dcases}
							x_1 x_2 \cdots x_m  & \text{if } m \leq k \\
							x_1 x_2 \cdots x_k  & \text{if } m > k,
						\end{dcases}
				\end{equation}
				and similarly for infinite strings.
				Regarding cutting maps as free transformations implements the idea that a prefix of any program contains no more information than the original uncut program.
				
				Furthermore, to implement the idea that a prefix of an output string has lower complexity than the uncut string\footnotemark{} we also can include cutting maps of the output tape as free transformations.
				\footnotetext{Whether this is sensible or not depends on the notion of complexity one has in mind.
				On the other hand, postulating that \emph{any} substring be less complex than the full string is certainly not sensible \cite{poulin2005comment}.}%
				
				The collection $\mathsf{Cut}$ of all cutting maps is idempotent and thus it can play the role of the transformation $B$ from previous considerations.
				That is, we can form the augmented representation $\mathrm{Aug}_{\mathsf{Cut}}(\mathcal{X})$ as in \cref{def:module_augment} or \cref{eq:free_augment}.
				It is order-generated by sets of strings of the form ($e$ denotes the empty string)
				\begin{equation}
					\mathsf{Cut} \apply (x_1 x_2 \cdots x_k) = \{ e, \, x_1, \, x_1 x_2, \, \ldots, \, x_1 x_2 \cdots x_k \}
				\end{equation}
				by taking unions in $\mathcal{P}(X)$.
				However, unlike $\mathcal{X}$, $\mathrm{Aug}_{\mathsf{Cut}}(\mathcal{X})$ is no longer an atomistic lattice.
				It only has one atom---the empty string $e$.
				This example thus also illustrates one of the reasons for us to consider resource theories which are not concrete in the sense of \cref{def:rt}.
				Namely, applying basic constructions (such as augmentation) to concrete resource theories can produce ones that are not concrete anymore.
			\end{example}
		
		\subsection{Asymmetry as a Resource}\label{sec:rt_asymmetry}
		
			Many interesting resource theories can be viewed as resource theories of asymmetry \cite{Marvian2013}, where the useful resources are those that break a particular notion of symmetry as specified by a group representation.
			This class of resource theories is general enough to include concepts such as quantum coherence \cite{marvian2016quantify}, distinguishability (see \cref{sec:Encodings}) as well as the quality of physical systems as reference frames \cite{bartlett2007reference}.
			At the same time, the tools of group representation theory allow a fruitful general study of these with few assumptions \cite{Marvian2014,Marvian2012}, especially if we think of the resources as elements of (topological) vector spaces as is often the case.
			
			We can introduce an analogous notion of a resource theory of asymmetry in the context of quantale modules.
			However, we don't assume resources carry a linear or convex structure here.
			Therefore, resource theories of asymmetry according to \cref{def:rt_asymmetry} constitute a much wider class than the commonly considered ones.
			The perspective we take is that in a specific context (such as that of quantum resources) the additional structure constrains the group of symmetries to the appropriately structure-preserving mappings and consequently also constrains what one would consider as a reasonable notion of a resource theory of asymmetry.
			In the abstract context we take the most general symmetry mapping of resources to be an element of $\mathrm{Aut}(\mathcal{X})$---the group of all suplattice isomorphisms $\mathcal{X} \to \mathcal{X}$.
			For concrete resource theories with the lattice of resources $\mathcal{X}$ given by $\mathcal{P}(X)$, this automorphism group includes \emph{all} permutations of the set $X$.
			As mentioned above, if the set $X$ of individual resources has additional structure, the symmetry mappings should be adjusted accordingly.
			For instance, if $X$ consists of quantum states, we would typically only consider symmetries that are elements of the corresponding (projective) unitary group.\footnotemark{}
			\footnotetext{When dealing with time-reversal symmetry, we may also include symmetry transformations that are antiunitary, see \cref{ex:time-reversal}.}%
			\begin{definition}\label{def:G-covariant}
				Let $G$ be a group and $\varphi \colon G \to \mathrm{Aut}(\mathcal{X})$ a group homomorphism.
				Given a quantale module $(\mathcal{T}, \mathcal{X}, \apply)$, a transformation $T \in \mathcal{T}$ is termed $\bm{G}$\textbf{-compatible} if for any $X \in \mathcal{X}$ and any $g \in G$ we have
				\begin{equation}\label{eq:G-covariant}
					\varphi_g ( T \apply X ) = T \apply \varphi_g (X).
				\end{equation}
				Furthermore, a transformation $S \in \mathcal{T}$ is said to be $\bm{G}$\textbf{-covariant}\footnotemark{} if every transformation $T$ satisfying $T \subseteq S$ is $G$-compatible.
				\footnotetext{Note that, despite the name, whether $S$ is $G$-covariant or not depends on the representation $\varphi$ and not just the group $G$.}%
			\end{definition}
			
			$G$-compatible elements form a subquantale of $\mathcal{T}$ that we denote by $\mathcal{C}_\phi$.
			In specific resource theories, it is sometimes the case that for every automorphism $\psi$ there is a transformation $U_{\psi} \in \mathcal{T}$ whose action on $\mathcal{X}$ coincides with $\psi$.
			If this holds, then a sufficient condition for $T$ to be $G$-compatible is
			\begin{equation}\label{eq:G-covariant_2}
				T \star U_{\varphi_g} = U_{\varphi_g} \star T,
			\end{equation}
			which is perhaps a more familiar way to define covariance.
			In this case, an equivalent way to express $G$-compatibility is that $T$ is an element of the centralizer (see \cref{def:centralizer}) of $\mathrm{im}(\varphi)$---the image of $\varphi$ as emebedded within $\mathcal{T}$.
			
			\begin{example}\label{ex:time-reversal}
				An example of a resource theory of asymmetry in which $G$-compatible transformations \emph{cannot} be defined via \cref{eq:G-covariant_2} is the quantum resource theory of time-reversal asymmetry \cite{gour2009time} or more generally the quantum resource theory of imaginarity with respect to a chosen basis \cite{hickey2018quantifying}.
				In this case, we are dealing with a quantum resource theory of states, so that the set of individual resources $X$ can be identified with quantum states.
				The symmetry group $G$ is just $\mathbb{Z}_2 = \{e,\tau\}$---the unique group of order $2$.
				Its non-identity element $\tau$ is represented by complex conjugation which is an \emph{antiunitary} operation and thus not a quantum channel.
				Therefore, there is no (physical) transformation $u_{\varphi_\tau} \in T$ whose action on quantum states is that of complex conjugation (or time-reversal if the chosen basis is an eigenbasis of the Hamiltonian).
				Consequently, we have to use the more general \cref{eq:G-covariant} in order to define transformations compatible with complex conjugation or time-reversal.
			\end{example}
			
			\begin{lemma}\label{lem:G-covariant_suprema}
				If $\mathcal{T}$ is localic, meaning that the underlying suplattice is a locale (see \cref{def:locale}), then the $G$-covariant elements are closed under arbitrary suprema. 
			\end{lemma}
			\begin{proof}
				Consider a collection $\mathcal{S}$ of $G$-covariant elements of $\mathcal{T}$.
				We want to show that $\bigvee \mathcal{S}$ is also $G$-covariant.
				To this end, note that thanks to the distributivity of $\mathcal{T}$ we can write any $T \subseteq \bigvee \mathcal{S}$ as
				\begin{equation}
					T = T \wedge \left( \bigvee \mathcal{S} \right) = \bigvee_{S \in \mathcal{S}} (T \wedge S).
				\end{equation}
				Now, since $T \wedge S$ is below $S$, it is $G$-compatible and 
				\begin{equation}
					\begin{split}
						\varphi_g ( T \apply X ) &= \varphi_g \left( \bigvee_{S \in \mathcal{S}} (T \wedge S) \apply X \right) \\
							&=  \bigvee_{S \in \mathcal{S}} \varphi_g \bigl(  (T \wedge S) \apply X \bigr) \\
							&= \bigvee_{S \in \mathcal{S}} (T \wedge S) \apply \varphi_g ( X ) \\
							&= T \apply \varphi_g (X) 
					\end{split}
				\end{equation}
				holds for all $g$ and all $X$.
				We conclude that $\bigvee \mathcal{S}$ is indeed $G$-covariant.
			\end{proof}
			Notice that $\bot$ is always $G$-covariant and that $S$ being $G$-covariant implies all elements below it are too.
			Thus, in a localic resource theory the $G$-covariant transformations form a suplattice $\mathcal{S}_{\varphi} \coloneqq \langle \bot, S_{\varphi} \rangle$ where $S_{\varphi}$ is the supremum of all $G$-covariant transformations.
			However, $G$-covariant elements need not be closed under $\star$ in general.
			At present, we do not know of a precise characterization of when $\mathcal{S}_{\varphi}$ forms a subquantale.
			As the following class of examples shows, it is a common phenomenon. 
			\begin{example}
				Let $\mathcal{T}$ be a uniquely atomistic quantale, whose composition $\star$ is a pointwise supremum of a collection $\cdot_i$ of (partial) semigroup multiplications of its atoms as in \cref{ex:free_quantale_multi}.
				Then we can show $S_{\varphi} \star S_{\varphi}$ to be $G$-covariant, so that $\mathcal{S}_{\varphi}$ is indeed a subquantale of $\mathcal{T}$.
				To see this, let $A_{\varphi}$ be the collection of all atoms that are $G$-compatible (and thus also $G$-covariant, since the two concepts coincide for atoms) and compute
				\begin{equation}
					\begin{split}
						S_{\varphi} \star S_{\varphi}  &=  \bigcup_{a, a' \in A_{\varphi}} a \star a' \\
							&= \bigcup_{a, a' \in A_{\varphi}} \bigcup_i  a \cdot_i a' \\
					\end{split}
				\end{equation}
				Since $G$-compatibility is preserved under $\cdot_i$, each $a \cdot_i a'$ is also $G$-compatible and thus $G$-covariant as it is an atom.
				By \cref{lem:G-covariant_suprema}, $S_{\varphi} \star S_{\varphi}$ is therefore $G$-covariant because it is a supremum of $G$-covariant elements and $\mathcal{T}$ is is clearly localic (being a power set lattice).
				It then follows from the definition of $S_{\varphi}$ as a supremum of all $G$-covariant elements that it must be idempotent and satisfy $S_{\varphi} \supseteq 1$.
			\end{example}
			\begin{definition}\label{def:rt_asymmetry}
				A \textbf{resource theory of $\bm{G}$-asymmetry} (with respect to $\varphi$) is a quantale module $(\mathcal{T}, \mathcal{X}, \apply)$ with the subquantale of free operations given by the $G$-covariant ones defined as above (as long as they indeed form a subquantale $\mathcal{S}_{\varphi}$ of $\mathcal{T}$).
			\end{definition}
			
			\begin{example}\label{ex:asymmetry}
				Consider $X = \{ s, a_1, a_2, a_3 \}$, $\mathcal{X} = \mathcal{P}(X)$ and an action of the cyclic group $\{e, \sigma, \sigma^2\}$ by permuting $a_1$, $a_2$ and $a_3$.
				Furthermore, let $\mathcal{T}$ be the free quantale (see \cref{ex:free_quantale}) of the monoid of functions $X \to X$.
				The $G$-covariant elements can be specified by their action on $s$ and $a_1$ alone.\footnotemark{}
				\footnotetext{In general, it suffices to specify how a $G$-covariant transformation acts on a single element of each $G$-orbit.}%
				There are 4 atomic ones given by\footnotemark{}
				\footnotetext{We once again use an abusive notation in which both $s$ and $\{s\}$ are denoted by $s$ and likewise for other singletons.}%
				\begin{equation}
					\begin{aligned}
						u_{e} (s) &= s  			&\qquad  u_{\sigma} (s) &= s 		&\qquad  u_{\sigma^2} (s) &= s 		&\qquad  t_s (s) &= s \\
						u_{e} (a_1) &= a_1		&\qquad  u_{\sigma} (a_1) &= a_2 	&\qquad  u_{\sigma^2} (a_1) &= a_3  	&\qquad  t_s (a_1) &= s
					\end{aligned}
				\end{equation}
				so that their action on $a_1$ can be visualized as 
				\begin{equation}
					\input{toy_rotations.tikz}
				\end{equation}
				In the pictorial representation, the symmetry group is that of rotations around the symmetric resource $s$ and the $G$-covariant transformations are given by ${S_{\varphi} = \{u_{e}, u_{\sigma}, u_{\sigma^2}, t_s\}}$.
				
				To get a bit more intuition about $G$-covariant transformations, let ${T_a \coloneqq \{t_1, t_2, t_3\}}$ where each $t_i$ is given by $t_i (s) = t_i(a_j) = a_i$ so that $t_1$, for example, acts as
				\begin{equation}
					\input{toy_rotations_2.tikz}
				\end{equation}
				and $T_a$ thus acts on $\mathcal{X}$ via
				\begin{equation}
					T_a \apply s = T_a \apply a_i = \{ a_1, a_2, a_3 \}.
				\end{equation}
				Then $T_a$ satisfies equation \eqref{eq:G-covariant}, i.e.\ it is $G$-compatible, but it is \emph{not} $G$-covariant.
				This is because there are transformations below it, such as $t_1$, which are not $G$-compatible.
				Intuitively, we can see that $T_a$ takes a symmetric state $s$ and maps it to a collection of asymmetric ones.
				Therefore, it should not be free in a resource theory of asymmetry.
			\end{example}
			
			The notion of covariant transformations can be generalized from the case of a group of symmetries to a monoid.
			Namely, there is the monoid of suplattice homomorphisms $\mathcal{X} \to \mathcal{X}$, denoted by $\mathrm{End}(\mathcal{X})$.
			For any monoid homomorphism $\varphi \colon M \to \mathrm{End}(\mathcal{X})$ the corresponding $M$-compatible transformations are those that satisfy
			\begin{equation}\label{eq:M-covariant}
				\varphi_m ( T \apply X ) = T \apply \varphi_m (X).
			\end{equation}
			for every $m \in M$ and every $X \in \mathcal{X}$.
			Just like in the case of groups, the $M$-covariant transformations are those whose $\subseteq$-downward closure in $\mathcal{T}$ consists of $M$-compatible elements only.
			We also have an analogous version of \cref{lem:G-covariant_suprema}, since its proof doesn't assume the existence of inverses.
			Many resource theories that cannot be expressed as resource theories of asymmetry do arise in this way.
			
			\begin{example}\label{ex:fixed_points}
				As an illustrative example, consider the two element monoid ${M = \{e, p\}}$ which satisfies $p \cdot p = p$.
				For simplicity, let us consider the representation $\varphi_p$ of $p$ to be a constant function $x \mapsto x_0$ for an atom $x_0$.
				In general, $\varphi_p$ could be any idempotent mapping.
				The (atomic) $M$-covariant transformations are then those for which $x_0$ is a fixed point, i.e.\ they ought to satisfy
				\begin{equation}
					t \apply x_0 = x_0.
				\end{equation}
				Already this elementary version of the construction includes many interesting cases, especially if we allow $x_0$ to depend on the type of the resource in question.
				For instance, in resource theories of nonuniformity \cite{Gour2015,Horodecki2003}, the fixed points are the uniform probability distributions.
				Similarly, in resource theories of athermality \cite{Janzing2000,Brandao2013}, the fixed points are thermal states.
			\end{example}
			A mild generalization of \cref{ex:fixed_points} allows one to describe resource theories where the free transformations are required to have multiple simultaneous fixed points, but there are also examples of higher complexity of course.

		\subsection{Morphisms of Resource Theories}\label{sec:rt_morphisms}
			
			There are multiple feasible ways to define structure-preserving maps between resource theories, depending on the structure that one requires to be preserved.
			Indeed, one could study relations that are only required to preserve the convertibility relation by free transformation on top of whatever compositional structure the resources may carry.
			This is what we do in \cref{sec:monotones}.
			
			Here, we describe perhaps the most restrictive notion of a resource theory morphism that can be of use.
			It has to preserve both the ways in which transformations are combined as well as the protocols for simulating non-free transformations.
			\begin{definition}\label{def:rt_morphism}
				Let $(\mathcal{T}, \mathcal{X}, \apply)$ and $(\mathcal{U}, \mathcal{Y}, \apply)$ be resource theories with subquatales of free transfomrations given by $T_{\rm free}$ and $U_{\rm free}$ respectively.
				A \textbf{resource theory morphism} is a quantale module homomorphism (see \cref{def:qmod_morphism}) $(\ell, f)$ that preserves free transformations.
				That is, $f \colon \mathcal{X} \to \mathcal{Y}$ and $\ell \colon \mathcal{T} \to \mathcal{U}$ are suplattice homomorphisms that satisfy
				\begin{subequations}\label{eq:rt_morphism}
				\begin{align}
					\label{eq:q_morphism}  \ell(T_1 \star T_2) &= \ell(T_1) \star \ell(T_2) , \\
					\label{eq:qm_morphism}  f \bigl( T \apply X \bigr) &= \ell (T) \apply f(X), \\
					\label{eq:free_preserving}  \ell(T_{\rm free}) &\subseteq U_{\rm free}.
				\end{align}
				\end{subequations}
				for arbitrary $X \in \mathcal{X}$ and $T,T_1,T_2 \in \mathcal{T}$.
			\end{definition}
			If \cref{eq:q_morphism,eq:qm_morphism} are weakened to 
			\begin{subequations}\label{eq:oplax_rt_morphism}
			\begin{align}
				\label{eq:oplax_q_morphism}  \ell(T_1 \star T_2) &\subseteq \ell(T_1) \star \ell(T_2) \text{, and} \\
				\label{eq:oplax_qm_morphism}  f \bigl( T \apply X \bigr) &\subseteq \ell (T) \apply f(X)
			\end{align}
			\end{subequations}
			respectively, we speak of an \textbf{oplax resource theory morphism}.
			
			To get a better intuition about \cref{def:rt_morphism}, we now look at various ways to construct morphisms of resource theories in the context of asymmetry.
			First, for any subgroup $H$ of $G$, one can simply restrict the asymmetry with respect to $G$ to asymmetry with respect to $H$.
			That is, given a group representation $\varphi \colon G \to \mathrm{Aut}(\mathcal{X})$ and a (not necessarily injective) group homomorphism $\mu \colon H \to G$, the $G$-covariant transformations are a subset of $H$-covariant ones:
			\begin{equation}
				\mathcal{S}_{\varphi} \subseteq \mathcal{S}_{\mu^* \varphi},
			\end{equation}
			where $\mu^* \varphi \colon H \to \mathrm{Aut}(\mathcal{X})$ is the pre-composition of $\varphi$ with $\mu$.
			The identity map is thus a morphism that maps a resource theory of $G$-asymmetry to a resource theory of $H$-asymmetry in this case.
			
			\begin{example}
				Consider $G$ to be $\mathrm{SO(3)}$, the group of rotations in 3 dimensions and $H$ to be $\mathrm{SO(2)}$, the subgroup of rotations around a fixed axis.
				Then the above procedure says that the transformations that are covariant with respect to $\mathrm{SO(3)}$ rotations in 3D space are necessarily also covariant with respect to 2D rotations around a fixed axis.
				Consequently, the resource theory of $\mathrm{SO(3)}$-asymmetry is canonically embedded in a resource theory of $\mathrm{SO(2)}$-asymmetry.
			\end{example}
			
			Now, consider two representations $\varphi, \psi \colon G \to \mathrm{Aut}(\mathcal{X})$.
			A suplattice homomorphism $f \colon \mathcal{X} \to \mathcal{X}$ is equivariant with respect to them if it satisfies 
			\begin{equation}
				\psi_g \circ f = f \circ \varphi_g
			\end{equation}
			for all $g \in G$.
			Let $f$ be an equivariant isomorphism and let $U_f \in \mathcal{T}$ be a transformation satisfying $U_f \apply X = f(X)$.
			Moreover, let $U_{f^{-1}}  \in \mathcal{T}$ be the an inverse of $U_f$.
			Then we can define a quantale homomorphism $\mathrm{Ad}_f \colon \mathcal{T} \to \mathcal{T}$ via
			\begin{equation}
				\mathrm{Ad}_f (T) \coloneqq U_f \star T \star U_{f^{-1}},
			\end{equation}
			and we have $\mathrm{Ad}_f (S_\varphi) \subseteq S_\psi$ because if $T \apply \varphi_g (X) = \varphi_g (T \apply X)$ holds, it follows that
			\begin{equation}
				\begin{split}
					\mathrm{Ad}_f (T) \apply \psi_g (X) &= (U_f \star T) \apply \bigl( f^{-1} \circ \psi_g (X) \bigr) \\
						&= (U_f \star T) \apply \bigl( \varphi_g \circ f^{-1} (X) \bigr) \\
						&= (f \circ \varphi_g) \bigl( T \apply f^{-1}(X) \bigr) \\
						&= \psi_g \bigl( \mathrm{Ad}_f (T) \apply X \bigr)
				\end{split}
			\end{equation}
			does as well.
			Consequently, the pair $(\mathrm{Ad}_f,f)$ constitutes a morphism of resource theories of $G$-asymmetry with respect to $\varphi$ and $\psi$ respectively.
			If we think of the two representations as specifying reference frames \cite{bartlett2007reference} for $G$, then this morphism corresponds to the change of a reference frame.
			
			\begin{example}\label{ex:intertwiner}
				Consider a resource theory of $\mathrm{U(1)}$-asymmetry and individual resource given by the states of a qubit---a $2$-dimensional quantum system.
				A $\mathrm{U(1)}$-representation specifies how the qubit states transform under shifts of the phase degree of freedom $e^{i \theta} \in \mathrm{U(1)}$.
				Therefore, it specifies which states can keep track of phase changes and thus how well they can act as a phase reference.
				Let the two representations, $\varphi$ and $\psi$ be given by 
				\begin{align}\label{eq:phase_rep}
					\varphi_\theta (\rho) &\coloneqq e^{\frac{i \theta}{2} \sigma_x} \, \rho \, e^{- \frac{i \theta}{2} \sigma_x}
						&
					\psi_\theta (\rho) &\coloneqq e^{\frac{i \theta}{2} \sigma_y} \, \rho \, e^{- \frac{i \theta}{2} \sigma_y}	
				\end{align}
				where $\sigma_x, \sigma_y, \sigma_z$ are the Pauli matrices and $\rho$ is a density matrix that specifies a qubit state.
				Then we have the following equivariant map:
				\begin{equation}
					f(\rho) \coloneqq e^{\frac{i \pi}{4} \sigma_z} \, \rho \, e^{- \frac{i \pi}{4} \sigma_z}	
				\end{equation}
				In particular, if we think of $\varphi$ and $\psi$ as rotations around the $x$ and $y$ axes of the Bloch sphere, then $f$ rotates the $x$-axis to the $y$-axis.
				Let us make the conventional choice that the state given by the Hilbert space vector $|0\rangle$ (one of the eigenvectors of $\sigma_z$) corresponds to $\theta = 0$.
				Relative to $\varphi$ and this choice of ``origin'', the qubit is a phase reference with respect to the orbit
				\begin{equation}
					\Set{ \cos\left( \tfrac{\theta}{2} \right) |0\rangle + i \sin\left( \tfrac{\theta}{2} \right) |1\rangle \given \theta \in [0,2 \pi) }
				\end{equation}
				of $|0\rangle$ under the $\mathrm{U(1)}$-action.
				On the other hand, relative to $\psi$, the qubit is a phase reference with respect to the orbit
				\begin{equation}
					\Set{ \cos\left( \tfrac{\theta}{2} \right) |0\rangle - \sin\left( \tfrac{\theta}{2} \right) |1\rangle \given \theta \in [0,2 \pi) }.
				\end{equation}
				The intertwiner $f$ then maps between the two qubits and its properties as $\mathrm{U(1)}$ reference frames relative to any choice of origin, not just the state $|0\rangle$ chosen here for concreteness.
				The ``1-comb'' $\mathrm{Ad}_f$ maps manipulations of the phase encoded in rotations around the $x$-axis (in the Bloch sphere picture) to manipulations of the phase encoded in rotations around the $y$-axis.
			\end{example}
			
			Given a resource theory of $G$-asymmetry, we can also restrict its quantale of transformations to the subquantale of $G$-compatible transformations, denoted by $\mathcal{C}_\varphi$ (recall \cref{def:G-covariant}).
			This removes all the individual transformations that are asymmetric.
			Nevertheless, recall that $\mathcal{C}_\varphi$ includes non-covariant sets of transformations such as $T_a$ from \cref{ex:asymmetry}, which can be thought of as an orbit of $t_1$ under the adjoint action of $G$ \mbox{on $\mathcal{T}$}.
			We can interpret $\mathcal{C}_\varphi$ as restricting the relevant (sets of) transformations to those whose power is identical whether a symmetry transformation is applied before or after them.
			Then, for any $g \in G$, the map $\varphi_g \colon \mathcal{X} \to \mathcal{X}$ is an endomorphism of the resource theory of $G$-asymmetry with enveloping quantale module $(\mathcal{C}_\varphi, \mathcal{X}, \apply)$. 
			In the reference frame point of view, this endomorphism admits the interpretation of a coordinate transformation.
			
			\begin{example}
				Consider the resource theory of $\mathrm{U(1)}$-asymmetry with respect to the representation $\varphi$ from \cref{ex:intertwiner}.
				For any phase $\theta \in [0,2 \pi)$, the rotation $\varphi_\theta$ as given \mbox{by \eqref{eq:phase_rep}} implements a coordinate transformation by translating the origin of the phase reference \mbox{by $\theta$}.
			\end{example}
			
			Further examples of resource theory morphisms appear in \cref{ex:bridge_lemma} and in \cref{sec:dist_poss}. 
			The former gives a morphism between a resource theory of $G$-asymmetry (see \cref{sec:rt_asymmetry}) and a resource theory of $G$-distinguishability (see \cref{sec:Encodings}), while the latter maps distinguishability of probabilistic behaviours to distinguishability of possibilistic behaviours.
			More details can be found in the respective sections.
			
			Besides quantale module homomorphisms, we can also relate resource theories by (op)lax quantale module homomorphisms as introduced in \eqref{eq:lax_qmod_morphism}.
			In the case of an oplax module homomorphism, we are effectively allowing the codomain theory to have additional ways of composing and applying transformations as compared to its domain.
			Since adding extra possibilities for composition can only aid in resource manipulations, we would expect that oplax morphisms preserve the resource ordering.
			\begin{proposition}\label{prop:isotone_from_morphism}
				Let $(\mathcal{T}, \mathcal{X}, \apply)$ and $(\mathcal{U}, \mathcal{Y}, \apply)$ be resource theories with free transfomrations given by $T_{\rm free}$ and $U_{\rm free}$ respectively.
				Consider an oplax resource theory moprhism consisting of 
				\begin{enumerate}
					\item a suplattice homomorphism $\ell \colon \mathcal{T} \to \mathcal{U}$ which satisfies, for all transformations $S, T \in \mathcal{T}$,\footnotemark{}
						\footnotetext{Note that we do not use \cref{eq:oplax_quantale_morphism} in the proof of \cref{prop:isotone_from_morphism}.
							Therefore, the same conclusion holds if $\ell$ is merely a suplattice homomorphism as opposed to an oplax quantale homomorphism.}%
						\begin{equation}\label{eq:oplax_quantale_morphism}
							\ell ( S \star T ) \subseteq \ell(S) \star \ell(T),
						\end{equation}
						
					\item and a suplattice homomorphism $f \colon \mathcal{X} \to \mathcal{Y}$ which satisfies, for all $X \in \mathcal{X}$ and all $T \in \mathcal{T}$,
						\begin{equation}
							f ( T \apply X ) \subseteq \ell(T) \apply f(X)
						\end{equation}
				\end{enumerate}
				such that $\ell$ preserves free transformations, i.e.\ $\ell(T_{\rm free}) \subseteq U_{\rm free}$ holds.
				Then $f$ preserves the respective resource orderings.
				That is, it is an isotone of type $(\mathcal{X},\freeconv) \to (\mathcal{Y}, \freeconv)$.
			\end{proposition}
			\begin{proof}
				If $X \freeconv W$ holds for some $X,W \in \mathcal{X}$, then we have
				\begin{equation}
					\begin{split}
						f(W) &\subseteq f \bigl( T_{\rm free} \apply X \bigr) \\
							&\subseteq  \ell(T_{\rm free}) \apply f(X) \\
							&\subseteq U_{\rm free} \apply f(X)
					\end{split}
				\end{equation}
				and conclude that $f(X) \freeconv f(W)$ follows.
			\end{proof}
			
			\begin{example}
				In \cref{ex:rt_states,ex:rt_states_2} we saw two ways to construct a quantale module from a process theory.
				The resulting modules share the same lattice of resources as well as the same lattice of transformations, given by the power set lattices of the set of all states and the set of all channels respectively.
				However, the latter allowes for strictly more possibilities both in the composition of transformations (see \cref{eq:rt_states_3}) and in the application of transformations to resources (see \cref{eq:rt_states_2}).
				Therefore, choosing $f$ and $\ell$ to be the identity functions leads to an oplax resource theory morphism.
			\end{example}

\chapter{Resource Monotones}\label{sec:monotones}
	
	\section{Evaluating Utility of Resources}
		
		As outlined in the Introduction, one of our motivations for introducing the framework, subsequently presented in \cref{sec:resource_theories}, is to be able to ``recognize patterns and formulate methods in a unified manner''.
		The concept of a resource measure is among the most basic tools used when studying resources.
		Having described aspects of the framework for resource theories in terms of quantale modules, we now turn to the study of measures. 
		In particular, we are interested in the general patterns that underlie their constructions.
		Most of the content in this chapter is a generalization of the work presented in \cite{Monotones}.
		While the content of \cite{Monotones} concerns universally combinable resource theories, here we present the results in the context of quantale modules.
		
		Traditionally, measures of resources are real-valued functions that satisfy certain desiderata.
		The requirements can vary in their specificity, but one property that every meaningful resource measure ought to satisfy is that it is an isotone.
		That is, it never assigns a higher value to a resource that is less useful than another one.
			
		\begin{definition}\label{def:res_monotone}
			Given a concrete resource theory with resource ordering $(\mathcal{P}(X), \freeconv)$, a \textbf{resource monotone} (alse referred to as a resource measure) is a monotone function $(X, \freeconv) \to \ordreals$.
		\end{definition}
		
		\begin{remark}\label{rem:monotones_downsets}
			Monotones are closely connected to the concept of downsets (see \cref{def:downset}).
			Specifically, for any real number $c$ and a monotone $f \colon X \to \reals$, the set of all resources whose $f$-value is bounded above by $c$, i.e.\
			\begin{equation}
				f^{-1} \bigl( \down (c) \bigr) = \Set{ x \in X  \given  f(x) \leq c },
			\end{equation}
			is a downward closed subset of $(X, \freeconv)$.
			Conversely, given a family $\{D_k\}_{k \in \mathbb{N}}$ of nested downsets that satisfy $D_k \in \mathcal{DC}(X)$ and 
			\begin{equation}
				k \leq l  \quad \implies \quad  D_k \subseteq D_l,
			\end{equation}
			the membership function $k \colon X \to \reals$ given by
			\begin{equation}
				k(x) \coloneqq \inf \Set{ k \in \mathbb{N}  \given  x \in D_k }
			\end{equation}
			is a resource monotone.
			
			Since bipartite quantum states with bounded entanglement rank (\cref{examplesMonotonesDCsets}) form a nested family of downsets, their membership function---also known as the Schmidt number---is a monotone that can be understood in this manner.
			A similar example is presented in \cref{sec:dist_rank} in the context of resource theories of distinguishability.
		\end{remark}
		
		Recall that the enhancement preorder applied to (extended) real numbers yields an ordering $\geq_{\rm enh}$, which compares sets of real numbers according to their suprema (\cref{ex:enh_reals}).
		Extending a monotone $f$ to the full preorder-extended lattice of sets of resources leads to a suplattice homomorphism $f_* \colon (\mathcal{P}(X), \freeconv) \to (\mathcal{P}(\reals), \geq_{\rm enh})$ defined as
		\begin{equation}
			f_*(Y) \coloneqq \Set{f(y)  \given y \in Y}.
		\end{equation}
		It is an isotone by \cref{lem:sup_isotone} so that $\sup \circ f_*$ gives a real-valued monotone of sets of resources.
		Conversely, every monotone function $(\mathcal{P}(X), \freeconv) \to \ordreals$ can be restricted to its atoms and yields a resource monotone.
		Therefore, in a resource theory that is not uniquely atomistic, we may choose to define resource monotones as the isotone suplattice homomorphisms $(\mathcal{X}, \freeconv) \to (\mathcal{P}(\reals), \geq_{\rm enh})$ and this reduces to the standard \cref{def:res_monotone} in the concrete case.
		
		A single monotone offers limited information about a resource theory of course.
		Whenever the resource ordering is not total, characterizing it by monotones requires more than one.
		We say that a set $\{f_i\}_{i \in I}$ of resource monotones is \textbf{complete} if we have
		\begin{equation}
			x \freeconv y  \quad \iff \quad  f_i (x) \geq f_i (y).
		\end{equation}
		For any complete set of monotones, one obtains an order-embedding $F \colon (X, \freeconv) \to (\reals^I, \geq)$ defined by
		\begin{equation}
			F(x) = \bigl( f_i(x) \bigr)_{i \in I}
		\end{equation}
		where the partial order $\geq$ on $\reals^I$ is defined entry-wise:
		\begin{equation}
			(m_i)_{i \in I} \geq (n_i)_{i \in I}  \quad \coloniff \quad  \forall \, i \in I  \,:\,  m_i \geq n_i. 
		\end{equation}
		
		In this chapter, we study how resource monotones can be constructed.
		We look at examples of common constructions of monotones appearing in the literature on resource theories and identify more general procedures, which they are instances of.
		This helps us organize various monotones, understand the connections between them, and obtain generally applicable methods for generating new interesting monotones in any resource theory of interest.
		In order to classify the monotone constructions we discuss, we use the following general procedure.
		\begin{named}{Broad Scheme}
			\label{broad scheme}
			Identify a preordered set $(\mathcal{B}, \succeq_{\mathcal{B}})$ and two order-preserving maps:
			\begin{align}
				\label{eq:general}
				\sigma_1 &\colon \ordresx \to (\mathcal{B}, \succeq_{\mathcal{B}})  &  \sigma_2 &\colon (\mathcal{B}, \succeq_{\mathcal{B}}) \to \ordreals
			\end{align}
			Composing the isotones then gives a resource monotone $\ordresx \to \ordreals$.
			$\sigma_2$ is called the \textbf{root monotone}, $(\mathcal{B}, \succeq_{\mathcal{B}})$ is called the \textbf{mediating preordered set} and $\sigma_1$ is called the \textbf{mediating isotone}.
			We can thus view the target resource monotone $\sigma_2 \circ \sigma_1$ as obtained by pulling back the root monotone $\sigma_2$ along the mediating isotone $\sigma_1$.
		\end{named}
		The overarching theme of the present chapter is to illuminate which choices of a mediating isotone and a root monotone lead to measures of resources that are either prevalent in the literature or interesting for other reasons.
		To cast an existing resource measure into the mould of the \ref*{broad scheme}, one may try to identify a choice of each of these ingredients that yield the given measure. 
		Alternative constructions can be obtained immediately by simply varying any of the ingredients.
		In this way, the \ref*{broad scheme} provides the means of classifying monotone constructions as well as generalizing the known ones to apply to novel situations.
		
		Throughout, we assume that $(\mathcal{T}, \mathcal{X}, \apply)$ denotes a uniquely atomistic quantale module with $\mathcal{T} = \mathcal{P}(T)$ and $\mathcal{X} = \mathcal{P}(X)$.
		As before, $T_{\rm free} \in \mathcal{P}(T)$ is the set of free transformations and $\freeconv$ the associated resource ordering on $\mathcal{X}$.
		Some of the techniques generalize to the case of merely atomistic (or even arbitrary) lattices, but precise conditions under which they can be applied are not of primary importance.
		The purpose of this chapter is not to merely present a list of techniques, but also to supply them with a class of applicable interpretations.
		Since the uniquely atomistic case comes with a more straightforward interpretation, it serves us better in that regard.
		Additionally, many of the resource theories one may wish to study at present are indeed concrete.
	
	\section{Cost and Yield}\label{sec:cost_yield}
	
		\subsection{Cost and Yield Relative to Free Transformations}\label{sec:free_cost_yield}
			
			The first type of monotone construction we study in depth is one expressing the cost of a particular resource and its counterpart---the yield.
			The idea of cost measures is to express the value of a resource $x$ in terms of how hard it is to procure it from a specified set of resources.
			Hereafter, we denote by $A$ the set of resources with respect to which cost (or yield) are measured.
			\begin{example}
				For instance, in the context of a market, we can let $A$ consist of all the monetary amounts in a currency (such as bitcoin) that one could use to exchange for other goods.\footnotemark{}
				\footnotetext{From the point of view of a single agent in a market at a fixed point in time, we can describe it as a resource theory in the following way.
					Individual resources model goods that the agent can potentially have in possession.
					Free transformations represent the available exchanges the agent can partake in given the state of the market at that particular time.}%
				The {\bitcoin}-cost of a particular marketable item (such as a carrot) is the smallest amount of bitcoins that one could buy a carrot for.
				On the other hand, its {\bitcoin}-yield is the largest amount of bitcoins that one could sell a carrot for.
			\end{example}
			\begin{example}
				In the context of thermodynamics, we can measure the work-cost of a thermodynamic process in terms of how much work needs to be spent to run the process.
				On the other, its work-yield is the optimal amount of work that can be extracted while running the process.
			\end{example}
			The following toy example illustrates the forthcoming formal treatment of cost and yield measures.
			\begin{example}\label{ex:toy_cost_yield}
				Consider the situation from \cref{ex:image_incompatible} with an extra resource $1'$ which plays the same role as $1$ but is not equivalent to it.
				That is, free transformations allow the following transitions (leaving the identities implicit):
				\begin{equation}\label{eq:toy_3+_free}
					\input{toy_3+_free.tikz}
				\end{equation}
				Let us use $A = \{0,1,2\}$ as the set of resources relative to which to measure cost and yield.
				Then the cost of $1'$ is $2$ because $2$ is the least valuable element of $A$ that can be used to generate $1'$ for free.
				Dually, its yield is $0$ because it is the best resource in $A$ which can be obtained from $1'$ via free transformations.
				In general, we have
				\begin{equation}
					\begin{aligned}
						\cost{A}{} (1') &= 2  			&\qquad  \cost{A}{} (2) &= 2 		&\qquad  \cost{A}{} (1) &= 1 		&\qquad  \cost{A}{} (0) &= 0 \\
						\yield{A}{} (1') &= 0  		&\qquad  \yield{A}{} (2) &= 2 		&\qquad  \yield{A}{} (1) &= 1 		&\qquad  \yield{A}{} (0) &= 0.
					\end{aligned}
				\end{equation}
				In this simple example, it turns out that the pair of a cost and a yield relative to $A$ constitutes a complete set of monotones.
				However, this is generally not the case for a fixed choice of $A$.
			\end{example}
			The choice of reference set $A$ is arbitrary to a large degree. 
			However, we might want it to have certain properties in order for the cost and yield to be useful and meaningful quantities.
			For instance, we may require any or all of the following.
			\begin{enumerate}
				\item Every resource can be obtained from one in $A$ for free.
				\item Every resource can be converted to one in $A$ for free.
				\item $A$ is totally ordered.
				\item The elements of $A$ have operational significance.
			\end{enumerate}
			One may think of the reference resources as providing a ``gold standard''.
			They form a basis for evaluating other resources in relation to them and they may have properties that make this evaluation useful.
			However, their choice is not necessarily indicative of any fundamental status in the resource theory.
			In our treatment of cost and yield measures, we keep the gold standard completely arbitrary.
			
			Let us first analyze the yield construction and use \cref{ex:toy_cost_yield} to make the discussion more concrete.
			We start by identifying the reference set $A \subseteq X$ and a particular assignment $f \colon A \to \reals$ of values to its elements.
			In \cref{ex:toy_cost_yield}, since $A$ is a totally ordered set, there is an essentially unique choice of $f$ that is ``most informative'' (as defined in \cref{sec:monotone_comparison}), which embeds the poset $\{0,1,2\}$ into $\reals$ via $x \mapsto x$.\footnotemark{}
			\footnotetext{Since we are only studying the properties of resources as elements of the preordered set $(X, \freeconv)$, any other order-embedding of $\{0,1,2\}$ into $\reals$ is just as valid.
			Any two such embeddings can be related by an order-isomorphism of $\ordreals$.}%
			If one demands resource monotones to be real-valued (as per \cref{def:res_monotone}) and $A$ is not a totally ordered set, there is a choice of $f$ to be made at this stage, which is largely arbitrary.
			Otherwise, one could also use elements of $A$ as the values of a ``resource isotone''.
			
			Once a choice of $f \colon A \to \reals$ is fixed, we find the $\yield{f}{}$ of $x$ by intersecting the free image (i.e.\ downward closure) of $x$
			\begin{equation}
				\down(x) = T_{\rm free} \apply x \in \mathcal{DC}(X)
			\end{equation}
			with $A$ and finding the maximum value of $f$ within this intersection.
			In \cref{ex:toy_cost_yield}, the free image of $1'$ is $\{0,1'\}$, its intersection with $A$ is $\{0\}$ and the maximal value of $f$ is $f(0) = 0$.
			The mediating isotone in this case is the free image map
			\begin{equation}
				\down \colon (X, \freeconv) \to \bigl( \mathcal{DC}(X), \supseteq \bigr),
			\end{equation}
			which is order-preserving by \cref{lem:down_isotone}.
			The root monotone is the function that maximizes $f$:
			\begin{align}
				\bigl(\mathcal{DC}(X), \supseteq \bigr) &\to \ordreals \\
				Y &\mapsto \sup f_*(Y)
			\end{align}
			which is order-preserving by \cref{lem:function_extensions}.
			In summary, we obtain a resource monotone $\yield{f}{} = \sup \circ f_* \circ \down$ which can be also expressed as
			\begin{equation}
				\yield{f}{} (x) = \sup \Set*[\big]{ f(y)  \given  y \in T_{\rm free} \apply x } .
			\end{equation}
			
			On the other hand we can find the $\cost{f}{}$ of $x$ by intersecting the \textbf{free preimage} of $x$
			\begin{equation}\label{eq:free_preimage}
				\up(x) \in \mathcal{UC}(X) 
			\end{equation}
			with $A$ and finding the minimum value of $f$ within this intersection.
			In \cref{ex:toy_cost_yield}, the free preimage of $1'$ is $\{1',2\}$, its intersection with $A$ is $\{2\}$ and the smallest value of $f$ is $f(2) = 2$.
			As our notation in \eqref{eq:free_preimage} suggests, the free preimage of $x$ is the upward closure thereof in the preordered set $(X, \freeconv)$ as introduced in \cref{def:order_closure}.
			It consists of all the individual resources that can be freely converted to $x$:
			\begin{equation}
				\up (x) = \Set{ y \in X  \given  x \in T_{\rm free} \apply y}
			\end{equation}
			and constitutes the mediating isotone for the cost construction.
			The root monotone is the function that minimizes $f$:
			\begin{align}
				\bigl( \mathcal{UC}(X), \subseteq \bigr) &\to \ordreals \\
				Y &\mapsto \inf f_*(Y)
			\end{align}
			which is order-preserving by \cref{lem:function_extensions} again.
			Putting these together gives a monotone $\cost{f}{} = \inf \circ f_* \circ \up$, which can be also expressed as
			\begin{equation}
				\cost{f}{} (x) = \inf \Set*[\big]{ f(y)  \given  x \in T_{\rm free} \apply y } .
			\end{equation}
			
			Both cost and yield constructions are profuse in the literature on resource theories and related studies.
			Here, we list a couple of examples.
			\begin{example}[monotones from dimension functions]\label{ex:dim_cost}
				The above constructions can be useful even when the reference set $A$ consists of all resources.
				In such a case it allows one to turn a candidate measure into an actual resource monotone.
				For instance, if every resource corresponds to an element of a vector space, we can set $f$ to be the dimension of the associated vector space and obtain the dimension cost and yield respectively.
				Evaluating the dimension cost can often be simplified, such as in the resource theory of \mbox{nonuniformity \cite{Gour2015}}, which possess a top element of the resource ordering for each fixed dimension in the form of deterministic states.
				Finding dimension cost then reduces to the question of which of the deterministic states has the smallest dimension while being able to produce the desired resource for free.
			\end{example}
			In \cref{ex:dim_cost}, the root valuation function $f$ is not itself a monotone.
			Indeed, if $A$ contains all resources and $f$ is a monotone, the cost and yield constructions provide no new information about the resource ordering, as \cref{prop:Extensions of monotones are the same on the original domain} below shows.
			However, for general reference sets $A$ we often do want $f$ to be a monotone.
			This is because the more information $f$ carries about the ordering of resources in $A$, the more informative the resulting cost and yield measures are (cf.\ \cref{thm:yield and cost from more informative functions}).
	
			\begin{proposition}
				\label{prop:Extensions of monotones are the same on the original domain}
				Let $f \colon A \to \reals$ be a monotone. 
				Then for all $a \in A$, we have
				\begin{equation}
					\yield{f}{}(a) = f(a) = \cost{f}{}(a).
				\end{equation}
			\end{proposition}
			\begin{proof}
				Since every $a$ is an element of both $\down (a)$ and $\up(a)$, we have
				\begin{equation}
					\cost{f}{}(a) \geq f_W(a) \geq \yield{f}{} (a)
				\end{equation}
				whenever $a$ is in $A$. 
				On the other hand, $f$ being a monotone on its domain $A$ implies that for all $x,y \in A$ satisfying $y \preceq a \preceq x$, we have $f(y) \leq f(a) \leq f(x)$. 
				Performing a supremum of the left inequality over all $y \in A \cap \down(a)$ yields $\cost{f}{}(a) \leq f(a)$, while taking the infimum of the right inequality over all $x \in A \cap \up(a)$ gives $f(a) \leq \yield{f}{} (a)$, so that the result follows.
			\end{proof}
			
			\begin{example}[yield and cost with respect to a chain]\label{ex:currencies}
				The ``currencies'' described in~\cite{Kraemer2016} are yield and cost monotones, wherein $A$ is totally ordered.
				A concrete example of this kind---from entanglement theory---is the cost of an entangled state measured in the number of e-bits (i.e.\ maximally entangled 2-qubit states) needed to produce it.
				It is called the single-shot entanglement cost.  
				In that case, $A$ is the set of $n$-fold tensor products of e-bits for different values of $n$ and $f$ just returns the integer $n$.
				Another example---from the classical resource theory of nonuniformity~\cite{Gour2015}---is the single-shot nonuniformity yield\footnotemark{} of a probability distribution, where $A$ is the set of sharp states, and $f$ is the Shannon nonuniformity.
				\footnotetext{Both of the examples presented here also have their dual counterparts of course.
				They are called the single-shot entanglement yield and single-shot nonuniformity cost respectively.}%
			\end{example}
			\begin{example}[axiomatic definitions of thermodynamic entropy]
				In the axiomatic approach to thermodynamics \cite{Lieb1999}, Lieb and Yngvason define the canonical entropy $S$---an essentially unique monotone among equilibrium states---as a currency.
				Moreover, central to the study of non-equilibrium states in this context are the monotones $S_{-}$ and $S_{+}$, defined in \cite{Lieb2013} as the $\yield{S}{}$ and $\cost{S}{}$, where the reference set $A$ consists of equilibrium states.
				In \cite{Weilenmann2016}, this approach to thermodynamics was directly related to the manifestly resource-theoretic approach \cite{Janzing2000,Brandao2013}, and it was shown that $S$, $S_{-}$, and $S_{+}$ correspond to versions of the Helmholtz free energy introduced in \cite{Horodecki2013}.
			\end{example}
			\begin{example}[entanglement rank as cost]
				In the resource theory of bipartite entanglement, entanglement rank \cite{Terhal2000} is a monotone that can be expressed as a cost for the Schmidt rank of pure states as explained in detail in \cite{gour2020optimal}.
			\end{example}
			\begin{example}[convex roof extension as a cost costruction]
				Similarly, the entanglement of formation \cite{Bennett1996} can also be seen as an instance of the generalized cost construction if we let $A$ be the set of all pure state ensembles and we let $f$ be the convex extension of the entanglement entropy of bipartite pure states to ensembles.
				That is, for an ensemble (or, equivalently, a hybrid quantum-classical state) of pure quantum states $\rho_j$ with probabilities $\lambda_j$, we define its extension to $A$ by
				\begin{equation}
					f_{\rm ens} \bigl( \{ \lambda_j, \rho_j \} \bigr) \coloneqq \sum_j \lambda_j f(\rho_j)
 				\end{equation}
 				where $f(\rho_j)$ is the standard entanglement entropy of $\rho_j$.
				More generally, any monotone obtained by the convex roof extension method \cite{Uhlmann1997,Vidal2000,Uhlmann2010} arises as a cost monotone in such a way.
 				One can recognize that the property of strong monotonicity, which is necessary for the convex roof construction to be applicable to a pure state monotone, corresponds to the monotonicity of $f_{\rm ens}$.
			\end{example}
			\begin{example}[yield and cost for nonclassical correlations]\label{ex:quantum_correlations}
				One of the resource theories where cost and yield monotones as presented here have been used explicitly is the resource theory of nonclassicality of common-cause boxes \cite{Wolfe2019}.
				The resources are bipartite classical channels (also known as ``boxes'' in this context) represented by a conditional probability distribution $P(X,Y | S,T)$ and often depicted as a process
				\begin{equation}
					\input{common_cause_box.tikz}
				\end{equation}
				and thought of as a stochastic map $S \otimes T \to X \otimes Y$.
				$X$ and $Y$ represent the spaces of outcomes for the two parties, while $S$ and $T$ are the respective settings.
				Moreover, the boxes are required to be non-signalling, so that they can be conceivably interpreted as describing a common-cause relationship between the two parties (with no direct causal influence $S \to Y$ or $T \to X$ mediating their interaction).
				This is the set-up relevant for experiments demonstrating violations of Bell inequalities.
				The free boxes are those that can be explained with a common cause mediated by a classical variable.
				The resource theory then allows us to compare nonclassical behaviours by means of their resourcefulness relative to free operations (see \cite{Wolfe2019} for more details) that respect the causal structure of the Bell scenario
				\begin{equation}\label{eq:Bell_DAG}
					\input{Bell_DAG.tikz}
				\end{equation}
				and only use classical variables $\Lambda$ as common causes.
				
				The maximal amount by which a given box violates a Bell inequality is a resource monotone, as it arises via the yield construction. 
				For instance, one may use the violation of the famous CHSH inequality.
				This gives a measure, which can be alternatively expressed also as a weight or robustness monotone as well \cite[corollary 18]{Wolfe2019}.
				While Bell inequalities delineate the set of free resources precisely, they are insufficient for the characterization of the resource ordering via such yield constructions.
				However, as shown in \cite[section 7.1]{Wolfe2019}, even a single additional monotone can help uncover a range of properties of the resource ordering inaccessible by considering violations of Bell inequalities alone.
				The specific one used therein uses the cost construction.
				Specifically, the gold standard resources ($A$) are given by the chain of boxes interpolating between a PR box at the top of the order and a classical box that is a noisy version thereof. 
				Notably, both of the monotones from \cite{Wolfe2019}, yield relative to CHSH inequality violations and cost with respect to noisy PR boxes, have a closed-form expression as shown in section 6.3 there.
				They may thus aid in discovering explicit formulas for yield and cost monotones in a broader context.
			\end{example}
			\begin{example}[changing the type of resources being evaluated]
				\label{ex:Changing type}
				The general constructions of $\yield{f}{}$ and the $\cost{f}{}$ allow one to extend monotones defined for a particular type of resources to other types.
				A particularly useful example of such a translation is the extension of monotones for states to monotones for higher-order processes within the same resource theory; such as channels, measurements, or combs.
				For instance, in entanglement theory, one can define a monotone for channels from a monotone for states, such as the cost of implementing a given channel (measured in the terms of the number of e-bits used). 
				Let us elaborate on this procedure in a general universally combinable resource theory obtained from a partitioned process theory \cite{Coecke2016} as in \cref{ex:quantale_for_UCRT2}.
				Let $f \colon A \to \reals$ be a function whose domain $A$ is the set of all states in the process theory.
				Then we can express $\yield{f}{}$ for a particular channel $\phi$ as
				\begin{equation}\label{eq:state_to_channel1}
					\input{state_to_channel1.tikz}
				\end{equation}
				The argument of $f$ in the optimization on the right-hand side is the most general state which can be obtained from $\phi$ for free.
				It consists of a correlated pre- and post-processing by free transformations $\rho$ and $\psi$ with the constraint that $\rho$ has a trivial input, i.e.\ it is a state.
				If furthermore $f$ is a monotone on its domain so that it satisfies $f(\psi \circ \tau) \leq f(\tau)$ for any state $\tau$ and any free channel $\psi$, we can simplify $\yield{f}{}$ to
				\begin{equation}\label{eq:state_to_channel2}
					\input{state_to_channel2.tikz}
				\end{equation}
				Likewise, we can express the $\cost{f}{}$ as
				\begin{equation}\label{eq:state_to_channel3}
					\input{state_to_channel3.tikz}
				\end{equation}
				These kinds of constructions have appeared in the works on resource theories of quantum channels \cite{Liu2019,liu2020operational,gour2019quantify}, where the proposed monotones for channels are defined via channel divergences (see example \ref{ex:cost for pairs}).
				However, as theorem 2 of \cite{gour2019quantify} shows, they can be also equivalently seen as originating from monotones for states (defined via state divergences).
			\end{example}
			\begin{example}[yield applied to pairs of resources]
				\label{ex:cost for pairs}
				Generalized channel divergences \cite{Leditzky2018} arise from the generalized yield construction when thinking about the resource theory of pairs of resources, i.e.\ a resource theory of distinguishability \cite{Wang2019}.
				More details on pairs (and other tuples) of resources and in what way they constitute a resource theory can be found in \cref{sec:Encodings}.
			\end{example}
	
		\subsection{Generalized Cost and Yield}\label{sec:general_cost_yield}

			Besides varying the root monotone in the \ref{broad scheme}, one can also vary the mediating isotone.
			In particular, $\down$ and $\up$ are not the only isotones mapping $\ordresx$ to $(\mathcal{DC}(X),\supseteq)$ and $(\mathcal{UC}(X),\subseteq)$ respectively.
			
			\begin{definition}\label{def:image_map}
				Let $U \in \mathcal{P}(T)$ be a set of individual transformations.
				The \textbf{$\bm{U}$-image map} is a suplattice homomorphism $\downwrt{U} \colon \mathcal{P}(X) \to \mathcal{P}(X)$ given by any of the following three identical expressions:
				\begin{equation}
					\downwrt{U}(Y) \coloneqq \bigcup_{y \in Y} U \apply \{y\} = \bigcup_{Z \subseteq Y} U \apply Z = U \apply Y.
				\end{equation}
				On the other hand, the \textbf{$\bm{U}$-preimage map} is a map $\upwrt{U} \colon \mathcal{P}(X) \to \mathcal{P}(X)$ given for each atom $y \in X$ by 
				\begin{equation}
					\upwrt{U}(y) \coloneqq \Set*[\big]{z \in X  \given  y \in U \apply z},
				\end{equation}
				and then extended to sets of individual resources by requiring that it gives a suplattice homomorphism, i.e.
				\begin{equation}
					\upwrt{U}(Y) \coloneqq \bigcup_{y \in Y} \upwrt{U}(y) = \Set*[\big]{z \in X  \given  \downwrt{U}(z) \cap Y \neq \emptyset }.
				\end{equation}
			\end{definition}
			With this notation, we can see that $\down$ is just $\downwrt{T_{\rm free}}$ while $\up$ coincides with $\upwrt{T_{\rm free}}$.	
			Unlike for $\down$ and $\up$, there is not neccesarily a preorder on $X$ for which the maps $\downwrt{U}$ and $\upwrt{U}$ are the downward and upward closure operations respectively.
			In order for that to be the case, $Y \mapsto U \apply Y$ would have to be both reflexive and transitive (recall \cref{def:idmptnt}) element of the quantale of all suplattice endomorphisms in which quantale multiplication is given by sequential composition.
			\begin{lemma}\label{lem:composing image maps}
				Let $U,V$ be sets of individual transformations in a concrete resource theory.
				The respective image and preimage maps satisfy
				\begin{align}
					\downwrt{U} \circ \downwrt{V} &= \downwrt{U \star V}  & \upwrt{U} \circ \upwrt{V} = \upwrt{V \star U}.
				\end{align}
			\end{lemma}	
			\begin{proof}
				The first equality follows directly from the definition of image maps.
				We can prove the second one as follows.
				Let $Y$ be an arbitrary set of resources.
				If $z$ is an element of $\upwrt{U} \circ \upwrt{V} (Y)$, then there must exist an $x \in \upwrt{V} (Y)$ such that $x \in U \apply z$ holds. 
				Additionally, there is also a corresponding $y \in Y$ satisfying $y \in V \apply x$.
				Therefore, we have $y \in V \apply (U \apply z)$ and thus $z \in \upwrt{V \star U} (Y)$, so that 
				\begin{equation}
					\upwrt{U} \circ \upwrt{V} (Y) \subseteq \upwrt{V \star U} (Y)
				\end{equation}
				holds.
				
				On the other hand, if $z$ is an element of $\upwrt{V \star U} (Y)$, then there exists a $y \in Y$ which is an element of $(V \star U) \apply z$.
				Thus, there is an $x \in U \apply z$ which is also in $\upwrt{V}(y)$, which in turn implies $z \in \upwrt{U} \circ \upwrt{V} (Y)$ and
				\begin{equation}
					\upwrt{U} \circ \upwrt{V} (Y) \supseteq \upwrt{V \star U} (Y).
				\end{equation}
				Consequently $\upwrt{U} \circ \upwrt{V} (Y)$ is equal to $\upwrt{V \star U} (Y)$ for all $Y \in \mathcal{P}(X)$, which completes the proof.
			\end{proof}
			
			\Cref{lem:composing image maps} shows that if $S \in \mathcal{P}(T)$ is left-invariant (see \cref{def:invariant_trans}), then $\downwrt{S}(Y)$ is downward closed for all $Y$.
			On the other hand, if $D$ is right-invariant, then $\upwrt{D}(Y)$ is upward closed for all $Y$.
			
			\begin{example}\label{ex:invariant_trans}
				Consider an extension of \cref{ex:image_incompatible} in which there are two copies of every individual resource, one primed and one unprimed.
				The free transformations are the same, but one cannot turn primed resources into unprimed ones or vice versa for free.
				Also let $u$ be the transformation with the following transitions (again, suppressing identities on the side of free transformations)
				\begin{equation}\label{eq:toy_3x2_levels}
					\input{toy_3x2_levels.tikz}
				\end{equation}
				where the fact that $0$ and $0'$ have no outgoing transitions on the right indicates that both $u \apply 0$ and $u \apply 0'$ are the empty set in $\mathcal{P}(X)$.
				The left (and right) invariant sets defined by left (and right) augmentation of $u$ act as
				\begin{equation}\label{eq:toy_3x2_levels_2}
					\input{toy_3x2_levels_2.tikz}
				\end{equation}
				If we denote $T_{\rm free} \star u$ by $S$ and $u \star T_{\rm free}$ by $D$, then the diagrams in \eqref{eq:toy_3x2_levels_2} give us an easy way to read off the $S$-preimage and $D$-image maps which will be useful for getting intuition about generalized cost and yield.
			\end{example}
			
			We now present sufficient conditions for $\downwrt{U}$ and $\upwrt{U}$ to be order-preserving.
			\begin{lemma}\label{lem:U-image is order-preserving}
				Given a right-invariant $D \in \mathcal{P}(T)$, the $D$-image map is an isotone of type $\downwrt{D} \colon (\mathcal{P}(X),\succeq_{\rm enh}) \to (\mathcal{P}(X),\supseteq)$.
			\end{lemma}
			\begin{proof}
				By \cref{lem:enh_down}, we have $Y \succeq_{\rm enh} Z  \iff  \down(Y) \supseteq \down(Z)$.
				Since $\downwrt{D}$ is a suplatice homomorphism by definition, it preserves the underlying lattice order \mbox{relation $\supseteq$} (see \eqref{eq:hom_pres_ord} for an explicit argument).
				As a consquence, relation $\down(Y) \supseteq \down(Z)$ implies $\downwrt{D} \circ \down (Y) \supseteq \downwrt{D} \circ \down (Z)$.
				By \cref{lem:composing image maps} and the assumption of right invariance of $D$, this is then equivalent to $\downwrt{D} (Y) \supseteq \downwrt{D} (Z)$, which is what we wanted to show.
			\end{proof}
			
			Note, however, that the set $\downwrt{D} (Y)$ is in general not downward closed and thus $\downwrt{D}$ need not be an isotone of type $(\mathcal{P}(X),\succeq_{\rm enh}) \to (\mathcal{DC}(X),\supseteq)$.
			Indeed, as we mention above, this is the case only if $D$ is \emph{both} left and right-invariant.
			
			\begin{lemma}\label{lem:V-preimage is order-preserving}
				Given a left-invariant $S \in \mathcal{P}(T)$, the $S$-preimage map is an isotone of type $\upwrt{S} \colon (\mathcal{P}(X),\succeq_{\rm deg}) \to (\mathcal{P}(X),\subseteq)$.
			\end{lemma}
			\begin{proof}
				The proof is analogous to the one of \cref{lem:U-image is order-preserving}.
				Instead of repeating it, let us present a visual proof for the special case of singleton sets of resources $Y = \{y\}$ and $Z= \{z\}$:
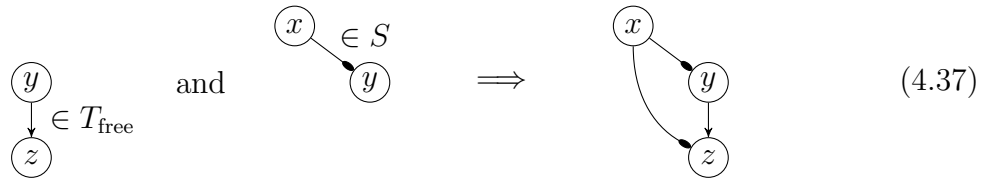
				\begin{equation}
					\input{preim_proof.tikz}
				\end{equation}
				That is, if $y \freeconv z$ holds, then for any $x$ that can be converted to $y$ by a transformation in $S$, there is also a transition $x \mapsto z$ afforded by $S$. 
			\end{proof}
			
			As a consequence, we can use the $\downwrt{D}$ and $\upwrt{S}$ as mediating isotones, restricting their domain to $X$.
			Indeed, both $\freeconv_{\rm enh}$ and $\freeconv_{\rm deg}$ coincide with the resource ordering on singleton sets by definition.
			
			\begin{example}\label{ex:toy_cost_yield2}
				We can illustrate the resulting monotones on \cref{ex:invariant_trans}, choosing $A$ to be $\{0,1,2\}$ and $f$ to be the function that returns the name of the resource as in \cref{ex:toy_cost_yield}.
				The standard $\yield{f}{}$ and $\cost{f}{}$ cannot differentiate between the primed resources---they are all assigned the same value.
				However, replacing $\down$ with $\downwrt{D}$ as the mediating isotone leads to a new monotone $\yield{f}{D}$ that can separate them.
				For a given $x$, $\yield{f}{D}(x)$ is the largest value of $f$ among the resources (in $A$) that can be obtained from $x$ by applying $D$ to it.
				In particular, if $D$ is $u \star T_{\rm free}$, then from diagram \eqref{eq:toy_3x2_levels} we can read-off:
				\begin{align}
					\yield{f}{D}(2') &= 1  &  \yield{f}{D}(1') &= 0  &  \yield{f}{D}(0') &= - \infty.
				\end{align}
				Similarly, replacing $\up$ with $\upwrt{S}$ as the mediating isotone leads to a new monotone $\cost{f}{S}$.
				Specifically, $\cost{f}{S}(x)$ is the smallest value of $f$ among the resources (in $A$) that can be converted to $x$ by applying $S$ to them.
				In our toy example, setting $S = T_{\rm free} \star u$, we have
				\begin{align}
					\cost{f}{S}(2') &= \infty  &  \cost{f}{S}(1') &= 2  &  \cost{f}{S}(0') &= 1.
				\end{align}
			\end{example}
		
			\begin{theorem}[generalized yield]\label{thm:yield}
				Let $(\mathcal{P}(T),\mathcal{P}(X),\apply)$ be a concrete resource theory with free transformations $T_{\rm free}$. 
				Furthermore, let $f \colon X \to \reals$ be a partial function with domain $A$, and let $D$ be a right-invariant subset of $T$.
				The $f$-yield relative to the $D$-image map, $\yield{f}{D} \colon X \to \reals$, defined via
				\begin{equation}
					\yield{f}{D}(x) \coloneqq \sup f_* \bigl( \downwrt{D} (x) \bigr)
				\end{equation}
				is a resource monotone.
			\end{theorem}
			\begin{proof}
				By \cref{lem:U-image is order-preserving,lem:function_extensions}, both $\downwrt{D}$ and $\sup \circ f_*$ are isotones and therefore $\yield{f}{D}$ is too.
			\end{proof}
			
			\begin{theorem}[generalized cost]\label{thm:cost}
				Given the same set-up as in \cref{thm:yield}, let $S$ be a left-invariant subset of $T$.
				The $f$-cost relative to the $S$-preimage map, $\cost{f}{S} \colon \mathcal{R} \to \reals$, defined via
				\begin{equation}
					\cost{f}{S}(x) \coloneqq \inf f_* \bigl( \upwrt{S} (x) \bigr)
				\end{equation}
				is a resource monotone.
			\end{theorem}
			\begin{proof}
				By \cref{lem:V-preimage is order-preserving,lem:function_extensions}, both $\upwrt{S}$ and $\inf \circ f_*$ are isotones and therefore $\cost{f}{S}$ is too.
			\end{proof}
			Unpacking the definitions, we can express the generalized yield and cost monotones as
			\begin{equation}
				\begin{split}
					\yield{f}{D} (x) &= \sup \Set*[\big]{ f(y) \given y \in D \apply x \text{ and } y \in A } \\
					\cost{f}{S} (x) &=\inf \Set*[\big]{ f(y) \given x \in S \apply y \text{ and } y \in A }.
				\end{split}
			\end{equation}
			Note that $\yield{f}{D}$ and $\cost{f}{S}$ are also order-preserving when we consider them as maps $(\mathcal{P}(X), \succeq_{\rm enh}) \to \ordreals$ and $(\mathcal{P}(X), \succeq_{\rm deg}) \to \ordreals$ respectively.
			This follows if we do not restrict the mediating isotones to the singletons.
			
			There are two basic reasons why one might want to use $\yield{f}{D}$ (or $\cost{f}{S}$) instead of $\yield{f}{}$ (or $\cost{f}{}$).
			On the one hand, an invariant set $D$ (or $S$) different from $T_{\rm free}$ can be easier to work with either algebraically or numerically when evaluating the monotone explicitly.
			This is a common practice in many resource theories in which $T_{\rm free}$ is not straightforward to work with.
			For example, LOCC operations in entanglement theory get replaced by separable operations, noisy operations in nonuniformity theory get replaced by unital operations, and thermal operations in athermality theory get replaced by Gibbs-preserving operations.
			
			On the other hand, $\yield{f}{D}$ and $\cost{f}{S}$ can give us new interesting monotones distinct from $\yield{f}{}$ and $\cost{f}{}$. 
			Here we give a simple example of how one could use these constructions for $D \neq T_{\rm free}$ in practice, similar in flavour to the toy \cref{ex:toy_cost_yield2}.
			\begin{example}
			\label{ex:Advantage of generalized yield}
				Consider a (universally combinable) resource theory defined from a partitioned process theory as in \cref{ex:quantale_for_UCRT2}, in which there are no free states.
				Such resource theories arise naturally when we consider multi-resource theories \cite{sparaciari2018multi,Sparaciari2018} such as the resource theory of work and heat \cite{Sparaciari2017}.
				Since a channel can only be converted to a state by applying it to a state as in the right-hand side of \cref{eq:state_to_channel1}, there is no way to convert a channel to a state for free in this case.
				Therefore, evaluating $\yield{f}{}$ for a function $f$ defined on states only leads to a trivial (i.e.\ constant) monotone for channels.
				We would not be able to use this construction to extend monotones for states to monotones for channels.
				One can instead use a right-invariant set $D$ that does include some states, in which case $\yield{f}{D}$ becomes a non-trivial monotone for channels in the resource theory.
				A choice of $D$ that is guaranteed to be invariant and include some states is the augmentation $D = R_{\rm free} \boxtimes U$ of a set of resources $U$ which includes states.
				The set $R_{\rm free} \boxtimes U$ can contain more states than those in $U$ of course.
				In particular, it contains any other state one could obtain from those in $U$ for free. 
			\end{example}
			Every set $R_{\rm free} \boxtimes U$ is invariant, but it need not be idempotent, in which case it is not a candidate for the set of free resources in a universally combinable resource theory. 
			Nevertheless, we can use it in generalized yield and cost constructions. 
			One way to interpret\footnotemark{} taking the images and preimages with respect to $R_{\rm free} \boxtimes U$ is as follows.
			\footnotetext{This interpretation is valid if the discarding operation is a free resource or else if $U$ contains the neutral set $1$.}%
			They specify what can be achieved by an agent who, in addition to having access to the free resources in unlimited supply, also has access to a single resource from $U$.
			Of course, if $R_{\rm free} \boxtimes U$ \emph{is} closed under $\boxtimes$, then we can think of it as describing access to both $R_{\rm free}$ and $U$ in unlimited supply. 
			The generalized cost and yield end up being standard cost and yield constructions in the universally combinable resource theory with free resources given by $D$. 
		
	\section{Translating Monotones}\label{sec:translating}
			
		\subsection{Cost and Yield as Translations}

			Let us change the mediating preordered sets in constructing monotones via the \ref{broad scheme}.
			In \cref{sec:cost_yield}, we looked at $\mathcal{DC}(X)$ and $\mathcal{UC}(X)$ as possible choices, and we made use of the fact that the ordering on each is defined in terms of subset inclusion in $\mathcal{P}(\mathcal{R})$. 
			In the present section, we investigate what can be said about the case when the mediating set corresponds to sets of resources in another resource theory $(\mathcal{P}(S), \mathcal{P}(W), \apply)$.
			That is, we consider mediating preordered sets $(\mathcal{P}(W), \freeconv_{\rm enh})$ and $(\mathcal{P}(W), \freeconv_{\rm deg})$.
			By \cref{lem:monotone_extensions}, any resource monotone $f \colon W \to \reals$ gives rise to two root monotones
			\begin{equation}
				\begin{split}
				\mathrm{sup} \circ f_* &\colon \bigl( \mathcal{P}(W), \enhgeq \bigr) \to \bigl(\reals , \geq \bigr) \\
				\mathrm{inf} \circ f_* &\colon \bigl( \mathcal{P}(W), \deggeq \bigr) \to \bigl(\reals , \geq \bigr).
				\end{split}
			\end{equation}
			Since the target monotones are going to be resource measures for the resource theory $(\mathcal{P}(T), \mathcal{P}(X), \apply)$, we can view this instance of the \ref{broad scheme} as giving prescriptions for translating monotones from one resource theory to another.
			
			First of all, let us look at a particularly simple case of the root and target resource theories being identical.
			One mediating isotone we can use is the $U$-image map whenever
			\begin{equation}\label{eq:oplax_natural}
				T_{\rm free} \star U \supseteq U \star T_{\rm free}
			\end{equation}
			holds, so that it is order-preserving by \cref{lem:Compatibility of order and combination}.
			This results in a target monotone given by the composition
			\begin{equation}
				\begin{tikzcd}
					(X, \freeconv) \ar[r, "\downwrt{U}"] 	& 	\bigl( \mathcal{P}(X), \freeconv_{\rm enh} \bigr) \ar[r, "\mathrm{sup} \, f_*"] 	& 	\ordreals
				\end{tikzcd}
			\end{equation}
			which is similar to $\yield{f}{U}$, but not quite the same.
			In particular, $U$ need not be right-invariant and $f$ has to be a monotone on its domain in this case.
			At the formal level, every target monotone of the above form can be expressed as a generalized yield monotone.
			However, in practice, using $\downwrt{U}$ for a set $U$ that is not right-invariant may be beneficial because it can result in reduction of the feasible set in the optimization.
			\begin{example}
				In \cref{ex:invariant_trans}, the set $U \coloneqq u \cup 1$ satisfies the desired relation \eqref{eq:oplax_natural}, where $1$ denotes the neutral set of transformations that leave all $x$ in their domain fixed.
				The monotone $\sup \circ f_* \circ \downwrt{U}$ is identical to $\yield{f}{D}$ with $D = u \star T_{\rm free}$, but since $U$ induces fewer non-trivial transitions, it may be preferred in certain contexts.
			\end{example}
			\begin{example}
				\label{ex:supplementation}
				In a universally combinable resource theory, condition \eqref{eq:oplax_natural} is automatically satisfied for any $U$ with equality of course.
				Thus, augmentation by an arbitrary set of resources $U \in \mathcal{P}(R)$ (see \cref{def:augment}) is a valid mediating isotone that allows one to translate monotones for the $U$-catalytic ordering defined in \eqref{eq:catalytic_order} to ones for the original resource ordering.
			\end{example}
			\begin{example}
				\label{ex:copy}
				Once again in the universally combinable case, we can define a map $\mathsf{Copy}_2 \colon (R, \freeconv) \to (\mathcal{P}(R), \succeq_{\rm enh})$ via
				\begin{equation}
					\mathsf{Copy}_2(r) \coloneqq r \boxtimes r.		
				\end{equation}
				Given that $s \in R_{\rm free} \boxtimes r$ implies $s  \boxtimes s \in (\mathcal{R}_{\rm free} \boxtimes r) \boxtimes (\mathcal{R}_{\rm free}  \boxtimes r) = \mathcal{R}_{\rm free} \boxtimes (r \boxtimes r)$, it follows that 
					$s \succeq r$ implies $s \boxtimes s \succeq r \boxtimes r$, and consequently $\mathsf{Copy}_2$ is an isotone.
				The same works for the map $\mathsf{Copy}_n \colon R \to \mathcal{P}(R)$ given by the $n$-fold universal combination. 
			\end{example}
			
			In order to obtain similar kinds of isotones for degradation ordering, we need a dual version of \cref{lem:Compatibility of order and combination}, which in effect generalizes \cref{lem:V-preimage is order-preserving}.
			\begin{lemma}\label{lem:preimage_isotone}
				Consider a concrete resource theory and a set of transformations $U$ that satisfies
				\begin{equation}\label{eq:lax_natural}
					T_{\rm free} \star U \subseteq U \star T_{\rm free}.
				\end{equation}
				Then we have
				\begin{equation}\label{eq:preimage_isotone}
					Y \degconv Z  \quad \implies \quad  \upwrt{U} (Y) \degconv \upwrt{U} (Z)
				\end{equation}
				or, in other words, the map $\upwrt{U} \colon (\mathcal{P}(X), \degconv) \to (\mathcal{P}(X), \degconv)$ is an isotone.
			\end{lemma}
			\begin{proof}
				By \cref{lem:deg_up} and the fact that preimage maps are suplattice homomorphisms, we have
				\begin{equation}
					Y \degconv Z  \iff  \up (Y) \subseteq \up (Z)  \implies  \upwrt{U} \circ \up (Y) \subseteq  \upwrt{U} \circ \up (Z).
				\end{equation}
				By \cref{lem:composing image maps}, the latter is equivalent to 
				\begin{equation}\label{eq:preimage_isotone_3}
					\upwrt{T_{\rm free} \star U} (Y) \subseteq  \upwrt{T_{\rm free} \star U} (Z).
				\end{equation}
				Now, for arbitrary set of transformations $V'$ and a subset $V$ thereof, definition of preimage maps directly implies
				\begin{equation}
					\upwrt{V} (S) \subseteq \upwrt{V'} (S)
				\end{equation}
				for all $S \in \mathcal{P}(X)$.
				Since we have $U  \subseteq T_{\rm free} \star U$ and $T_{\rm free} \star U \subseteq U \star T_{\rm free}$, using this fact in conjunction with relation \eqref{eq:preimage_isotone_3} gives
				\begin{equation}
					\upwrt{U} (Y) \subseteq \upwrt{T_{\rm free} \star U} (Y) \subseteq  \upwrt{T_{\rm free} \star U} (Z) \subseteq \upwrt{U \star T_{\rm free}} (Z).
				\end{equation}
				Applying $\up$ to this relation and using $\up \circ \up = \up$ as well as \cref{lem:composing image maps} again we obtain
				\begin{equation}
					\up \bigl( \upwrt{U} (Y)  \bigr)  \subseteq \up \circ \upwrt{U \star T_{\rm free}} (Z) = \up \bigl( \upwrt{U} (Z) \bigr),
				\end{equation}
				which is equivalent to the consequent of implication \eqref{eq:preimage_isotone} by \cref{lem:deg_up}.
			\end{proof}
			Consequently, for any $U$ satsfying \eqref{eq:lax_natural} and any (partial) monotone $f$, we get another monotone given by the composition
			\begin{equation}
				\begin{tikzcd}
					(X, \freeconv) \ar[r, "\upwrt{U}"] 	& 	\bigl( \mathcal{P}(X), \freeconv_{\rm deg} \bigr) \ar[r, "\inf \, f_*"] 	& 	\ordreals
				\end{tikzcd}
			\end{equation}
		
			\subsection{Translating Measures of Distinguishability: Examples}
			\label{sec:Monotones from Information Theory}
				
				Many resource theories of interest either have an information-theoretic flavour or explicitly describe resources of information.
				It is no surprise then, that in these resource theories, measures of information often crop up as resource monotones or as building blocks for them.
				As we describe below, to many concrete resource theories it is possible to associate an ``information theory'' where the resources constitute the alphabet used to encode a classical message.
				This association can then be used to understand such results in greater generality.
		
				Consider a quantum resource theory of states constructed as in \cref{ex:rt_states}.
				That is, $X$ is the set of all quantum states, $T$ is the set of all quantum channels, and $X_{\rm free}$ contains states considered to be free.
				In particular, $X_{\rm free}$ is the free image of the unique state on the trivial system $I$.
				A contraction in such a quantale module is a real-valued function $f$ of pairs of quantum states that satisfies the data processing inequality
				\begin{equation}
					\label{eq:data processing inequality}
					f(\rho, \sigma) \geq f \bigl(\Phi(\rho), \Phi(\sigma) \bigr)
				\end{equation}
				for all states $\rho$ and $\sigma$ and all channels $\Phi$.
				Usually, we consider $f$ to be defined only when $\rho$ and $\sigma$ are states on the same system.
				Thus, $f$ is a partial function of type $X \times X \to \reals$.
		
				We now provide a couple of examples of monotone constructions based on contractions that we aim to understand and generalize here.
		
				\begin{example}[minimal distinguishability from free resources]
					\label{ex:Monotones from contractions}
					Given a contraction, it is well-known that one can obtain a monotone by minimizing its value over the set of all free states in one of its arguments.
					That is, the function $m \colon X \to \reals$ given by
					\begin{align}
						\label{eq:monotone from contraction}
						m (\rho) = \inf \Set*[\big]{f(\rho, \sigma)  \given  \sigma \in X_{\rm free} }
					\end{align}
					is a resource monotone.
					In the resource theory of quantum entanglement, a popular measure of this type is the relative entropy of entanglement \cite{vedral1997quantifying}.
					Various other monotones based on contractive distance measures such as the trace distance, relative R\'{e}nyi entropies \cite{Petz1986}, and many others arise in this way as well.
					An extensive overview of these kinds of monotones can be found in \cite{Chitambar2018}.
				\end{example}
				
				\begin{example}
				\label{ex:Monotones from contractions 2}
					
					Consider a quantum resource theory of $G$-asymmetry with respect to a unitary representation of a group $G$ given by
					\begin{equation}
						\varphi_g (\rho) = U_g \rho U_g^\dagger
					\end{equation}
					for each $g \in G$.
					It is one where the free channels are $G$-covariant (see \cref{def:G-covariant}) with respect to the $G$-action $\varphi$ on quantum states.
					Let $\gamma_{\mu}$ be the twirling map weighted by a probability measure $\mu$ over the group, i.e.\
					\begin{equation}
						\gamma_{\mu} (\rho) \coloneqq \int_G U_g \rho U^{\dag}_g \, \mu(dg).
					\end{equation}
					Given a contraction $f$, the following function is a monotone \cite{Marvian2012}:
					\begin{equation}\label{eq:twirl_monotone}
						m(\rho) \coloneqq f \bigl( \rho, \gamma_{\mu} (\rho) \bigr).
					\end{equation}  
					The proof of monotonicity relies on the fact that any free operation mapping $\rho$ to $\sigma$ also maps $\gamma_{\mu} (\rho)$ to $\gamma_{\mu} (\sigma)$, since a twirling map commutes with every $G$-covariant operation.
					It is worth considering some special cases of this monotone.
					If $\mu$ is the Dirac delta measure supported on $g_0 \in G$, so that $m(\rho)= f(\rho, \varphi_{g_0} (\rho) )$, then the monotone quantifies how distinguishable $\rho$ is from its image under the action of $g_0$.
					If $\mu$ is the Haar measure, then $m$ quantifies how distinguishable $\rho$ is from its ``uniformly twirled'' counterpart.
					One can then understand the monotonicity of these functions intuitively as the statement that more asymmetric states are more distinguishable from their rotated and uniformly twirled counterparts.
					Furthermore, note that one can also obtain a monotone from a contraction and a pair of distributions, $\mu$ and $\nu$, given by $m(\rho) = f ( \gamma_{\mu} (\rho), \gamma_{\nu} (\rho))$ \cite{Marvian2012}.
				\end{example}
				
				\subsection{Resource Theories of Distinguishability}
				\label{sec:Encodings}
				
					In order to understand monotones obtained from contractions (\cref{ex:Monotones from contractions,ex:Monotones from contractions 2}) as special cases of the \ref{broad scheme} and to thereby generalize them, we introduce a class of resource theories in which contractions are resource monotones.
					Thus, we now take a detour from the discussion of resource measures to describe a general construction of resource theories of distinguishability.
					That is, for a concrete resource theory with underlying quantale module $(\mathcal{P}(T), \mathcal{P}(X), \apply)$ and a set of hypothses $H$, the corresponding resources and transformations are given by assignments
					\begin{align}
						H &\to R  & &\text{and} &   H &\to T.
					\end{align}
					Thinking of these as tuples indexed by $H$, we write $R^H$ and $T^H$ for the sets of all encodings of individual resources and transformations respectively.
					In the resource theory of $H$-distinguishability, the suplattices of resources and transformations are thus power sets $\mathcal{P}(R^H)$ and $\mathcal{P}(T^H)$.
					
					The free transformations among $T^H$ are the constant encodings---ones that do not depend on $H$.
					That is, they can be implemented without knowing the value of the hypothesis.
					The resources that are close to the top of the resulting resource ordering are thus thought to be highly informative about the hypothesis $H$.
					Indeed, any information about $H$ carried by an encoding after a free transformation is applied must have already been present in the original encoding before the transformation.
					Alternatively, we can think of the free transformations as those that preserve indistinguishability of resources.
					However, to prescribe the right notion of composition under which free transformations are closed, we need the original quantale $(\mathcal{P}(T),\star)$ to have extra properties.
					
					\begin{example}\label{ex:tuple_composition}
						Let us look at a toy example of what might be the issue in the completely general case.
						Consider $s, t \in T$ whose composite $s \star t$ is not a singleton.
						For concreteness, let $s \star t$ be equal to $\{u_1, u_2\}$ for some $u_1, u_2 \in T$.
					
						Let $H$ be a binary hypothesis.
						Then $(s,s)$ and $(t,t)$ are both constant encodings and therefore free transformations in the resource theory of $H$-distinguishability.
						What should the composite $(s,s) \star (t,t)$ be as an element of $\mathcal{P}(T^H)$?
						It cannot be the set 
						\begin{equation}\label{eq:cand_comp}
							 \bigl\{ (u_1,u_1),  (u_1,u_2), (u_2,u_1), (u_2,u_2) \bigr\} \approx  (s\star t, s\star t)
						\end{equation}
						because it would contain elements that are not free!
						In other words, if we think of $u_1$ and $u_2$ as different ways in which $s$ and $t$ can be combined, then to obtain \eqref{eq:cand_comp}, one would need to have the ability to combine $s$ and $t$ in a way that depends on the value of $H$.
						Instead, we might expect that $(s,s) \star (t,t)$ should be equal to $\{ (u_1,u_1), (u_2,u_2) \}$, so that $H$ cannot inform one about the way in which to combine $s$ and $t$.
						
						This raises the question of what the composition of non-constant encodings ought to be.
						For illustrative purposes, consider another transformation $r \in T$ which satisfies $s \star r = \{ v \}$ for some $v \in T$.
						There are at least four alternatives for what the composition $(s,s) \star (r,t)$ might reasonably be:
						\begin{equation}\label{eq:cand_comp_2}
							\begin{aligned}
								\bigl\{ (v &, u_1) \bigr\}  			& &\qquad \qquad &  	\bigl\{ (v &, u_2) \bigr\} \\
								\bigl\{ (v,u_1) &, (v, u_2) \bigr\}  	& &\qquad \qquad &   	&\emptyset 
							\end{aligned}
						\end{equation}
						and it is not clear which one is the right choice. 
						We claim that this question is underdetermined.
						What we are lacking is the data specifying whether the way in which $s$ and $r$ are to be combined to obtain $v$ is \emph{the same} as the way in which $s$ and $t$ can be combined to obtain $u_1$ (or $u_2$).
					\end{example}
					
					One could argue that the right setting to include this kind of data is in the context of resource theories with ``types''.
					This approach is sketched in \cite[appendix C]{Monotones}.
					However, typed resource theories are beyond the scope of this thesis, we merely mention the steps one could take to extend the present framework to include resource types in \cref{sec:future}.
					Instead, we make a more abstract assumption about the nature of the quantale composition that allows us to carry out this construction and subsumes typed resource theories.\footnotemark{}
					\footnotetext{Strictly speaking, it subsumes only those typed resource theories that are ``small'' in the sense that they form a small category.}%
					
					\begin{definition}\label{def:listable}
						A uniquely atomistic quantale module $(\mathcal{P}(T), \mathcal{P}(X), \apply)$ is \textbf{listable} if there is a set $\mathcal{I}$ such that for all $s,t \in T$ and $x \in X$ we have
						\begin{align}
							\label{eq:listable_star}  s \star t &= \bigcup_{i \in \mathcal{I}} s \star_i t  	& 	 t \apply x &= \bigcup_{i \in \mathcal{I}} t \apply_i x
						\end{align}
						where each $s \star_i t$ and each $t \apply_i x$ is either a singleton or an empty set.
						Furthermore, we require that for all $i,j \in \mathcal{I}$ we have
						\begin{align}\label{eq:listable_associativity}
							(s \star_i t) \star_j u &= s \star_i (t \star_j u)  &   (s \star_i t) \apply_j x &= s \apply_i (t \apply_j x)
						\end{align}
						so that, in particular, each $(\mathcal{P}(T), \star_i)$ is a (not necessarily unital) quantale and each $(\mathcal{P}(T), \mathcal{P}(X), \apply_i)$ is a (not necessarily unital) quantale module.
					\end{definition}
					In other words, $(\mathcal{P}(T), \star)$ can be seen as the quantale associated to a collection of partial semigroups (one for each $i$) as in \cref{ex:free_quantale_multi}.
					Additinally, \cref{def:listable} postulates that each of the semigroups comes with a (partial) action on $X$.

					We think of the index set $\mathcal{I}$ as part of the data specifying a listable quantale module.
					That is, when speaking of a listable quantale module (or a listable resource theory), we presume that it comes with a choice of the index set $\mathcal{I}$.
					In this sense, a fixed quantale module could be associated with \emph{multiple} listable quantale modules.
					Many concrete resource theories from the literature are resource theories of states in the sense of \eqref{ex:rt_states}, in which case there is only one mode of composition and $\mathcal{I}$ is thus a singleton.
					We use these in many examples that follow since they tend to be easier to get an intuitive understanding of.
					For future reference, listable resource theories with $\abs{\mathcal{I}} = 1$ are called \textbf{unimodular}.
					
					\begin{example}\label{ex:listable_rt_states}
						Quantale modules arising from process theories as described in \cref{ex:rt_states_2,ex:rt_channels} provide further examples of listable ones.
						To provide more details, we need to be careful about what we mean by a diagram in a process theory.
						Specifying a list of contractions between processes requires us to be able to identify \emph{which wires} are being contracted.
						That is why we need to give names to them.
						These names differ from the wire \emph{types} that are commonly used in process theories.
						Types of wires tell us which contractions are not implementable.
						Two distinct wires of a process can have the same type.
						This encodes that they share the restriction on which other wires they may be plugged into.
						Since we want to use wire names to refer to a specific wire, the names of wires should be distinct.
						
						Let us illustrate this on the quantale module from \cref{ex:rt_states_2}.
						We denote the name of a specific wire inside a circle attached to said wire.
						For instance the diagram
						\begin{equation}\label{eq:process_names}
							\input{process_names.tikz}
						\end{equation}
						depicts a process whose input wire is called ``1'' and whose output wire is called ``2''.
						A contraction between two such labelled processes can be specified by a relation between their names.
						Not every wire has to be contracted, but a given wire can only be plugged into at most one wire.
						That is why the relation specify a contraction should be a partial injective function.
						If we use the natural number to label wires, then the index set $\mathcal{I}$ consists of all partial injective functions $\mathbb{N} \to \mathbb{N}$.
						
						Consider two labelled processes as follows, where $\alpha$ and $\beta$ indicate wire types as before.
						\begin{equation}\label{eq:process_names_2}
							\input{process_names_2.tikz}
						\end{equation}
						Their composition $s \star t$ according to \cref{eq:rt_states_5} consists of two non-trivial contractions:
						\begin{equation}\label{eq:process_names_3}
							\input{process_names_3.tikz}
						\end{equation}
						and
						\begin{equation}\label{eq:process_names_4}
							\input{process_names_4.tikz}
						\end{equation}
						where $g$ and $h$ are the relations (specified as subsets of $\mathbb{N} \times \mathbb{N}$ here)
						\begin{align}
							g &\coloneqq \emptyset   &  h &\coloneqq \{ (1,4) \}
						\end{align}
						respectively.
						By convention, the first label in the pair gives the name of the wire of the left process (i.e.\ $s$ here) to be contracted with the wire of the right process (i.e.\ $t$) referred to by the second label in the pair.
						In the case of the quantale analyzed here (i.e.\ the one from \cref{ex:rt_states_2}), the first label thus has to refer to an input wire while the second label refers to an output wire, otherwise the contraction is not allowed.
						This would not be true in other quantales such as the one from \cref{ex:quantale_for_UCRT2}.
						Note that other choices of $g$ and $h$ can also give rise to the same contractions, such as 
						\begin{align}
							g &= \{ (5,5) \}   &  h &= \{ (1,4) , (5,5) \}
						\end{align}
						but not
						\begin{align}
							g &= \{ (1,1) \}   &  h &= \{ (1,4) , (2,3) \}.
						\end{align}
						In the latter case, both $s \star_g t$ and $s \star_h t$ would be the empty set of processes, as opposed to a singleton set of process as in \eqref{eq:process_names_3} and \eqref{eq:process_names_4}.
						We also have
						\begin{equation}\label{eq:process_names_5}
							\input{process_names_5.tikz}
						\end{equation}
						for $f = \{(3,2)\}$, because the output wire of $s$ has a different type compared to the input wire of $t$ so that they cannot be contracted.
						
						The very same indexing can be used for the action $\apply$ of channels on states.
						The only difference is that the result is an empty set of states whenever there are some input wires of the channel left uncontracted.
						As for the second instance of composition of process with multiple wires, depicted by \cref{eq:rt_states_3}, we can use the following indexing choice.
						\begin{equation}\label{eq:process_names_6}
							\input{process_names_6.tikz}
						\end{equation}
						The composition of the above $s$ and $t$ consists of 4 distinct processes which can be described by the following contractions that correspond to the choice of $g,h,e,f$ from the index set $\mathcal{I}$ as indicated below.
						\begin{equation*}\label{eq:process_names_7}
							\input{process_names_7.tikz}
						\end{equation*}
						
						The above discussion can be made more precise by introducing interfaces and ports as defined in \cite[section 1.3B]{houghton2021mathematical}.
						However, we do not need these notions here and thus we now return to the general discussion building up to the definition of resource theories of distinguishability.
					\end{example}
					
					\begin{example}\label{ex:tuple_composition_2}
						Returning to \cref{ex:tuple_composition}, we can now show how to answer the question of what the composition $(s,s) \star (r,t)$ should be, provided that the original quantale module is listable.
						Let $\mathcal{I}$ be the set $\{1,2,3\}$ and let
						\begin{align}
							s \star_1 t &= u_1  &  s \star_2 t &= u_2  &  s \star_3 t &= \emptyset
						\end{align}
						for instance.
						\Cref{tab:tuple_composition} depicts four different cases of how $s \star r = v$ could arise via the left equation in \eqref{eq:listable_star} and the resulting $(s,s) \star (r,t)$ that corresponds to each.
						The idea is that when we combine $(s,s)$ with $(r,t)$, any of the choices in $\mathcal{I}$ is allowed as long as it is the same choice for all hypotheses in $H$.
						As we can see, all the candidates from \eqref{eq:cand_comp_2} do arise.
						\bgroup
						\renewcommand{\arraystretch}{1.2}%
						\begin{table}[h]
						\centering
							\begin{tabular}
							{>{\raggedleft}m{0.15\textwidth}>{\centering}m{0.17\textwidth}>{\centering}m{0.17\textwidth}>{\centering}m{0.17\textwidth}>{\centering\arraybackslash}m{0.17\textwidth}}
								\hline
								$s \star_1 r :$        & $v$    			& $\emptyset$ 	& $v$    			& $\emptyset$ 	\\
								$s \star_2 r :$        & $\emptyset$ 	& $v$    			& $v$    			& $\emptyset$ 	\\
								$s \star_3 r :$       & $\emptyset$ 	& $\emptyset$ 	& $\emptyset$ 	& $v$    			\\ \hline
								$(s,s) \star (r,t)$ : & 
									$ (v, u_1) $ 						& 
									$ (v, u_2) $ 						& 
									$ (v,u_1) \cup (v, u_2) $ 	& 
									$\emptyset $ 						\\ \hline
							\end{tabular}
							\caption[Example of Composition in a Resource Theory of Distinguishability]{First three rows show four of the possibilities for what $s \star_i r$ could be. 
								The last row gives the associated $ (s,s) \star (r,t)$ in each case.}
							\label{tab:tuple_composition}
						\end{table}
						\egroup
					\end{example}
		
					More generally, given a finite hypothesis $H = \{1, 2, \dots, k\}$ and two encodings of transformations
					\begin{align}
						s &= \bigl( s_1, s_2, \dots, s_k \bigr) \in T^H  &  t &= \bigl( t_1, t_2, \dots, t_k \bigr) \in T^H
					\end{align}
					in a listable quantale module, we define 
					\begin{equation}\label{eq:tuple_comp}
						s \star t \coloneqq \bigcup_{i \in \mathcal{I}}  \Bigl\{ \bigl( s_1 \star_i t_1, s_2 \star_i t_2, \dots , s_k \star_i t_k \bigr)  \Bigr\}.
					\end{equation}
					Similarly for any encoding of resources $x \in X^H$, we define
					\begin{equation}\label{eq:tuple_action}
						s \apply x \coloneqq \bigcup_{i \in \mathcal{I}}  \Bigl\{ \bigl( s_1 \apply_i x_1, s_2 \apply_i x_2, \dots , s_k \apply_i x_k \bigr)  \Bigr\}.
					\end{equation}
					Furthermore, both of these can be extended to sets of encodings by requiring compatibility with unions.
					Properties \eqref{eq:listable_associativity} then ensure that we get a quantale module $(\mathcal{P}(T^H), \mathcal{P}(X^H), \apply)$ which, when equipped with free transformations given by
					\begin{equation}\label{eq:cons_encodings}
						T^H_{\rm unif} \coloneqq \Set*[\big]{ (t, t, \dots , t)  \given  t \in T }.
					\end{equation}
					defines the \textbf{resource theory of (unconstrained) $\bm{H}$-distinguishability} associated to a listable quantale module $(\mathcal{P}(T), \mathcal{P}(X), \apply)$.
					Its resource ordering is denoted by $(\mathcal{P}(X^H), \succeq_{\rm unif})$.
					
					In the context of quantum theory, there have been several works investigating such resource theories \cite{Wang2019,wang2019resource,katariya2020evaluating,salzmann2021symmetric}.
					A relevant earlier investigation of the conditions for conversions of tuples of quantum states can be found in \cite{chefles2004existence}.
					On the other hand, for classical probability theory, the resource ordering of distinguishability is known as the Blackwell order for comparison of statistical \mbox{experiments \cite{blackwell1951comparison}} as well as matrix majorization \cite{Dahl1999} in the case of finite sample spaces.
					We give more details on the latter in \cref{sec:distinguishability}.
					
					Alternatively, we can think of the resource theory of $H$-distinguishability as a resource theory of asymmetry with respect to permutations of $H$.
					For example, in a unimodular theory, permutation-covariance of reads
					\begin{equation}
						\forall \, x \in X , \; \forall i,j  \, : \quad  t_i \apply x = t_j \apply x
					\end{equation}
					which in many ciscumstances implies that the individual transformations $t_i$ are in fact identical in $\mathcal{T}$.
					
					\begin{remark}
						It should be noted that \cref{def:listable} as given here is not quite strong enough to guarantee that $(\mathcal{P}(T^H), \star)$ has a unit and that $T^H_{\rm unif}$ is reflexive.
						We nevertheless assume this to be the case so that we can use the framework of resource theories in the sense of \cref{def:rt}.
						There are different ways in which this assumption can be born out in practice.
						For example, in quantales arising from process theories as in \cref{ex:quantale_for_UCRT1}, there is a single transformation in the neutral set---the empty diagram.
						Moreover, there is a single element $i$ of $\mathcal{I}$---the parallel composition (i.e.\ diagram placement with no contractions)---such that $1 \star_i t = t = t \star_i 1$.
						In such situations, $(\mathcal{P}(T^H), \star)$ is guaranteed to be unital and $T^H_{\rm unif}$ is necessarily reflexive.
					\end{remark}

					If we now return to the definition of contractions for quantum states via data processing inequality \eqref{eq:data processing inequality}, we can see that monotones in the resource theory of $\{0,1\}$-distinguishability provide a suitable generalization of this notion.
					We thus refer to monotones in this resource theory as \textbf{contractions}. 
					Analogously, a monotone in the resource theory of $H$-distinguishability is termed a $\bm{k}$\textbf{-contraction}, where $k$ denotes the cardinality of $H$.
						
					These concepts have been implicitly present in some of the definitions in quantum theory.
					For instance, the diamond norm\footnotemark{} distance is a contraction in the binary hypothesis case.
					\footnotetext{Diamond norm is also known as the completely bounded trace norm.}%
					In particular, it is a contraction for quantum channels obtained via the yield construction from a contraction for states---the trace norm distance.
					
					\begin{example}[diamond norm as yield]
						\label{ex:diamond norm}
						Let $\phi, \psi$ be two quantum channels with identical input and output systems. 
						The completely bounded trace norm distance between $\phi$ and $\psi$ is given by
						\begin{equation}
							\label{eq:diamond norm}
							\bigl\| \phi - \psi \bigr\|_{\diamond} = \sup \Set*[\Big]{ \bigl\| ( \phi \otimes \id ) \circ \rho  - (\psi \otimes \id ) \circ \rho \bigr\|_1  \given  \rho \text{ is a state} }
						\end{equation}
						where $\rho$ is a quantum state on a bipartite system, one part of which is the input system of the channels.
						This matches the form of \cref{eq:state_to_channel2} if we identify $f$ as the trace norm distance.
						The fact that $f$ is indeed a contraction is necessary in order to simplify the expression from \eqref{eq:state_to_channel1} to the one of \eqref{eq:state_to_channel2}.
					\end{example}
					
					Further restriction on the free encodings can be imposed if we require that these have to come from some set of free transformations $T_{\rm free}$ specified independently.
					That is, we may define the free transformations to be\footnotemark{}
					\footnotetext{Note that $T^H_{\rm free}$ as given here is very different from $(T_{\rm free})^H$.
					Since we don't use the latter, this notation clash should not be too confusing.}
					\begin{equation}
						T^H_{\rm free} \coloneqq \Set*[\big]{ (t, t, \dots , t)  \given  t \in T_{\rm free} }
					\end{equation}
					instead of those given by \cref{eq:cons_encodings}.
					In this way, we get a \textbf{resource theory of (constrained) $\bm{H}$-distinguishability} associated to a listable resource theory.
					Its resource ordering is usually denoted by $(\mathcal{P}(X^H), \succeq)$.
					
					A resource theory of constrained distinguishability describes the capabilities of an agent who, in addition to having no a priori knowledge of the hypothesis, also has a restriction on the allowed processings in the form of $T_{\rm free}$ as a subset $T$.
					Since we can recover the unconstrained theory by setting $T_{\rm free} = T$, in what follows we work with the constrained case (without explicit mention) unless stated otherwise.
					As such, when we mention ($k$-)contractions in the subsequent sections of \cref{sec:monotones}, we only require that they are monotones under $T^H_{\rm free}$.
					Of course, any contraction with respect to all transformations satisfies that requirement, since $T^H_{\rm free}$ is a subset of $T^H_{\rm unif}$ for any set of free transformations $T_{\rm free}$.

				\subsection{Monotones from Equivariant Maps}
				\label{sec:monotones from twirling generalized}
				
					We now describe a generalization of the monotone construction in \cref{ex:Monotones from contractions 2} which applies to listable resource theories.
					Although similar ideas can be used for any hypothesis, we focus on the binary case with $H$ given by $\bm{2} = \{0,1\}$.
					We take the resource ordering $(\mathcal{P}(X^{\bm{2}}), \succeq)$ in the resource theory of \emph{constrained} distinguishability to be the mediating preordered set. 
					A root monotone is therefore $\sup \circ f_*$ for a particular choice of a contraction $f \colon X^{\bm{2}} \to \reals$. 
					
					Consider a function $\phi \colon X \to \mathcal{P}(X)$ or, in other words, an endomorphism of the suplattice of resources.
					We say that $\phi$ is \textbf{oplax equivariant} (with respect to free transformations) if it satisfies
					\begin{equation}\label{eq:oplax_atom_natural}
						t \apply_i \phi(x)  \supseteq  \phi \bigl( t \apply_i x \bigr)
					\end{equation}
					for all $i \in \mathcal{I}$, all $t \in T_{\rm free}$ and all $x \in X$.
					In some sense, this condition bears resemblance to \eqref{eq:oplax_natural}, but here it is applied atom-wise and for each $i$ separately.
					\begin{example}\label{ex:group_action_oplax_covariant}
						Given a unimodular\footnotemark{} resource theory of asymmetry (\cref{sec:rt_asymmetry}), the free transformations are defined so that every representation $\varphi_g$ of an element of the symmetry group satisfies \eqref{eq:oplax_atom_natural} with equality.
						\footnotetext{The same holds for a listable resource theory of asymmetry that is not unimodular, but one has to adapt \cref{def:G-covariant} of covariant transformations to require covariance for each $i$ independently.}%
						If the resource theory comes with a convex structure (\cref{sec:rt_convex}), then any twirling map $\gamma_\mu$ given by an average over $\{\varphi_g\}_{g \in G}$ does as well.
					\end{example}
					In analogy to \cref{ex:Monotones from contractions 2}, we take the mediating isotone to be the map that sends $x$ to the binary encoding of $x$ versus $\phi(x)$.
					\begin{lemma}\label{lem:oplax_covariant_isotone}
						For a listable resource theory and an oplax equivariant $\phi$, the map given by
						\begin{equation}\label{eq:oplax_covariant_isotone}
							x \mapsto  \bigl( x , \phi(x)  \bigr) \coloneqq  \Set{ (x, a) \given a \in \phi(x) }
						\end{equation}
						is an isotone of type $(X, \freeconv) \to (\mathcal{P}(X^{\bm{2}}), \succeq)$.
					\end{lemma}
					Note that $( x , \phi(x))$ is not an element of $\mathcal{P}(X^{\bm{2}})$, we merely use it as a shorthand for the expression on the right.
					\begin{proof}
						If $x \succeq y$ holds, then there is a free transformation $t$ such that $y$ is an element of $t \apply x$.
						Since the theory is listable, there ought to be an $i \in \mathcal{I}$ for which $y$ equals $t \apply_i x$.
						Thus, we get
						\begin{equation}\label{eq:oplax_covariant_isotone_2}
							\phi(y) = \phi( t \apply_i x) \subseteq t \apply_i \phi(x)
						\end{equation}
						as $\phi$ is presumed to be oplax equivariant with respect to free transformations.
						In particular, for every $b \in \phi(y)$, there exists an $a_b \in \phi(x)$ for which $b = t \apply_i a_b$ holds.
						The assignment of $(x,a_b)$ to $(y,b)$ is therefore an enhancement of type $(y , \phi(y)) \to ( x , \phi(x))$, since we have
						\begin{equation}
							(t,t) \apply (x, a_b) \supseteq (t,t) \apply_i (x, a_b) = (t \apply_i x, t \apply_i a_b ) = (y,b)
						\end{equation}
						for each $b$.
						As a result, $(x , \phi(x)) \freeconv (y , \phi(y))$ holds by \cref{lem:access=enh}, and the statement of \cref{lem:oplax_covariant_isotone} follows.
					\end{proof}
		
					Using the \ref{broad scheme}, we can construct a resource monotone for any contraction $f$ and any oplax equivariant $\phi$, defined as
					\begin{equation}\label{eq:oplax_covariant_monotone}
						m(x) \coloneqq \sup f_* \bigl( r, \phi(r) \bigr) = \sup \Set{ f(x, a) \given a \in \phi(x) }.
					\end{equation}
					Whenever $\phi$ is a single-valued function $X \to X$, there is no optimization to perform.
					Choosing $\phi$ to be the twirling map, we recover the monotone from \cref{ex:Monotones from contractions 2}.
					\begin{remark}\label{rem:joint_oplax_covariant_isotone}
						Basically the same proof as that of \cref{lem:oplax_covariant_isotone} can be also used to show that for any two oplax equivariant maps $\phi_1$ and $\phi_2$, the function
						\begin{equation}
							x \mapsto  \bigl( \phi_1(x), \phi_2(x) \bigr) \coloneqq  \Set*[\big]{ (a,c) \given a \in \phi_1(x), \; c \in \phi_2(x) }
						\end{equation}
						is an isotone of the same type.
						In the context of resource theories of asymmetry, we then recovers the monotone construction with two twirling maps mentioned at the very end of \cref{ex:Monotones from contractions 2}.
					\end{remark}
					\begin{example}[bridge lemma]\label{ex:bridge_lemma}
						One can extend the argument from \cref{lem:oplax_covariant_isotone} to hypothesis variables of higher cardinality, such as the symmetry group $G$ in a resource theory of asymmetry.\footnotemark{}
						\footnotetext{In order to avoid unnecessary detail, let us consider finite groups only here.}%
						Using the oplax equivariant maps provided by representations $\varphi_g$ of the group's elements, we thus obtain an isotone $\mathrm{orb} \colon (X,\freeconv) \to (X^{G}, \freeconv)$ given by
						\begin{equation}
							\mathrm{orb}(x) \coloneqq \bigl( \varphi_g(x) \bigr)_{g \in G}.
						\end{equation}
						To each resource $x$, $\mathrm{orb}(x)$ assigns the corresponging $G$-orbit.
						Its isotonicity can be also viewed as a consequence of \cref{prop:isotone_from_morphism} because $\mathrm{orb}$ constitutes the resource part of a resource theory morphism $(\mathrm{unif},\mathrm{orb})$ where $\mathrm{unif} \colon T \to T^G$ sends each individual transformation to the uniform tuple
						\begin{equation}
							\mathrm{unif}(t)  \coloneqq  (t)_{g \in G}
						\end{equation}
						and its action can be extended to sets of transformations so as to impose compatibility with unions.
						
						Since the free transformations in a resource theory of constrained distinguishability form a subset of the unconstrained ones, the identity map is an isotone of type ${(X^{G}, \freeconv) \to (X^{G}, \freeconv_{\rm unif})}$ where $\freeconv_{\rm unif}$ is the resource ordering of unconstrained distinguishability.
						Thus, we get an isotone $\mathrm{orb} \colon (X,\freeconv) \to (X^{G}, \freeconv_{\rm unif})$ that maps $G$-asymmetry properties of resources in $X$ to the distinguishability properties of their orbits under the $G$-action.
						It can thus act as the mediating isotone in the \ref{broad scheme} that allows us to translate distinguishability monotones (i.e.\ contractions) to asymmetry monotones.
						If the resource theory of $G$-asymmetry is also a convex resource theory (see \cref{sec:rt_convex}), then $\mathrm{orb}$ is an order-embedding.
						That is, when restricted to encodings that are $G$-orbits, the distinguishability preorder $(X^{G}, \freeconv_{\rm unif})$ is equivalent to the asymmetry preorder $(X,\freeconv)$.
						The proof can be found in \cite[lemma 7]{Marvian2012} for the special case of quantum resource theories, but it is no different in the general case.
						In this way, $\mathrm{orb}$ establishes a ``bridge'' between asymmetry and distinguishability.
					\end{example}
					
					Besides twirling maps in resource theories of asymmetry, there are functions in other resource theories which are oplax equivariant with respect to free operations.
					\begin{example}\label{ex:athermality_oplax_covariant}
						In the resource theory of athermality, one can take the oplax equivariant map $\phi$ to be the thermalization map $\varphi_p$.
						Specifically, it is the discarding map followed by a preparation of the thermal state $\rho_{\rm therm}$ for the system just discarded.
						Monotones constructed in this way describe the thermo-majorization order \cite{Horodecki2013}, a special case of relative majorization first introduced in \cite{Veinott1971}.
						They arise by pulling back the matrix majorization order along the mediating isotone $x \mapsto (x,\varphi_p(x)) = (x, \rho_{\rm therm})$.
						In the special case of trivial Hamiltonians, we obtain the resource theory of nonuniformity, in which the general thermal states	are replaced by uniform probability distributions.
					\end{example}
					
					\begin{example}
						A similar example arises in resource theories wherein the free operations are \emph{local}, i.e.\ those which are of tensor product form for some prespecified partition.
						The free transformations are $M$-covariant with respect to the map that sends resources to the tensor product of their projections. 
						These are generally \emph{not} resource theories of $M$-asymmetry, since we exclude certain covariant transformations (such as swapping) from the free set.
						Nevertheless, the same monotone construction can be applied, because monotones for $M$-covariant transformations are automatically order-preserving also for any subset thereof (such as local operations).
						If $\rho$ denotes a bipartite state on the joint system $\alpha \otimes \beta$ and $\rho_{\alpha} \otimes \rho_{\beta}$ is the tensor product of its projections (or ``marginals''), the resulting monotone is
						\begin{equation}
							m(\rho) =  f ( \rho , \rho_{\alpha} \otimes \rho_{\beta} ).
						\end{equation}
						Choosing $f$ to be the relative entropy, for instance, gives rise to the mutual information $I(\alpha \mathbin{;} \beta)$.
					\end{example}
		
				\subsection{Minimal Distinguishability as a Monotone}
				\label{sec:Monotones From Contractions in General}
				
					Let us turn to the question of generalizing the monotone construction presented in \cref{ex:Monotones from contractions}, which quantifies minimal distinguishability between a given resource and the set of all free resources.
					Because of the infimum in \cref{eq:monotone from contraction}, we expect the mediating preordered set to be $(\mathcal{P}(X^{\bm{2}}), \freeconv_{\rm deg})$.
					The root monotone is then $\inf \circ f_*$ for a contraction $f$.
					\begin{remark}
						Note that the contraction need not be defined everywhere. 
						As long as its domain of definition is downward closed, we get a valid root monotone by \cref{lem:monotone_extensions}. 
						This approach was used in \cite{Monotones} to understand monotones of the form as in \eqref{eq:min_dis_downset} via the \ref{broad scheme}.
						Here, we take a more general approach.
					\end{remark}
					
					The right choice of mediating isotone so as to obtain the monotone in \cref{ex:Monotones from contractions} via the \ref{broad scheme} is the function $\mathcal{E}_{X_{\rm free}}$ defined by
					\begin{equation}\label{ERfree}
						\mathcal{E}_{X_{\rm free}} (x) \coloneqq \Set*[\big]{ (x, y)  \given  y \in X_{\rm free} },
					\end{equation}
					where $X_{\rm free}$ denotes the set of free states as before. 
					There are multiple avenues to understand why it is order-preserving.
					We choose to follow an approach that mimics our discussion in \cref{sec:monotones from twirling generalized}.
					
					We say that a function $\psi \colon X \to \mathcal{P}(X)$ is \textbf{lax equivariant} (with respect to free transformations) if it satisfies
					\begin{equation}\label{eq:lax_atom_natural}
						t \apply_i \psi(x)  \subseteq  \psi \bigl( t \apply_i x \bigr)
					\end{equation}
					for all $i \in \mathcal{I}$, all $t \in T_{\rm free}$ and all $x \in X$.
					\begin{example}
						In a unimodular resource theory, any constant function that sends each $x$ to a fixed downset $D \in \mathcal{DC}(X)$ is lax equivariant.
						Indeed, we have
						\begin{equation}
							t \apply_i \psi(x) \subseteq T_{\rm free} \apply D = D = \psi \bigl( t \apply_i x \bigr)
						\end{equation}
						in such case.
						Specifically, choosing $D$ to be the set of free states, the map $x \mapsto X_{\rm free}$ is lax equivariant with respect to free transformations.
					\end{example}
					
					\begin{lemma}\label{lem:lax_covariant_isotone}
						For a listable resource theory and a lax equivariant $\psi$, the map given by
						\begin{equation}\label{eq:lax_covariant_isotone}
							x \mapsto  \bigl( x , \psi(x)  \bigr) =  \Set*[\big]{ (x, a) \given a \in \psi(x) }
						\end{equation}
						is an isotone of type $(X, \freeconv) \to (\mathcal{P}(X^{\bm{2}}), \freeconv_{\rm deg})$.
					\end{lemma}
					\begin{proof}
						The proof is analogous to that of \cref{lem:oplax_covariant_isotone}.
						The reverse direction of the relation \eqref{eq:oplax_covariant_isotone_2} allows us to construct a degradation as opposed to an enhancement.
					\end{proof}
					
					The upshot of this discussion is that for every contraction $f \colon X^{\bm{2}} \to \reals$ and every downset $D$ of $(X, \freeconv)$, the function 
					\begin{equation}\label{eq:min_dis_downset}
						\inf \circ f_* \circ \mathcal{E}_{D} (x) = \inf \Set*[\big]{ f(x,a)  \given  a \in D}
					\end{equation}
					is a resource monotone in the original resource theory.
					For every downset distinct from the set of free resources, therefore, one obtains a corresponding variation on a monotone expressing the minimal distinguishability from the free set---one that quantifies the distinguishability from the chosen downset $D$ according to $f$.
						
					\begin{example}[monotones quantifying distinguishability from a downset]
					\label{examplesMonotonesDCsets} \hspace{0pt}
						\begin{enumerate}
							\item Consider the resource theory of bipartite quantum entanglement \cite{Horodecki2009} with the free operations given by LOCC processes \cite{Chitambar2014}.
							For any contraction $f$ of quantum states, the distinguishability (according to $f$) between a state and the set of separable states is a popular entanglement monotone. 
							Given that the set of states with entanglement rank at most $r$ is a downset in this resource theory,\footnotemark{} for every $r$, the distinguishability (according to $f$) between a state and the set of states with entanglement rank at most $r$ is also an entanglement monotone.
							\footnotetext{This result can be found in \cite[theorem 6.23]{watrous2018theory} and it has been first shown in \cite{Terhal2000} where entanglement rank has been introduced under the name ``Schmidt number''.}%

							\item Resource theories of quantum coherence \cite{streltsov2017colloquium} describe resources of ``quantum superposition'' in a given basis so that the free states are those given by diagonal density matrices.
								Besides these, there are downsets of bounded ``coherence number''. 
								In particular, a pure state is $r$-incoherent if it is block-diagonal with blocks of size no larger than $r \times r$.
								For general states, we can take the convex hull of these to obtain all the $r$-incoherent states \cite{johnston2018evaluating}.
								The coherence number $r$ introduced in \cite{sperling2015convex} plays the analogous role of Schmidt number from entanglement theory \cite{killoran2016converting,regula2018converting}.
								Indeed, $r$-incoherent states form a downset for any $r$ and can be thus used to generate monotones via \cref{eq:min_dis_downset}.
							
							\item In the resource theory of multipartite quantum entanglement, a partition of the $p$ parties is said to have radius at most $q$ if each of its elements consists of at most $q$ parties.
							Then, for any given $q$, the set of states that are convex combinations of pure states separable for a partition of radius at most $q$ forms a downward closed set.
							In particular, the specific case of $p = 3$ and $q=2$ defines the set of states that are deemed to be \emph{not} intrinsically 3-way entangled, and so this set is downward closed.
							Consequently, the distinguishability (according to $f$) between a state and this set is a multipartite entanglement monotone that quantifies intrinsic 3-way entanglement.
							
							\item Resource theories of $G$-asymmetry provide another illustrative example. 
							In particular, the set of states that are symmetric under a subgroup $H$ of $G$ form a downset. 
							Therefore, the distinguishability (according to $f$) between a state and the states symmetric under $H$ is also an asymmetry monotone. 
							Roughly speaking, of all the ways that a state may break $G$-symmetry, the extent to which it does so by breaking $H$-symmetry is quantified by this monotone.\footnotemark{}
							\footnotetext{If the contraction $f$ is the relative entropy, then this monotone becomes $S(\rho \mathrel{\|} \gamma_H(\rho))$ and is equivalent to the Holevo asymmetry monotone $S(\gamma_H(\rho)) - S(\rho)$ associated to the uniform twirling $\gamma_H$ over $H$.
								This equivalence follows from \cite[proposition 2]{Gour2009}. 
								Note that the simplest case of such a monotone, $S(\gamma_G(\rho)) - S(\rho)$, was introduced in \cite{Vaccaro2008}.}%
						\end{enumerate}
					\end{example}
				
					We can repeat the construction for $f$ which is a $k$-contraction instead of a $2$-contraction.
					Just like \cref{rem:joint_oplax_covariant_isotone} points out in the oplax equivariant case, given a family of lax equivariant functions $\psi_j$, the map
					\begin{equation}
						x  \mapsto \bigl(  \psi_1(x), \psi_2(x), \dots , \psi_k(x) \bigr)
					\end{equation}
					is an isotone of type $(X, \freeconv) \to (\mathcal{P}(X^{\bm{k}}), \freeconv_{\rm deg})$ by the same argument as the one from \cref{lem:lax_covariant_isotone}.
					Thus, for every $k$-contraction $f \colon X^{\bm{k}} \to \reals$ and a collection of downsets $D_j$, the function 
					\begin{equation}\label{eq:min_k-dist_monotone}
						x \mapsto \inf \Set*[\big]{ f(x,a_2, a_3, \dots a_k)  \given  a_j \in D_j}
					\end{equation}
					is a resource monotone in the original resource theory.
					Of course, one can also permute the arguments of $f$ arbitrarily in the above expression.
			
				\subsection{Convex Resource Theories}\label{sec:rt_convex}
		
					As an example of how the generalized construction of monotones from $k$-contractions appears in a more concrete setting, we examine arguably two of the most ubiquitous monotones---resource weight \cite{lewenstein1998separability,Barrett2006,Abramsky2017} and resource robustness \cite{vidal1999robustness,harrow2003robustness}---within the context of convex resource theories.
					By connecting them to a monotone in the resource theory of distinguishability, we complement the results of \cite{takagi2019general,Skrzypczyk2019,Ducuara2019}.
					In the first two articles, robustness measures are connected to discrimination tasks, while the last one describes a similar connection between weight measures and the state exclusion tasks.
					Before we discuss weight and robustness, however, we take a brief detour to introduce convexity within our framework.
					
					Many resource theories, particularly those studied at present, come with a convex-linear structure and are generally even embedded in inner product spaces.
					This is a consequence of the probabilistic behaviour of quantum systems and stochastic classical systems, which provide the most common context for resource theories.
					As a result, many of the methods are based on ideas of convex optimization \cite{Chitambar2018} and convex geometry \cite{Regula2017,zhou2020general} more broadly.
					To express these in our abstract set-up, we need to supply the framework with a convex structure.
					
					One could do so in an abstract way as in \cite[appendix F]{DelRio2015} that incorporates notions of uncertainty beyond the standard probabilistic calculus.
					However, our aim here is not to develop a general framework of \emph{convex resource theories}.
					The goal is to use convex structure to illustrate how some of the concrete results that require convexity fit into the quantale module picture.
					Therefore, we restrict our attention to the ``probabilistic convexity''.
					Additionally, we consider listable resource theories only. 
					They are concrete resource theories in which $\star$ operation can be decomposed into a union of partial semigroup operations as in \cref{ex:free_quantale_multi} and similarly for $\apply$.
					In particular, if the semigroups are labelled by an index $i$, then each $s \star_i t$ and each $t \apply_i x$ is an individual resource or the empty set.
					
					In this context, we demand that in a convex resource theory, both\footnotemark{} $X$ and $T$ are convex sets in (topological) affine spaces which we think of as vector spaces without origin.
					\footnotetext{There is an intermediate notion in which only the resources carry a convex structure (see \cite{DelRio2015}).
					Quantum resource theories with transformations by unitary operations would generally fall into this class that we do not consider in this thesis.}%
					Note that we do not the need convex combination of arbitrary resources (or transformations) to be a valid resource (or transformation).
					Just like it makes sense to restrict hypothesis encodings to tuples of elements of the same type \cite[appendix C]{Monotones}, we would only expect to be able to take convex combinations of elements whose types match.
					In other words, we do not presume that a convex combination of states of distinct systems exists, for instance.
					Without using typed resource theories (which we avoid in this thesis), we can implement such an idea by requiring that convex combinations are merely partial operations.
					Moreover, the operations in a convex resource theory should respect convex combinations.
					That is, for all $i$, all $\lambda \in [0,1]$, and all individual transformations and resources, we should have
					\begin{equation}\label{eq:comp_conv_comp}
						\begin{alignedat}{9}
							\bigl( \lambda \, s + ( 1 - \lambda ) \, t \bigr) &\star_i u &\;&= \;\;& \lambda \, (s \star_i u) &+ ( 1 - \lambda ) \, (t \star_i u) \\
							s &\star_i \bigl( \lambda \, t + ( 1 - \lambda ) \, u \bigr) &\;&= \;\;& \lambda \, (s \star_i t) &+ ( 1 - \lambda ) \, ( s \star_i u ) \\
							\bigl( \lambda \, s + ( 1 - \lambda ) \, t \bigr) &\apply_i x &\;&= \;\;& \lambda \, (s \apply_i x) &+ ( 1 - \lambda ) \, (t \apply_i x) \\
							s &\apply_i \bigl( \lambda \, x + ( 1 - \lambda ) \, y \bigr) &\;&= \;\;& \lambda \, (s \apply_i x) &+ ( 1 - \lambda ) \, ( s \apply_i y ) 
						\end{alignedat}
					\end{equation}
					as so-called Kleene equalities in each case.
					That is, the relations impose that the left-hand side is defined if and only if the corresponding right-hand side is defined.
					
					\begin{remark}
						In certain convex resource theories, we may want to impose that an agent's access to resources is closed under convex combinations.
						That is, instead of using the powerset suplattices $\mathcal{P}(X)$ and $\mathcal{P}(T)$, we may use suplattices which only contain the convex subsets of $X$ and $T$ respectively.
						Suprema therein are given by the convex hull of the union of respective sets.
					\end{remark}
			
				\subsection{Resource Weight and Robustness}
				\label{sec:weight and robustness}
								
					\begin{figure}[h]
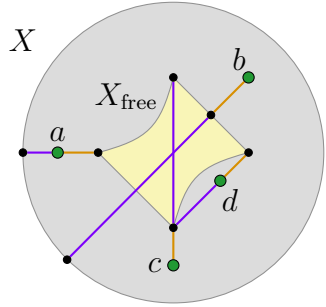

						\ctikzfig{weight_rob}
						\caption[Resource Weight and Robustness]{A pictorial depiction of the optimal convex decompositions for each of the four monotones mentioned in this section: (a) resource weight $m_{\rm w}$, (b) resource robustness $m_{\rm rob}$, (c) free robustness $m_{\rm f.\,rob}$, and (d) resource non-convexity $m_{\rm nc}$.
						Grey disc represents the set $X$ of all individual resources, while the yellow ``hourglass'' witin represents the free resources among them.
						In order to illustrate each of the four optimal decompositions, we select a distinct resource (element of $X$), depicted by a green node.
						These demopositions are given by the three points along one of the line segments with an orange and purple portion.
						The value of each of the monotones for these; $m_{\rm w}(a), m_{\rm rob}(b), m_{\rm f.\,rob}(c)$, and $m_{\rm nc}(d)$; can be read off as the length of the respective orange segment divided by the total lenth of the orange and purple segments combined.}
						\label{fig:weight_rob}
					\end{figure}
				
					Returning to weight and robustness measures, let us consider a listable resource theory with a convex structure. 
					We can then define a particularly simple $3$-contraction that we call \textbf{convex alignment}:\footnotemark{}
					\footnotetext{Recall from \cref{eq:empty_bound} that our convention for the infimum of the empty set is the top element.
					The value of $\mathsf{cva}(x,y,z)$ whenever $x$ is not a convex combination of $y$ and $z$ thus depends on the codomain of $\mathsf{cva}$.
					If we think of it as a function $X \to \reals$, then $\mathsf{cva}(x,y,z)$ would be defined as $\infty$ in such cases.}%
					\begin{equation}\label{eq:alignment}
						\mathsf{cva} (x,y,z) \coloneqq \inf \Set*[\big]{\lambda \in [0,1] \given x = \lambda \, y + (1-\lambda) \, z }
					\end{equation}
					Indeed, this function is a monotone in any resource theory of $\bm{3}$-distinguishability.
					To show this, it suffices to consider the unconstrained case.
					That is, we want to establish that, for any transformation $t \in T$, any $i \in \mathcal{I}$ and any triple of resources $x,y,z$, we have
					\begin{equation}\label{eq:alignment_monotone}
						\mathsf{cva}(x,y,z) \geq  \mathsf{cva} (t \apply_i x, \, t \apply_i y, \, t \apply_i z).
					\end{equation}
					If $\mathsf{cva}(x,y,z)$ is infinite, then its value cannot increase and \cref{eq:alignment_monotone} holds trivially.
					Otherwise, for any $\lambda \in [0,1]$ such that $x$ equals $\lambda \, y + (1-\lambda) \, z$, we have
					\begin{equation}
						t \apply_i x = t \apply_i \bigl( \lambda \, y + (1-\lambda) \, z \bigr) =  \lambda \, t \apply_i y + (1-\lambda) \, t \apply_i z
					\end{equation}
					thanks to the convex-linearity of $\apply$ as given by \eqref{eq:comp_conv_comp}.
					Therefore, inequality \eqref{eq:alignment_monotone} also holds whenever $\mathsf{cva}(x,y,z)$ is finite.
					
					Now we can use the construction of monotones from $k$-contractions described in \cref{sec:Monotones From Contractions in General} to obtain resource monotones. 
					In particular, using $\mathsf{cva}$ as the $3$-contraction in \eqref{eq:min_k-dist_monotone} leads to a number of interesting monotones for different choices of downsets. 
					There are many downsets one could use for each $D_j$, but here we will restrict our attention to the two most obvious choices---$X_{\rm free}$ and $X$.
					Even with this restriction, one can obtain 12 potentially distinct monotones. 
					Out of these, 8 give rise to a constant monotone and are therefore uninteresting.
					The other 4 are as follows.
					\begin{enumerate}[(a)]
						\item The \textbf{resource weight} (also known as the resource fraction) $m_{\rm w} \colon X \to \reals$ is defined as 
							\begin{align}
								m_{\rm w}(x) &\coloneqq \inf \Set*[\big]{ \mathsf{cva}(x,y,z)  \given y \in X , \, z \in X_{\rm free} }. \\
							\intertext{Explicitly, its value is}
								\label{eq:weight_2} m_{\rm w} (x) &= \inf \Set*[\big]{ \lambda  \given  x \in \lambda \, X + (1-\lambda) \, X_{\rm free} }.
							\end{align}
							It corresponds to the smallest weight of a resource $y$ that can be used to form $x$ by convex mixture with some free resource $z$.
						
						\item The \textbf{resource robustness} (also known as global robustness) $m_{\rm rob} \colon X \to \reals$ is defined as 
							\begin{align}
								m_{\rm rob}(z) &\coloneqq \inf \Set*[\big]{ \mathsf{cva}(x,y,z)  \given x \in X_{\rm free} , \, y \in X }. \\
							\intertext{Explicitly, its value is}
									m_{\rm rob} (z) &= \inf \Set*[\big]{ \lambda  \given  \lambda \, y + (1-\lambda) \, z \in X_{\rm free} , \, y \in X }.
							\end{align}
							It is the smallest weight of a resource $y$ that one needs to convexly mix with $z$ in order to obtain a free resource.
							
						\item The \textbf{free robustness} (also known as standard robustness) $m_{\rm f.\,rob} \colon X \to \reals$ is defined as 
							\begin{align}
								m_{\rm f.\,rob}(z) &\coloneqq \inf \Set*[\big]{ \mathsf{cva}(x,y,z)  \given x \in X_{\rm free} , \, y \in X_{\rm free} }. \\
							\intertext{Explicitly, its value is}
									m_{\rm f.\,rob} (z) &= \inf \Set*[\big]{ \lambda  \given  \lambda \, y + (1-\lambda) \, z \in X_{\rm free} , \, y \in X_{\rm free} }.
							\end{align}
							It is the smallest weight of a free resource $y$ that one needs to convexly mix with $z$ in order to obtain another free resource.
						
						\item The \textbf{resource non-convexity} $m_{\rm nc} \colon X \to \reals$ is defined as 
							\begin{align}
								m_{\rm nc}(x) &\coloneqq \inf \Set*[\big]{ \mathsf{cva}(x,y,z)  \given y \in X_{\rm free} , \, z \in X_{\rm free} }. \\
							\intertext{Explicitly, its value is}
									m_{\rm nc} (x) &= \inf \Set*[\big]{ \lambda  \given  x \in \lambda \, X_{\rm free} + (1-\lambda) \, X_{\rm free} }.
							\end{align}
							It is trivial if the set of free resources (for each resource type) is a convex set. 
							Otherwise, it tells us about the ordering of resources that are within the convex hull of the free resources, but are not free themselves.
							While many resource theories have convex free subtheories, some do not.
							Examples include resource theory of non-Gaussianity \cite{zhuang2018resource,lami2018gaussian}, as well as resource theories in which free operations are those compatible with a causal structure without a common source of randomness.
							Resource non-convexity thus quantifies the relative distance of a resource from the set of free resources in terms of its convex decompositions into free resources. 
							Its value is set to $0$ if the resource in question is free itself, and $\infty$ if it is outside of the convex hull of the free resources.
					\end{enumerate}
					
					As a consequence of our discussion in \cref{sec:Monotones From Contractions in General}, all four functions above are resource monotones. 
					However, being able to prove the monotonicity of these four functions is not where the value of the general construction of monotones from contractions lies.
					What it provides is an understanding of the assumptions required for these functions to be monotones.
					Furthermore, it gives us a unified picture, within which we can adjust various elements of the monotone constructions according to the question we are interested in.
					In this case, there are many more monotones one can obtain from $\mathsf{cva}$ in a similar way, since $X$ or $X_{\rm free}$ can be replaced by any other downset in the optimization.
		
			\subsection{Translating Monotones in General}
			\label{sec:General Translating}
			
				In preceding sections, we saw that many monotones used in practice arise as translations of measures of distinguishability. 
				We now generalize these techniques to be able to translate monotones between arbitrary concrete resource theories with underlying quantale modules $(\mathcal{P}(T), \mathcal{P}(X), \apply)$ and $(\mathcal{P}(S), \mathcal{P}(W), \apply)$, provided that they can be related suitably.
				The choices of the mediating preorders are again going to be the enhancement and degradation preorders constructed from the resource ordering of individual resources.
				For any monotone $f$ on $W$, we thus have corresponding root monotones $\sup \circ f_*$ and $\inf \circ f_*$ as defined in \cref{lem:monotone_extensions}.
				
				In order to construct maps that can be used as the mediating isotones of type $(X, \freeconv) \to (\mathcal{P}(W), \succeq_{\rm enh})$, we can use the following sufficient and necessary condition. 
			\begin{lemma}[mediating isotones for $\succeq_{\rm enh}$]
				\label{lem:mediating maps for enhancement order}
				Let $(\mathcal{P}(T), \mathcal{P}(X), \apply)$, $(\mathcal{P}(S), \mathcal{P}(W), \apply)$ be two concrete resource theories with free transformations given by $T_{\rm free}$ and $S_{\rm free}$ respectively.
				Let $\mathcal{F} \colon \mathcal{P}(X) \to \mathcal{P}(W)$ be a suplattice homomorphism.
				We have
				\begin{equation}
					\label{eq:mediating maps for enhancement order}
					S_{\rm free} \apply \mathcal{F}(x) \supseteq \mathcal{F}(T_{\rm free} \apply x) 
				\end{equation}
				for all $x \in X$ if and only if $\mathcal{F}$ is an isotone of type $(\mathcal{P}(X), \succeq) \to (\mathcal{P}(W), \succeq)$.
			\end{lemma}
			Note that we can express relation \eqref{eq:mediating maps for enhancement order} succinctly as $\downwrt{S_{\rm free}} \circ \mathcal{F} \supseteq \mathcal{F} \circ \downwrt{T_{\rm free}}$, where the relation $\supseteq$ among maps is defined point-wise:
			\begin{equation}
				\mathcal{G} \supseteq \mathcal{G}'  \quad \coloniff \quad   \forall \, x \in X \; : \; \mathcal{G} (x) \supseteq \mathcal{G}' (x)
			\end{equation}
			\begin{proof}
				First of all, note that relation \eqref{eq:mediating maps for enhancement order}, in conjunction with the fact that both $\mathcal{F}$ and $\apply$ preserve unions, implies that the same relation holds if we replace an individual resource $x$ by any subset $Y$ of $X$.
				
				We need to show that for $\mathcal{F}$ as above, the implication $Y \succeq Z \implies \mathcal{F}(Y) \succeq \mathcal{F}(Z)$ holds for any $Y,Z \in \mathcal{P}(X)$.
				This can be broken down as follows:
				\begin{equation}
					\begin{alignedat}{999}
						Y \succeq Z 
							\iff 		&\;& 	T_{\rm free} \apply Y &\supseteq Z \\
							\implies 	&\;&	\mathcal{F}(T_{\rm free} \apply Y) &\supseteq \mathcal{F}(Z)   \\
							\implies 	&\;&	S_{\rm free} \apply \mathcal{F}(Y) &\supseteq \mathcal{F}(Z) \\
							\iff 		&\;& 	\mathcal{F}(Y) &\succeq \mathcal{F}(Z),
					\end{alignedat}
				\end{equation}
				where the second implication follows from property \eqref{eq:mediating maps for enhancement order}.
				
				Conversely, if $\mathcal{F}$ is an isotone, then we have
				\begin{equation}
					\begin{alignedat}{999}
						x \freeconv T_{\rm free} \apply x 
							&\implies & 	 \mathcal{F}(x) &\freeconv \mathcal{F}(T_{\rm free} \apply x) \\
							&\implies & 	 S_{\rm free} \apply \mathcal{F}(x) &\supseteq \mathcal{F}(T_{\rm free} \apply x),
					\end{alignedat}
				\end{equation}
				and thus $\mathcal{F}$ must satisfy \eqref{eq:mediating maps for enhancement order} since $x \freeconv \down (x)$ is true for any $x \in X$.
			\end{proof}
			There is of course an analogous lemma for resource theories that are not uniquely atomistic, but we cannot restrict the applicability of condition \eqref{eq:mediating maps for enhancement order} to atoms then.
			
			As we have seen in \cref{prop:isotone_from_morphism}, for any oplax morphism of resource theories $(\ell,f)$, the resource mapping $f$ is an isotone.
			Indeed, relation \eqref{eq:mediating maps for enhancement order} follows for any pair $(\ell,f)$ satisfying the assumptions of \cref{prop:isotone_from_morphism} by
			\begin{equation}
				f(T_{\rm free} \apply x) \subseteq \ell(T_{\rm free}) \apply f(x) \subseteq S_{\rm free} \apply f(x).
			\end{equation}
			
			\Cref{lem:oplax_covariant_isotone}, which establishes that mapping a resource $x$ to $(x, \phi(x))$ for some oplax equivariant endomorphism of suplattices $\phi$ is order-preserving, can also be recovered from \cref{lem:mediating maps for enhancement order} if we choose $\mathcal{F}$ to be given by $\mathcal{F}(x) = (x, \phi(x))$.
			
			As one would expect, we have a corresponding dual version of the above condition. 
			However, unlike in the case of listable resource theories where lax equivariant maps satisfy an opposite set inclusion to oplax equivariant ones, we cannot express the dual condition as the opposite of \eqref{eq:mediating maps for enhancement order} here.
			\begin{lemma}[mediating isotones for $\succeq_{\rm deg}$]
				\label{lem:mediating maps for degradation order}
				Given the same set-up as in \cref{lem:mediating maps for enhancement order} (and writing $\mathcal{G}$ instead of $\mathcal{F}$); we have
				\begin{equation}
					\label{eq:mediating maps for degradation order}
					 \mathcal{G} \bigl( \upwrt{T_{\rm free}} (x)  \bigr)  \subseteq \upwrt{S_{\rm free}} \bigl( \mathcal{G}(x) \bigr)
				\end{equation}
				for all $x \in X$ if and only if $\mathcal{G}$ is an isotone of type $(\mathcal{P}(X), \succeq_{\rm deg}) \to (\mathcal{P}(W), \succeq_{\rm deg})$.
			\end{lemma}
			\begin{proof}
				The proof is very much in line with that of \cref{lem:mediating maps for enhancement order}.
				Given subsets $Y,Z$ of $X$, we have
				\begin{equation}
					\begin{alignedat}{999}
						Y \freeconv_{\rm deg} Z  
							\implies 	&\;& 		Y &\subseteq  \upwrt{T_{\rm free}} (Z) \\
							\implies 	&\;& 	\mathcal{G}(Y)  &\subseteq \mathcal{G}\bigl( \upwrt{T_{\rm free}} (Z) \bigr) \\
							\implies 	&\;& 	\mathcal{G}(Y)  &\subseteq \upwrt{S_{\rm free}} \bigl( \mathcal{G} (Z) \bigr) \\
							\implies 	&\;& 	\upwrt{S_{\rm free}} \bigl( \mathcal{G}(Y) \bigr) &\subseteq \upwrt{S_{\rm free}} \bigl( \mathcal{G} (Z) \bigr)  \\
							\iff 		&\;& 		\mathcal{G}(Y) &\freeconv_{\rm deg} \mathcal{G}(Z),
					\end{alignedat}
				\end{equation}
				which proves the claim.
				
				For the converse, we argue as before.
				If $\mathcal{G}$ is an isotone, then we can infer
				\begin{equation}
					\begin{split}
						\upwrt{T_{\rm free}} (x) \freeconv_{\rm deg}  x  &\implies   \mathcal{G}\bigl( \upwrt{T_{\rm free}} (x) \bigr)  \freeconv_{\rm deg}  \mathcal{G}(x) \\
							&\implies \mathcal{G} \bigl( \upwrt{T_{\rm free}} (x)  \bigr)  \subseteq \upwrt{S_{\rm free}} \bigl( \mathcal{G}(x) \bigr),
					\end{split}
				\end{equation}
				and thus relation \eqref{eq:mediating maps for degradation order} holds for any $x$.
			\end{proof}
			
			Recall \cref{lem:lax_covariant_isotone}, which establishes that mapping a resource $x$ to $(x, \psi(x))$ for some lax equivariant endomorphism of suplattices $\phi$ is order-preserving.
			Defining $\mathcal{G}$ to be the (unique) suplattice homomorphism that satisfies
			\begin{equation}
				\mathcal{G} (x) = \bigl( x, \psi(x) \bigr)
			\end{equation} 
			we can show that the assumptions of \cref{lem:lax_covariant_isotone} ensure that condition \eqref{eq:mediating maps for degradation order} holds.
			Explicitly, we have
			\begin{align}
				\mathcal{G} \bigl( \up (x)  \bigr) &= \Set*[\big]{ (y,b)  \given  \exists \, t \in T_{\rm free}, \; \exists \, i \in \mathcal{I} \; : \;  t \apply_i y = x , \; b \in \psi(y) } \\
			\intertext{and}
				\up \bigl( \mathcal{G} (x) \bigr) &= \Set*[\big]{ (y,b)  \given   \exists \, t \in T_{\rm free}, \; \exists \, i \in \mathcal{I} \; : \;  t \apply_i y = x, \; t \apply_i b \in \psi(x) }.
			\end{align}
			The assumption that $\psi$ is lax equivariant says that $b \in \psi(y)$ implies $t \apply_i b \in \psi(x)$ and thus we have the desired relation 
			\begin{equation}
				\mathcal{G} \bigl( \up (x)  \bigr) \subseteq \up \bigl( \mathcal{G} (x) \bigr).
			\end{equation}
			Consequently, \cref{lem:lax_covariant_isotone} thus follows from \cref{lem:mediating maps for degradation order}.
			
	\section{Comparing Monotones}\label{sec:monotone_comparison}

		The general monotone constructions from \cref{sec:cost_yield,sec:translating} have several inputs that need to be specified to obtain the target resource monotone.
		For instance, cost and yield constructions require a choice of the gold standard $A$ and an evaluation thereof, as well as a right (or left) invariant set of transformations.
		On the other hand, measures of minimal distinguishability correspond to a choice of a contraction and a (family of) downset(s).
		
		In this section, we aim to address the following question:
		Which choices of these input parameters are good in the sense that they result in a useful resource monotone?
		In order to assess the usefulness of monotones as far as characterizing a preordered set $(X,\succeq)$ is concerned, we define a preorder $\sqsupseteq_{X}$ on the set of monotones---denoted by $\mathsf{Mon}(X)$.
		It is just the set of all order-preserving maps from $(X,\succeq)$ to $\ordreals$.
		In this context, we consider a monotone $m$ to be more ``useful'' than a monotone $\ell$ if it contains all of the information about $(X,\succeq)$ that $\ell$ does and possibly more.	
		We now formalize what we mean by the amount of information a monotone has about a preordered set.
		
		Note that a function $m \colon X \to \reals$ is a monotone if and only if for all pairs $(a,b) \in X \times X$, the following implication holds:
		\begin{equation}
			\label{eq:monotone inference}
			m(a) < m(b) \implies a \not \succeq b.
		\end{equation}
		That is, monotones contain information about the order relation ${\succeq} \subseteq X \times X$ insofar as they witness when pairs of elements of $X$ are \emph{not} related by $\succeq$.
		Of course, if $m(a) \geq m(b)$ holds, then the implication above does not allow us to learn anything about whether $a \succeq b$ is true or not.
		Given a monotone $m$, a pair $(a,b)$ is henceforth called \mbox{$\bm{m}$-\textbf{interesting}} if $m(a) < m(b)$ holds. 
		The $m$-interesting pairs are those for which we can learn whether $a \not \succeq b$ holds from the fact that $m$ is a monotone.
		
		The set of all $m$-interesting pairs for a monotone $m$ is denoted by 
		\begin{equation}
			\label{eq:m-interesting relation for monotones}
			\Interesting{m}{X, \succeq} \coloneqq \Set*[\big]{(a,b) \in X \times X  \given  m(a) < m(b) }.  
		\end{equation}
		We also refer to $\Interesting{m}{X, \succeq}$ as the $m$-interesting relation on $X$.
		
		\begin{definition}
			\label{def:Monotone preorder}
			Let $(X,\succeq)$ be a preordered set and let $\mathsf{Mon}(X)$ be the set of monotones $(X,\succeq) \to \ordreals$.
			We define the preorder $\sqsupseteq_{X}$ on $\mathsf{Mon}(X)$ via
			\begin{equation}
				\label{eq:Monotone preorder}
				m \sqsupseteq_{X} \ell  \quad\iff\quad  \Interesting{m}{X, \succeq} \supseteq \Interesting{\ell}{X, \succeq}
			\end{equation}
			and we say that $m$ is \textbf{more informative about} $\bm{(X,\succeq)}$ than $\ell$ is if $m \sqsupseteq_{X} \ell$ holds.
			Whenever the preordered set $(X, \succeq)$ is clear from context, we denote it by $\sqsupseteq$.
		\end{definition}
		For functions $m$ and $\ell$ which are monotones, we can express $m \sqsupseteq \ell$ also as
		\begin{subequations}
			\label{eq:Informative relation for monotones}
			\begin{alignat}{9}
				\label{eq:Informative relation for monotones 1}
				m \sqsupseteq \ell  
					\quad&\iff\quad  \forall \, a,b \in X \, &&: &\; \ell(a) &< \ell(b) \,&&\implies \,& m(a) &< m(b) \\
				\label{eq:Informative relation for monotones 2}
					\quad&\iff\quad  \forall \, a,b \in X \, &&: &\; m(a) &\geq m(b) \,&&\implies \,& \ell(a) &\geq \ell(b).
			\end{alignat}
		\end{subequations}
	
		We would like to compare the constructions of monotones appearing in \cref{sec:cost_yield,sec:translating} in terms of how informative the resulting monotones are, depending on the choices one can make.
		One of the input elements for cost and yield constructions is a partial function $f \colon X \to \reals$.
		Although it need not be a monotone on its domain $A$, it can still be understood as witnessing inconvertibility between some resources within $A$. 
		In \cref{thm:yield and cost from more informative functions} below, we prove that whenever $f$ witnesses all the pairs of inconvertible resources that $g$ does, then $f$ is at least as useful as $g$ is, when thought of as an input to the generalized yield and cost constructions.
		That is, in such a case $\yield{f}{D}$ is more informative about $(X,\succeq)$ than $\yield{g}{D}$ is and similarly for the cost construction.
		To make these kinds of statements more precise, we now introduce the notion of the resource inconvertibility that a partial function (as opposed to a monotone) witnesses.
		
		Let $f \colon X \to \reals$ be a partial function with domain $A$.
		We say that $f$ witnesses the inconvertibility of a pair of resources $(x,y)$ if both $f(x) < f(y)$ and $x \not \succeq y$ hold.
		As far as this property is concerned, we call such a pair of resources $(x,y)$ \mbox{$\bm{f}$-\textbf{interesting}}.
		
		The set of all $f$-interesting pairs for a partial function $f$ is denoted by 
		\begin{equation}
			\Interesting{f}{X, \succeq} \coloneqq \Set*[\big]{ (x,y) \in A \times A  \given  f(x) < f(y) \;\land\;  x \not \succeq y }.
		\end{equation}
		We also refer to $\Interesting{f}{X, \succeq}$ as the $f$-interesting relation on $X$.
		Note that this definition coincides with the $f$-interesting relation for a monotone $f$ given by \cref{eq:m-interesting relation for monotones}, whenever $f$ is indeed a monotone.
		This justifies the use of identical notation for these two concepts.	
	
		\begin{definition}
			\label{def:Function preorder}
			Let $(X,\succeq)$ be a preordered set and let $f, g \colon X \to \reals$ be partial functions with domains $A$ and $B$ respectively.
			We say that $f$ \textbf{witnesses more resource inconvertibility in} $\bm{(X, \succeq)}$ than $g$ does if $f \sqsupseteq g$ holds, where
			\begin{equation}
				\label{eq:Function preorder}
				f \sqsupseteq g  \quad\iff\quad  \Interesting{f}{X,\succeq} \supseteq \Interesting{g}{X,\succeq}.
			\end{equation}
		\end{definition}
		
		\begin{proposition}
			\label{thm:yield and cost from more informative functions}
			Consider a concrete resource theory $(\mathcal{P}(T), \mathcal{P}(X), \apply)$ with resource preorder $(X,\succeq)$.
			Let $D$ be a right-invariant and $S$ a left-invariant subset of $T$.
			Furthermore, let $f \colon X \to \reals$ and $g \colon X \to \reals$ be two partial functions with domains $A$ and $B$ respectively.
			
			If $f$ witnesses more resource inconvertibility in $(X, \succeq)$ than $g$ does, then $\yield{f}{D}$ is more informative about $(X, \succeq)$ than $\yield{g}{D}$ is and also $\cost{f}{S}$ is more informative about $(X, \succeq)$ than $\cost{g}{S}$ is.
			That is, we have
			\begin{subequations}
				\label{eq:yield and cost from more informative functions}
				\begin{alignat}{3}
					\label{eq:yield from more informative functions}
					f \sqsupseteq g &\,\implies \, & \yield{f}{D} &\sqsupseteq \yield{g}{D}, \\
					\label{eq:cost from more informative functions}
					f \sqsupseteq g &\,\implies \, & \cost{f}{S} &\sqsupseteq \cost{g}{S}.
				\end{alignat}
			\end{subequations}
			
			Moreover, if $f$ and $g$ are monotones on their respective domains, their domains $A, B$ coincide, and $D = S = T_{\rm free}$, then the converse of both implications holds as well. 
			That is, in such a case we have
			\begin{equation}
				\label{eq:yield and cost from more informative monotones}
				\begin{alignedat}{3}
					f \sqsupseteq g &\,\iff \, & \yield{f}{} &\sqsupseteq \yield{g}{}, \\
					f \sqsupseteq g &\,\iff \, & \cost{f}{} &\sqsupseteq \cost{g}{}.
				\end{alignedat}
			\end{equation}
		\end{proposition}
		
		\begin{proof}
			In order to prove claim \eqref{eq:yield from more informative functions}, we need to show that $\yield{g}{D}(x) < \yield{g}{D}(y)$ implies $\yield{f}{D}(x) < \yield{f}{D}(y)$ for all $x,y \in X$ for which we have $x \not\succeq y$. 
			This follows via
			\begin{subequations}
				\begin{alignat}{9}
					\yield{g}{D}(x) &< \yield{g}{D}(y) \\ 
						\label{eq:yield from more informative functions 2}
						\iff  &\forall \, x_2 \in B \cap \downwrt{D} (x) ,\; \exists \, y_2 \in B \cap \downwrt{D} (y) \;&&:\; g(x_2) < g(y_2) \\
						\label{eq:yield from more informative functions 2'}
						\implies  &\forall \, x_2 \in B \cap \downwrt{D} (x) ,\; \exists \, y_2 \in B \cap \downwrt{D} (y) \;&&:\; g(x_2) < g(y_2) \text{ and } x_2 \not \succeq y_2 \\
						\label{eq:yield from more informative functions 3}
						\implies &\forall \, x_2 \in A \cap \downwrt{D} (x) ,\; \exists \, y_2 \in A \cap \downwrt{D} (y) \;&&:\; f(x_2) < f(y_2) \text{ and } x_2 \not \succeq y_2 \\
						\iff &\yield{f}{D}(x) < \yield{f}{D}(y).&&
				\end{alignat}
			\end{subequations}
			Implication \mbox{(\ref{eq:yield from more informative functions 2}) $\Rightarrow$ (\ref{eq:yield from more informative functions 2'})} follows for the following reason.
			Note that if $x_2 \succeq y_2$ holds, then we also have
			\begin{equation}
				y_2 \in \down (x_2) \subseteq \down \circ \downwrt{D} (x) = \downwrt{D} (x)
			\end{equation}
			where the last equality is by \cref{lem:composing image maps}.
			Thus, if there was an $x_2$ for which all feasible $y_2$ satisfying $g(x_2) < g(y_2)$ also satisfied $x_2 \succeq y_2$, then $g(y_2)$ would be bounded above by $\yield{g}{D}(x)$. 
			Since $\yield{g}{D}(y)$ is strictly larger than $\yield{g}{D}(x)$, for each $x_2$ there must be a feasible $y_2$ outside $\down (x_2)$.
			
			Implication \mbox{(\ref{eq:yield from more informative functions 2'}) $\Rightarrow$ (\ref{eq:yield from more informative functions 3})} follows from the assumption that $f$ witnesses more resource inconvertibility in $(X, \succeq)$ than $g$ does.
			Thus, we obtain the desired implication \eqref{eq:yield from more informative functions}.
			
			In order to prove claim \eqref{eq:cost from more informative functions}, we need to show the analogous statement for cost monotones:
			\begin{subequations}
				\begin{alignat}{9}
					\cost{g}{S}(x) &< \cost{g}{S}(y) \\ 
						\label{eq:cost from more informative functions 2}
						\iff & \forall \, y_2 \in B \cap \upwrt{S} (y) ,\; \exists \, x_2 \in B \cap \upwrt{S} (x) \;&&:\; g(x_2) < g(y_2) \text{ and } x_2 \not \succeq y_2 \\
						\label{eq:cost from more informative functions 3}
						\implies & \forall \, y_2 \in A \cap \upwrt{S} (y) ,\; \exists \, x_2 \in A \cap \upwrt{S} (x) \;&&:\; f(x_2) < f(y_2) \text{ and } x_2 \not \succeq y_2 \\
						\iff & \cost{f}{S}(x) < \cost{f}{S}(y)
				\end{alignat}
			\end{subequations}
			Once again, given a fixed $y_2$, the relation $x_2 \succeq y_2$ would imply 
			\begin{equation}
				x_2 \in \up (y_2) \subseteq \up \circ \upwrt{S} (y) = \upwrt{S} (y),
			\end{equation}
			and consequently $\cost{g}{S}(x) \geq \cost{g}{S}(y)$ if it were true for all feasible $x_2$.
			That is why the first equivalence holds.
			As before, implication \mbox{(\ref{eq:cost from more informative functions 2}) $\Rightarrow$ (\ref{eq:cost from more informative functions 3})} follows from $f \sqsupseteq g$.
			This concludes the proof of the first half of \cref{thm:yield and cost from more informative functions}.
			
			In order to obtain the converse, we can show that $f \not\sqsupseteq g$ implies both $\yield{f}{} \not\sqsupseteq \yield{g}{}$ and 
				$\cost{f}{} \not\sqsupseteq \cost{g}{}$, provided that $f$ and $g$ are monotones on their domain $A$. 
			Under these assumptions, the statement $f \not\sqsupseteq g$ can be expressed as
			\begin{equation}
				\label{eq:f not above g}
				\exists \, x,y \in A \; : \; g(x) < g(y) \text{ and } f(x) \geq f(y).
			\end{equation}
			By \cref{prop:Extensions of monotones are the same on the original domain}, the values of $\yield{f}{}$ and $\cost{f}{}$ coincide with the value of $f$ on $A$, and similarly for $g$. 
			Therefore, $f \not\sqsupseteq g$ implies the following two statements:
			\begin{equation}
				\begin{alignedat}{9}
					\exists \, x,y \in A \; &: &\;   \yield{g}{}(x) &< \yield{g}{}(y) &&\text{ and }& \yield{f}{}(x) &\geq \yield{f}{}(y), \\
					\exists \, x,y \in A \; &: &\;   \cost{g}{}(x) &< \cost{g}{}(y) &&\text{ and }& \cost{f}{}(x) &\geq \cost{f}{}(y).
				\end{alignedat}
			\end{equation}
			Since the yields and costs are monotones themselves, these imply $\yield{f}{} \not\sqsupseteq \yield{g}{}$ and \mbox{$\cost{f}{} \not\sqsupseteq \cost{g}{}$}.
			Consequently, the proof of the second half of \cref{thm:yield and cost from more informative functions} is also complete.
		\end{proof}
		
		\begin{corollary}
			As a consequence of \cref{thm:yield and cost from more informative functions}, sufficient and necessary conditions for the ordering (via $\sqsupseteq$) of yields and costs relative to $T_{\rm free}$ are given by the ordering \mbox{(via $\sqsupseteq$)} of their restrictions to $A \cup B$.
			These facts can also be expressed in terms of the order relation with respect to informativeness about $(A \cup B, \succeq)$ as:
			\begin{equation}
				\label{eq:Informativeness on domain is equivalent to informativeness on R}
				\begin{alignedat}{9}
					\yield{f}{} &\sqsupseteq_{A \cup B} \yield{g}{} \;&&\iff &\; \yield{f}{} &\sqsupseteq_X \yield{g}{} \\
					\cost{f}{} &\sqsupseteq_{A \cup B} \cost{g}{} \;&&\iff &\;  \cost{f}{} &\sqsupseteq_X \cost{g}{}
				\end{alignedat}
			\end{equation}
		\end{corollary}
		
		Therefore, if one wishes to characterize the resource preorder by monotones generated by the generalized yield and cost constructions, then using functions $A \to \reals$ that are informative about $(X, \succeq)$ according to $\sqsupseteq$ should be preferred.
			
		A function $f$ cannot witness more resource inconvertibility than a complete set of monotones because such a set captures all the information in the preordered set $(A, \succeq)$.
		Nonetheless, a function $f$ can witness more resource inconvertibility than \emph{any single} monotone. 
		The simplest example is provided by $A$ with 4 elements, two pairs of which are ordered as in the following Hasse diagram:
		\begin{equation}
			\label{eq:preorder example}
			\begin{tikzcd}
				x_1 \arrow[d]	&	y_1 \arrow[d] \\
				x_2				&	y_2
			\end{tikzcd}
		\end{equation}
		If we let $f$ be defined as follows 
		\begin{equation}
			\begin{split}
				f(x_1) &= 0  \\
				f(x_2) &= 1	
			\end{split}
			\qquad
			\begin{split}
				f(y_1) &= 0 \\
				f(y_2) &= 1
			\end{split}
		\end{equation}
		then it clearly fails to be a monotone.
		Note that $f$ witnesses inconvertibility for the two pairs of resources, $(x_1,y_2)$ and $(y_1,x_2)$, while no single monotone can do so simultaneously.
		Indeed, any monotone must satisfy $m(x_1) \geq m(x_2)$ and $m(y_1) \geq m(y_2)$. 
		If it witnesses the inconvertibility of the pair $(x_1,y_2)$, then $m(x_1) < m(y_2)$ holds and these three inequalities together imply that $m(y_1)$ is greater than $m(x_2)$, so that $m$ then \emph{cannot} witness the inconvertibility of the other pair $(y_1,x_2)$. 
		The function $f$ is capable of witnessing the inconvertibility of both of these pairs of resources thanks to the fact that it fails to be order-preserving.
	
		\begin{remark}
			Note that this function has another interesting property in that both $\yield{f}{}$ and $\cost{f}{}$ are constant, i.e.\ least informative about $(A, \succeq)$ among all monotones $A \to \reals$.
			It can thus serve as a potentially fruitful counterexample to conjectures about the yield and cost constructions.
		\end{remark}
		
		\begin{example}[chains admit a most informative function]\label{ex:chains_unique_value}
			If $A$ is a totally ordered subset of $X$, then there \emph{is} a single monotone $m$ that forms a complete set of monotones by itself. 
			Therefore, it is at least as informative about $(A, \succeq)$ than any other function $A \to \reals$.
			As a consequence, for each right invartiant set $D \in \mathcal{P}(T)$, there is a unique most informative yield monotones with respect to gold standard specified by $A$, namely $\yield{m}{D}$. 
			Similarly, $\cost{m}{S}$ is more informative than $\cost{f}{S}$ for any other choice of $f \colon A \to \reals$.s
			Given the choice of $D = S = X_{\rm free}$, these correspond to currencies defined in \cite{Kraemer2016}, as we mentioned in \cref{ex:currencies}.
		\end{example}

\chapter{Distinguishability of Probabilistic Behaviours}\label{sec:distinguishability}
	
	In the final chapter, we use a specific resource theory of distinguishability to illustrate some of the concepts introduced in previous chapters.
	Specifically, we look at encodings of hypotheses in classical probabilistic systems, whose behaviour is characterized by a probability distribution over a sample space.
	Transformations of such probabilistic systems are provided by stochastic maps which assign a probabilistic behaviour of the system in their codomain to each element of the sample space in their domain.
	
	After introducing the basic notation, in \cref{sec:dist_zonotope} we give a brief account of image maps (see \cref{def:image_map}) in this context and of their utility when we want to learn about the resource ordering.
	The image maps correspond to objects known as Markotopes and zonotopes, the latter of which offer a particularly simple and visual way to read off whether an encoding is freely convertible to another one or not.
	In \cref{sec:dist_Shor_enc}, we give an example of two encodings for which the zonotope condition is not sufficient to decide the fact that they are unordered in the resource theory of distinguishability.
	The following two sections present alternative methods that witness their inconvertibility.
	The method in \cref{sec:dist_poss} uses the notion of a morphism of resource theories from \cref{sec:rt_morphisms} and the one in \cref{sec:dist_rank} shows how one can use the resource weight construction from \cref{sec:weight and robustness} in practice.
	Finally, in \cref{sec:dist_channels} we illustrate how to extend monotones defined for encodings of distributions to monotones for encodings of channels, the general case of which has been introduced in \cref{ex:Changing type}.
	
	\section{The Set-Up}\label{sec:dist_set-up}
		
		In \cref{sec:Encodings}, we saw a construction of a resource theory of $H$-distinguishability, which can be associated to any (listable) quantale module.
		While there are other consistent interpretations, we may think of $H$ as a set of hypotheses about the state of some variable of interest.
		The resource objects, tuples indexed by $H$, model how the hypotheses manifest themselves in the world. 
		Constant tuples are independent of $H$, while those for which the identity of a resource is highly dependent on $H$ allow one to learn about $H$, provided that the alternatives can be distinguished.
		In this context, the manipulations are performed by an agent who has no a priori knowledge of $H$.
		This is why the free transformations are required to be independent of $H$---they are constant tuples.
		
		Consider a collection of resources $\{d_i\}_{i\in I}$ that can be interpreted as unambiguously distinguishable, such as distinct read-outs of an (ideal) measurement device. 
		In constructor theory \cite{deutsch2015constructor}, such collections are called \emph{information variables}.
		An encoding of $H$ that can be freely transformed to
		\begin{equation}
			(d_1, d_2, \dots, d_{\mathfrak{h}})
		\end{equation}
		is a maximally informative one---it encodes all the information about $H$ there is to know.
		Many encodings, however, fall somewhere in between, they are neither independent of $H$ nor maximally informative. 
		We can compare their informativeness via the resource ordering of $H$-distinguishability.
		
		In the context of classical probability theory, these ideas correspond to the comparison of statistical experiments in terms of their informativeness \cite{blackwell1951comparison}.
		Hereafter, we thus refer to statistical experiments and resource encodings interchangeably.
		Likewise, we identify the concept of \emph{informativeness about hypotheses} with that of \emph{distinguishability of behaviours} (which encode said hypotheses).
		When talking about statistical experiments, the set $H$ may be identified with a statistical parameter that we wish to determine by conducting experiments. 
		\begin{example}
			As one of an endless list of examples, consider $H$ to represent the time of day.
			That is, different elements of $H$ corresponds to different hypotheses about what is the time.
			There are many ways in which one could try to learn about $H$.
			For instance, we may look at a well-calibrated clock and use the fact that different times are encoded in easily distinguishable positionings of the clock arms.
			A somewhat less informative encoding of $H$ is that of a shadow cast by a solar clock.
			A statistical experiment for learning about $H$ with very low precision would be to measure the temperature---the elements of $H$ are typically encoded in highly overlapping air temperature probability distributions.
			However, since it may be used at night, as opposed to a solar clock, we would expect that a thermometer and a solar clock are \emph{incomparable} as encodings of $H$.
			Neither contains strictly more information about $H$ than the other does.
		\end{example}
		While realistic statistical experiments (also known as statistical models) carry more intricate structure \cite{mccullagh2002statistical}, a rudimentary way to describe them is as encodings of $H$ that specify a probabilistic behaviour of some system $A$ that depends on the value of $H$.
		We may think of $A$ as the \emph{data variable} of the statistical experiment.
		Whether one experiment carries more information about $H$ than another one is given by the resource ordering of (unconstrained) $H$-distinguishability, denoted by $\freeconv_{\rm unif}$.
		There are other definitions of the ordering of informativeness of statistical experiments, but in many contexts they coincide with $\freeconv_{\rm unif}$ \cite{le1996comparison}.
		
		More concretely, if we restrict our attention to finite sample spaces, the individual resources in the original quantale module (elements of $X$) are probability distributions over finite sets, which we may think of as column vectors with non-negative entries that add up to 1.
		We write $X[A]$ for the space of all probability distribution over a fixed sample space $A$.
		Correspondingly, the transformations (elements of $T$) are stochastic maps.
		In other words, they are matrices with non-negative entries, every column of which adds up to 1.
		We use the notation $T[A,B]$ to refer to all the stochastic maps with domain $A$ and codomain $B$.
		
		The quantale composition $\star$ and quantale action $\apply$ are both given by matrix multiplication as one would expect.
		It is thus a unimodular (and listable) quantale module.
		We can also use a diagrammatic language of a process theory\footnotemark{} in which probability distributions are states and stochastic maps are general processes.
		\footnotetext{In particular, the relevant process theory is the \emph{Markov category} of finite sets and stochastic \mbox{maps \cite{fritz2019synthetic}}.}%
		The quantale module operations can be identified with the first construction of \cref{ex:rt_states}:
		\begin{equation}\label{eq:rt_states_4}
			\input{rt_states_4.tikz}	
		\end{equation}
		
		In the associated resource theory of $H$-distinguishability, elements of $X^H$ are tuples of distributions.
		Let us restrict our attention to those tuples for which the sample space is independent of $H$, as we would do in a resource theory of distinguishability whose types\footnotemark{} are given by the underlying sample spaces.
		\footnotetext{Construction of typed resource theories of distinguishability is described in \cite[appendix C]{Monotones}.}%
		Thus, an element of $X^H$ can be described by a stochastic map $H \to A$ for some sample space $A$, depicted as
		\begin{equation}\label{eq:stoch_map}
			\input{stoch_map.tikz}
		\end{equation}
		in the diagrammatic language with the wires representing the hypothesis distinguished for better readability. 
		In other words, we have $X^H[A] = T[H,A]$.
		
		Tuples of transformations in $T^H$ of uniform type can be described as tensors or equivalently as stochastic maps $H \otimes A \to B$ for some sample spaces $A$ and $B$.
		A generic transformation encoding can be thus represented diagrammatically as
		\begin{equation}\label{eq:gen_trans_enc}
			\input{gen_trans_enc.tikz}
		\end{equation}
		Free transformations in $T^H_{\rm unif}$ are those conversions of $A$ to $B$, whose behaviour is independent of the hypothesis. 
		That is, as stochastic maps $H \otimes A \to B$, we require that they can be decomposed as
		\begin{equation}\label{eq:trans_enc}
			\input{trans_enc.tikz}
		\end{equation}
		for some stochastic map $t$, where the process
		\begin{equation}
			\begin{tikzpicture}
	\begin{pgfonlayer}{nodelayer}
		\node [style=bn] (0) at (0, 0.75) {};
		\node [style=none] (1) at (0, -0.25) {};
		\node [style=none] (2) at (0, -0.75) {$H$};
	\end{pgfonlayer}
	\begin{pgfonlayer}{edgelayer}
		\draw [style=dense] (1.center) to (0);
	\end{pgfonlayer}
\end{tikzpicture}

		\end{equation}
		denotes discarding of $H$.
		It is given by the stochastic map that sends each $h \in H$ to the unique probability distribution over the trivial system $I$ which, in this process theory, corresponds to the sample space with a single element.
		That is, diagram \eqref{eq:trans_enc} expresses that free encodings are independent of $H$, so that elements of $T^H_{\rm unif}[A,B]$ can be identified with elements of $T[A,B]$.
		
		Transformation encodings compose element-wise with respect to the hypotheses, as given by \cref{eq:tuple_comp,eq:tuple_action}.
		Diagrammatically, we can depict the composition of transformation encodings as\footnotemark{}
		\begin{equation}\label{eq:param_composition}
			\input{param_composition.tikz}
		\end{equation}
		and the action of transformation encodings on resource encodings as
		\begin{equation}\label{eq:param_composition_2}
			\input{param_composition_2.tikz}
		\end{equation}
		where the process
		\begin{equation}
			\input{copy.tikz}	
		\end{equation}
		denotes the copying of $H$ given by the stochastic map that sends each $h \in H$ to the point distribution supported on $(h,h) \in H \times H$.
		Copying the hypothesis in expressions \eqref{eq:param_composition} and \eqref{eq:param_composition_2} imposes the requirement that the value of $H$ is kept fixed when composing encodings.
		\footnotetext{In the language of Markov categories, this is the composition in the Markov category parametrized by $H$ \cite[section 2.1]{fritz2020representable}.}%
		The ordering of individual resouce encodings is thus given by 
		\begin{equation}\label{eq:ord_enc}
			\input{ord_enc.tikz}
		\end{equation}
		which says that $y$ can be obtained from $x$ by a \emph{post-processing}, i.e.\ as a composition $t \circ x$ with some stochastic map $t$.
		Within the context of finite sample spaces that we focus on here, this ordering is also known as \emph{matrix majorization} \cite{Dahl1999}.
		
	\section{Markotopes and Zonotopes}\label{sec:dist_zonotope}
		
		Matrix majorization can be studied geometrically in terms of so-called Markotopes \cite{Dahl1999}.
		For an integer $k$, consider the set of free transformation encodings whose output space has cardinality $k$.
		Since they can be identified with stochastic maps of type $B \to K$ for fixed $K = \{1, 2, \dots, k \}$ and arbitrary $B$, we denote them by $T[\ph,K]$.
		By construction, these form a right-invariant set of transformation encodings.
		Indeed, for any free encoding $t_{\rm unif}$ as in \eqref{eq:trans_enc} and any $s \in T[\ph,K]$, we have
		\begin{equation}
			\input{k-decision.tikz}
		\end{equation}
		so that $T[\ph,K] \star T^H_{\rm unif} \subseteq T[\ph,K]$ holds, which is the required form of right-invariance.
		
		Thus, by \cref{lem:U-image is order-preserving}, the map
		\begin{equation}
			\downwrt{T[\phsm,K]}  \colon  \bigl( X^H, \freeconv_{\rm unif} \bigr)  \to  \Bigl( \mathcal{P}\bigl( X^H[K] \bigr), \supseteq \Bigr)
		\end{equation}
		is an isotone for any choice of $k$ and it can be used, for instance, to construct generalized yield monotones.
		This fact can be found in \cite[theorem 3.3]{Dahl1999} where $\downwrt{T[\phsm,K]}(x)$, i.e.\ the $T[\ph,K]$-image of an encoding $x$, is a convex subset of $T[H,K]$ called the Markotope associated to $x$.
		Markotopes with $k=2$ are also referred to as \emph{zonotopes}.
		
		In the case of two hypotheses, so that we can write $H = \bm{2} = \{0,1\}$, it suffices to consider zonotopes as the map $\downwrt{T[\phsm,K]}$ for $K = \bm{2}$ is then an order embedding \cite[theorem 4.1]{Dahl1999}.
		That is, two encodings of a binary hypothesis; $x,y \in X^{\bm{2}}[\bm{2}]$; are ordered if and only if their zonotopes are included one within the other:
		\begin{equation}\label{eq:zonotope_inclusion}
			x \freeconv_{\rm unif} y  \quad \iff \quad   \downwrt{T[\phsm,\bm{2}]} (x) \supseteq \downwrt{T[\phsm,\bm{2}]} (y)
		\end{equation}
		The elements of $X^{\bm{2}}[\bm{2}]$ (i.e.\ $2 \times 2$ stochastic matrices) can be parametrized by two elements of the unit interval ($\lambda, \nu \in [0,1]$) via
		\begin{equation}\label{eq:matrix_rep}
			\begin{pmatrix}
				\lambda  		& 		\nu 			\\
				1 - \lambda 	&		1 - \nu 	\\
			\end{pmatrix},
		\end{equation}
		so that the zonotope condition offers a simple visual representation of $\bm{2}$-distinguishability.
		Two examples are given in \cref{fig:zonotope}.
		
		\begin{figure}[h]
			\begin{center}
				\begin{subfigure}[b]{.4\textwidth}\centering
					\begin{tikzpicture}
						\begin{pgfonlayer}{nodelayer}
							\node [style=none] (0) at (-2, -2) {};
							\node [style=none] (1) at (-2, 2) {};
							\node [style=none] (2) at (2, 2) {};
							\node [style=none] (3) at (2, -2) {};
							\node [style=none] (4) at (3, -2) {};
							\node [style=none] (5) at (-2, 3) {};
							\node [style=none] (6) at (-2.5, -1.5) {$\tfrac{1}{8}$};
							\node [style=none] (7) at (-2, -2.5) {$0$};
							\node [style=none] (8) at (2, -2) {};
							\node [style=none] (9) at (3.5, -2) {$\lambda$};
							\node [style=none] (10) at (-2, 3.5) {$\nu$};
							\node [style=none] (11) at (0, 1.5) {};
							\node [style=none] (12) at (0, -1.5) {};
							\node [style=none] (13) at (-2, -1.5) {};
							\node [style=none] (14) at (-2, 1.5) {};
							\node [style=none] (15) at (0, -2) {};
							\node [style=none] (16) at (0, -2.5) {$\tfrac{1}{2}$};
							\node [style=none] (17) at (2, -2.5) {$1$};
							\node [style=none] (18) at (-2.5, 1.5) {$\tfrac{7}{8}$};
						\end{pgfonlayer}
						\begin{pgfonlayer}{edgelayer}
							\draw [style=arrow] (8.center) to (4.center);
							\draw [style=arrow] (1.center) to (5.center);
							\draw [style=blue arrow, fill={rgb,255: red,100; green,140; blue,255}] (0.center) -- (11.center) -- (2.center) -- (12.center) -- cycle;
							\draw [style=segment] (0.center) to (15.center);
							\draw [style=segment] (15.center) to (8.center);
							\draw [style=segment] (0.center) to (13.center);
							\draw [style=segment] (13.center) to (14.center);
							\draw [style=segment] (14.center) to (1.center);
							\draw [style=dashed box] (1.center) to (2.center);
							\draw [style=dashed box] (3.center) to (2.center);
						\end{pgfonlayer}
					\end{tikzpicture}
				\end{subfigure}\hspace{0.05\textwidth}
				\begin{subfigure}[b]{.4\textwidth}\centering
					\begin{tikzpicture}
						\begin{pgfonlayer}{nodelayer}
							\node [style=none] (0) at (-2, -2) {};
							\node [style=none] (1) at (-2, 2) {};
							\node [style=none] (2) at (2, 2) {};
							\node [style=none] (3) at (2, -2) {};
							\node [style=none] (4) at (3, -2) {};
							\node [style=none] (5) at (-2, 3) {};
							\node [style=none] (6) at (-2.5, -1) {$\tfrac{1}{4}$};
							\node [style=none] (7) at (-2, -2.5) {$0$};
							\node [style=none] (8) at (2, -2) {};
							\node [style=none] (9) at (3.5, -2) {$\lambda$};
							\node [style=none] (10) at (-2, 3.5) {$\nu$};
							\node [style=none] (11) at (-1, 1) {};
							\node [style=none] (12) at (1, -1) {};
							\node [style=none] (13) at (-2, -1) {};
							\node [style=none] (14) at (-2, 1) {};
							\node [style=none] (15) at (1, -2) {};
							\node [style=none] (16) at (1, -2.5) {$\tfrac{3}{4}$};
							\node [style=none] (17) at (2, -2.5) {$1$};
							\node [style=none] (18) at (-2.5, 1) {$\tfrac{3}{4}$};
							\node [style=none] (19) at (-1, -2.5) {$\tfrac{1}{4}$};
							\node [style=none] (20) at (-1, -2) {};
							\node [style=none] (21) at (-2.5, 2) {$1$};
						\end{pgfonlayer}
						\begin{pgfonlayer}{edgelayer}
							\draw [style=arrow] (8.center) to (4.center);
							\draw [style=arrow] (1.center) to (5.center);
							\draw [-, draw={rgb,255: red,150; green,0; blue,15}, fill={rgb,255: red,255; green,120; blue,120}] (0.center) -- (11.center) -- (2.center) -- (12.center) -- cycle;
							\draw [style=segment] (0.center) to (20.center);
							\draw [style=segment] (20.center) to (15.center);
							\draw [style=segment] (15.center) to (8.center);
							\draw [style=segment] (0.center) to (13.center);
							\draw [style=segment] (13.center) to (14.center);
							\draw [style=segment] (14.center) to (1.center);
							\draw [style=dashed box] (1.center) to (2.center);
							\draw [style=dashed box] (3.center) to (2.center);
						\end{pgfonlayer}
					\end{tikzpicture}
				\end{subfigure}
			\end{center}
			\caption[Example of Zonotopes]{Two examples of zonotopes of binary encodings of a binary hypothesis are provided here.
				They are depicted as subsets of $X^{\bm{2}}[\bm{2}]$, which is the unit square in the representation provided by \eqref{eq:matrix_rep}.
				The blue one (on the left) corresponds to the encoding with $(\lambda,\nu) = (1/8,1/2)$.
				The red one (on the right) corresponds to the encoding with $(\lambda,\nu) = (1/4,3/4)$.
				Points in the shaded regions are elements of the $T[\bm{2},\bm{2}]$-image of the respective encoding.
				Vertices of the polytopes correspond to ``deterministic'' post-processings of the given encoding by transformations that are $\{0,1\}$-valued.
				That is, for an encoding with a given pair $(\lambda,\nu)$, they are the points $(0,0)$, $(1,1)$, $(\nu,\lambda)$, and $(\lambda,\nu)$ itself.
				Since neither of the two zonotopes includes the other as a subset, we conclude that these two encodings are incomparable in the resource theory of distinguishability.
				In other words, neither of the two encodings is more informative about the hypothesis than the other.}
			\label{fig:zonotope}
		\end{figure}
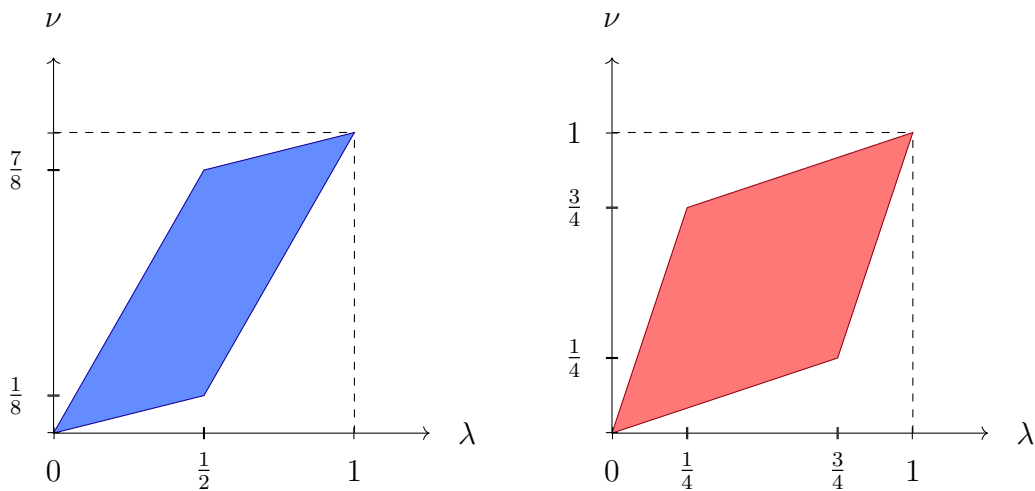
		
		As we saw in \cref{ex:athermality_oplax_covariant}, resource ordering of nonuniformity can be faithfully embedded in the informativeness order of encodings of a binary hypothesis via the map
		\begin{equation}
			X \to X^{\bm{2}} , \quad  x \mapsto (x,\mu)
		\end{equation}
		where $\mu$ denotes the uniform distribution over the same sample space that $x$ is defined on.
		In this context, the zonotope of the tuple $(x,\mu)$ is commonly viewed through its equivalent representation in terms of its ``upper boundary''---the Lorenz curve (cf.\ \cref{fig:Lorenz_curves}).
		The left diagram in \cref{fig:zonotope} provides a concrete example.
		In terms of Lorenz curves, the zonotope inclusion condition \eqref{eq:zonotope_inclusion} boils down to the question of whether one Lorenz curve is ``nowhere below'' another one.
		
		Similarly, resource ordering of athermality\footnotemark{} can be viewed as distinguishability of a given distribution $x$ from the thermal state via the embedding
		\footnotetext{While nonuniformity ordering of probability distributions corresponds to majorization, the athermality ordering can be viewed as a generalization thereof known as relative majorization \cite[section 14.B]{Marshall1979} or as thermo-majorization \cite{Horodecki2013}---the latter name used specifically in the context of thermodynamics.
			Namely, instead of requiring that stochastic post-processings preserve uniform distributions (as in the case of standard majorization), one requires that they fix some other choice of distributions---one for each system.}%
		\begin{equation}
			x \mapsto (x, \rho_{\rm therm})
		\end{equation}
		where $\rho_{\rm therm}$ denotes the thermal state of the system which $x$ is a state of.
		Just as for nonuniformity, this point of view gives an equivalent representation of the athermality ordering.
		There is also a corresponding notion of a Lorenz curve that can be used to decide the ordering of states graphically.
		For instance, if the thermal state and $x$ are given by
		\begin{align}
			\rho_{\rm therm} &= 
				\begin{pmatrix}
					3/4 \\
					1/4
				\end{pmatrix} &
			x &= 	
				\begin{pmatrix}
					1/4 \\
					3/4
				\end{pmatrix}	
		\end{align}
		then the zonotope of the encoding $(x, \rho_{\rm therm})$ is the one depicted in the right diagram of \cref{fig:zonotope}.
		Its Lorenz curve is the upper boundary thereof.
		
	\section{Shor Encodings}\label{sec:dist_Shor_enc}
		
		For the purposes of resource theories with a single fixed point such as nonuniformity and athermality (recall \cref{ex:fixed_points}), it is sufficient to understand the zonotope inclusion condition \eqref{eq:zonotope_inclusion}.
		This is because they can be embedded in a resource theory of $\bm{2}$-distinguishability whose resource ordering, as we mentioned, boils down to the zonotope inclusion condition for \emph{any} data variable $A$ in which the hypothesis $H$ are encoded.
		However, other resource theories of $M$-asymmetry require the full understanding of $H$-distinguishability.
		Indeed, for hypotheses with cardinality greater than two, there is no known fixed integer $k$ for which the Markotope inclusion condition would be sufficient for determining the ordering of encodings.\footnotemark{}
		\footnotetext{This holds when we consider encodings with arbitrary data variables.
		For a fixed choice of a data variable $A$ with cardinality $a$, the inclusion of Markotopes with $k = a$ is both necessary and sufficient condition for the informativeness ordering of encodings in $X^{H}[A]$.
		This is easy to see as the Markotope inclusion is then just a different way to write the definition of the resource ordering of encodings.
		Moreover, when comparing two encodings with distinct data variables $A$ and $B$, it suffices to look at Markotope inclusion at the level of the smaller of the minimal sufficient statistics for the two encodings.
		Markotopes with higher $k$ provide no additional information.
		This is because every encoding is equivalent to its minimal sufficient statistic (as well as any other sufficient statistic) in the resource theory of distinguishability.}%
		
		In the rest of \cref{sec:distinguishability}, we use a particular example of two encodings to illustrate methods that can offer additional insights about informativeness of encodings, as compared to the zonotope inclusion condition.
		They are encodings of a three-valued hypothesis $H = \bm{3}$ in three and four-valued data variables respectively, given by
		\begin{align}\label{eq:Shor_codes}
			x &= \begin{pmatrix}[1.2]
				\frac{1}{2} & \frac{1}{2} & \frac{1}{2} \\
				\frac{1}{2} & 0 & 0 \\
				0 & \frac{1}{2} & 0 \\
				0 & 0 & \frac{1}{2}
			\end{pmatrix} & 
			y &= \begin{pmatrix}[1.2]
				\frac{1}{2} & 0 & \frac{1}{2} \\
				\frac{1}{2} & \frac{1}{2} & 0 \\
				0 & \frac{1}{2} & \frac{1}{2}
			\end{pmatrix}.
		\end{align}
		We refer to $x$ and $y$ from \eqref{eq:Shor_codes} as ``Shor encodings'' as they have been considered by Peter Shor as a classical analogue of the corresponding encoding of quantum states \cite{jozsa2000distinguishability}.
		See also \cite[appendix C]{Marvian2013} for an explicit consideration of these.
		
		One can check that the zonotope $\downwrt{T[\phsm,\bm{2}]}(x)$ of $x$ includes that of $y$.
		This can be established, for example, by listing all the deterministic post-processings of $x$ and $y$ which correspond to (potential)\footnotemark{} zonotope vertices and observing that every vertex of $\downwrt{T[\phsm,\bm{2}]}(y)$ is also among the vertices of $\downwrt{T[\phsm,\bm{2}]}(x)$.
		\footnotetext{Deterministic post-processings generate the zonotope as its convex hull, but a single encoding given by such post-processing may not be a vertex if it is a convex combination of other ones.}%
		However, this inclusion is strict. 
		Indeed, the encoding
		\begin{equation}\label{eq:Shor_codes_2}
			\begin{pmatrix}[1.2]
				\frac{1}{2} & 0 & 0 \\
				\frac{1}{2} & 1 & 1
			\end{pmatrix} 
				= 
			\begin{pmatrix}[1]
				1 & 0 & 1 & 1 \\
				0 & 1 & 0 & 0 
			\end{pmatrix}
			\begin{pmatrix}[1.2]
				\frac{1}{2} & \frac{1}{2} & \frac{1}{2} \\
				\frac{1}{2} & 0 & 0 \\
				0 & \frac{1}{2} & 0 \\
				0 & 0 & \frac{1}{2}
			\end{pmatrix}
				\in \downwrt{T[\phsm,\bm{2}]}(x)
		\end{equation}
		obtained by a deterministic post-processing of $x$ is in the zonotope of $x$, but not in the zonotope of $y$.
		By the zonotope inclusion condition, i.e.\ because $\downwrt{T[\phsm,\bm{2}]}$ is an isotone, we conclude $y \not \succeq_{\rm unif} x$.
		That is, $y$ is not more informative about the hypothesis than $x$ is.
		
	\section{Possibilistic Encodings}\label{sec:dist_poss}
		
		The two zonotopes cannot be used to determine whether the Shor encodings are incomparable, since we have
		\begin{equation}
			\downwrt{T[\phsm,\bm{2}]}(x) \supseteq \downwrt{T[\phsm,\bm{2}]}(y).
		\end{equation}
		We would need to consider the Markotope with $k=3$ to establish this fact.
		However, there are other methods one can use as well.
		
		Let the ceiling function $\mathrm{ceil} \colon [0,1] \to \{0,1\}$ be defined by
		\begin{equation}
			\mathrm{ceil}(\lambda) \coloneqq 
				\begin{cases}
					0 & \text{if } \lambda = 0 \\
					1 & \text{if } \lambda \in (0,1] 
				\end{cases}
		\end{equation}
		Let us consider $\{0,1\}$ to be the semiring with multiplication as usual and addition given by
		\begin{align}
			0 + 0 &= 0  &  0 + 1 = 1 + 0 &= 1  &  1 + 1 &= 1.
		\end{align}
		This semiring is commonly termed the \textbf{Boolean semiring} and we denote it by $\mathbb{B}$.
		The ceiling function is then a homomorphism of semirings of type $\mathbb{I} \to \mathbb{B}$ where $\mathbb{I}$ denotes the unit real interval thought of as a semiring.
		Let us denote the space of $\mathbb{B}$-valued distributions by $Y$ in order to differentiate them from the $\mathbb{I}$-valued distributions (i.e.\ probability distributions) which are denoted by $X$.
		The space of all $\mathbb{B}$-valued stochastic matrices is denoted by $U$ in contrast to the standard $\mathbb{I}$-valued stochastic matrices denoted by $T$.
		
		Applying the ceiling function to the entries of probability distributions and stochastic matrices leads to two functions $X \to Y$ and $T \to U$ respectively, which we both denote by $\mathrm{ceil}$ as well.
		Nothing changes when we work with tuples.
		In this way, we get a \emph{resource theory of possibilistic distinguishability} whose quantale module is $(\mathcal{P}(U^H), \mathcal{P}(Y^H), \apply)$ with $\star$ and $\apply$ given by matrix multiplication and the set of free transformation is given by the hypothesis-independent tuples $U^H_{\rm unif}$.
		Moreover, the two functions, ${\mathrm{ceil} \colon T^H \to U^H}$ and ${\mathrm{ceil} \colon X^H \to Y^H}$, define a morphism of resource theories $(\mathrm{ceil},\mathrm{ceil})$ that maps distinguishability of probabilistic encodings to distinuishability of possibilistic encodings.
		Specifically, they are both defined to be suplattice homomorphisms; they satisfy \cref{eq:q_morphism,eq:qm_morphism} in \cref{def:rt_morphism} because $\mathrm{ceil}$ applied to matrices entry-wise preserves matrix multiplications; and the free transformations are preserved as well.
		By \cref{prop:isotone_from_morphism}, $\mathrm{ceil}$ is thus an isotone of type $(X^H,\freeconv_{\rm unif}) \to (Y^H,\freeconv_{\rm unif})$.
		
		When we apply the ceiling function to Shor encodings, we obtain
		\begin{align}\label{eq:Shor_codes_pos}
			\mathrm{ceil}(x) &= \begin{pmatrix}
				1 & 1 & 1 \\
				1 & 0 & 0 \\
				0 & 1 & 0 \\
				0 & 0 & 1
			\end{pmatrix} & 
			\mathrm{ceil}(y) &= \begin{pmatrix}
				1 & 0 & 1 \\
				1 & 1 & 0 \\
				0 & 1 & 1
			\end{pmatrix}.
		\end{align}
		The question of whether $\mathrm{ceil}(x)$ can be converted to $\mathrm{ceil}(y)$ by a $\mathbb{B}$-valued (i.e.\ possibilistic) stochastic post-processing is now a combinatorial one with a finite number of feasible transformations.
		One can indeed check that there is no $\mathbb{B}$-valued stochastic $3 \times 4$ matrix that gives $\mathrm{ceil}(y)$ when composed with $\mathrm{ceil}(x)$.
		Consequently, since $\mathrm{ceil}$ is an isotone, we can also conclude that $x \not \succeq_{\rm unif} y$ holds in the original resource theory of (probabilistic) distinguishability.
		
		Note that the possibilistic encodings can be also viewed as hypergraphs\footnotemark{} with nodes given by elements of the data variable $A$.
		\footnotetext{Strictly speaking, they are hypergraphs with hyperedges labelled by a fixed index set $H$.}%
		For each hypothesis value $h \in H$, we get one hyperedge that corresponds to the subset of elements of $A$ that are consistent with said hypothesis $h$.
		In other words, each column of the matrix specifies a hyperedge that consists of those rows for which the corresponding matrix element is non-zero.
		Transformations can be thought of as (multivalued) functions between the nodes of hypergraphs, such that the image of a hyperedge for each $h$ is required to be the hyperedge corresponding to the same $h$ in the codomain.
		Explicitly, the two hypergraphs that correspond to Shor encodings are:
		
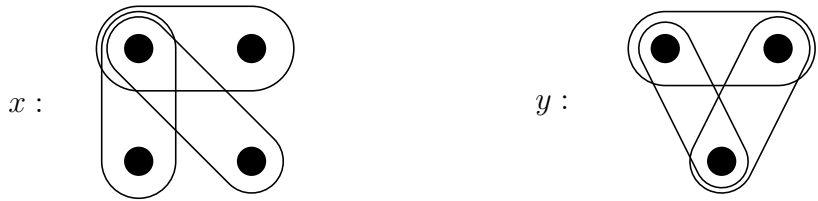
\begin{center}
			\begin{tikzpicture}[align=center,thick,
					resource/.style={circle,fill=black,draw=none},
					he/.style={draw, semithick},
				]
				\begin{pgfonlayer}{nodelayer}
					\node[resource]	(a)		at (-2,0.75)			{};
					\node[resource]	(b)		at (-0.5,0.75)				{};
					\node[resource]	(c)		at (-2,-0.75)			{};
					\node[resource]	(d)		at (-0.5,-0.75)			{};
					
					\node[resource]	(e)		at (5,0.75)				{};
					\node[resource]	(f)		at (6.5,0.75)				{};
					\node[resource]	(g)		at (5.75,-0.75)			{};
					
					\node						at (-3.5,0)					{$x :$};
					\node						at (3.5,0)					{$y :$};
				\end{pgfonlayer}
				\begin{pgfonlayer}{edgelayer}
					\draw[he] 	\hedgeii{a}{b}{16pt};
					\draw[he] 	\hedgeii{a}{c}{14pt};
					\draw[he] 	\hedgeii{a}{d}{12pt};
					
					\draw[he] 	\hedgeii{e}{f}{14pt};
					\draw[he] 	\hedgeii{f}{g}{12pt};
					\draw[he] 	\hedgeii{g}{e}{10pt};
				\end{pgfonlayer}
			\end{tikzpicture}
		\end{center}
		In this way, it is easy to see that there is no such transformation from the hypergraph for $x$ to the hypergraph for $y$, or vice versa.
				
	\section{Weight of a Deterministic Encoding as a Monotone}\label{sec:dist_rank}
		
		The analysis in terms of possibilistic hypergraphs is qualitative by nature.
		In order for it to be useful, there has to be sufficient structure merely in the supports of the encodings.
		More quantitative measures are provided by the resource weight construction (introduced in \cref{sec:weight and robustness}), which we explore hereafter.
		
		First of all, note that the rank of the matrix representing an encoding is a distinguishability monotone.
		This is because $x \succeq_{\rm unif} y$ implies $\mathrm{rank}(y) = \mathrm{rank}(t \circ x) \leq \mathrm{rank}(x)$, where $t$ is the free transformation that achieves the conversion of $x$ to $y$.
		In particular, the rank is a monotone for the deterministic encodings, which are those elements of $X^H$ that are $\{0,1\}$-valued.
		We denote them by $D[H,A]$ where $A$ is the data variable.
		For a fixed $A$, the set $D[H,A]$ of deterministic stochastic matrices constitutes the vertices of the polytope $T[H,A]$ of all stochastic matrices.
		Let us denote the set of elements of $D[H,A]$ whose rank is bounded above by $k$ by $D_k[H,A]$.
		
		For example, if both the hypothesis $H$ and the data variable $A$ are binary, then the deteministic encodings are
		\begin{equation}
			\begin{pmatrix}
				1 & 1 \\ 0 & 0
			\end{pmatrix},
			\begin{pmatrix}
				1 & 0 \\ 0 & 1
			\end{pmatrix},
			\begin{pmatrix}
				0 & 1 \\ 1 & 0
			\end{pmatrix},
			\begin{pmatrix}
				0 & 0 \\ 1 & 1
			\end{pmatrix}.
		\end{equation}
		The first and last have rank $1$, while the other two are of rank $2$.
		
		The above discussion implies that, for any positive integer $k$, the convex hull of deterministic encodings of rank at most $k$ is a downset, denoted by\footnotemark{}
		\begin{equation}
			T_k[H,A] \coloneqq \mathrm{conv} \bigl( D_k[H,A] \bigr).
		\end{equation}
		\footnotetext{The notation $T_k[H,A]$ for the convex hull of the set $D_k[H,A]$ can be motivated by recognizing that $T_k[H,A]$ coincides with the set of all stochastic matrices of type $H \to A$ that have rank at most $k$.
			Proving this fact is another way to establish that $T_k[H,A]$ is a downset, via the argument from \cref{rem:monotones_downsets}, since rank is a monotone.}%
		In the case of binary encodings of a binary hypothesis, $T_1$ consists of encodings of the form:
		\begin{equation}
			\begin{pmatrix}
				\mu & \mu \\ 1-\mu & 1-\mu
			\end{pmatrix}.
		\end{equation}
		The argument showing that $T_k$ is downward closed is as follows.
		If a given encoding $x$ (viewed as a stochastic matrix of type $H \to A$) can be expressed as a convex mixture of deterministic encodings of rank at most $k$ via
		\begin{equation}
			x = \sum_i \lambda_i \, d_i \in T_k[H,A]
		\end{equation}
		where each $d_i$ is an element of $D_k[H,A]$, then for any post-processing $t \in T[A,B]$ the transformed encoding satisfies
		\begin{equation}
			t \circ x = \sum_i \lambda_i  \, t \circ d_i \in T_k[H,A],
		\end{equation}
		as the rank of each $t \circ d_i$ is upper bounded by $k$.
		
		Therefore, we can use the downsets $T_k$ in the resource weight construction.
		Let us write $\mathfrak{h}$ for the cardinality of the hypothesis variable $H$, so that we can use $H = \{ 1, 2, \ldots, \mathfrak{h} \} $ without loss of generality.
		Then, for each pair of positive integers $m,k \in H$ such that $m < k$ holds,\footnotemark{} we have a weight monotone (see \cref{eq:weight_2}) given by
		\footnotetext{If $m$ is larger than or equal to $k$, then the resulting weight monotone is trivial (i.e.\ constant).}%
		\begin{equation}
			\label{eq:Weight monotones from rank}
			f_{m,k}(x) \coloneqq \inf \Set*[\big]{ \lambda \in [0,1]  \given  x \in \lambda\, T_k + (1 - \lambda) \, T_m  }.
		\end{equation}
		We can express this function as an optimization over the convex decompositions of $x$ in terms of deterministic encodings.
		Specifically, we have
		\begin{equation}
			f_{m,k}(x) = \inf \Set{ \sum_i \lambda_i \given  x = \sum_i \lambda_i \, d_i + \sum_j \nu_j \, e_j },
		\end{equation}
		where each $e_j$ is a deterministic encoding with rank at most $m$, each $d_i$ is a deterministic encoding with rank at most $k$ but greater than $m$, and the coefficients add up to $1$, i.e.\ they satisfy:
		\begin{equation}
			\sum_i \lambda_i + \sum_j \nu_j = 1.
		\end{equation}
		
			If we now return to the two Shor encodings introduced in \eqref{eq:Shor_codes}, we can provide another explanation of their incomparability via the weight monotones $f_{m,k}$.
			In particular, the values of the two relevant non-trivial monotones are
			\begin{align}
				\begin{aligned}
					f_{1,3} (x) &= 0.5 \\
					f_{2,3} (x) &= 0.25
				\end{aligned}
				&&
				\begin{aligned}
					f_{1,3} (y) &= 1 \\
					f_{2,3} (y) &= 0
				\end{aligned}
			\end{align}
			with optimal convex decompositios of $x$ are given by
			\begin{equation}
				x = \frac{1}{2} 
					\begin{pmatrix}
						0 & 0 & 0 \\
						1 & 0 & 0 \\
						0 & 1 & 0 \\
						0 & 0 & 1
					\end{pmatrix} + \frac{1}{2}
					\begin{pmatrix}
						1 & 1 & 1 \\
						0 & 0 & 0 \\
						0 & 0 & 0 \\
						0 & 0 & 0
					\end{pmatrix}
			\end{equation}
			for $m = 1, k = 3$ and 
			\begin{equation}
				x = \frac{1}{4} 
					\begin{pmatrix}
						0 & 0 & 0 \\
						1 & 0 & 0 \\
						0 & 1 & 0 \\
						0 & 0 & 1
					\end{pmatrix} + \frac{1}{4} 
					\begin{pmatrix}
						0 & 1 & 1 \\
						1 & 0 & 0 \\
						0 & 0 & 0 \\
						0 & 0 & 0
					\end{pmatrix} + \frac{1}{4} 
					\begin{pmatrix}
						1 & 0 & 1 \\
						0 & 0 & 0 \\
						0 & 1 & 0 \\
						0 & 0 & 0
					\end{pmatrix} + \frac{1}{4} 
					\begin{pmatrix}
						1 & 1 & 0 \\
						0 & 0 & 0 \\
						0 & 0 & 0 \\
						0 & 0 & 1
					\end{pmatrix}
			\end{equation}
			for $m = 2, k = 3$.
			We can also observe that there is no convex decomposition of $y$ with non-zero weight of rank $1$ encodings because there is no value of the data variable with which all hypothesis values are compatible (i.e.\ the matrix representation of $y$ has no rows with strictly positive values). 
			Finally, since we can write $y$ purely as a convex combination of rank two deterministic encodings via
			\begin{equation}
				y = \frac{1}{2} \begin{pmatrix}
						1 & 0 & 1 \\
						0 & 1 & 0 \\
						0 & 0 & 0
					\end{pmatrix} + \frac{1}{2} \begin{pmatrix}
						0 & 0 & 0 \\
						1 & 0 & 0 \\
						0 & 1 & 1
					\end{pmatrix},
			\end{equation}
			we conclude that $f_{2,3} (y)$ is zero.
			
			As we can see, $f_{1,3}$ has higher value for $y$, while $f_{2,3}$ has higher value for $x$. 
			That implies that neither $x$ can be converted to $y$ by stochastic post-processing nor $y$ can be converted to $x$.

	\section{Monotones for Channels From Monotones for States}\label{sec:dist_channels}
		
		Besides encodings of states (such as probability distributions), one can also study encodings of channels (such as stochastic maps).
		The resulting \emph{resource theory of channel distinguishability} can be used to understand asymmetry of channels via \cref{ex:bridge_lemma} (bridge lemma) as well as other aspects and applications of channel discrimination.
		Here, we briefly look at the resource theory of encodings of a hypothesis in stochastic processes to illustrate how monotones can be extended from a resource theory of states to resource theories of channels (cf.\ \cref{ex:Changing type}).
		While our discussion is in the context of distinguishability, one can take a similar approach to use the well-understood properties of quantum states (such as entanglement or athermality) to learn about the corresponding properties of quantum channels.
		
		The set-up is very much like what we describe in \cref{sec:dist_set-up}.
		Resources are tuples of stochastic maps, indexed by a set $H$ of hypothesis values.
		We can think of them as stochastic maps with a composite input.
		One part is the hypothesis variable and the complement is the actual input of the channel when the hypothesis is fixed.
		That is, they can be depicted as processes of the following type:
		\begin{equation}\label{eq:map_enc}
			\input{map_enc.tikz}
		\end{equation}
		Recall that the transformations are typically given by 1-combs (\cref{ex:rt_channels}).
		In our case, their diagrammatic representation with an explicit depiction of the hypothesis variable is
		\begin{equation}\label{eq:comb_enc}
			\input{comb_enc.tikz}	
		\end{equation}
		where $\rho, \sigma$ are stochastic maps of type $H \otimes C \to Z \otimes A$ and $H \otimes Z \otimes B \to D$ respectively.
		Therefore, $Z$ denotes an arbitrary variable mediating the side-channel communication between the pre- and post-processing here.
		Just as in \cref{eq:param_composition,eq:param_composition_2}, the composition of combs and their action on channels is given by supplying a copy of the hypothesis variable $H$ to all processes involved and contracting the relevant wires.
		
		Recall from \cref{ex:rt_channels} that there is not a unique choice of free transformations in a resource theory of channels.
		We can, for instance, impose that both pre- and post-processings are hypothesis-independent, in which case a generic free comb can be depicted as
		\begin{equation}\label{eq:comb_enc_free}
			\input{comb_enc_free.tikz}	
		\end{equation}
		where $\rho \colon C \to Z \otimes A$ and $\sigma \colon Z \otimes B \to D$ are arbitrary stochastic maps.
		
		Just like in \cref{ex:Changing type}, we can think of the monotones (for states) introduced in \cref{sec:dist_rank} as partial monotones for channels.
		The cost and yield constructions can be then used to extend their domain to all channels, resulting in monotones $\yield{f_{m,k}}{}$ and $\cost{f_{m,k}}{}$ in the resource theory of channel distinguishability.
		In general, these may be tricky to evaluate, as they involve two optimizations---one in the definition of the original monotone $f_{m,k}$ and another as part of the yield (or cost) construction.
		However, for specific classes of channel encodings, we can find their values easily.
		
		Consider any channel encoding $\psi$ represented by a stochastic map of type $H \otimes A \to B$, which is equivalent to a state encoding $z$ under the action of free (i.e.\ hypothesis-independent) 1-combs.
		Concretely, $\psi$ and $z$ should satisfy
		\begin{equation}
			\input{channel_state_eq.tikz}
		\end{equation}
		and
		\begin{equation}
			\input{channel_state_eq_2.tikz}
		\end{equation}
		for some stochastic maps $\rho$, $\sigma$, and $\sigma'$ of appropriate types.
		Given any monotone $f$ in the resource theory of state distinguishability (such as a weight monotone $f_{m,k}$), we have 
		\begin{equation}
			\yield{f}{}(z) = f(z) = \cost{f}{}(z)
		\end{equation}
		by \cref{prop:Extensions of monotones are the same on the original domain}.
		Moreover, a monotone in the resource theory of channel distinguishability (such as $\yield{f}{}$ or $\cost{f}{}$) has to assign the same value to both $\psi$ and $z$ since they are equivalent resources.
		As a result, we have 
		\begin{equation}\label{eq:cost_for_channels}
			\yield{f}{}(\psi) = f(z) = \cost{f}{}(\psi).
		\end{equation}
		
		The simplest kinds of channel encodings that are equivalent to state encodings are ones in which the input and output of the channel are independent for each hypothesis value.
		For example, consider
		\begin{equation}\label{eq:Shor_channels}
			\input{Shor_channels.tikz}
		\end{equation}
		where $x$ and $y$ are the Shor encodings from \cref{sec:dist_Shor_enc}, $B$ is the sample space $\bm{4} = \{0,1,2,3\}$, and $B'$ is the sample space $\bm{3} = \{0,1,2\}$.
		The channel encoding $\tilde{x}$ is an equivalent resource of distinguishability to the state encoding $x$, so that
		\begin{equation}
			\yield{f_{m,k}}{} \left( \tilde{x} \right) = f_{m,k} (x) = \cost{f_{m,k}}{} \left( \tilde{x} \right)
		\end{equation}
		holds for any of the weight monotones.
		Since $x$ and $y$ are incomparable in the resource theory of state distinguishability, we can also immediately conclude that $\tilde{x}$ and $\tilde{y}$ are incomparable in the resource theory of channel distinguishability.
		
		For a slightly less trivial example, consider channel encodings $\psi_x \colon H \otimes A \to B$ and $\psi_y \colon H \otimes A \to B'$ given by 
		\begin{equation}
			\begin{split}
				\psi_x(h,a) &= \frac{1}{2} [0] + \frac{1}{2} [a \oplus h] \\
				\psi_y(h,a) &= \frac{1}{2} [a \oplus h] + \frac{1}{2} [a \oplus h \oplus 2]
			\end{split}
		\end{equation}
		with $H = \{1,2,3\}$, $A = B' = \{0,1,2\}$, and $B = \{0,1,2,3\}$.
		Here $[b]$ denotes the delta distribution supported on $b \in B$ (or $b \in B'$) and the symbol $\oplus$ denotes modular addition within $B$ and $B'$ respectively. 
		Once again, we have $\psi_x \sim_{\rm unif} x$ and $\psi_y \sim_{\rm unif} y$ for the Shor encodings $x,y$ so that $\psi_x$ and $\psi_y$ are also incomparable. 
		We can prove these equivalences as follows.
		For $\psi_x \succeq_{\rm unif} x$ and $\psi_y \succeq_{\rm unif} y$, note that applying either $\psi_x$ or $\psi_y$ to the delta distribution $[0]$ on $A$ yields the corresponding Shor encoding:
		\begin{equation}
			\begin{split}
				\psi_x(h,0) &= \frac{1}{2} [0] + \frac{1}{2} [h] = x(h) \\
				\psi_y(h,0) &= \frac{1}{2} [h] + \frac{1}{2} [h \oplus 2] = y(h)
			\end{split}
		\end{equation}
		For the opposite conversions, $\psi_x \succeq_{\rm unif} x$ and $\psi_y \succeq_{\rm unif} y$, we give explicit post-processings $\sigma_x \colon B \otimes A \to B$ and  $\sigma_y \colon B' \otimes A \to B'$ so that we have
		\begin{equation}\label{eq:Shor_enc_to_channel}
			\input{Shor_enc_to_channel.tikz}
		\end{equation}
		and similarly for $y$.
		Specifically, both $\sigma_x$ and $\sigma_y$ can be chosen to be deterministic and given by
		\begin{equation}
			\begin{split}
				\sigma_x(b,a) &= \begin{cases} [0] & \text{if b = 0} \\ [b \oplus a] & \text{otherwise,} \end{cases} \\
				\sigma_y(b,a) &= [b \oplus a].
			\end{split}
		\end{equation}
		One can check that with this choice, \cref{eq:Shor_enc_to_channel} holds for both $x$ and $y$.
		Thus, by the advertized equivalences $\psi_x \sim_{\rm unif} x$ and $\psi_y \sim_{\rm unif} y$, values of any yield and cost monotones obtained from a monotone $f$ for state encodings satisfy
		\begin{equation}
			\begin{split}
				\yield{f}{}(\psi_x) &= f(x) = \cost{f}{}(\psi_x) \\
				\yield{f}{}(\psi_y) &= f(y) = \cost{f}{}(\psi_y).
			\end{split}
		\end{equation}
		
		In general, as we mentioned, evaluating these yield and cost monotones requires solving an optimization problem. 
		However, in many contexts, it is not an unbounded optimization as \cref{eq:state_to_channel2} might suggest at first sight.
		In the case of a resource theory of distinguishability, \cref{eq:state_to_channel2} reads
		\begin{equation}
			\input{state_to_channel4.tikz}
		\end{equation}
		Note that we can write the bipartite state $\rho \colon I \to A \otimes Z$ as a product of its marginal distribution $\mu$ on $A$ and the conditional distribution of $Z$ given $A$ which specifies a stochastic map $\upsilon \colon A \to Z$ via
		\begin{equation}
			\input{conditional.tikz}
		\end{equation}
		Since we assume $f$ to be a monotone on its domain, we can remove $\upsilon$ from the optimization, so that it reduces to
		\begin{equation}\label{eq:state_to_channel6}
			\input{state_to_channel6.tikz}
		\end{equation}
		and the optimization is thus restricted to distributions over a fixed sample space $A$---the input system of the channel encoding $\psi$.

\chapter{Future Directions}\label{sec:future}
	
	\subsubsection{Landscape of resource theories}

	In \cref{sec:resource_theories}, we introduced a somewhat combinatorial framework for resource theories.
	It incorporates structure to model the degree of access to resources that an agent may have, which is often used implicitly when reasoning about resource conversions.
	The abstract nature of the framework allows for many quite disparate examples as instances.
	All of the commonly studied resource theories postulate individual resources as mutually exclusive entities and are thus uniquely atomistic lattices in our framework.
	We believe that shifting the perspective to more general lattices can lead to new applications of resource-theoretic ideas.
	There are many possibilities, even though the details may be unclear at this point.
	Let us mention algorithmic information and complexity theory, and evolutionary dynamics, to name a few.
	
	Additionally, the abstract point of view is suitable to analyze connections between resource theories that may look quite different when described in more concrete terms.
	This is the main purpose of the abstraction in the thesis, as we try to systematically understand how resource orderings can be studied via resource measures in \cref{sec:monotones}.
	We take some of the common ways that resources are evaluated in practice and frame them as instances of general constructions that can be applied with fairly minimal assumptions.
	In this way, connections between resource theories can be established and some methods for studying them can be translated.
	Moreover, fitting these into the \ref{broad scheme} for modularizing monotone constructions, we can both characterize existing resource measures as well as develop a systematic procedure to come up with new ones.
	
	Besides the construction of resource measures, there are other aspects of the ``landscape of resource theories'' worth investigating. 
	As we briefly sketch in \cref{sec:ucrt_free}, general frameworks can be used to understand the interactions of restrictions placed on the free transformations.
	For instance, one may study emergent consequences of a conjunction of assumptions about ``what is allowed''.
	Yet another topic in this domain is that of trade-offs between different restrictions on the free operations, such as trade-offs between asymmetry properties for different group actions, recently also investigated in \cite{ying2021asymmetry}.
	
	\subsubsection{Additional structure}
	
	It should be mentioned that the framework used here is stripped of a lot of structure one may presume in applications.
	While this is intentional and makes the ideas presented more transparent and easier to describe, it limits the inferences and methods one can accommodate.
	It would certainly be interesting to add more structure to the framework presented here and study the new phenomena that arise.
	
	One of the main pieces missing from our analysis here is the concept of resource types. 
	In practice, one can often make several assumptions about the nature of the resource, even in situations in which the precise behaviour is uncertain.
	In physics and engineering, these assumptions can be collated into what we call a ``system'' (such as a particle, a collection of interacting particles, a mechanical device, etc.).
	Other examples include registers in computer science, molecules in chemistry, chromosomes and genes in genetics, and these merely scratch the surface of course.
	Concrete transformations of resources will likewise come with a notion of type.
	If $\alpha, \beta$, etc.\ denote types of resources, a type of transformation is of the form $\alpha \to \beta$, meaning that it transforms resources of type $\alpha$ to resources of type $\beta$.
	Crucially, its behaviour is unspecified when acting on resources that are not of type $\alpha$.
	For example, a transformation that is afforded by a prism specifies what happens to resources that are beams of light, but it says nothing about what happens to beams of electrons.
	On the basis that transformation of type $\alpha \to \beta$ can \emph{only} transform resources of type $\alpha$ to resources of type $\beta$, they allow for high-level reasoning about impossibility.
	
	One way in which types can be incorporated in the mathematical description is to model types as objects in a category. 
	This is precisely what process theories in \cref{sec:proc_th} do.
	In our context, incorporating types into the description of resource theories as presented may be conceived as horizontal categorification.
	While all of the details are not clear to us at present, the most immediate approach would correspond to replacing quantale modules with quantaloid modules \cite{stubbe2008q,rosenthal2014theory} as described in \cite{Abramsky1993} as well.
	Another, closely related, approach is modelling resource theories in the context of quantale-enriched categories \cite{marsden2018quantitative}.
	
	Yet another piece of structure one would generally like to have in their toolkit is convex or linear structure more broadly.
	We presented a rudimentary discussion of these in \cref{sec:rt_convex}.
	Convex structure and barycentric calculus become important when the very description of individual resources refers to (probabilistic) beliefs about their identity.
	As such, resources of information (described probabilistically) as well as quantum resources carry convex-linear structure.
	Indeed, we have seen in \cref{sec:weight and robustness} that some of the most popular resource measures rely on convexity.
	More generally, one can then employ the tools of convex geometry to study resource theories \cite{Regula2017,uola2019quantifying}.
	
	While universally combinable resource theories (\cref{sec:ucrt}) model a conjunction of resources via the universal composition $\boxtimes$, in the general framework of \cref{sec:rt} no such operation is assumed.
	Nevertheless, many resource-theoretic questions involve conjunctions of resources in one way or another.
	As a toy example, consider the case of stoichiometry in chemistry, where one is interested to model which conjunctions (and ratios) of reactants can yield desired products in a chemical reaction such as
	\begin{equation}
		2 \mathrm{Na} + \mathrm{Cl_2} \longrightarrow 2 \mathrm{NaCl}.
	\end{equation}
	Therefore, it would be desirable to equip our framework with the structure of a ``tensor product'', such as that of symmetric monoidal categories used in process theories.
	Such notion is also necessary for questions of catalytic and asymptotic many-copy conversions, as explored in \cite{Fritz2017}.
	
	A large fraction of the research of resource theories focuses on whether resources can be converted approximately.
	The main reason is that while we often use idealized descriptions of resources as elements of a set, in practice one can never achieve perfect precision and thus we loosen the idealized description to allow for `errors'.
	Additionally, approximations are relevant for developing a theory of conversions between many copies of the same resource (\cref{rem:iid}).
	This loosening of the requirement of exact convertibility is commonly achieved by equipping resources with metric structure specifying a notion of how close they are.
	Alternatively, one may choose to use a non-idealized description to begin with.
	In this case, we may think of resources as elements of a ``point-free topological space'', such as a locale.
	However, to derive quantitative statements about errors and approximate conversions, additional metric structure would be needed.
	
	\subsubsection{Related ideas}
	
	With or without additional structure, we would like to establish tighter links to some other more structured approaches to resource theories and related ideas in the future.
	Throughout this manuscript, we highlighted the connections to ordered commutative monoids \cite{Fritz2017}, partitioned process theories \cite{Coecke2016}, resource theories of knowledge \cite{DelRio2015} and others.
	The discussion of connections of resource theories to constructor theory \cite{deutsch2013constructor,marletto2016constructor_therm,marletto2016constructor_prob,deutsch2015constructor} and linear logic \cite{Girard1995} from \cite[section 10]{Fritz2017} applies to our framework as well.
	Let us mention, however, that quantales specifically have been used to describe the semantics of linear logic in \cite{yetter1990quantales}.
	Yet another model of linear logic is provided by Petri nets \cite{marti1989petri,engberg1990petri}, which have been connected more explicitly to quantales in \cite{brown1989petri,Tsinakis2004}.
	Given the common resource-like interpretation of markings and transitions on a Petri net, this is a connection that should certainly be pursued, deepening the ties between resource theories and theoretical computer science. 
	Finite-state automata and labelled transition systems offer yet another related flavour of computation, the latter of which has been modelled by quantale modules explicitly in \cite{Abramsky1993}.
	Often, resources are also conceptualized in terms of their utility in a game-theoretic sense.
	As such, recent developments in compositional game theory \cite{Gutoski2007,hedges2016towards,ghani2018compositional} may offer yet another connection to explore.
	All that is to say that one of the outcomes that we hope to achieve by studying resource theories in the abstract setting is to facilitate and make concrete the connections with investigations in related fields.



\bibliographystyle{plain}

\clearpage  			

\phantomsection  
\renewcommand*{\bibname}{References}

\addcontentsline{toc}{chapter}{\textbf{References}}

\bibliography{references}

\begin{thebibliography}{100}

\bibitem{Abramsky2017}
Samson Abramsky, Rui~Soares Barbosa, and Shane Mansfield.
\newblock
  \href{https://journals.aps.org/prl/abstract/10.1103/PhysRevLett.119.050504}{Contextual
  fraction as a measure of contextuality}.
\newblock {\em Physical review letters}, 119(5):050504, 2017.

\bibitem{Abramsky1993}
Samson Abramsky and Steven Vickers.
\newblock \href{https://doi.org/10.1017/S0960129500000189}{Quantales,
  observational logic and process semantics}.
\newblock {\em Mathematical structures in computer science}, 3(2):161--227,
  1993.

\bibitem{arnold1987majorization}
Barry~C. Arnold.
\newblock {\em Majorization and the Lorenz Order: A Brief Introduction}.
\newblock Springer-Verlag, 1987.

\bibitem{arnold2018majorization}
Barry~C Arnold and Jos{\'e}~Mar{\'\i}a Sarabia.
\newblock {\em Majorization and the Lorenz order with applications in applied
  mathematics and economics}.
\newblock Springer, 2018.

\bibitem{Barrett2006}
Jonathan Barrett, Adrian Kent, and Stefano Pironio.
\newblock
  \href{https://journals.aps.org/prl/abstract/10.1103/PhysRevLett.97.170409}{Maximally
  nonlocal and monogamous quantum correlations}.
\newblock {\em Physical review letters}, 97(17):170409, 2006.

\bibitem{bartlett2007reference}
Stephen~D. Bartlett, Terry Rudolph, and Robert~W. Spekkens.
\newblock
  \href{https://journals.aps.org/rmp/abstract/10.1103/RevModPhys.79.555}{Reference
  frames, superselection rules, and quantum information}.
\newblock {\em Reviews of Modern Physics}, 79(2):555, 2007.

\bibitem{bell1964einstein}
John~S. Bell.
\newblock \href{https://doi.org/10.1103/PhysicsPhysiqueFizika.1.195}{On the
  Einstein Podolsky Rosen Paradox}.
\newblock {\em Physics Physique Fizika}, 1(3):195, 1964.

\bibitem{bell1966problem}
John~S. Bell.
\newblock \href{https://doi.org/10.1103/RevModPhys.38.447}{On the Problem of
  Hidden Variables in Quantum Mechanics}.
\newblock {\em Reviews of Modern Physics}, 38(3):447, 1966.

\bibitem{bennett1993teleporting}
Charles~H. Bennett, Gilles Brassard, Claude Cr{\'e}peau, Richard Jozsa, Asher
  Peres, and William~K. Wootters.
\newblock \href{https://doi.org/10.1103/PhysRevLett.70.1895}{Teleporting an
  unknown quantum state via dual classical and Einstein-Podolsky-Rosen
  channels}.
\newblock {\em Physical review letters}, 70(13):1895, 1993.

\bibitem{Bennett1996}
Charles~H. Bennett, David~P. DiVincenzo, John~A. Smolin, and William~K.
  Wootters.
\newblock
  \href{https://journals.aps.org/pra/abstract/10.1103/PhysRevA.54.3824}{Mixed-state
  entanglement and quantum error correction}.
\newblock {\em Physical Review A}, 54(5):3824, 1996.

\bibitem{bennett1992communication}
Charles~H. Bennett and Stephen~J. Wiesner.
\newblock \href{https://doi.org/10.1103/PhysRevLett.69.2881}{Communication via
  one-and two-particle operators on Einstein-Podolsky-Rosen states}.
\newblock {\em Physical review letters}, 69(20):2881, 1992.

\bibitem{blackwell1951comparison}
David Blackwell.
\newblock \href{https://projecteuclid.org/euclid.bsmsp/1200500222}{Comparison
  of Experiments}.
\newblock In {\em Proceedings of the Second Berkeley Symposium on Mathematical
  Statistics and Probability}, pages 93--102. University of California Press,
  1951.

\bibitem{Brandao2015}
Fernando Brandão, Micha{\l} Horodecki, Nelly Ng, Jonathan Oppenheim, and
  Stephanie Wehner.
\newblock \href{https://www.pnas.org/content/112/11/3275}{The second laws of
  quantum thermodynamics}.
\newblock {\em Proceedings of the National Academy of Sciences},
  112(11):3275--3279, 2015.

\bibitem{Brandao2013}
Fernando Brandão, Micha{\l} Horodecki, Jonathan Oppenheim, Joseph~M. Renes,
  and Robert~W. Spekkens.
\newblock \href{https://doi.org/10.1103/PhysRevLett.111.250404}{Resource Theory
  of Quantum States Out of Thermal Equilibrium}.
\newblock {\em Physical review letters}, 111(25):250404, 2013.

\bibitem{brandao2015reversible}
Fernando~G.S.L. Brandão and Gilad Gour.
\newblock
  \href{https://journals.aps.org/prl/abstract/10.1103/PhysRevLett.115.070503}{Reversible
  framework for quantum resource theories}.
\newblock {\em Physical review letters}, 115(7):070503, 2015.

\bibitem{brown1989petri}
Carolyn Brown.
\newblock {\em
  \href{http://www.lfcs.inf.ed.ac.uk/reports/89/ECS-LFCS-89-96/ECS-LFCS-89-96.pdf}{Petri
  nets as quantales}}.
\newblock University of Edinburgh, Department of Computer Science, 1989.

\bibitem{buscemi2012all}
Francesco Buscemi.
\newblock \href{https://doi.org/10.1103/PhysRevLett.108.200401}{All entangled
  quantum states are nonlocal}.
\newblock {\em Physical review letters}, 108(20):200401, 2012.

\bibitem{chefles2004existence}
Anthony Chefles, Richard Jozsa, and Andreas Winter.
\newblock \href{https://arxiv.org/abs/quant-ph/0307227}{On the existence of
  physical transformations between sets of quantum states}.
\newblock {\em International Journal of Quantum Information}, 2(01):11--21,
  2004.

\bibitem{chiribella2009theoretical}
Giulio Chiribella, Giacomo~Mauro D’Ariano, and Paolo Perinotti.
\newblock
  \href{https://journals.aps.org/pra/abstract/10.1103/PhysRevA.80.022339}{Theoretical
  framework for quantum networks}.
\newblock {\em Physical Review A}, 80(2):022339, 2009.

\bibitem{Chitambar2018}
Eric Chitambar and Gilad Gour.
\newblock \href{https://arxiv.org/abs/1806.06107}{Quantum resource theories}.
\newblock {\em Reviews of Modern Physics}, 91(2):025001, 2019.

\bibitem{Chitambar2014}
Eric Chitambar, Debbie Leung, Laura Man{\v{c}}inska, Maris Ozols, and Andreas
  Winter.
\newblock \href{https://doi.org/10.1007/s00220-014-1953-9}{Everything You
  Always Wanted to Know About LOCC (But Were Afraid to Ask)}.
\newblock {\em Communications in Mathematical Physics}, 328(1):303--326, May
  2014.

\bibitem{Coecke2016}
Bob Coecke, Tobias Fritz, and Robert~W. Spekkens.
\newblock
  {\href{http://www.sciencedirect.com/science/article/pii/S0890540116000353}{A
  mathematical theory of resources}}.
\newblock {\em Information and Computation}, 250:59 -- 86, 2016.

\bibitem{coecke2015categorical}
Bob Coecke and Aleks Kissinger.
\newblock \href{https://arxiv.org/abs/1510.05468}{Categorical quantum mechanics
  I: causal quantum processes}.
\newblock {\em Categories for the Working Philosopher}, pages 286--328, 2015.

\bibitem{coecke2018picturing}
Bob Coecke and Aleks Kissinger.
\newblock Picturing quantum processes.
\newblock In {\em International Conference on Theory and Application of
  Diagrams}, pages 28--31. Springer, 2018.

\bibitem{Coecke2010}
Bob Coecke and Keye Martin.
\newblock
  \href{https://link.springer.com/chapter/10.1007/978-3-642-12821-9_10}{A
  partial order on classical and quantum states}.
\newblock In {\em New Structures for Physics}, pages 593--683. Springer, 2010.

\bibitem{Coecke2009}
Bob {Coecke} and Eric~O. {Paquette}.
\newblock \href{https://arxiv.org/abs/0905.3010}{Categories for the practising
  physicist}.
\newblock {\em arXiv:0905.3010}, 2009.

\bibitem{Dahl1999}
Geir Dahl.
\newblock
  \href{http://www.sciencedirect.com/science/article/pii/S0024379598101751}{Matrix
  majorization}.
\newblock {\em Linear Algebra and its Applications}, 288:53 -- 73, 1999.

\bibitem{DelRio2015}
L{\'\i}dia Del~Rio, Lea Kraemer, and Renato Renner.
\newblock \href{https://arxiv.org/abs/1511.08818}{Resource theories of
  knowledge}.
\newblock {\em arXiv:1511.08818}, 2015.

\bibitem{deutsch2013constructor}
David Deutsch.
\newblock \href{https://doi.org/10.1007/s11229-013-0279-z}{Constructor theory}.
\newblock {\em Synthese}, 190(18):4331--4359, 2013.

\bibitem{deutsch2015constructor}
David Deutsch and Chiara Marletto.
\newblock \href{https://doi.org/10.1098/rspa.2014.0540}{Constructor theory of
  information}.
\newblock {\em Proceedings of the Royal Society A: Mathematical, Physical and
  Engineering Sciences}, 471(2174):20140540, 2015.

\bibitem{Devetak2008}
Igor Devetak, Aram~W. Harrow, and Andreas~J. Winter.
\newblock \href{https://ieeexplore.ieee.org/document/4626055}{A resource
  framework for quantum Shannon theory}.
\newblock {\em IEEE Transactions on Information Theory}, 54(10):4587--4618,
  2008.

\bibitem{Ducuara2019}
Andr{\'e}s~F. Ducuara and Paul Skrzypczyk.
\newblock \href{https://arxiv.org/abs/1908.10347}{Weight of informativeness,
  state exclusion games and excludible information}.
\newblock {\em arXiv:1908.10347}, 2019.

\bibitem{eklund2018semigroups}
Patrik Eklund, Javier~Guti{\'e}rrez Garc{\'\i}a, Ulrich H{\"o}hle, and Jari
  Kortelainen.
\newblock {\em Semigroups in complete lattices: quantales, modules and related
  topics}, volume~54.
\newblock Springer, 2018.

\bibitem{engberg1990petri}
Uffe Engberg and Glynn Winskel.
\newblock
  \href{http://citeseerx.ist.psu.edu/viewdoc/download?doi=10.1.1.37.3821&rep=rep1&type=pdf}{Petri
  nets as models of linear logic}.
\newblock In {\em Colloquium on Trees in Algebra and Programming}, pages
  147--161. Springer, 1990.

\bibitem{fong2018seven}
Brendan Fong and David~I. Spivak.
\newblock {\em \href{https://arxiv.org/abs/1803.05316}{An invitation to applied
  category theory: seven sketches in compositionality}}.
\newblock Cambridge University Press, 2019.

\bibitem{fraser2020functoriality}
Patrick Fraser.
\newblock \href{https://arxiv.org/abs/2006.16350}{Functoriality of Quantum
  Resource Theory and Variable-Domain Modal Logic}.
\newblock {\em arXiv:2006.16350}, 2020.

\bibitem{Fritz2017}
Tobias Fritz.
\newblock
  \href{https://www.cambridge.org/core/services/aop-cambridge-core/content/view/C0D57AE7683BAA911F3F270ABD82A4E9/S0960129515000444a.pdf}{Resource
  convertibility and ordered commutative monoids}.
\newblock {\em Mathematical Structures in Computer Science}, 27(6):850–938,
  2017.

\bibitem{fritz2019synthetic}
Tobias Fritz.
\newblock \href{https://arxiv.org/abs/1908.07021}{A synthetic approach to
  Markov kernels, conditional independence, and theorems on sufficient
  statistics}.
\newblock {\em arXiv:1908.07021}, 2019.

\bibitem{fritz2020local}
Tobias Fritz.
\newblock \href{https://arxiv.org/abs/2003.13835}{A local-global principle for
  preordered semirings and abstract Positivstellensätze}.
\newblock {\em arXiv:2003.13835}, 2020.

\bibitem{fritz2021finetti}
Tobias Fritz, Tom{\'a}{\v{s}} Gonda, and Paolo Perrone.
\newblock \href{https://arxiv.org/abs/2105.02639}{De Finetti's Theorem in
  Categorical Probability}.
\newblock {\em arXiv:2105.02639}, 2021.

\bibitem{fritz2020representable}
Tobias Fritz, Tom{\'a}{\v{s}} Gonda, Paolo Perrone, and Eigil~Fjeldgren
  Rischel.
\newblock \href{https://arxiv.org/abs/2010.07416}{Representable Markov
  Categories and Comparison of Statistical Experiments in Categorical
  Probability}.
\newblock {\em arXiv:2010.07416}, 2020.

\bibitem{ghani2018compositional}
Neil Ghani, Jules Hedges, Viktor Winschel, and Philipp Zahn.
\newblock \href{https://arxiv.org/abs/1603.04641}{Compositional game theory}.
\newblock In {\em Proceedings of the 33rd Annual ACM/IEEE Symposium on Logic in
  Computer Science}, pages 472--481, 2018.

\bibitem{Girard1995}
Jean-Yves Girard.
\newblock \href{http://dl.acm.org/citation.cfm?id=212876.212880}{Linear Logic:
  Its Syntax and Semantics}.
\newblock In {\em Proceedings of the Workshop on Advances in Linear Logic},
  pages 1--42, New York, NY, USA, 1995. Cambridge University Press.

\bibitem{Monotones}
Tomáš Gonda and Robert~W. Spekkens.
\newblock \href{https://arxiv.org/abs/1912.07085}{Monotones in General Resource
  Theories}.
\newblock {\em arXiv:1912.07085}, 2019.

\bibitem{gonda2018almost}
Tom{\'a}{\v{s}} Gonda, Ravi Kunjwal, David Schmid, Elie Wolfe, and
  Ana~Bel{\'e}n Sainz.
\newblock \href{https://quantum-journal.org/papers/q-2018-08-27-87/}{Almost
  Quantum Correlations are Inconsistent with Specker's Principle}.
\newblock {\em Quantum}, 2:87, 2018.

\bibitem{Gour2009}
Gilad Gour, Iman Marvian, and Robert~W. Spekkens.
\newblock
  \href{https://journals.aps.org/pra/abstract/10.1103/PhysRevA.80.012307}{Measuring
  the quality of a quantum reference frame: The relative entropy of frameness}.
\newblock {\em Physical Review A}, 80(1):012307, 2009.

\bibitem{Gour2015}
Gilad Gour, Markus~P. M{\"u}ller, Varun Narasimhachar, Robert~W. Spekkens, and
  Nicole~Y. Halpern.
\newblock
  \href{http://www.sciencedirect.com/science/article/pii/S037015731500229X}{The
  resource theory of informational nonequilibrium in thermodynamics}.
\newblock {\em Physics Reports}, 583:1 -- 58, 2015.

\bibitem{gour2009time}
Gilad Gour, Barry~C. Sanders, and Peter~S. Turner.
\newblock \href{https://doi.org/10.1063/1.3243821}{Time-reversal frameness and
  superselection}.
\newblock {\em Journal of mathematical physics}, 50(10):102105, 2009.

\bibitem{gour2020optimal}
Gilad Gour and Marco Tomamichel.
\newblock
  \href{https://journals.aps.org/pra/abstract/10.1103/PhysRevA.102.062401}{Optimal
  extensions of resource measures and their applications}.
\newblock {\em Physical Review A}, 102(6):062401, 2020.

\bibitem{gour2019quantify}
Gilad Gour and Andreas Winter.
\newblock \href{https://arxiv.org/abs/1906.03517}{How to quantify a dynamical
  quantum resource}.
\newblock {\em Physical review letters}, 123(15):150401, 2019.

\bibitem{Gutoski2007}
Gus Gutoski and John Watrous.
\newblock \href{https://arxiv.org/abs/quant-ph/0611234}{Toward a general theory
  of quantum games}.
\newblock In {\em Proceedings of the thirty-ninth annual ACM symposium on
  Theory of computing}, pages 565--574. ACM, 2007.

\bibitem{harrow2003robustness}
Aram~W. Harrow and Michael~A. Nielsen.
\newblock
  \href{https://journals.aps.org/pra/abstract/10.1103/PhysRevA.68.012308}{Robustness
  of quantum gates in the presence of noise}.
\newblock {\em Physical Review A}, 68(1):012308, 2003.

\bibitem{hedges2016towards}
Julian~M. Hedges.
\newblock {\em
  \href{https://qmro.qmul.ac.uk/xmlui/bitstream/handle/123456789/23259/HEDGES_Julian_PhD031016.pdf}{Towards
  compositional game theory}}.
\newblock PhD thesis, Queen Mary University of London, 2016.

\bibitem{hickey2018quantifying}
Alexander Hickey and Gilad Gour.
\newblock \href{https://doi.org/10.1088/1751-8121/aabe9c}{Quantifying the
  imaginarity of quantum mechanics}.
\newblock {\em Journal of Physics A: Mathematical and Theoretical},
  51(41):414009, 2018.

\bibitem{Horodecki2003}
Micha{\l} Horodecki, Pawe{\l} Horodecki, and Jonathan Oppenheim.
\newblock
  \href{https://journals.aps.org/pra/abstract/10.1103/PhysRevA.67.062104}{Reversible
  transformations from pure to mixed states and the unique measure of
  information}.
\newblock {\em Physical Review A}, 67(6):062104, 2003.

\bibitem{Horodecki2013}
Micha{\l} Horodecki and Jonathan Oppenheim.
\newblock \href{https://arxiv.org/abs/1111.3834}{Fundamental limitations for
  quantum and nanoscale thermodynamics}.
\newblock {\em Nature communications}, 4:2059, 2013.

\bibitem{Horodecki2009}
Ryszard {Horodecki}, Pawe{\l} {Horodecki}, Micha{\l} {Horodecki}, and Karol
  {Horodecki}.
\newblock \href{https://arxiv.org/abs/quant-ph/0702225v2}{Quantum
  entanglement}.
\newblock {\em Reviews of Modern Physics}, 81:865--942, April 2009.

\bibitem{houghton2021mathematical}
Nicholas~Gauguin Houghton-Larsen.
\newblock \href{https://arxiv.org/abs/2103.02302}{A Mathematical Framework for
  Causally Structured Dilations and its Relation to Quantum Self-Testing}.
\newblock {\em arXiv:2103.02302}, 2021.

\bibitem{howard2017application}
Mark Howard and Earl Campbell.
\newblock
  \href{https://journals.aps.org/prl/abstract/10.1103/PhysRevLett.118.090501}{Application
  of a resource theory for magic states to fault-tolerant quantum computing}.
\newblock {\em Physical review letters}, 118(9):090501, 2017.

\bibitem{Janzing2000}
Dominik Janzing, Pawel Wocjan, Robert Zeier, Rubino Geiss, and Thomas Beth.
\newblock \href{https://arxiv.org/abs/quant-ph/0002048}{Thermodynamic cost of
  reliability and low temperatures: tightening Landauer's principle and the
  second law}.
\newblock {\em International Journal of Theoretical Physics},
  39(12):2717--2753, 2000.

\bibitem{johnston2018evaluating}
Nathaniel Johnston, Chi-Kwong Li, Sarah Plosker, Yiu-Tung Poon, and Bartosz
  Regula.
\newblock \href{https://arxiv.org/abs/1806.00653}{Evaluating the robustness of
  k-coherence and k-entanglement}.
\newblock {\em Physical Review A}, 98(2):022328, 2018.

\bibitem{jozsa2000distinguishability}
Richard Jozsa and J{\"u}rgen Schlienz.
\newblock \href{https://doi.org/10.1103/PhysRevA.62.012301}{Distinguishability
  of states and von Neumann entropy}.
\newblock {\em Physical Review A}, 62(1):012301, 2000.

\bibitem{katariya2020evaluating}
Vishal Katariya and Mark~M. Wilde.
\newblock \href{https://arxiv.org/abs/2001.05376}{Evaluating the advantage of
  adaptive strategies for quantum channel distinguishability}.
\newblock {\em arXiv:2001.05376}, 2020.

\bibitem{killoran2016converting}
Nathan Killoran, Frank E.~S. Steinhoff, and Martin~B. Plenio.
\newblock
  \href{https://journals.aps.org/prl/abstract/10.1103/PhysRevLett.116.080402}{Converting
  nonclassicality into entanglement}.
\newblock {\em Physical review letters}, 116(8):080402, 2016.

\bibitem{kissinger2017categorical}
Aleks Kissinger and Sander Uijlen.
\newblock \href{https://arxiv.org/abs/1701.04732}{A categorical semantics for
  causal structure}.
\newblock In {\em 2017 32nd Annual ACM/IEEE Symposium on Logic in Computer
  Science (LICS)}, pages 1--12. IEEE, 2017.

\bibitem{Kraemer2016}
Lea Kraemer and L{\'\i}dia del Rio.
\newblock \href{https://arxiv.org/abs/1605.01064}{Currencies in resource
  theories}.
\newblock {\em arXiv:1605.01064}, 2016.

\bibitem{kristjansson2020resource}
Hl{\'e}r Kristj{\'a}nsson, Giulio Chiribella, Sina Salek, Daniel Ebler, and
  Matthew Wilson.
\newblock
  \href{https://iopscience.iop.org/article/10.1088/1367-2630/ab8ef7/meta}{Resource
  theories of communication}.
\newblock {\em New Journal of Physics}, 22(7):073014, 2020.

\bibitem{lami2018gaussian}
Ludovico Lami, Bartosz Regula, Xin Wang, Rosanna Nichols, Andreas Winter, and
  Gerardo Adesso.
\newblock
  \href{https://journals.aps.org/pra/abstract/10.1103/PhysRevA.98.022335}{Gaussian
  quantum resource theories}.
\newblock {\em Physical Review A}, 98(2):022335, 2018.

\bibitem{landauer1991information}
Rolf Landauer et~al.
\newblock \href{http://dx.doi.org/10.1063/1.881299}{Information is physical}.
\newblock {\em Physics Today}, 44(5):23--29, 1991.

\bibitem{le1996comparison}
Lucien Le~Cam.
\newblock \href{https://www.jstor.org/stable/4355942?seq=1}{Comparison of
  experiments: A short review}.
\newblock {\em Lecture Notes-Monograph Series}, pages 127--138, 1996.

\bibitem{Leditzky2018}
Felix Leditzky, Eneet Kaur, Nilanjana Datta, and Mark~M. Wilde.
\newblock \href{https://link.aps.org/doi/10.1103/PhysRevA.97.012332}{Approaches
  for approximate additivity of the Holevo information of quantum channels}.
\newblock {\em Phys. Rev. A}, 97:012332, Jan 2018.

\bibitem{lewenstein1998separability}
Maciej Lewenstein and Anna Sanpera.
\newblock \href{https://doi.org/10.1103/PhysRevLett.80.2261}{Separability and
  entanglement of composite quantum systems}.
\newblock {\em Physical review letters}, 80(11):2261, 1998.

\bibitem{Li2008}
Ming Li and Paul Vit{\'a}nyi.
\newblock {\em An introduction to Kolmogorov complexity and its applications},
  volume~3.
\newblock Springer, 2008.

\bibitem{Lieb1999}
Elliott~H. Lieb and Jakob Yngvason.
\newblock \href{https://arxiv.org/abs/cond-mat/9708200}{The physics and
  mathematics of the second law of thermodynamics}.
\newblock {\em Physics Reports}, 310(1):1--96, 1999.

\bibitem{Lieb2013}
Elliott~H. Lieb and Jakob Yngvason.
\newblock \href{https://doi.org/10.1098/rspa.2013.0408}{The entropy concept for
  non-equilibrium states}.
\newblock {\em Proceedings of the Royal Society A: Mathematical, Physical and
  Engineering Sciences}, 469(2158):20130408, 2013.

\bibitem{liu2020operational}
Yunchao Liu and Xiao Yuan.
\newblock \href{https://arxiv.org/abs/1904.02680}{Operational resource theory
  of quantum channels}.
\newblock {\em Physical Review Research}, 2(1):012035, 2020.

\bibitem{Liu2019}
Zi-Wen Liu and Andreas Winter.
\newblock \href{https://arxiv.org/abs/1904.04201}{Resource theories of quantum
  channels and the universal role of resource erasure}.
\newblock {\em arXiv:1904.04201}, 2019.

\bibitem{lorenz1905methods}
Max~O. Lorenz.
\newblock \href{https://doi.org/10.2307/2276207}{Methods of measuring the
  concentration of wealth}.
\newblock {\em Publications of the American statistical association},
  9(70):209--219, 1905.

\bibitem{lostaglio2019introductory}
Matteo Lostaglio.
\newblock \href{https://arxiv.org/abs/1807.11549}{An introductory review of the
  resource theory approach to thermodynamics}.
\newblock {\em Reports on Progress in Physics}, 82(11):114001, 2019.

\bibitem{marletto2016constructor_therm}
Chiara Marletto.
\newblock \href{https://arxiv.org/abs/1608.02625}{Constructor Theory of
  Thermodynamics}.
\newblock {\em arXiv:1608.02625}, 2016.

\bibitem{marletto2016constructor_prob}
Chiara Marletto.
\newblock \href{https://doi.org/10.1098/rspa.2015.0883}{Constructor theory of
  probability}.
\newblock {\em Proceedings of the Royal Society A: Mathematical, Physical and
  Engineering Sciences}, 472(2192):20150883, 2016.

\bibitem{marsden2018quantitative}
Dan Marsden and Maaike Zwart.
\newblock
  \href{https://drops.dagstuhl.de/opus/volltexte/2018/9699/}{Quantitative
  Foundations for Resource Theories}.
\newblock In {\em 27th EACSL Annual Conference on Computer Science Logic},
  volume 119 of {\em Leibniz International Proceedings in Informatics}, pages
  32:1--32:17, 2018.

\bibitem{Marshall1979}
Albert~W. Marshall, Ingram Olkin, and Barry~C. Arnold.
\newblock {\em \href{https://doi.org/10.1007/978-0-387-68276-1}{Inequalities:
  Theory of Majorization and Its Applications}}.
\newblock Springer Series in Statistics. Springer, New York, NY, second edition
  edition, 2011.

\bibitem{marti1989petri}
Narciso Marti-Oliet and Jos{\'e} Meseguer.
\newblock \href{https://link.springer.com/chapter/10.1007/BFb0018359}{From
  Petri nets to linear logic}.
\newblock In {\em Category Theory and Computer Science}, pages 313--340.
  Springer, 1989.

\bibitem{Marvian2012}
Iman Marvian.
\newblock {\em \href{https://uwspace.uwaterloo.ca/handle/10012/7088}{Symmetry,
  asymmetry and quantum information}}.
\newblock PhD thesis, University of Waterloo, 2012.

\bibitem{Marvian2013}
Iman Marvian and Robert~W. Spekkens.
\newblock \href{http://stacks.iop.org/1367-2630/15/i=3/a=033001}{The theory of
  manipulations of pure state asymmetry: I. Basic tools, equivalence classes
  and single copy transformations}.
\newblock {\em New Journal of Physics}, 15(3):033001, 2013.

\bibitem{Marvian2014}
Iman Marvian and Robert~W. Spekkens.
\newblock \href{https://www.nature.com/articles/ncomms4821}{Extending
  Noether’s theorem by quantifying the asymmetry of quantum states}.
\newblock {\em Nature communications}, 5:3821, 2014.

\bibitem{marvian2016quantify}
Iman Marvian and Robert~W. Spekkens.
\newblock
  \href{https://journals.aps.org/pra/abstract/10.1103/PhysRevA.94.052324}{How
  to quantify coherence: Distinguishing speakable and unspeakable notions}.
\newblock {\em Physical Review A}, 94(5):052324, 2016.

\bibitem{mccullagh2002statistical}
Peter McCullagh.
\newblock \href{https://doi.org/10.1214/aos/1035844977}{What is a statistical
  model?}
\newblock {\em The Annals of Statistics}, 30(5):1225--1310, 2002.

\bibitem{methot2006anomaly}
Andr{\'e}~Allan M{\'e}thot and Valerio Scarani.
\newblock \href{https://arxiv.org/abs/quant-ph/0601210}{An Anomaly of
  Non-locality}.
\newblock {\em Quantum Information and Computation}, 7:157--170, 2007.

\bibitem{Nielsen1999}
Michael~A. Nielsen.
\newblock
  \href{https://journals.aps.org/prl/abstract/10.1103/PhysRevLett.83.436}{Conditions
  for a class of entanglement transformations}.
\newblock {\em Physical Review Letters}, 83(2):436, 1999.

\bibitem{Petz1986}
Dénes Petz.
\newblock
  \href{https://www.sciencedirect.com/science/article/pii/0034487786900674}{Quasi-entropies
  for finite quantum systems}.
\newblock {\em Reports on Mathematical Physics}, 23(1):57--65, 1986.

\bibitem{poulin2005comment}
David Poulin and Hugo Touchette.
\newblock \href{https://arxiv.org/abs/cs/0503034}{Comment on" Some
  non-conventional ideas about algorithmic complexity"}.
\newblock {\em arxiv:cs/0503034}, 2005.

\bibitem{Regula2017}
Bartosz Regula.
\newblock \href{https://arxiv.org/abs/1707.06298}{Convex geometry of quantum
  resource quantification}.
\newblock {\em Journal of Physics A: Mathematical and Theoretical},
  51(4):045303, Dec 2017.

\bibitem{regula2018converting}
Bartosz Regula, Marco Piani, Marco Cianciaruso, Thomas~R Bromley, Alexander
  Streltsov, and Gerardo Adesso.
\newblock
  \href{https://iopscience.iop.org/article/10.1088/1367-2630/aaae9d}{Converting
  multilevel nonclassicality into genuine multipartite entanglement}.
\newblock {\em New Journal of Physics}, 20(3):033012, 2018.

\bibitem{rosenthal1990quantales}
Kimmo~I. Rosenthal.
\newblock {\em Quantales and their applications}, volume 234.
\newblock Longman Scientific and Technical, 1990.

\bibitem{rosenthal2014theory}
Kimmo~I. Rosenthal.
\newblock {\em The theory of quantaloids}.
\newblock Chapman and Hall/CRC, 2014.

\bibitem{russo2010quantale}
Ciro Russo.
\newblock \href{https://arxiv.org/abs/1002.0968}{Quantale modules and their
  operators, with applications}.
\newblock {\em Journal of Logic and Computation}, 20(4):917--946, 2010.

\bibitem{Russo2017}
Ciro Russo.
\newblock \href{https://arxiv.org/abs/1706.00135}{Quantales and their modules:
  projective objects, ideals, and congruences}.
\newblock {\em South American Journal of Logic}, 2(2):405--424, 2016.

\bibitem{Sadrzadeh2006}
Mehrnoosh Sadrzadeh.
\newblock {\em \href{https://eprints.soton.ac.uk/262823/1/all.pdf}{Actions and
  resources in epistemic logic}}.
\newblock PhD thesis, University of Quebec at Montreal, 2006.

\bibitem{salzmann2021symmetric}
Robert Salzmann, Nilanjana Datta, Gilad Gour, Xin Wang, and Mark~M. Wilde.
\newblock \href{https://arxiv.org/abs/2102.12512}{Symmetric distinguishability
  as a quantum resource}.
\newblock {\em arXiv:2102.12512}, 2021.

\bibitem{scandi2021undecidability}
Matteo Scandi and Jacopo Surace.
\newblock \href{https://arxiv.org/abs/2105.09341}{Undecidability in resource
  theory: can you tell theories apart?}
\newblock {\em arXiv:2105.09341}, 2021.

\bibitem{schmid2020understanding}
David Schmid, Thomas~C. Fraser, Ravi Kunjwal, Ana~Belen Sainz, Elie Wolfe, and
  Robert~W. Spekkens.
\newblock \href{https://arxiv.org/abs/2004.09194}{Understanding the interplay
  of entanglement and nonlocality: motivating and developing a new branch of
  entanglement theory}.
\newblock {\em arXiv:2004.09194}, 2020.

\bibitem{schmid2020type}
David Schmid, Denis Rosset, and Francesco Buscemi.
\newblock \href{https://doi.org/10.22331/q-2020-04-30-262}{The type-independent
  resource theory of local operations and shared randomness}.
\newblock {\em Quantum}, 4:262, 2020.

\bibitem{sep-process-philosophy}
Johanna Seibt.
\newblock \href{https://plato.stanford.edu/entries/process-philosophy/}{Process
  Philosophy}.
\newblock In Edward~N. Zalta, editor, {\em The {Stanford} Encyclopedia of
  Philosophy}. Metaphysics Research Lab, Stanford University, {S}ummer 2020
  edition.

\bibitem{selinger2010survey}
Peter Selinger.
\newblock \href{https://arxiv.org/abs/0908.3347}{A survey of graphical
  languages for monoidal categories}.
\newblock In {\em New structures for physics}, pages 289--355. Springer, 2010.

\bibitem{Shannon1948}
Claude~E. Shannon.
\newblock
  \href{https://onlinelibrary.wiley.com/doi/abs/10.1002/j.1538-7305.1948.tb01338.x}{A
  mathematical theory of communication}.
\newblock {\em Bell system technical journal}, 27(3):379--423, 1948.

\bibitem{shannon1958note}
Claude~E. Shannon.
\newblock \href{https://core.ac.uk/download/pdf/82709959.pdf}{A note on a
  partial ordering for communication channels}.
\newblock {\em Information and control}, 1(4):390--397, 1958.

\bibitem{Skrzypczyk2019}
Paul Skrzypczyk and Noah Linden.
\newblock
  \href{https://journals.aps.org/prl/abstract/10.1103/PhysRevLett.122.140403}{Robustness
  of measurement, discrimination games, and accessible information}.
\newblock {\em Physical review letters}, 122(14):140403, 2019.

\bibitem{sparaciari2018multi}
Carlo Sparaciari.
\newblock {\em
  \href{https://discovery.ucl.ac.uk/id/eprint/10056289/1/main.pdf}{Multi-resource
  theories and applications to quantum thermodynamics}}.
\newblock PhD thesis, University College London, 2018.

\bibitem{Sparaciari2018}
Carlo Sparaciari, L{\'\i}dia Del~Rio, Carlo~Maria Scandolo, Philippe Faist, and
  Jonathan Oppenheim.
\newblock \href{https://doi.org/10.22331/q-2020-04-30-259}{The first law of
  general quantum resource theories}.
\newblock {\em Quantum}, 4:259, 2020.

\bibitem{Sparaciari2017}
Carlo Sparaciari, Jonathan Oppenheim, and Tobias Fritz.
\newblock \href{https://link.aps.org/doi/10.1103/PhysRevA.96.052112}{Resource
  theory for work and heat}.
\newblock {\em Phys. Rev. A}, 96:052112, Nov 2017.

\bibitem{sperling2015convex}
J.~Sperling and W.~Vogel.
\newblock
  \href{https://iopscience.iop.org/article/10.1088/0031-8949/90/7/074024}{Convex
  ordering and quantification of quantumness}.
\newblock {\em Physica Scripta}, 90(7):074024, 2015.

\bibitem{streltsov2017colloquium}
Alexander Streltsov, Gerardo Adesso, and Martin~B. Plenio.
\newblock
  \href{https://journals.aps.org/rmp/abstract/10.1103/RevModPhys.89.041003}{Colloquium:
  Quantum coherence as a resource}.
\newblock {\em Reviews of Modern Physics}, 89(4):041003, 2017.

\bibitem{stubbe2008q}
Isar Stubbe.
\newblock \href{https://arxiv.org/abs/0809.4343}{Q-modules are Q-suplattices}.
\newblock {\em Theory and Applications of Categories}, 19(4):50--60, 2007.

\bibitem{Szilard1929}
Le{\'o} Szilard.
\newblock \href{https://doi.org/10.1007/BF01341281}{{\"U}ber die
  Entropieverminderung in einem thermodynamischen System bei Eingriffen
  intelligenter Wesen}.
\newblock {\em Zeitschrift f{\"u}r Physik}, 53(11):840--856, Nov 1929.

\bibitem{takagi2019general}
Ryuji Takagi and Bartosz Regula.
\newblock \href{https://arxiv.org/abs/1901.08127}{General resource theories in
  quantum mechanics and beyond: operational characterization via discrimination
  tasks}.
\newblock {\em Physical Review X}, 9(3):031053, 2019.

\bibitem{Takagi2019}
Ryuji Takagi, Bartosz Regula, Kaifeng Bu, Zi-Wen Liu, and Gerardo Adesso.
\newblock \href{https://arxiv.org/abs/1809.01672}{Operational advantage of
  quantum resources in subchannel discrimination}.
\newblock {\em Physical review letters}, 122(14):140402, 2019.

\bibitem{Terhal2000}
Barbara~M. Terhal and Pawe\l{} Horodecki.
\newblock \href{https://link.aps.org/doi/10.1103/PhysRevA.61.040301}{Schmidt
  number for density matrices}.
\newblock {\em Phys. Rev. A}, 61:040301, Mar 2000.

\bibitem{torgersen1991comparison}
Erik Torgersen.
\newblock {\em Comparison of statistical experiments}, volume~36.
\newblock Cambridge University Press, 1991.

\bibitem{Tsinakis2004}
Constantine Tsinakis and Han Zhang.
\newblock
  \href{https://my.vanderbilt.edu/constantinetsinakis/files/2014/03/petri.pdf}{Order
  Algebras as Models of Linear Logic}.
\newblock {\em Studia Logica}, 76(2):201--225, Mar 2004.

\bibitem{turgut2007catalytic}
Sadi Turgut.
\newblock \href{https://doi.org/10.1088/1751-8113/40/40/012}{Catalytic
  transformations for bipartite pure states}.
\newblock {\em Journal of Physics A: Mathematical and Theoretical},
  40(40):12185, 2007.

\bibitem{Uhlmann1997}
Armin Uhlmann.
\newblock \href{https://arxiv.org/abs/quant-ph/9704017}{Entropy and optimal
  decompositions of states relative to a maximal commutative subalgebra}.
\newblock {\em Open Systems \& Information Dynamics}, 5(3):209--228, 1998.

\bibitem{Uhlmann2010}
Armin Uhlmann.
\newblock \href{http://www.mdpi.com/1099-4300/12/7/1799}{Roofs and Convexity}.
\newblock {\em Entropy}, 12(7):1799--1832, 2010.

\bibitem{uola2019quantifying}
Roope Uola, Tristan Kraft, Jiangwei Shang, Xiao-Dong Yu, and Otfried G{\"u}hne.
\newblock
  \href{https://journals.aps.org/prl/abstract/10.1103/PhysRevLett.122.130404}{Quantifying
  quantum resources with conic programming}.
\newblock {\em Physical review letters}, 122(13):130404, 2019.

\bibitem{Vaccaro2008}
Joan~A. Vaccaro, Fabio Anselmi, Howard~M. Wiseman, and Kurt Jacobs.
\newblock
  \href{https://journals.aps.org/pra/abstract/10.1103/PhysRevA.77.032114}{Tradeoff
  between extractable mechanical work, accessible entanglement, and ability to
  act as a reference system, under arbitrary superselection rules}.
\newblock {\em Physical Review A}, 77(3):032114, 2008.

\bibitem{vanrietvelde2021routed}
Augustin Vanrietvelde, Hl{\'{e}}r Kristj{\'{a}}nsson, and Jonathan Barrett.
\newblock \href{https://doi.org/10.22331/q-2021-07-13-503}{Routed quantum
  circuits}.
\newblock {\em Quantum}, 5:503, 2021.

\bibitem{vedral1997quantifying}
Vlatko Vedral, Martin~B. Plenio, Michael~A. Rippin, and Peter~L. Knight.
\newblock \href{https://doi.org/10.1103/PhysRevLett.78.2275}{Quantifying
  entanglement}.
\newblock {\em Physical Review Letters}, 78(12):2275, 1997.

\bibitem{Veinott1971}
Arthur~F. Veinott~Jr.
\newblock Least d-majorized network flows with inventory and statistical
  applications.
\newblock {\em Management Science}, 17(9):547--567, 1971.

\bibitem{Vidal2000}
Guifre {Vidal}.
\newblock \href{https://arxiv.org/abs/quant-ph/9807077v2}{Entanglement
  monotones}.
\newblock {\em Journal of Modern Optics}, 47:355--376, February 2000.

\bibitem{vidal1999robustness}
Guifr{\'e} Vidal and Rolf Tarrach.
\newblock
  \href{https://journals.aps.org/pra/abstract/10.1103/PhysRevA.59.141}{Robustness
  of entanglement}.
\newblock {\em Physical Review A}, 59(1):141, 1999.

\bibitem{vidick2011more}
Thomas Vidick and Stephanie Wehner.
\newblock \href{https://doi.org/10.1103/PhysRevA.83.052310}{More nonlocality
  with less entanglement}.
\newblock {\em Physical Review A}, 83(5):052310, 2011.

\bibitem{wang2019resource}
Xin Wang and Mark~M. Wilde.
\newblock
  \href{https://journals.aps.org/prresearch/abstract/10.1103/PhysRevResearch.1.033169}{Resource
  theory of asymmetric distinguishability for quantum channels}.
\newblock {\em Physical Review Research}, 1(3):033169, 2019.

\bibitem{Wang2019}
Xin Wang and Mark~M. Wilde.
\newblock
  \href{https://journals.aps.org/prresearch/abstract/10.1103/PhysRevResearch.1.033170}{Resource
  theory of asymmetric distinguishability}.
\newblock {\em Physical Review Research}, 1(3):033170, 2019.

\bibitem{watrous2018theory}
John Watrous.
\newblock {\em \href{https://cs.uwaterloo.ca/~watrous/TQI/}{The theory of
  quantum information}}.
\newblock Cambridge University Press, 2018.

\bibitem{Weilenmann2016}
Mirjam Weilenmann, Lea Kraemer, Philippe Faist, and Renato Renner.
\newblock
  \href{https://journals.aps.org/prl/pdf/10.1103/PhysRevLett.117.260601}{Axiomatic
  relation between thermodynamic and information-theoretic entropies}.
\newblock {\em Physical review letters}, 117(26):260601, 2016.

\bibitem{wilson2021causality}
Matt Wilson and Giulio Chiribella.
\newblock \href{https://dx.doi.org/10.4204/EPTCS.343.12}{Causality in higher
  order process theories}.
\newblock In {\em Electronic Proceedings in Theoretical Computer Science},
  volume 343, pages 265--300. EPTCS, 2021.

\bibitem{wilson2021mathematical}
Matt Wilson and Giulio Chiribella.
\newblock
  \href{https://qift.weebly.com/uploads/7/4/5/1/74515901/weebly___a_categorical_framework_for_higher_order_physics__1_.pdf}{A
  Mathematical Framework for Higher Order Physical Theories}.
\newblock 2021.

\bibitem{Wofsey2016}
Eric Wofsey.
\newblock \href{https://math.stackexchange.com/q/1745430}{Is a finite lattice
  uniquely atomistic iff it is boolean?}
\newblock Mathematics Stack Exchange, 2016.

\bibitem{Wolfe2019}
Elie Wolfe, David Schmid, Ana~Bel{\'e}n Sainz, Ravi Kunjwal, and Robert~W
  Spekkens.
\newblock \href{https://doi.org/10.22331/q-2020-06-08-280}{Quantifying Bell:
  The resource theory of nonclassicality of common-cause boxes}.
\newblock {\em Quantum}, 4:280, 2020.

\bibitem{yetter1990quantales}
David~N. Yetter.
\newblock \href{https://www.jstor.org/stable/2274953?seq=1}{Quantales and
  (noncommutative) linear logic}.
\newblock {\em The Journal of Symbolic Logic}, 55(1):41--64, 1990.

\bibitem{ying2021asymmetry}
Y\`{i}l\`{e} Y\={\i}ng.
\newblock Asymmetry trade-off relations: a new kind of uncertainty relations.
\newblock Master's thesis, University of Waterloo, 2021.

\bibitem{zhou2020general}
Wenbin Zhou and Francesco Buscemi.
\newblock \href{https://arxiv.org/abs/2005.09188}{General state transitions
  with exact resource morphisms: a unified resource-theoretic approach}.
\newblock {\em Journal of Physics A: Mathematical and Theoretical},
  53(44):445303, 2020.

\bibitem{zhuang2018resource}
Quntao Zhuang, Peter~W. Shor, and Jeffrey~H. Shapiro.
\newblock
  \href{https://journals.aps.org/pra/abstract/10.1103/PhysRevA.97.052317}{Resource
  theory of non-Gaussian operations}.
\newblock {\em Physical Review A}, 97(5):052317, 2018.

\end{thebibliography}



\end{document}